\documentclass[a4paper,UKenglish,numberwithinsect,cleveref, autoref, thm-restate]{lipics-v2021}
\usepackage{amsfonts,amsthm,amsmath}
\usepackage{amstext}
\usepackage[justification=centering]{caption}
\usepackage{subcaption}
\usepackage[ruled, linesnumbered]{algorithm2e}
\usepackage{graphics}

\usepackage{algorithmic}
\newcommand{\OO}{\mathcal{O}}

\newcommand{\tw}{\mathsf{tw}}

\newcommand{\fc}{d_\mathsf{y}}
\newcommand{\fr}{d_\mathsf{x}}

\newcommand{\av}{\mathsf{VerAsso}}
\newcommand{\gi}{\mathsf{EnrichedGridPoints}}
\newcommand{\gis}{\mathsf{EnrichedGridPointsIn}}
\newcommand{\gp}{\mathsf{GridPoints}}
\newcommand{\gps}{\mathsf{StrictGridPoints}}

\newcommand{\fin}{f_{\mathsf{init}}}
\newcommand{\ir}{\mathsf{IntRegion}}
\newcommand{\er}{\mathsf{ExtRegion}}
\newcommand{\tp}{\mathsf{TurnPoints}}
\newcommand{\pp}{\mathsf{PlanePoints}}
\newcommand{\frc}{\mathsf{frame}(f,c)}

\newcommand{\pw}{\mathsf{pw}}
\newcommand{\cw}{\mathsf{cw}}
\newcommand{\dcw}{\mathsf{dcw}}
\newcommand{\ew}{\mathsf{emw}}
\newcommand{\calC}{\mathcal{C}}
\newcommand{\sizef}{\mathsf{cost}}
\hideLIPIcs  %uncomment to remove references to LIPIcs series (logo, DOI, ...), e.g. when preparing a pre-final version to be uploaded to arXiv or another public repository

\title{Drawn Tree Decomposition: New Approach for Graph Drawing Problems} %TODO Please add

\author{Siddharth Gupta}{University of Warwick, UK}{siddharth.gupta.1@warwick.ac.uk}{https://orcid.org/0000-0003-4671-9822}{Supported by Engineering and Physical Sciences Research Council (EPSRC) grant EP/V007793/1.}

\author{Guy Sa'ar}{Ben Gurion University of the Negev, Israel}{saag@post.bgu.ac.il}{}{Supported in part by the Israeli Smart Transportation Research Center and by the Lynne and William Frankel Center for Computing Science at Ben-Gurion University.}

\author{Meirav Zehavi}{Ben Gurion University of the Negev, Israel}{meiravze@bgu.ac.il}{https://orcid.org/0000-0002-3636-5322}{Supported by the European Research Council (ERC) grant titled PARAPATH.}

\authorrunning{S. Gupta, G. Sa'ar, and M. Zehavi} %TODO mandatory. First: Use abbreviated first/middle names. Second (only in severe cases): Use first author plus 'et al.'

\Copyright{Siddharth Gupta, Guy Sa'ar, and Meirav Zehavi} %TODO mandatory, please use full first names. LIPIcs license is "CC-BY";  http://creativecommons.org/licenses/by/3.0/

 %TODO mandatory; please add comma-separated list of keywords

\relatedversion{} %optional, e.g. full version hosted on arXiv, HAL, or other respository/website
\nolinenumbers %uncomment to disable line numbering

\EventEditors{Neeldhara Misra and Magnus Wahlstr\"{o}m}
\EventNoEds{2}
\EventLongTitle{18th International Symposium on Parameterized and Exact Computation (IPEC 2023)}
\EventShortTitle{IPEC 2023}
\EventAcronym{IPEC}
\EventYear{2023}
\EventDate{September 6--8, 2023}
\EventLocation{Amsterdam, The Netherlands}
\EventLogo{}
\SeriesVolume{285}
\ArticleNo{7}
\begin{document}
	
\maketitle
	
\begin{abstract}
%!TEX root =Main-Movement.tex

Over the past decade, we witness an increasing amount of interest in the design of exact exponential-time and parameterized algorithms for problems in Graph Drawing. Unfortunately, we still lack knowledge of general methods to develop such algorithms. An even more serious issue is that, here, ``standard'' parameters very often yield intractability. In particular, for the most common structural parameter, namely, treewidth, we frequently observe NP-hardness already when the input graphs are restricted to have constant (often, being just $1$ or $2$) treewidth. 

Our work deals with both drawbacks simultaneously. We introduce a novel form of tree decomposition that, roughly speaking, does not decompose (only) a graph, but an entire drawing. As such, its bags and separators are of geometric (rather than only combinatorial) nature. While the corresponding parameter---like treewidth---can be arbitrarily smaller than the height (and width) of the drawing, we show that---unlike treewidth---it gives rise to efficient algorithms. Specifically, we get slice-wise polynomial (XP) time algorithms parameterized by our parameter. We present a general scheme for the design of such algorithms, and apply it to several central problems in Graph Drawing, including the recognition of grid graphs, minimization of crossings and bends, and compaction. 
Other than for the class of problems we discussed in the paper, we believe that our decomposition and scheme are of independent interest and can be further extended or generalized to suit even a wider class of problems.
Additionally, we discuss classes of drawings where our parameter is bounded by $\OO(\sqrt{n})$ (where $n$ is the number of vertices of the graph), yielding subexponential-time algorithms. Lastly, we prove which relations exist between drawn treewidth and other width measures, including treewidth, pathwidth, (dual) carving-width and embedded-width.
\end{abstract}

\newpage

%!TEX root =Main-Movement.tex

\section{Introduction}\label{sec:intro}

Over the past decade, we witness an increasing amount of interest in the design of exact exponential-time and parameterized algorithms for problems in Graph Drawing. For a few illustrative examples, let us mention that this includes studies of crossing minimization~\cite{grohe2004computing,DBLP:conf/gd/HlinenyS19,kawarabayashi2007computing}, recognition of planar graph families such as upward planarity testing~\cite{chan2004parameterized,healy2006two} and grid graph recognition~\cite{DBLP:conf/isaac/0002SZ21}, as well as recognition of beyond planar graph families~\cite{DBLP:journals/jgaa/BannisterCE18}, turn-minimization~\cite{DBLP:conf/wg/FellowsGKPRWY10}, linear layouts such as books embeddings~\cite{DBLP:journals/jgaa/BannisterE18,bhore2019parameterized}, clustered planarity and hybrid
planarity~\cite{DBLP:conf/wg/LozzoEG018,LiRuTa-PCGPRCO-22,DBLP:journals/algorithmica/LozzoEGG21}, and bend minimization \cite{didimo1998computing,di2019sketched}. For more information on recent progress on these and other topics, we refer to the report \cite{DBLP:journals/dagstuhl-reports/GanianMNZ21} and surveys such as~\cite{suvey/zehavi}. Unfortunately, still, we have very limited knowledge of general methods to develop exact exponential-time and parameterized algorithms for problems in Graph Drawing.

An even more serious issue is that, for Graph Drawing problems, ``standard'' parameters very often yield intractability. In particular, for the most common structural parameter, namely, treewidth,\footnote{Definitions of standard terms and notations used in the Introduction can be found in Section \ref{sec:prelims}.} we frequently observe NP-hardness already when the input graphs are restricted to have constant (often, being just $1$ or $2$) treewidth. The same result holds even for the larger parameter pathwidth. For example, {\sc Grid Recognition} is NP-hard on graphs of treewidth $1$ (being trees)~\cite{DBLP:journals/ipl/BhattC87} or pathwidth $2$~\cite{DBLP:conf/isaac/0002SZ21}, {\sc Orthogonal Compaction} is NP-hard even on cycles~\cite{DBLP:journals/comgeo/EvansFKSSW22} and hence on graphs of pathwidth (and treewidth) $2$, {\sc Min-Area Planar Straight-line Drawing} is NP-complete on outerplanar graphs and hence on graphs of treewidth $2$~\cite{DBLP:conf/icalp/Biedl14, DBLP:conf/gd/KrugW07}, and {\sc Grid Upward Drawing} is NP-complete on graphs of treewidth $1$ (being trees)~\cite{DBLP:conf/gd/AkitayaLP18,DBLP:conf/gd/BiedlM17a}. In light of this, we must seek parameterizations that are larger (or incomparable) to treewidth. Due to the nature of the problems at hand, it is natural to seek parameters of geometric flavors. Here, perhaps, the first choice that comes to mind is the {\em height} (or, rather, the minimum among the height and width) of the sought (or given) drawing. In particular, we can easily observe that this parameter for planar orthogonal grid drawings is bounded by $\Omega(\mathsf{tw})$, where $\mathsf{tw}$ is the treewidth of the drawn graph, and that it gives rise to the use of dynamic programming. However, denoting the number of vertices by $n$, we can also easily observe that this parameter can be as large as $\Omega(n)$ for ridiculously simple planar orthogonal grid drawings (and graphs)! For example, consider the path drawn in Figure~\ref{fig:simplePathA}---here, already, both height and width are equal to (roughly) $n/2$.

Our work deals with both drawbacks mentioned above simultaneously. We introduce a novel form of tree decomposition that, roughly speaking, does not decompose (only) a graph, but an entire drawing. As such, its bags and separators are of geometric (rather than only combinatorial) nature. We further discuss this concept (still informally but in more detail) in Section \ref{sec:introConcept} ahead. While the corresponding parameter---like treewidth---can be arbitrarily smaller than the height (and width) of the drawing (e.g., for the aforementioned example in Figure~\ref{fig:simplePathA}, our parameter is a fixed constant), we show that---unlike treewidth---it gives rise to efficient (that is, XP) algorithms. Specifically, we present a general scheme for the design of such algorithms (described in Section \ref{sec:introScheme}), and apply it to several central problems in Graph Drawing, including the recognition of grid graphs, minimization of crossings and bends, and compaction (see Section \ref{sec:introApplications}). 
We believe that our new concept of geometric tree decomposition is interesting on its own, and exploring the connections between it and notions concerning (classical) tree decompositions is a promising research direction (see Section~\ref{sec:conclusion}). Furthermore, we believe that this concept and our scheme can be further extended or generalized to be applicable to problems other than those discussed in this paper.
%and the outline of the rest of the paper is given in Section \ref{sec:roadmap}.

\subsection{The Concept of Drawn Tree Decomposition}\label{sec:introConcept}

Here, we discuss (informally) our main conceptual contribution: the introduction and study of the concepts of {\em drawn tree decomposition} and {\em drawn treewidth}, which we believe to be of independent interest. A formal definition of these concepts can be found in Section~\ref{sec:drawnSep}. Then, in Section~\ref{sec:introComparison}, we compare our parameter with several seemingly related graph parameters. Later, in Sections~\ref{sec:introScheme} and~\ref{sec:introApplications}, we discuss our main technical contribution (which has been our initial motivation for these concepts): our general algorithmic scheme and its applications to problems in Graph Drawing. Our focus is on a class of rather general drawings of graphs on the Euclidean plane (allowing drawings of edges to have both crossings and bends, as well as to consist of segments that are not necessarily parallel to the axes), called {\em polyline grid drawings}. Roughly speaking, a polyline grid drawing $d$ of a graph $G$ is a mapping of the vertices of $G$ to distinct grid points (being points of the form $(i,j)$ where $i,j\in\mathbb{Z}$) and edges to {\em straight-line paths} between their endpoints.
%, where the drawings of different edges may intersect only in a finite number of points. 
That is, the drawing of an edge is a simple curve that is the concatenation of straight-line segments (e.g, see Figure~\ref{fi:DrawEx5} in Section~\ref{sec:prelims2}). Towards the (informal) definition of a drawn tree decomposition ahead, we first introduce three critical terms: {\em frame}, {\em cutter}, and {\em rectangular}.

\begin{figure}[t]
	\centering
	\begin{subfigure}{0.3\textwidth}
		\includegraphics[page=69, width=\textwidth]{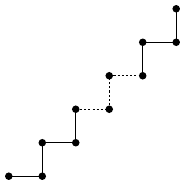}
		\subcaption{}
		\label{fig:simplePathA}
	\end{subfigure}
	\hfil
		\begin{subfigure}{0.3\textwidth}
		\includegraphics[page=70, width=\textwidth]{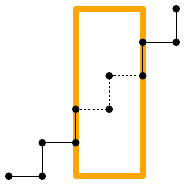}
		\subcaption{}
		\label{fig:simplePathB}
	\end{subfigure}
	\caption{(a) A drawing of a path on $n$ vertices with height and width $(n-1)/2$. However, the drawn treewidth is $16$. (b) An illustration of a frame shown in orange with width $16$.}
	\label{fig:simplePath}
\end{figure}

\medskip\noindent{\bf Frame.} A {\em frame} is, simply, a straight-line cycle (defined analogously to a straight-line path above) whose segments are axis-parallel (see the orange polygon in Figure~\ref{fig:TP}). In other words, it is a simple rectilinear polygon whose vertices lie on grid points.

For the definition of the width of our decomposition (presented later), we define the {\em width of a frame}. Roughly speaking, the width of a frame $f$, denoted by $\mathsf{width}(f)$, is the sum of measures of the complexities of {\em (i)} the frame itself, and {\em (ii)} the ``way'' in which the drawing ``traverses'' the frame.  For {\em (i)}, we simply count the number of vertices of the frame (ignoring ``superfluous'' vertices, being those where the angle between incident edges is of 180 degrees). For {\em (ii)}, we regard the drawings of vertices and edges separately (and sum up the two corresponding numbers). Specifically, for vertices, we simple count the number of vertices drawn on the frame. However, for edges, the measure is somewhat more complex, based on the notion of {\em turning points} (defined immediately); for each edge, we count the number of its turning points on the frame, and, then, the measure is the sum (over all edges) of these counters. We remark that some points on the plane might be counted multiple times---at the extreme case, the same point might be {\em (a)} a vertex of the frame, {\em (b)} a point on which a vertex of the graph is drawn, and {\em (c)} a turning point for one (or more) edges. We find this multi-count to be justified: the more complicated the frame and the drawing are at a certain point, the more that point ``contributes'' to the complexity of~the~measure.

Now, let us define the notion of a turning point. For this purpose, consider some edge $e=\{u,v\}$ of the graph and some point $p$ on the frame. Then, roughly speaking, we refer to $(p,e)$ as a turning point if, when we traverse the drawing of $e$ from $u$ to $v$ or from $v$ to $u$, we encounter $p$, and ``immediately'' before this encounter, we were in the strict interior or exterior of the frame. Additionally, we refer to $(p,e)$ as a turning point if $u$ or $v$ themselves are drawn on $p$.  An illustrative example is given in Figure~\ref{fig:TP}.

\begin{figure}[t]
	\centering
	\includegraphics[page=8, width=0.35\textwidth]{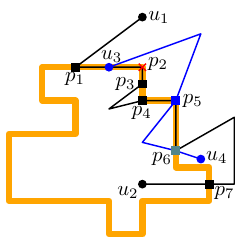}
	\caption{The turning points of the black edge $e=\{u_1,u_2\}$ in the orange frame $f$ are $(p_1,e),(p_3,e),(p_4,e)$, $(p_6,e)$ and $(p_7,e)$. Note that, $(p_2,e)$ and $(p_5,e)$ are not turning points in $f$. Similarly, the turning points of the blue edge $e'=\{u_3,u_4\}$ in the frame $f$ are $(d(u_3),e'),(p_5,e')$ and $(p_6,e')$. Note that $(p_5,e')$ is a turning point in $f$, but $(p_5,e)$ is not a turning point in $f$.}
	\label{fig:TP}
\end{figure} 

For an example of the definition of the width of a frame, we refer to Figure~\ref{fig:simplePathB}. Here, the frame itself (being a rectangle) consists of exactly $4$ vertices. Second, the path contains exactly $4$ vertices that are drawn on the frame. Third, every  point on which one of these vertices, say, $v$, is drawn is a turning point of $2$ edges, being the two edges incident to $v$. So, the width of the frame is $4+4+4\cdot 2=16$.

\medskip\noindent{\bf Cutter.} A {\em cutter of a frame} is, simply, a straight-line path whose segments are axis-parallel and which intersects the frame in exactly two points, which are the endpoints of the cutter. Later, we discuss the ``futility'' of two simpler definitions for a cutter. The utility of a cutter of a frame $f$ is, as its name suggests, in ``cutting'' $f$ into (exactly) two frames $f_1$ and $f_2$. Roughly speaking, we obtain one of $f_1$ and $f_2$ by the concatenation of the cutter with one path among the two subpaths of $f$ between the endpoints of the cutter, and we obtain the other of $f_1$ and $f_2$ by the concatenation of the cutter with the other path among the two subpaths of $f$ between the endpoints of the cutter. For more intuition, we refer the reader to Figure~\ref{fig:cutterIntro}.

\begin{figure}[t]
	\centering
	\includegraphics[page=71, width=0.35\textwidth]{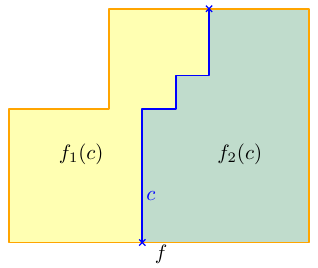}
	\caption{An illustration for a cutter $c$, shown in blue, of a frame $f$, shown in orange, and its associated frames $f_1(c)$ and $f_2(c)$.}
	\label{fig:cutterIntro}
\end{figure} 

\medskip\noindent{\bf Rectangular.} For the sake of intuition, the construction of a drawn tree decomposition may be thought of as a recursive process where, for a given frame, we compute a cutter that cuts it into two, and then proceed (recursively) with each of these two resulting frames. Then, two questions arise: What is the initial frame, and when does this process terminate? For the first question, the answer is simply the {\em rectangular} of the drawing (defined immediately). For the second question, the answer is even simpler---we stop when the current frame does not contain any grid point in its strict interior.  Roughly speaking, the rectangular of a drawing is the (unique) frame whose interior is minimized among all frames whose ``shape'' is a rectangle and which contain the given drawing in their strict interior (see Figure~\ref{fig:Rd}).

\begin{figure}[t]
	\centering
	\includegraphics[page=7, width=0.31\textwidth]{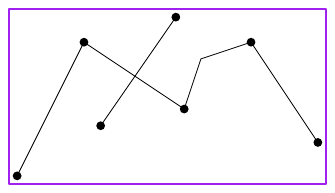}
	\caption{The frame shown in purple is $R_d$ where $d$ is the drawing inside the frame.}
	\label{fig:Rd}
\end{figure} 

\medskip\noindent{\bf Drawn Tree Decomposition.} At the heart of the concept of a drawn tree decomposition, lies our definition of a {\em frame-tree} (abbreviation for {\em tree of frames}). Informally, for a graph $G$ and a polyline grid drawing $d$ of $G$, a frame-tree is a pair $({\cal T},\alpha)$ where $\cal T$ is a binary rooted tree and $\alpha$ maps each vertex of $\cal T$ to a frame, such that: {\em (i)} the root is mapped to the rectangular of $d$; {\em (ii)} for every internal vertex $v$ of $\cal T$, there exists a (unique) cutter $c_v$ of $\alpha(v)$ so that the frames mapped to the children of $v$ are those obtained by cutting $\alpha(v)$ by $c_v$; {\em (iii)} the leaves of $\cal T$ (and none of the internal vertices of $\cal T$) are mapped to frames whose strict interior does not contain any grid point. For an illustrative example, see Figure~\ref{fi:proofDTD}. 

\begin{figure}[!t]
	\centering
	\begin{subfigure}{0.3\textwidth}
		\includegraphics[width = \textwidth, page = 18]{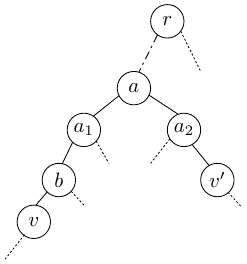}
		\subcaption{A frame-tree $\cal T$.}
		\label{fi:DTD}
	\end{subfigure}
	\hfil
	\begin{subfigure}{0.32\textwidth}
		\includegraphics[width = \textwidth, page = 10]{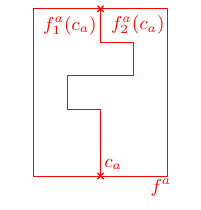}
		\subcaption{The frame $f^a$, the cutter $c_a$ and the sub-frames $f^a_1(c_a)$ and $f^a_2(c_a)$ associated with the vertex $a$ of $\cal T$.}
		\label{d1}
	\end{subfigure}
	\begin{subfigure}{0.36\textwidth}
		\includegraphics[width = \textwidth, page = 11]{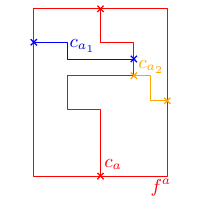}
		\subcaption{The cutters $c_{a_1}$ and $c_{a_2}$ associated with the children $a_1$ and $a_2$ of $a$ in $\cal T$, respectively.}
		\label{d2}
	\end{subfigure}
	\begin{subfigure}{0.32\textwidth}
		\includegraphics[width = \textwidth, page = 12]{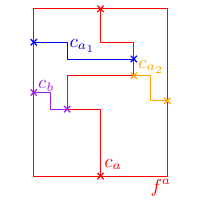}
		\subcaption{The cutter $c_{b}$ associated with the child $b$ of $a_1$ in $\cal T$.}
		\label{d3}
	\end{subfigure}
	\hfil
	\begin{subfigure}{0.32\textwidth}
		\includegraphics[width = \textwidth, page = 13]{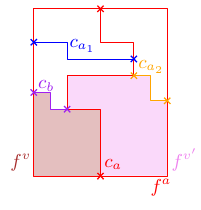}
		\subcaption{The frames $f^v$ (bounding the brown region) and $f^{v'}$(bounding the pink region) associated with the vertices $v$ and $v'$ of $\cal T$, respectively.}
		\label{d4}
	\end{subfigure}
	\begin{subfigure}{0.32\textwidth}
		\includegraphics[width = \textwidth, page = 14]{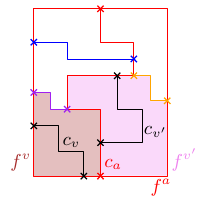}
		\subcaption{The cutters $c_v$ and $c_{v'}$ associated with the vertices $v$ and $v'$ in $\cal T$, respectively.}
		\label{d5}
	\end{subfigure}
	
	\caption{Example of frames and cutters of a frame-tree. For clarity, the polyline grid drawing is not shown.} 
	\label{fi:proofDTD}
\end{figure}

Now, for the definition of a drawn tree decomposition, we consider a frame-tree $({\cal T},\alpha)$. Then, we ``enrich'' the frame-tree by the introduction of an additional mapping, $\beta$, from the vertex set of $\cal T$ to subsets of vertices of $G$. In particular, we define $\beta$ so that we can: {\em (P1)} prove that $({\cal T},\beta)$ is a tree decomposition (this proof is slightly technical, based on case analysis); {\em (P2)} prove that, for every vertex $v$ of $\cal T$, $|\beta(v)|$ is at most twice the sum of the widths of the frames of $v$ and its two children (if they exist). For the definition of $\beta$, we (next) define the {\em  set of vertices associated with a frame}, and the {\em set of vertices associated with a cutter} of a frame. Then, for a vertex $v$ of $\cal T$, $\beta(v)$ is simply the union of the set of vertices associated with $\alpha(v)$, and the set of vertices associated with the cutter $c_v$ of $\alpha(v)$. Correspondingly, the triple $({\cal T},\alpha,\beta)$ is a drawn tree decomposition.

So, consider a graph $G$, a polyline grid drawing $d$ of $G$, a frame $f$ and a cutter $c$ of $f$. Then, the {\em set of  vertices associated with $f$} is the union of the set of vertices of $G$ that $d$ draws on $f$  and the set of endpoints of edges of $G$ whose drawing (by $d$) is {\em separated} by $f$---that is, edges having one endpoint in the strict interior of $f$ and the other endpoint in the strict exterior of $f$ (see Figure~\ref{fig:assFr}). Similarly, the {\em set of  vertices associated with $c$} is the union of the set of vertices of $G$ that $d$ draws on $c$  and the set of endpoints of edges of $G$ whose drawing (by $d$) is {\em separated} by $c$---that is, edges having one endpoint in the strict interior of one of the frames obtained by cutting $f$ by $c$, and the other endpoint  in the strict interior of the other frame obtained by cutting $f$ by $c$ (see Figure~\ref{fig:assCu}).

\begin{figure}[t]
	\centering
	\begin{subfigure}{0.35\textwidth}
		\includegraphics[page=2, width=\textwidth]{figures/drawnTreewidth}
		\subcaption{}
		\label{fig:assFr}
	\end{subfigure}
	\hfil
	\begin{subfigure}{0.35\textwidth}
		\includegraphics[page=1, width=\textwidth]{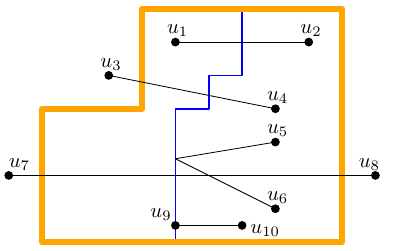}
		\subcaption{}
		\label{fig:assCu}
	\end{subfigure}
	\caption{Example of vertices associated with a frame and a cutter. (a) The edge $\{u_5,u_6\}$ is the only edge separated by the orange frame. The vertices associated with the orange frame are $u_3,u_5$ and $u_6$. (b) The vertices associated with the blue cutter of the orange frame are $u_1,u_2$ and $u_9$.}
	\label{fig:assVerIntro}
\end{figure}

\medskip\noindent{\bf Drawn Treewidth.} The {\em  width} of a drawn tree decomposition $({\cal T}=(V_T,E_T),\alpha,\beta)$ is the maximum width of its frames, that is, $\max_{v\in V_T}\mathsf{width}(\alpha(v))$. Accordingly, the drawn treewidth of a polyline grid drawing $d$ of a graph $G$ is the minimum width of a drawn tree decomposition of $d$. Notably, due to {\em (P1)} and {\em (P2)} mentioned above (proved in Section~\ref{sec:relToTreewidth}), we can easily conclude that the treewidth of $G$ is at most $6$ times its drawn treewidth.

We remark that the usage of frames bears similarity to that of {\em cycle separators of planar graphs} (being a central player in proofs of the planar separator theorem; see, e.g.,~\cite{DBLP:journals/siamdm/AlonST94,DBLP:journals/jcss/Miller86}). However, the corresponding widths (drawn treewidth versus treewidth) can be critically different: While treewidth is bounded from above by the order of drawn treewidth, we have already pointed out that for various problems where treewidth yields intractability, drawn treewidth does not---this, of course, implies that treewidth can, often, be arbitrarily smaller than drawn treewidth; for a concrete example, see Figure~\ref{fig:pathBadDTWIntro}. Further, treewidth depends only on the graph, while drawn treewidth depends (as desired) on the drawing; for example, notice that Figures~\ref{fig:simplePathA} and~\ref{fig:pathBadDTWIntro1} depict the same graph, but the corresponding drawings have radically different drawn treewidths.

Besides its above-mentioned relation to treewidth, drawn treewidth for planar orthogonal grid drawings can also be related to height (and width). On the one hand, we prove (in Section~\ref{sec:uppBound}) the desirable property that---like treewidth---drawn treewidth is bounded from above by the order of the minimum among the height and width of the drawing. Notably, various central graph width measures do {\em not} have this property. For example, one of the most commonly used relaxations of pathwidth is  {\em treedepth} (see, e.g.,~\cite{DBLP:books/sp/CyganFKLMPPS15} for information on  treedepth); however, the treedepth of an $n$-vertex path is  $\lceil\log_2(n+1)\rceil$, while it can be easily drawn so that the height (or, symmetrically, width) of the drawing is $1$. On the other hand, we have already observed that the drawn treewidth can be arbitrarily smaller than the minimum among the height and width of a drawing (see Figure~\ref{fig:simplePathA}).

\begin{figure}[t]
	\centering
	\begin{subfigure}{0.3\textwidth}
		\includegraphics[page=72, width=\textwidth]{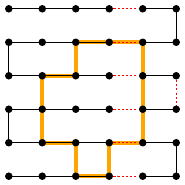}
		\subcaption{}
		\label{fig:pathBadDTWIntro1}
	\end{subfigure}
	\hfil
	\begin{subfigure}{0.3\textwidth}
		\includegraphics[page=73, width=\textwidth]{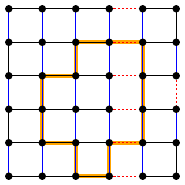}
		\subcaption{}
		\label{fig:pathBadDTWIntro2}
	\end{subfigure}
	\caption{(a) A path $P$ on $n$ vertices and a frame $f$ shown in orange. (b) A grid  graph $G$ on the same set of vertices and the frame $f$ shown in orange. Consider a frame, say $f$, in $P$ with width $w$. Observe that $f$ is also a frame in $G$. Moreover, the width of $f$ in $G$ is at most $3w$ as every vertex has exactly $2$ more edges in $G$ compared to $P$ so the vertex may be counted $2$ more times in the width of $f$ in $G$ as the turning points of those $2$ extra edges. As treewidth is a lower bound for drawn treewidth and the treewidth of a grid graph is $\sqrt{n}$, the drawn treewidth of $P$ is $\Omega(\sqrt{n})$ (while its treewidth is $1$).}
	\label{fig:pathBadDTWIntro}
\end{figure}

\medskip\noindent{\bf Bounds for Specific Types of Drawings.}  For some classes of drawings (being subclasses of polyline grid drawings), we are able to prove that drawn treewidth is bounded by a sublinear function of $n$ (the number of vertices of the graph). For example, for grid drawings---which are mappings of vertices to distinct grid points and of edges to unit-length straight lines between their endpoints (see Figure~\ref{fi:DrawEx2} in Section~\ref{sec:prelims2})---we prove that the drawn treewidth (and even the {\em straight-line drawn treewidth}, defined ahead) is bounded by $\OO(\sqrt{n})$. More generally, we prove (in Section~\ref{sec:uppBound}) that given a graph $G$ and an orthogonal grid drawing $d$ of $G$, drawn treewidth of $d$ is $\OO (\Delta \cdot \sqrt{\Delta\cdot \ell\cdot n}\cdot \mathsf{maxInt})$, where (i) $\Delta$ is the maximum degree in $G$, (ii) $\ell$ is the average length of the edges of $G$ in $d$, and (iii) $\mathsf{maxInt}$ is the maximum number of edges and vertices intersected in a grid point in $d$.

At this point, a short discussion is in order. 
One of the most well-known results in graph theory about planar graphs is that every $n$-vertex planar graph has pathwidth (and hence treewidth) bounded by $\OO(\sqrt{n})$~\cite{DBLP:journals/tcs/Bodlaender98}. In particular, this result and generalizations thereof have found impactful applications in algorithm design, particularly of parameterized and approximation algorithms. In fact, (almost) all subexponential-time algorithms for problems on planar graphs rely on it. Here, a central component in several proofs is the planar separator theorem~\cite{DBLP:journals/siamdm/AlonST94,lipton1979separator,DBLP:journals/jcss/Miller86} (briefly mentioned earlier), which states that every $n$-vertex planar graph contains an $\OO(\sqrt{n})$-sized subset of vertices (called separator) whose removal from the graph yields connected components that are each of size at most $2n/3$.
 Thus, due to the above-mentioned sub-quadratic bound on drawn treewidth for grid drawings, the following {\em conjecture} seems tempting: the drawn treewidth of any {\em planar} polyline grid drawing is $\OO(\sqrt{n})$. However, we observe (in Section~\ref{sec:uppBound}) that the statement analogous to the planar separator theorem does not hold in our case, where our notion of a separator is that of a cutter and their sizes is, in particular,  bounded from below by the size of the set of vertices associated with the cutter.

\medskip\noindent{\bf Drawbacks of Simpler Definitions for a Cutter.} Lastly, we present and discuss two alternative restricted forms of cutters: {\em horizontal (or vertical) cutters} and {\em straight-line cutters}. A horizontal cutter (resp., vertical cutter) of a frame is a cutter of that frame where all vertices have the same $y$-coordinate (resp., $x$-coordinate). Then, a straight-line cutter is a cutter that is either horizontal or vertical. The replacement of cutters by horizontal/vertical cutters or straight-line cutters  yields corresponding definitions of horizontal/vertical drawn tree decompositions and straight-line drawn tree decompositions, and, accordingly, of horizontal/vertical drawn treewidth and straight-line drawn treewidth. In particular, when we use these restricted forms of cutters, every frame has the shape of a rectangle. In turn, this significantly simplifies the visualization (and, possibly, also the use) of these concepts.

\begin{figure}[t]
	\centering
	\includegraphics[page=66, width=0.4\textwidth]{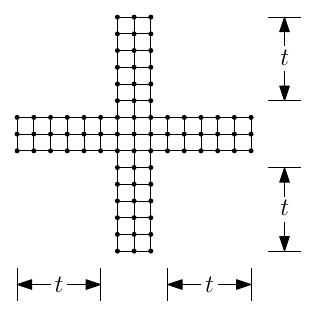}
	\caption{Example of a rectilinear drawing of a graph on $n$ vertices with $t = \Omega(n)$. The horizontal/vertical drawn treewidth of this drawing is $\Omega(n)$.}
	\label{fig:simpleCutterIntro1}
\end{figure} 

Unfortunately, horizontal/vertical drawn treewidth and even straight-line drawn treewidth can be arbitrarily larger than drawn treewidth. To see this, let us first consider horizontal cutters (or, symmetrically, vertical cutters), and the graph depicted in Figure~\ref{fig:simpleCutterIntro1}. Notably, this graph, in fact, admits {\em exactly one} grid drawing (up to isomorphism)---the one depicted in the figure. Now, notice that the horizontal drawn treewidth of this drawing is $\Omega(n)$. To see this, notice that, for any horizontal tree decomposition and for each of the three horizontal straight lines in the ``middle'' of the drawing, the rooted tree will have to contain a vertex whose associate cutter ``coincides'' with that line. However, the drawn treewidth of this drawing is only $\OO(1)$, and, more generally, recall that we prove that for any grid drawing, the straight-line drawn treewidth (and hence also the drawn treewidth) is  $\OO(\sqrt{n})$. So, for example, by using only horizontal (or vertical) cutters, we will not be able to attain the subexponential-time algorithm for {\sc Grid Recognition} mentioned in Section \ref{sec:introScheme}.

\begin{figure}[t]
	\centering
	\begin{subfigure}{0.4\textwidth}
		\includegraphics[page=67, width=\textwidth]{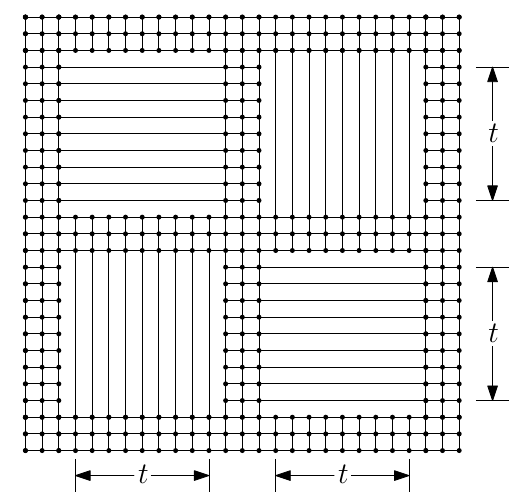}
		\subcaption{}
		\label{fig:simpleCutterIntro2_1}
	\end{subfigure}
	\hfil
	\begin{subfigure}{0.4\textwidth}
		\includegraphics[page=68, width=\textwidth]{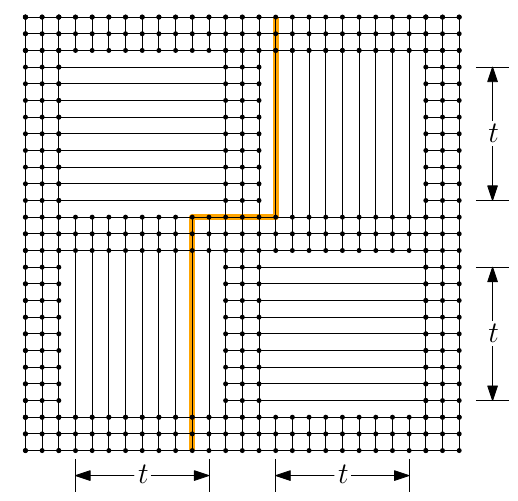}
		\subcaption{}
		\label{fig:simpleCutterIntro2_2}
	\end{subfigure}
	\caption{(a) Example of a rectilinear drawing of a graph on $n$ vertices with $t = \Omega(n)$. The straight-line drawn treewidth of this drawing is $\Omega(n)$. (b) Example of a cutter used in the drawn tree decomposition of width $O(1)$.}
	\label{fig:simpleCutterIntro2}
\end{figure}

Nevertheless, the straight-line treewidth of the drawing in Figure~\ref{fig:simpleCutterIntro1} can be seen to be bounded by $\OO(1)$ as well. However, regarding straight-line cutters, we consider the graph depicted in Figure~\ref{fig:simpleCutterIntro2_1}. Notably, every {\em rectilinear grid drawing} of this graph (being a generalization of a grid drawing, where edges are straight-lines of arbitrary lengths) can be obtained from the one depicted in the figure by ``stretching'' the drawings of some of its edges (and up to isomorphism). Now, notice that the straight-line drawn treewidth of this drawing is $\Omega(n)$.  To see this, notice that every axis-parallel straight-line that intersects this graph, intersects the drawings of at least $\Omega(n)$ distinct vertices and edges of this graph. However, the drawn treewidth of this drawing is only $\OO(1)$. To see this, consider the usage of cutters as the one depicted in Figure~\ref{fig:simpleCutterIntro2_2}.

%!TEX root =Main-Movement.tex

\subsection{Comparison with Other Graph Width Parameters}\label{sec:introComparison}

Recall that, drawn tree decomposition is based on decomposing a given polyline grid drawing of a graph. Therefore, the drawn treewidth is dependent on the polyline grid drawing of the graph. For e.g.,  Figures~\ref{fig:simplePathA} and~\ref{fig:pathBadDTWIntro1} depicts two different drawings of the same path which have different drawn treewidth. As path has a unique embedding, this also shows that {\bf different drawings of the same embedded graph may have different drawn treewidth. To the best of our knowledge, our parameter is the only one that depends on the drawing (rather than the embedding or just the graph).} Thus, we compare and discuss the differences between the drawn treewidth of a given polyline drawing of the graph and some seemingly related graph width parameters, namely: treewidth, pathwidth, carving-width, dual carving-width and embedded-width. Note that, the dual carving-width and the embedded-width is only defined when the given graph is a plane graph. Specifically, we prove the following theorem.
\begin{theorem}
	Given a graph $G$ and a polyline drawing $d$ of $G$, we have the following.
	\begin{enumerate}[(a)]
		\item The treewidth of $G$ is at most $6$ times the drawn treewidth of $d$. Moreover, the drawn treewidth of $d$ might be arbitrary larger than the treewidth of $G$.
		\item The pathwidth of $G$ and the drawn treewidth of $d$ are incomparable.
		\item The drawn treewidth of $d$ might be arbitrary larger than the carving-width of $G$.
		\item If $G$ is a plane graph, the dual carving-width and the embedded-width of $G$ might be arbitrary larger than the drawn treewidth of $d$.
	\end{enumerate}
\end{theorem}

We now give the proof of the above theorem. Let $\Delta$, $\tw$, $\pw$, and $\cw$ be the maximum degree, treewidth, pathwidth and the carving-width of $G$, respectively. Further, if $G$ is a plane graph, let $\ell$, $\dcw$ and $\ew$ be the maximum face size, the dual carving-width (the carving width of the dual graph), and the embedded-width of $G$, respectively.

\medskip\noindent{\bf Comparison with Treewidth.} As mentioned earlier in Section~\ref{sec:introConcept}, we prove that given a graph and a polyline drawing of it, $\tw$ is at most $6$ times the drawn treewidth (in Section~\ref{sec:relToTreewidth}). Moreover, we also show that given a graph and a polyline drawing of it, the drawn treewidth of the drawing might be arbitrary larger than the treewidth of the graph (see Figure~\ref{fig:pathBadDTWIntro}).

\begin{figure}[!t]
	\centering
	\begin{subfigure}{0.45\textwidth}
		\includegraphics[width = \textwidth, page = 79]{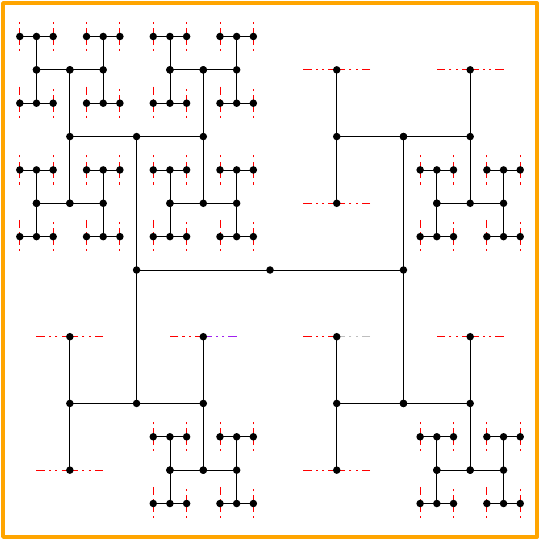}
		\subcaption{}
		\label{fig:pathwidthComA}
	\end{subfigure}
	\hfil
	\centering
	\begin{subfigure}{0.45\textwidth}
		\includegraphics[width = \textwidth, page = 80]{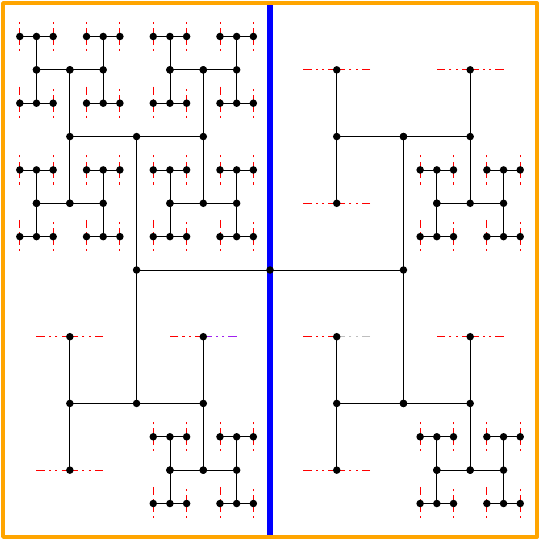}
		\subcaption{}
		\label{fig:pathwidthComB}
	\end{subfigure}
	\hfil
	\centering
	\begin{subfigure}{0.45\textwidth}
		\includegraphics[width = \textwidth, page = 81]{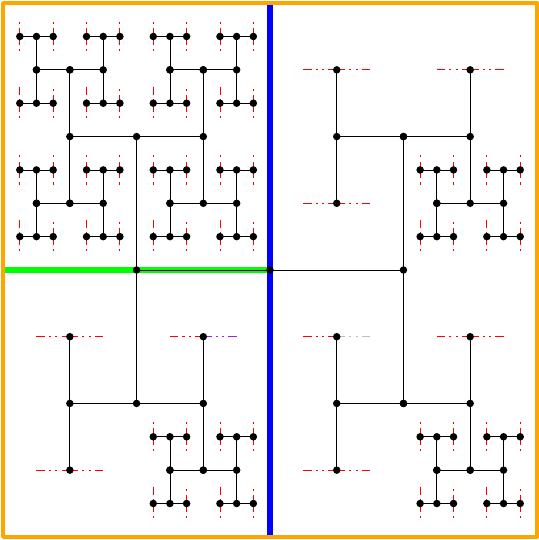}
		\subcaption{}
		\label{fig:pathwidthComC}
	\end{subfigure}
	\hfil
	\centering
	\begin{subfigure}{0.45\textwidth}
		\includegraphics[width = \textwidth, page = 82]{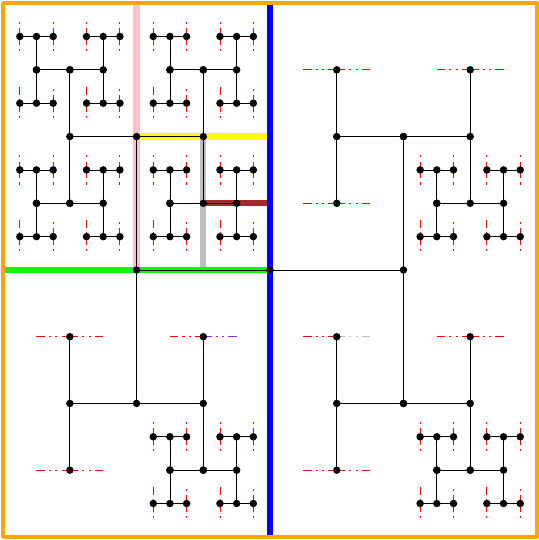}
		\subcaption{}
		\label{fig:pathwidthComD}
	\end{subfigure}
%	\centering
%	\includegraphics[page=76, width=0.5\textwidth]{figures/drawnTreewidth}
	\caption{Example of a rectilinear drawing (in black) of a binary tree on $n$ vertices. The rectangular is shown in orange. Examples of cutters are shown in blue, green, pink, yellow, grey and brown. Each one of them intersects $O(1)$ vertices and edges. Overall, the pathwidth of the tree is $\Omega(\log n)$, while the drawn treewidth of this drawing is $O(1)$.}
	\label{fig:pathwidthCom}
\end{figure} 

\medskip\noindent{\bf Comparison with Pathwidth.} In Figure~\ref{fig:pathwidthCom}, we have a rectilinear grid drawing of a binary tree. By using cutters as illustrated in the figure (in orange), we can get a drawn tree decomposition of constant width. In particular, one can see that each cutter intersects a constant number of edges and vertices. Since we use only straight cutters, and the maximum degree of the graph is $3$, by Observation~\ref{obs:costF}, we conclude that the width of each frame in such a drawn tree frame is bounded by a constant. Therefore, we get that the drawn treewidth of the drawing is bounded by a constant. Observe that this example can be expanded to a binary tree of any size. Furthermore, the pathwidth of a binary tree with $n$ vertices is $\Omega(\mathsf{log}_2(n))$. So, given a graph and a polyline drawing of it, the pathwidth of the graph might be arbitrary larger than the drawn treewidth of the drawing.

On the other hand, {\sc Grid Recognition} is NP-hard on graphs of pathwidth $2$, and we show in this paper that the problem is XP with respect to drawn treewidth. So, given a graph and a polyline drawing of it, the drawn treewidth of the drawing might be arbitrary larger than the pathwidth of the graph. Thus, we conclude that the two parameters, pathwidth and drawn treewidth, are incomparable.

\medskip\noindent{\bf Comparison with Carving-width, Dual Carving-width and Embedded-width.} It is known that $\cw \leq \Delta(\tw +1)$~\cite{DBLP:journals/ijcga/BiedlV13}. As the {\sc Grid Recognition} problem is NP-hard even for binary trees, we get that it is NP-hard even for graphs of carving-width at most $6$. In this paper, we show that the problem is XP with respect to drawn treewidth. So, given a graph and a polyline drawing of it, the drawn treewidth of the drawing might be arbitrary larger than the carving-width of the graph.

If the given graph is plane, it is known that $\ell \leq \dcw$ and $\ell \leq \ew$~\cite{DBLP:journals/algorithmica/LozzoEGG21}. Therefore, we get that both the dual carving-width and the embedded-width of a path are at least the size of its vertex set. In this paper, we show that there exists a drawing of any path with drawn treewidth at most $16$ (see Figure~\ref{fig:simplePathB}). So, given a plane graph and a polyline drawing of it, the dual carving-width and the embedded-width of the graph might be arbitrary larger than the drawn treewidth of the drawing.
Thus, we conclude that drawn treewidth differs from carving-width, dual carving-width and embedded-width.

%!TEX root =Main-Movement.tex

\subsection{Our Scheme}\label{sec:introScheme}

Here, we present (informally) our general scheme for the design of algorithms for problems in Graph Drawing parameterized by the drawn treewidth of the sought drawing (that should be, in particular, a polyline grid drawing), based on dynamic programming. Formal definitions and further details can be found in Section~\ref{sec:ProbSche}. For the clarity of the discussion, we first introduce the four main definitions required for the scheme and its proof of correctness. Then, we discuss the usage of our scheme---specifically, which two procedures the user should design in order to apply the scheme as a black box. Afterwards, we specify the properties that a problem should satisfy so that our scheme will solve it correctly, and the running time that will be attained. Lastly, we present some technical details concerning the scheme itself, that is, how it is executed.

\medskip\noindent{\bf Key Players: Info-Frames, Info-Cutters, Splitting and Glueing.} {\bf\em Info-Frames.} The most basic definition required for our scheme is that of an {\em info-frame}. Briefly, an info-frame encodes information about the ``behaviour'' of the restriction of some drawing $d$ to the interior of some particular frame $f$ (see Figure~\ref{fig:infFandD}). For that purpose, the info-frame consists of five components, where the first one is, simply, the frame $f$. The second component is a drawing $d_f$ that specifies the drawings of the vertices and edges (by $d$) of the graph on $f$ itself. More precisely, $d_f$ specifies which vertices of the graph are drawn on $f$ and where are they drawn on $f$. Additionally, for every edge $e$ of the graph, it specifies which are the turning points of $e$ on $f$ and where are these turning points drawn on $f$. Moreover, for the aforementioned turning points, it specifies the order in which they are encountered (when we ``walk'' along the drawing of $e$ from one end to the other), and for each maximal subcurve of the drawing of $e$ that does not contain a turning point internally, it specifies whether this subcurve is drawn on $f$ (i.e., being a subcurve of $f$ as well), and if yes, then it specifies the drawing of this subcurve (for which, knowing the drawings of its endpoints, we have only two options).

The third and fourth components, denoted by $U_f$ and $E_f$, concern the strict interior of $f$. Specifically, $U_f$ specifies which vertices of the graph are drawn strictly inside $f$. As for $E_f$, for every edge $e$ of the graph and for each maximal subcurve of the drawing of $e$ that does not contain a turning point internally, it specifies whether this subcurve is drawn in the strict interior of $f$ (except for, possibly, the endpoints of the curve). We remark that the number of ``sensible'' choices for $U_f$ and $E_f$ is much smaller than it might appear to be at first glance, supposing that the graph at hand is connected. The fifth component, roughly speaking, describes the ``angles'' in which drawings of edges cross $f$ using straight line segments attached to turning points. Such information is necessary, for example, to ensure that some subcurves corresponding to the drawings of the same edge lie in a single straight line, so that no bend---if forbidden by the problem at hand---occurs.

Importantly, the definition of an info-frame is independent of a specific drawing, being an ``abstract'' tuple of five components. Every drawing that can be described by the tuple (as discussed above) is said to be a drawing of the info-frame. So, one info-frame may describe multiple drawings, or none at all. We note that for an ``abstract'' five-component tuple to be an info-frame, it should satisfy various (considerably technical) properties, which, in particular, any info-frame that does describe at least one drawing must satisfy. On the one hand, these properties bound the number of possible info-frames, and, on the other hand, they are also used in the proof of correctness of our scheme.

Lastly, observe that the restriction of some drawing $d$ to the interior of some particular frame $f$ is {\em not} a drawing of a {\em graph}. Indeed, some edges are drawn (by $d$) partially in the interior of $f$ and partially in the strict exterior of $f$. However, if we ``enrich'' the graph by placing ``virtual'' vertices on turning points, then the restriction of $d$ to the interior of $f$ will be a drawing of a graph (being a subgraph of the enriched graph). So, for technical reasons, this is exactly what we do. For this purpose, we define and work with so-called $G^\star$-drawings; however, to keep the overview short and simple, we will not discuss $G^\star$-drawings and related technical terms in this overview.

\medskip\noindent{\bf\em Info-Cutters.} Just as we use an info-frame to encode information about the ``behaviour'' of a drawing $d$ with respect to a frame $f$, we use an {\em info-cutter} of an info-frame to encode information about the ``behaviour'' of $d$ with respect to a cutter $c$ of $f$. Rather than directly describing how $d$ is drawn on $c$ and how $d$ is ``split'' by $c$ inside $f$, we find it easier to indirectly describe this information by defining an info-cutter based on two info-frames corresponding to the frames obtained by cutting $f$ with $c$ (later, for the dynamic programming implementation, we can thus immediately know to which already computed entries to refer). Observe that, in particular, the two frames being part of these two info-frames contain $c$, and, thus, these two info-frames capture the aforementioned information.

To be more precise, an info-cutter $C$ of an info-frame $F$, where the first component of $F$ is some frame $f$, is a triple $(c,F_1,F_2)$, where, in particular, $c$ is a cutter of $f$, and $F_1$ and $F_2$ are info-frames for the two frames obtained by cutting $f$ with $c$. Additionally, for such a triple to be an info-cutter, it should satisfy (considerably technical) properties, which, in particular, any info-cutter that does describe at least one drawing must satisfy. Very briefly, these properties validate consistency between the information described by $F$, $F_1$ and $F_2$. This is more complicated than it might appear to be at first glance, since, even on the cutter $c$, $F_1$ and $F_2$ might describe the existence of different virtual vertices (having different turning points).  For the sake of simplicity, we do not discuss these details in the overview.

\medskip\noindent{\bf\em Splitting and Glueing.} First, let us consider the {\em splitter function}, which, for our scheme, is used only for the proof of correctness (where its input is assumed to contain a subdrawing of a hypothetical solution drawing). Given an info-frame $F$ whose first component (being a frame) is $f$, a drawing $d$ restricted to the interior of $f$ that is compatible with the description encoded by $F$, and a cutter $c$ of  $f$, the splitter function returns an info-cutter $C=(c,F_1,F_2)$ and two drawings, $d_1$ and $d_2$. Let $f_1$ ($f_2$) be the first component of $F_1$ ($F_2$). Briefly, we define the output such that $d_1$ and $d_2$ would be the subdrawings of $d$ restricted to the interiors of $f_1$ and $f_2$, respectively, and $F_1$ and $F_2$ would be the info-frames that describe $d_1$ and $d_2$, respectively.

The {\em glue function} is, intuitively, the ``inverse'' of the split function, and it is used algorithmically in our scheme. Its input consists of an info-frame $F$, an info-cutter $C = (c, F_1, F_2)$ of $F$, a drawing $d_1$ of $F_1$ and a drawing $d_2$ of $F_2$. Roughly speaking, this function aims to ``glue'' $d_1$ and $d_2$ into a single drawing $d$ that is restricted to the interior of $f$, being the first component of $F$, and that should be compatible with the description encoded by $F$; of course, this operation might be impossible, and then the function simply announces that. We refer to Figures~\ref{fig:splitter} and~\ref{fig:glue} for a high-level illustrative description of splitting and glueing. Among other proofs concerning these functions, we show, in particular, that the specific way in which we define the splitter and glue functions (not described in the overview)  ensures that, if we apply the glue function on an output of the splitter function, we are able to reconstruct the drawing given as input to the splitter function.

\medskip\noindent{\bf The User's Point of View.} For the execution of the scheme, we expect the user to provide four components: some universe denoted by $\mathsf{INF}$, and three algorithmic procedures (that will be defined immediately). All of these components are problem-dependent.
\begin{itemize}
\item The first procedure, termed {\em classifier} and denoted by $\mathsf{Classifer}$, is given an info-frame $F$ and a corresponding drawing $d$, and it returns an element from $\mathsf{INF}$. Intuitively, this element describes the equivalence class of $d$. So, we say that two drawings corresponding to the same info-frame are {\em equivalent} if the classifier associates them with the same element. 
\item The second procedure, termed {\em classifier algorithm},   is given an info-frame $F$, an info-cutter $C = (c, F_1 ,F_2)$ of $F$ and  $I_1 , I_2 \in \mathsf{INF}$, and it returns  $I'\in \mathsf{INF}$ such that: For any two drawings $d_1$ and $d_2$ corresponding to $F_1$ and $F_2$, respectively, such that $\mathsf{Classifier}(F_1, d_1) = I_1$ and $\mathsf{Classifier}(F_2, d_2) = I_2$, we have $\mathsf{Classifier}(F, d) = I'$ where $d = \mathsf{Glue}(F, C, d_1, d_2)$. In particular, notice that any two drawings of the same two equivalence classes always yield (when being glued) a drawing of the same equivalence class---this justifies our usage of the term {\em equivalence} in this context.
\item The third procedure, termed {\em leaf solver}, is given an info-frame $  F$ whose frame does not contain any grid point in its strict interior, and for every $I'\in\mathsf{INF}$,  it returns ``yes'' if and only if there exists a drawing $d$ corresponding to $F$ such that $\mathsf{Classifier}(F, d) = I'$. Practically, we require this procedure to solve the basis of our dynamic programming computation, corresponding to info-frames whose frames do not contain any grid point in their strict interiors.
\end{itemize}

The scheme, once given these components, can be executed in a black box fashion. For the sake of simplicity of the overview, we do not discuss the technical details of the execution itself (as a white box) here.

\medskip\noindent{\bf To Which Type of Problems Does Our Scheme Apply?} Roughly speaking, we prove that our scheme can be applied to any graph drawing problem $\Pi$ such that: 
\begin{enumerate}
\item Every instance of $\Pi$ contains, in particular, a connected graph $G$, dimensions $h$ and $w$ for the sought drawing (which are, usually, bounded from above by the number of vertices $n$ of $G$), and the parameter $k$ (being any non-negative integer).
\item The objective is to determine whether $G$ admits a polyline grid drawing bounded by rectangle of dimensions $h\times w$, whose drawn treewidth is at most $k$, and that satisfies various problem-specific properties (for some examples, see Section \ref{sec:introApplications}).
\item The user can design the three algorithmic procedures discussed above.
\end{enumerate}

For any such problem $\Pi$, we prove that the runtime of the scheme is bounded by
\[\OO(k\cdot h\cdot w\cdot n)^{\OO(k)}\cdot|\mathsf{INF}|^{\OO(1)}\cdot\left(2^{\OO(\Delta\cdot k)}\cdot\mathsf{T2} + \mathsf{T3}\right),\]
where $\mathsf{T2}$ and $\mathsf{T3}$ bound the runtimes of the second and third procedures provided by the user, and $\Delta$ is the maximum degree of $G$. In particular, if $h,w,|\mathsf{INF}|,\mathsf{T2}$  and $\mathsf{T3}$ can be bounded by $n^{\OO(1)}$ (which is the case for many applications, such as grid recognition and orthogonal compaction), then the runtime above simplifies to $n^{\OO(k)}$, that is, we obtain an XP-algorithm.

%!TEX root =Main-Movement.tex

\subsection{Applications of Our Scheme to Problems in Graph Drawing}\label{sec:introApplications}

For most of the problems considered in this paper, the time complexity of our scheme can be bounded by $n^{\OO(k)}$, where $k$ is the input parameter that upper bounds the drawn treewidth of the output drawing (the refined upper bound on the running time of the scheme can be found in Theorem~\ref{the:AlgSch}). We remark that the formal definitions of these problems are relegated to Section~\ref{sec:problemDef}. 

%and a discussion of related literature is relegated to Section~\ref{sec:relatedWorks}.
\bigskip\noindent{\bf Grid Recognition.} We first consider the relatively simple {\sc Grid Recognition} problem in order to demonstrate the application of our scheme. Here, given a (connected) graph $G$, the objective is to determine whether $G$ is a grid graph, that is, whether it admits a grid drawing. The {\sc Grid Recognition} problem was first proved to be NP-hard in 1987, on ternary trees of pathwidth $3$~\cite{DBLP:journals/ipl/BhattC87}. Two years later in 1989, the problem was proved to be NP-hard even on binary trees~\cite{gregori1989unit}. Recently in 2021, the problem was proved to be NP-hard even on trees of pathwidth $2$~\cite{DBLP:conf/isaac/0002SZ21}. In the same paper, it was also proved that the problem is polynomial time solvable on graphs of pathwidth $1$. A year later in 2022, it was proved that even if we require all the internal faces of the drawing to be rectangles, the problem is still NP-hard even for biconnected graphs~\cite{alegria2022unit}. In the same paper, it was also proved that if we require all the faces of the drawing to be rectangles (including outer face), the problem is cubic time solvable.

As we deal with the parameterized version of this problem where the parameter is the drawn treewidth of the sought drawing (or, more precisely, an upper bound on it), we are also given $k$ as input. We prove the following result.

\begin{theorem}\label{the:gridRecRunTime}
	There exists an algorithm that solves the {\sc Grid Recognition} problem in time $n^{\OO(k)}$.
\end{theorem}

Since for grid drawings, we also prove that $k\leq\OO(\sqrt{n})$ (see Corollary~\ref{cor:uppBouGr}), we get the following corollary. 

\begin{corollary}
	There exists an algorithm that solves the {\sc Grid Recognition} problem in time $n^{\OO(\sqrt{n})}$.
\end{corollary}  

Thus, we obtain a subexponential-time algorithm for {\sc Grid Recognition}, matching the running time of the current best known algorithm for this problem~\cite{damaschke2020enumerating}.

\bigskip\noindent{\bf Crossing and Bend Minimization.} For our second application, we study a variant of the {\sc Crossing Minimization} problem. The {\sc Crossing Minimization} problem  is one of the most fundamental graph layout problems. It was shown to be NP-complete by Garey and Johnson in 1983~\cite{garey1983crossing}. Later, it was proved to be NP-complete even on graph of maximum degree 3~\cite{hlinveny2006crossing} and also on \emph{almost planar graphs} which are graphs that can be made planar by  removing a single edge~\cite{cabello2013adding}. It was also shown that the problem remains NP-hard even if the cyclic order of the neighbours around each vertex is fixed and to be respected by the resulting drawing~\cite{DBLP:journals/algorithmica/PelsmajerSS11}. On the positive side, it is known the problem is FPT with respect to the number of crossings~\cite{grohe2004computing,kawarabayashi2007computing} and also with respect to the vertex cover~\cite{DBLP:conf/gd/HlinenyS19}.
There are many other variants of this problem which are studied in the literature. One of them concerns with minimizing the number of pairwise crossing edges in any straight-line drawing of the graph. This problem is known to be NP-hard~\cite{bienstock1991some} (and even $\exists \mathbb{R}$-complete~\cite{schaefer2009complexity}). For more information about the crossing minimization and its variants, we refer to the survey~\cite{suvey/zehavi}.

A related problem is the {\sc Bend Minimization} problem. Given a graph $G$, the {\sc Bend Minimization} problem asks for an orthogonal grid drawing of $G$ with minimum number of total bends. The problem was proved to be NP-complete in 2001, even when there are no bends~\cite{DBLP:journals/siamcomp/GargT01}. On the positive side, if the input graph is plane, the problem can be solved in polynomial time~\cite{DBLP:journals/siamcomp/Tamassia87}.  When the input graph is not planar, there are polynomial time algorithms for subclasses of planar graphs, namely planar graphs with maximum degree 3~\cite{DBLP:conf/compgeom/ChangY17,DBLP:journals/siamcomp/BattistaLV98,DBLP:conf/gd/DidimoLP18,DBLP:journals/ieicet/RahmanEN05} and series-parallel graphs~\cite{DBLP:journals/siamdm/ZhouN08}.

We study the {\sc Straight-line Grid Crossing Minimization} problem where the sought drawing should be a straight-line grid drawing. Here, given a (connected) graph $G$ and $h,w\in\mathbb{N}$, the objective is to determine a straight-line grid drawing of $G$ bounded by a rectangle of dimension $h \times w$ with minimum number of crossings, if one exists. Similar to the previous example, as we study the parameterized version of this problem, we are also given $k$ as input. We prove the following result.

\begin{theorem}\label{the:crossingMinTime}
	There exists an algorithm that solves {\sc Straight-line Grid Crossing Minimization} problem in time $\OO ((k\cdot h\cdot w\cdot n)^{\OO (k)}\cdot 2^{\OO (\Delta\cdot k)})$, where $\Delta$ is the maximum degree of the input graph.
\end{theorem}

More generally, our scheme can be applied to a very wide class of problems of such flavor; in particular, every problem where:
\begin{itemize}
	\item The input consists of (some or all of)  the following: a graph $G$; $\mathsf{cross}: E(G)\rightarrow\mathbb{N}_0\cup{\infty}$; $\mathsf{bend}: E(G)\rightarrow\mathbb{N}_0\cup{\infty},$ and $C,B,k\in\mathbb{N}_0\cup\{\infty\}$. Here, $E(G)$ is the edge set of $G$, and $\mathbb{N}_0=\mathbb{N}\cup\{0\}$.
	\item The objective is to determine whether $G$ admits a drawing that is {\em (i)} a grid drawing, or  {\em (ii)} a rectilinear drawing, or {\em (iii)} an orthogonal grid drawing, or {\em (iv)} a straight-line grid drawing, or {\em (v)} a polyline grid drawing, such that:
	\begin{itemize}
		\item For every edge $e\in E(G)$, the drawing of $e$ has at most $\mathsf{cross}(e)$ crossings and at most $\mathsf{bend}(e)$ bends.
		\item In total, we have at most $C$ crossings and at most $B$ bends.
	\end{itemize}
\end{itemize}
Further, the scheme can be applied to various variants of the above generic problem that were studied in the literature. For example, we can specify, for every edge, whether it should be crossed an even or odd number of times. Similarly, we can also consider the weighted crossing number.

\bigskip\noindent{\bf Orthogonal Compaction.} Lastly, we note that our scheme can also be applied to problems of flavors quite different than the above. As an example, we consider the {\sc Orthogonal Compaction} problem. Here, given a planar orthogonal representation $H$ of a connected planar graph $G$, the objective is to compute a minimum-area drawing of $H$. The {\sc Orthogonal Compaction} problem was first proved to be NP-hard on general graphs in 2001~\cite{patrignani2001complexity}. Later, it was shown that the problem is NP-hard even on cycles~\cite{DBLP:journals/comgeo/EvansFKSSW22}, ruling out an FPT algorithm with respect to treewidth, unless P=NP. On the positive side, it was proved that the problem is linear time solvable for a restricted class of planar orthogonal representation~\cite{bridgeman2000turn}. Recently, it was also shown that the problem is FPT with respect to number of ``kitty corner vertices", a parameter central to the problem~\cite{DBLP:conf/sofsem/DidimoGKLWZ23}.

Similar to the previous examples, as we study the parameterized version of this problem, we are also given $k$ as input. We prove the following result.

\begin{theorem}\label{the:orthoComTime}
	There exists an algorithm that solves the {\sc Orthogonal Compaction} problem in time $n^{\OO(k)}$.
\end{theorem}

\section{Preliminaries}\label{sec:prelims}
In this paper, we only consider finite simple undirected graphs, unless stated otherwise. Moreover, we refer to straight line segments as line segments, unless stated otherwise. 
Let $\mathbb{N}_0=\mathbb{N}\cup \{ 0 \}$. 
For $k,i,j \in \mathbb{N}$, we denote $[k] = \{1,2,\ldots k\}$ and $[i,j] = \{i, i+1, \ldots, j\}$.  

\subsection{\bf{Graph Notation and Decompositions}}\label{sec:prelims1}

For a graph $G=(V,E)$ and a subset of vertices $U\subseteq V$, we denote by $G[U]$ the subgraph of $G$ induced by $U$.  
For a given subset $V'\subseteq V$ of vertices, we define the {\em boundary of $V'$} as the set of vertices in $V'$ that are adjacent to a vertex in $V\setminus V'$:

\begin{definition} [{\bf Boundary}] \label{def:Boundary}
Let $G=(V,E)$ be a graph. Let $V'\subseteq V$. Then the {\em boundary} of $V'$ in $G$, denoted by $B_G(V')$, is the set of vertices of $V'$ that have a neighbor in $V\setminus V'$, i.e., $B_G(V')=\{v'\in V'~|~$ there exists $v\in V\setminus V'$ such that  $\{v,v'\} \in E \}$.   
\end{definition}

When the graph $G$ is clear from the context, we drop it from the subscript. Given a path $P$, we represent $P$ as a sequence of vertices $v_1,v_2,\ldots, v_k$, such that $\{v_i,v_{i+1}\}$ is an edge in $P$ for every $1\leq i\leq k-1$. Similarly, given a cycle $C$, we represent $C$ as a sequence of vertices $v_1,v_2,\ldots, v_k$, such that $v_1=v_k$ and $\{v_i,v_{i+1}\}$ is an edge in $C$ for every $1\leq i\leq k-1$. Note that we use the terms path and cycle to refer to simple path and cycles.
We now define the concepts of a {\em tree decomposition} and a {\em path decomposition}.

\begin{definition} [{\bf Tree Decomposition}] \label{def:treeDeco}
A {\em tree decomposition} of a graph $G=(V,E)$ is a pair $({\cal T}=(V_T,E_T),\beta: V_T\rightarrow 2^V)$ where $\cal T$ is a tree such that:
\begin{enumerate}
\item For every $v\in V$, the subgraph of $\cal T$ induced by $\{x\in V_T~|~ v\in \beta(x)\}$ is non-empty and connected. \label{def:treeDecocon1}
\item For every $\{u,v\}\in E$, there exists $x\in V_T$ such that $\{u,v\}\subseteq \beta(x)$. \label{def:treeDecocon2}
\end{enumerate} 
The {\em width} of $({\cal T},\beta)$ is defined to be $\max_{x\in V_T}|\beta(x)|-1$. For every $x\in V_T$, $\beta(x)$ is called a {\em bag}. The {\em treewidth} of a graph $G$ is the minimum width of any tree decomposition of $G$.
\end{definition}
 
 \begin{definition} [{\bf Path Decomposition}] \label{def:pathDeco}
A {\em path decomposition} of a graph $G=(V,E)$ is a pair $({\cal P}=(V_P,E_P),\beta: V_P\rightarrow 2^V)$ where $\cal P$ is a path such that:
\begin{enumerate}
\item For every $v\in V$, the subgraph of $\cal P$ induced by $\{x\in V_P~|~ v\in \beta(x)\}$ is non-empty and connected. \label{def:pathDecocon1}
\item For every $\{u,v\}\in E$, there exists $x\in V_P$ such that $\{u,v\}\subseteq \beta(x)$. \label{def:pathDecocon2}
\end{enumerate} 
The {\em width} of $({\cal P},\beta)$ is defined to be $\max_{x\in V_P}|\beta(x)|-1$. For every $x\in V_P$, $\beta(x)$ is called a {\em bag}. The {\em pathwidth} of a graph $G$ is the minimum width of any path decomposition of $G$.
\end{definition}

%Let $f:V(G) \rightarrow \mathbb{N} \times \mathbb{N}$ be an injection that maps each vertex $v$ of $G$ to a point $(i,j)$ of an integer grid; then, $i$ and $j$ are also denoted as $\fr(v)$ and $\fc(v)$, respectively, that is, $f(v) = (\fr(v), \fc(v))$.

\subsection{\bf{Graph Drawing}}\label{sec:prelims2}
For a given graph $G$, a {\em drawing of} $G$ on the plane is a mapping of the vertices to distinct points of $\mathbb{R}^2$ and of the edges to simple curves in $\mathbb{R}^2$, connecting the images of their endpoints. A drawing of a graph is {\em planar} if no pair of edges, or an edge and a vertex, cross except at a common endpoint. Two planar drawings of the same graph are {\em equivalent} if they determine the same {\em rotation} at each vertex, that is, the same circular ordering for the edges around each vertex. An {\em embedding} is an equivalence class of planar drawings.

Given a drawing $d$ of $G$, we represent $d$ as a pair of functions $(d_V,d_E)$ as follows. The function $d_V:V\rightarrow \mathbb{R}\times \mathbb{R}$ is an injection, which maps each vertex $v$ of $G$ to a point $(i,j)$ in the plane; then, $i$ and $j$ are also denoted as $\fr(v)$ and $\fc(v)$, respectively, that is, $d_V(v) = (\fr(v), \fc(v))$. The function $d_E:E\rightarrow \calC$, where $\calC$ is the set of all simple curves in the plane, maps each edge $\{u,v\}\in E$ to a simple curve $c\in \calC$ between $d_V(u)$ and $d_V(v)$. For simplicity, we refer to $(d_V,d_E)$ as one function, $d:V\cup E \rightarrow \{\mathbb{R}\times \mathbb{R}\} \cup \calC$, such that $d(v)=d_V(v)$ for every $v\in V$, and $d(\{u,v\})=d_E(\{u,v\})$ for every $\{u,v\}\in E$.  We call $V$ and $E$ the {\em vertex set} and the {\em edge set associated with $d$}, respectively. Let $d$ be a drawing of a graph $G$, and let $p\in\mathbb{R}^2$ be a point. We say that {\em $p$ is on $d$} if $p$ is on the image of an edge of $G$ in $d$ or $p$ is the image of a vertex of $G$ in $d$. We denote by $\pp(d)$ the set of points on $d$.

For two points $p_1=(x_1,y_1)$ and $p_2=(x_2,y_2)$ in the plane, we denote the line segment joining the points by $\ell(p_1,p_2)$. 
For four points, $p_i=(x_i,y_i)\in \mathbb{R}^2$ for every $1\leq i\leq 4$, we say that $\ell(p_1,p_2)$ {\em crosses} $\ell(p_3,p_4)$ if the line segments $\ell(p_1,p_2)$ and $\ell(p_3,p_4)$ cross except at $p_i=(x_i,y_i)$ for every $1\leq i\leq 4$. 
Let $a$ and $b$ be two points in $\mathbb{R}^2$ and let $\epsilon > 0$. We denote $\ell(a,a_\epsilon)$ by $\mathsf{line}_\epsilon(a,b)$, where $a_\epsilon$ is the point on the line $\ell(a,b)$ at distance $\epsilon$ from $a$ if it exists.
For a pair of points $(p_1,p_2)$, and a point $p'$, where $p_1,p_2,p'\in \mathbb{R}^2$, we say that {\em $\ell(p_1,p_2)$ intersects $p'$} if $p'$ is on the line $\ell(p_1,p_2)$, including its endpoints.  
We use the term {\em grid points} to refer to the infinite set of points $(x,y)\in \mathbb{R}^2$ where $x,y\in \mathbb{N}_0$. 
Given two distinct grid points $p_1=(x_1,y_1)$ and $p_2=(x_2,y_2)$, we say that $p_1<p_2$ if $x_1<x_2$ or $x_1=x_2$ and $y_1<y_2$.

%\begin{definition} [{\bf Axis Parallel Drawn Path}] \label{def:DrawnWalk}
%	An {\em axis parallel drawn path} is a drawn path, with a prescribed orthogonal grid drawing.
%\end{definition}

%Let $a$ and $b$ be two points in $\mathbb{R}^2$ and let $\epsilon > 0$. We denote by $\mathsf{point}_\epsilon(\ell(a,b))$ the point on the line $\ell(a,b)$ at distance $\epsilon$ from $a$ if it exists.

A {\em drawn graph} is a graph with a prescribed drawing. A {\em plane graph} is a drawn graph whose prescribed drawing is planar. A drawing of a graph is called a {\em straight-line} drawing if the edges are mapped to line segments, connecting the images of their endpoints. We define a {\em straight-line path (cycle)} as a plane path (cycle), where the vertices are mapped to grid points and edges are mapped to line segments connecting the images of their endpoints. We denote by $\cal{P}\subset \calC$ the (infinite) set of straight-line paths in $\mathbb{R}^2$. Moreover, we alternatively denote any path $P=(v_1,\ldots,v_k)\in\cal{P}$ by the sequence $(p_1,\dots,p_k)$, where $p_i\in \mathbb{R}^2$ is the image of the vertex $v_i$ in $P$, for every $1\leq i\leq k$.
We define an {\em axis-parallel path (cycle)} as a straight-line path (cycle), where every edge of the path is  parallel to the $X$- or $Y$- axis. For an axis-parallel path $P=(p_1,\ldots, p_k)$, we denote by $|P|$ the (Euclidean) {\em length} of $P$, that is, $|P|=|p_2-p_1|+|p_3-p_2|+\ldots+|p_k-p_{k-1}|$. Next, we define a {\em grid drawing} of a graph $G$ as a straight-line drawing of $G$ where the vertices are mapped to grid points and the edges are mapped to (axis-parallel) unit length line segments (e.g., see Figure~\ref{fi:DrawEx2}):

%$\mathsf{point}_\epsilon(\ell(a,b))$.
%nd the edges are mapped to line segments connecting its endpoints, which are non-intersecting, except at their endpoints. 

\begin{definition}[{\bf Straight-Line Grid Drawing}]
	Let $G$ be a graph. A {\em straight-line grid} drawing $d$ of $G$ is a straight-line drawing $d$ of $G$ such that (i) for every $u\in V$, $d(u)$ is a grid point (ii) For every $\{u,v\},\{u',v'\}\in E$, $d(\{u,v\})$ and $d(\{u',v'\})$ are intersected in at most one point.
\end{definition}

\begin{definition} [{\bf Grid Drawing}] \label{def:Grid graph}
Let $G=(V,E)$ be a graph. A {\em grid drawing} $d$ of $G$ is a drawing $d:V\cup E\rightarrow \mathbb{N}_0\times \mathbb{N}_0\cup \cal{P}$ such that if $\{u,v\} \in E$ then $|\fr(u)-\fr(v)|+|\fc(u)-\fc(v)|=1$.
%, is a function, that maps each vertex $v$ of $G$ to a point $(i,j)$ of an integer grid; then, $i$ and $j$ are also denoted as $\fr(v)$ and $\fc(v)$, respectively, that is, $d_V(v) = (\fr(v), \fc(v))$. The function $d_E:E\rightarrow \cal{P}$, maps every edge $\{u,v\}\in E$ to a path $p\in \cal{P}$ between $d_V(u)$ and $d_V(v)$. 
%Let $G=(V,E)$ be a graph. A {\em grid drawing} $d$ of $G$ is a straight-line drawing, where $d:V\rightarrow \mathbb{N}_0\times \mathbb{N}_0$ of $G$ 
\end{definition}

We now extend the concept of a grid drawing to a {\em rectilinear grid drawing}, where the edges are mapped to variable length line segments parallel to the axes (e.g., see Figure~\ref{fi:DrawEx3}):

\begin{definition} [{\bf Rectilinear Grid Drawing}] \label{def:rect draw}
	Let $G=(V,E)$ be a graph. A {\em rectilinear grid drawing} $d$ of $G$ is a drawing $d:V\cup E\rightarrow \mathbb{N}_0\times \mathbb{N}_0\cup \cal{P}$ of $G$, such that for every edge $\{u,v\}\in E$, $d(\{u,v\})$ is a line segment between $d(u)$ and $d(v)$ such that $\fr(u)=\fr(v)$ or $\fc(u)=\fc(v)$. 
\end{definition}

Further, we extend the concept of a rectilinear grid drawing to an {\em orthogonal grid drawing}, where the edges are mapped to straight-line paths, such that the edges of these paths are mapped to line segments parallel to the axes (e.g., see Figure~\ref{fi:DrawEx4}):

\begin{definition} [{\bf Orthogonal Grid Drawing}] \label{def:orthDraw}
	Let $G=(V,E)$ be a graph. An {\em orthogonal grid drawing} $d$ of $G$ is a drawing $d:V\cup E\rightarrow \mathbb{N}_0\times \mathbb{N}_0\cup \cal{P}$ of $G$, such that for every edge $\{u,v\}\in E$, $d(\{u,v\})$ is an axis-parallel path between $d(u)$ and $d(v)$. 
\end{definition}

Finally, we extend the concept of an orthogonal grid drawing to a {\em polyline grid drawing}, where the edges are mapped to straight-line paths instead of axis-parallel paths (e.g., see Figure~\ref{fi:DrawEx5}).

\begin{definition} [{\bf Polyline Grid Drawing}] \label{def:StrightD}
	Let $G=(V,E)$ be a graph. A {\em polyline grid drawing} $d$ of $G$ is a drawing $d:V\cup E\rightarrow \mathbb{N}_0\times \mathbb{N}_0\cup \cal{P}$ of $G$. 
\end{definition}

\begin{figure}[!t]
	\centering
	\begin{subfigure}{0.35\textwidth}
		\includegraphics[width = \textwidth, page = 32]{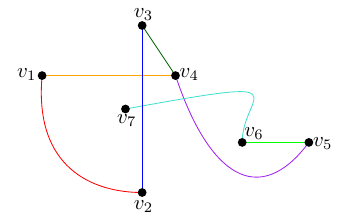}
		\subcaption{}
		\label{fi:DrawEx1}
	\end{subfigure}
	\hfil
	\begin{subfigure}{0.3\textwidth}
		\includegraphics[width = \textwidth, page = 33]{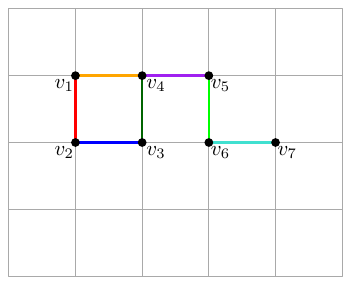}
		\subcaption{}
		\label{fi:DrawEx2}
	\end{subfigure}
\hfil
	\begin{subfigure}{0.3\textwidth}
		\includegraphics[width = \textwidth, page = 34]{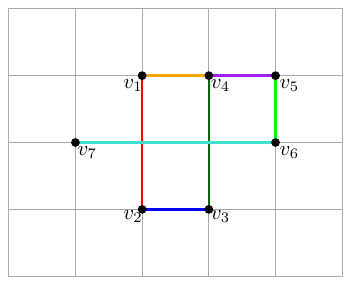}
		\subcaption{}
		\label{fi:DrawEx3}
	\end{subfigure}
	\begin{subfigure}{0.3\textwidth}
		\includegraphics[width = \textwidth, page = 35]{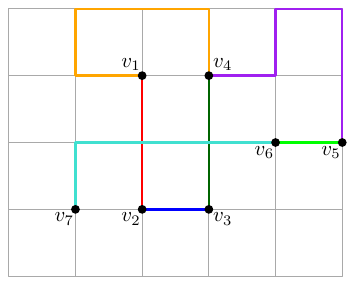}
		\subcaption{}
		\label{fi:DrawEx4}
	\end{subfigure}
\hfil
	\begin{subfigure}{0.3\textwidth}
		\includegraphics[width = \textwidth, page = 36]{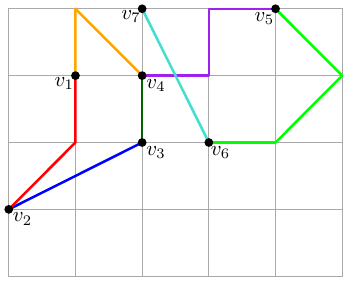}
		\subcaption{}
		\label{fi:DrawEx5}
	\end{subfigure}
	
	\caption{Different drawings (defined in Definitions \ref{def:Grid graph}-\ref{def:StrightD}) of the graph $G$ shown in (a). A grid, a rectilinear grid, an orthogonal grid and a polyline grid drawings of $G$ are shown in (b), (c), (d) and (e), respectively.} 
	\label{fi:DrawEx}
\end{figure}

\subsection{Problem Definitions}\label{sec:problemDef}
In this subsection, we give the definitions for the problems we will solve in Section~\ref{sec:ExampleScheme3} using our new concept.

%{\sc Grid Recognition} problem with runtime $n^{\OO (k)}$ where $k$, given as input, bounds the drawn treewidth of the sought realization (if one exists). In turn, this will also yield a runtime of $n^{\OO(\sqrt{n)}}$. Recall that in the {\sc Grid Recognition} problem parameterized $k$, given a graph $G=(V,E)$, the goal is to determine whether $G$ has a grid drawing of drawn treewidth at most $k$. We assume that $G$ is connected; otherwise, we apply the algorithm on each of the different connected components separately. 
%Additionally, recall that a drawing $d$ is a grid drawing if for every $u\in V(d)$, $d(u)\in \mathbb{N}_0\times \mathbb{N}_0$ and for every $\{u,v\} \in E(d)$, $|\fr(u)-\fr(v)|+|\fc(u)-\fc(v)|=1$. 

\begin{definition}[{\bf Grid Recognition Problem}]\label{def:gridRec}
	The {\sc Grid Recognition} problem is, given a graph $G$, to determine whether $G$ has a grid drawing.
\end{definition}

\begin{definition}[{\bf Crossing Minimization Problem on Straight-Line Grid Drawings}]\label{def:crossMin}
	The {\sc Straight-line Grid Crossing Minimization} problem is, given a graph $G$ and $h,w\in \mathbb{N}$, to construct a straight-line grid drawing $d$ of $G$ (if one exists) such that: (i) $d$ is strictly bounded by $R_{h,w}$,(ii) $d$ has minimum number of crossings out of all the straight-line grid drawings of $G$ which are strictly bounded by $R_{h,w}$. If such a drawing does not exists, return ``no-instance''.
\end{definition}

In the {\sc Orthogonal Compaction} problem we get a connected graph $G$. We assume to have an order on the vertices, that is, for every $u,v\in V$ such that $u\neq v$, either $u>v$ or $v<u$. In addition to $G$, we have, for every $\{u,v\}\in E$ where $u>v$, the relative position of $v$ compered to $u$, that is, the {\em direction} of the $\{u,v\}$ from $u$ to $v$. We denote these directions by $\mathsf{U},\mathsf{D},\mathsf{L}$ and $\mathsf{R}$; this stands for ``up'', ``down'', ``left'' and ``right'', respectively. We assume that there exists a planar rectilinear grid drawing of $G$ such that for every $\{u,v\}\in E$, the relative position of $v$ compered to $u$ is as given as input. Our goal is to find such a drawing of minimum area. We start by defining the problem formally. For this purpose, we first have the following definition:

%\begin{definition}[{\bf Vertex Environment}]
%	Let $G$ be a connected graph with maximum degree $4$ and let $u\in V$. A function $\mathsf{Env}_u:\{\mathsf{N},\mathsf{S},\mathsf{E},\mathsf{W}\}\rightarrow N(u)\cup \{\mathsf{Null}\}$ is a {\em vertex environment} of $u$.
%\end{definition}

%\begin{definition}[{\bf Drawing Respects a Vertex Environment}]
%	Let $G$ be a connected graph with maximum degree $4$, let $u\in V$, let $\mathsf{Env}_u$ be a vertex environment of $u$ and let $d$ be an orthogonal grid drawing of $G$. We say that $d$ {\em respects $\mathsf{Env}_u$} if the following conditions are satisfied:
%\begin{enumerate}
%\item If $\mathsf{Env}_u(\mathsf{N})\neq \mathsf{Null}$, then $(d(u)+(1,0))\in d(\{u,\mathsf{Env}_u(N)\})$.
%\item If $\mathsf{Env}_u(\mathsf{S})\neq \mathsf{Null}$, then $(d(u)+(-1,0))\in d(\{u,\mathsf{Env}_u(S)\})$.
%\item If $\mathsf{Env}_u(\mathsf{W})\neq \mathsf{Null}$, then $(d(u)+(0,-1))\in d(\{u,\mathsf{Env}_u(W)\})$.
%\item If $\mathsf{Env}_u(\mathsf{E})\neq \mathsf{Null}$, then $(d(u)+(0,1))\in d(\{u,\mathsf{Env}_u(E)\})$.
%\end{enumerate}
%\end{definition}

%\begin{definition}[{\bf Edge Turning Description}]
%	Let $G$ be a connected graph with maximum degree $4$ and let $\{u,v\}\in E$. A sequence $T_{\{u,v\}}=(t_1\ldots,t_\ell)$ where $t_i\in \{\mathsf{U},\mathsf{D},\mathsf{L},\mathsf{R}\}$ for every $1\leq i\leq \ell$ is an {\em edge turning description} of $\{u,v\}$.
%\end{definition}

\begin{definition}[{\bf Drawing Respects an Edge Direction}]\label{def:dir}
	Let $G$ be a connected graph, let $\{u,v\}\in E$ such that $u>v$, and let $\mathsf{dir}_{\{u,v\}}\in \{\mathsf{U},\mathsf{D},\mathsf{L},\mathsf{R}\}$. Let $d$ be a rectilinear grid drawing of $G$. We say that $d$ {\em respects $\mathsf{dir}_{\{u,v\}}$} if the following conditions are satisfied
	\begin{enumerate}
		\item If $\mathsf{dir}_{\{u,v\}}=\mathsf{U}$, then $\fr(v)=\fr(u)$ and $\fc(v)>\fc(u)$.
		\item If $\mathsf{dir}_{\{u,v\}}=\mathsf{D}$, then $\fr(v)=\fr(u)$ and $\fc(v)<\fc(u)$.
		\item If $\mathsf{dir}_{\{u,v\}}=\mathsf{L}$, then $\fc(v)=\fc(u)$ and $\fr(v)<\fr(u)$.
		\item If $\mathsf{dir}_{\{u,v\}}=\mathsf{R}$, then $\fc(v)=\fc(u)$ and $\fr(v)>\fr(u)$.
	\end{enumerate}  
\end{definition}

Now, we define the problem {\sc Orthogonal Compaction} as follows:

\begin{definition}[{\bf Orthogonal Compaction Problem}]\label{def:OrtCom}
	Let $G$ be a connected graph. For every $\{u,v\}\in E$ let $\mathsf{dir}_{\{u,v\}}\in \{\mathsf{U},\mathsf{D},\mathsf{L},\mathsf{R}\}$. The {\sc Orthogonal Compaction} problem is to find a planar rectilinear grid drawing $d$ of $G$ such that (i) for every $\{u,v\}\in E$, $d$ respects $\mathsf{dir}_{\{u,v\}}$, and (ii) $d$ is strictly bounded by $R_{h,w}$ such that $(h-1)\cdot (w-1)$ is minimum. 
\end{definition}

%!TEX root =Main-Movement.tex

\section{The Concept of Drawn Decomposition}\label{sec:drawnSep}

%\begin{definition} [{\bf Interior of a Cycle}] \label{def:Frame}
%Let $G=(V,E)$ be a graph and let $f$ be a rectilinear drawing of $G$. Let $H\subseteq V$ is a {\em frame} of $G\setminus H$ in $f$ if  . Then the {\em boundary} of $V'$, denoted by $B(V')$, is the set of vertices of $V'$ that have a neighbor in $V\setminus V'$, i.e., $B(V')=\{v'\in V'~|~$ there exists $v\in V\setminus V'$ such that  $\{v,v'\} \in E \}$.   
%\end{definition}

%\begin{definition} [{\bf Drawn Walk}] \label{def:DrawnWalk}
%Let $G=(V,E)$ be a graph and let $f$ be a rectilinear drawing of $G$. A sequence of points $(p_1,\dots, p_k)$ is a {\em drawn walk} in $f$ if for every $0\leq i\leq k-1$, $\fr(p_i)=\fr(p_{i+1})$ or $\fc(p_i)=\fc(p_{i+1})$.
%\end{definition}

\begin{figure}[t]
	\centering
	\includegraphics[page=37, width=0.35\textwidth]{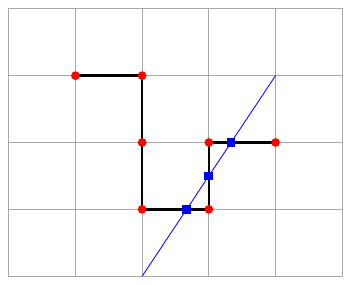}
	\caption{An example for $\gp(P)$ (drawn in red) and some of the points in $\gi(P)$ (drawn in blue) for the axis-parallel path $P$ drawn in black.}
	\label{fig:GI}
\end{figure}

 \medskip\noindent{\bf Frames.} To define a drawn decomposition of a given graph, we first define the concept of a {\em frame}, which is simply an axis-parallel cycle. Still, we will use the term frame since cycles identified as frames will play a special role in our decomposition. 

\begin{definition} [{\bf Frame}] \label{def:Frame}
A {\em frame} is an axis-parallel cycle $f=(p_1,\ldots, p_k)$.  
\end{definition}

Observe that every frame $f$ is a Jordan curve, so it divides the plane into an interior region $\ir(f)$, bounded by $f$, and an exterior region $\er(f)$, which is unbounded.  Moreover, $f$ is the boundary of both $\ir(f)$ and $\er(f)$. Let $f$ be a frame and let $p\in \mathbb{R}^2$ be a point.  We say that $p$ is {\em inside} (resp., {\em outside}) $f$ if $p\in \ir(f)$ ($p\in \er(f)$). We say that $p$ is {\em strictly inside} (resp., {\em strictly outside}) $f$ if $p\in \ir(f)$ ($p\in\er(f)$) but not on $f$. Similarly, given a frame $f$ and a drawn graph $G$ with a prescribed drawing $d$, we say that $d$ is {\em inside} (resp., {\em outside}) $f$ if every point $p\in\mathbb{R}^2$ on $d$ is inside (resp., {\em outside}) $f$. Moreover, we say that $d$ is {\em strictly inside} (resp., {\em strictly outside}) $f$ if every point $p\in\mathbb{R}^2$ on $d$ is strictly inside (resp., {\em strictly outside}) $f$.  

We denote by $\mathsf{Frames}$ the set of all frames.

Let $P$ be an axis-parallel path and let $f$ be a frame. We denote by $\gp(P)$ the set of intersection points of $P$ with the set of grid points (e.g., see the red points in Figure~\ref{fig:GI}). We denote by $\gps(f)$ the set of grid points that are strictly inside $f$. We denote by $\gi(P)$ the set of intersection points of every edge $e$ of $P$ with every line segment connecting two grid points that is not parallel to $e$ (e.g., see the blue points in Figure~\ref{fig:GI}). %In particular, see the line segment connecting two grid points (drawn in blue), and its intersection points (drawn in blue) with edges of the path (drawn in black).  
Observe that $\gp(P)\subseteq \gi(P)$. 

We denote by $\gis(f)$ the set of all the points in $\gi(P)$ for any axis-parallel path $P$ inside $f$.   

%We denote by $V_d(f)$ the set of vertices of $G$ that are drawn inside the interior (excluding the boundary) of $f$ in $d$. When $d$ is clear from the context, we refer to $V_d(f)$ as $V(f)$.    

%Let $f=(p_1,\ldots, p_k)$ be a frame, let $G=(V,E)$ be a graph and let $d$ be a segmented grid drawing of $G$. 

%We say that $p$ is {\em on} $f$, if $p\in \ir(f)\cap \er(f)$.

%Let $P=(p_1,\ldots, p_k)$ be an axis-parallel path, let $p\in \gi(P)$ and let $0<\epsilon<|(p_1,\ldots p)|$. We denote by $[p-\epsilon,p]$ the drawn walk of length $\epsilon$ that is on $w$ and $p$ is the endpoint that is closer to $p_k$. Similarly, for $0<\epsilon<|(p,\ldots p_k)|$, we denote by $[p,p+\epsilon]$ the drawn walk of length $\epsilon$ that is on $w$ and $p$ is the endpoint that is closer to $p_1$.       

%$\mathsf{point}_\epsilon(\ell(a,b))$.

%Let $a$ and $b$ be two points in $\mathbb{R}^2$ and let $\epsilon > 0$. We denote $\ell(a,a_\epsilon)$ by $\mathsf{line}_\epsilon(a,b)$, where $a_\epsilon$ is the point on the line $\ell(a,b)$ at distance $\epsilon$ from $a$ if it exists.

Let $G$ be a graph, let $d$ be a polyline grid drawing of $G$, let $f$ be a frame, and let $\{u,v\}$ be an edge of $G$ that intersects $f$. Now, we would like to consider some specific set of points on $d(\{u,v\})$ with respect to the frame $f$, that, roughly speaking, includes every point where the path $d(\{u,v\})$ ``goes'' from $\ir(f)$ to $\er(f)$, and where at least one of the containment is strict. In other words, the points included are the intersection points of $d(\{u,v\})$ and $f$ such that ``right after them'' or ``right before them'' there is no intersection between $d(\{u,v\})$ and $f$. Observe that these points belong to the set $\gi(d(\{u,v\}))\cap \gi(f)$. Remind that, for $a,b\in \mathbb{R}^2$ and $\epsilon > 0$, we denote $\ell(a,a_\epsilon)$ by $\mathsf{line}_\epsilon(a,b)$, where $a_\epsilon$ is the point on the line $\ell(a,b)$ at distance $\epsilon$ from $a$ if it exists. Formally, we have the following definition.

\begin{definition} [{\bf Turning Points of a Drawn Edge in a Frame}] \label{def:IntPoint}
Let $G=(V,E)$ be a graph and let $d$ be a polyline grid drawing of $G$. Let $f$ be a frame and let $\{u,v\}$ be an edge of $G$. First, $(d(u),\{u,v\})$ (and similarly $(d(v),\{u,v\})$) is a {\em turning point} in $f$ if $d(u)\in \gp(f)$ ($d(v)\in \gp(f))$ (see $(d(u_3),\{u_3,u_4\})$ in Figure~\ref{fig:TP}). Second, let $p\in \gi(d(\{u,v\}))\cap \gi(f)$ such that $p\notin \{d(u),d(v)\}$. Let $p_i$ and $p_j$ be the two vertices of the path $d(\{u,v\})=(p_1,p_2,\ldots,p_k)$ such that $j>i$, $p\notin \{p_i,p_j\}$, and $(p_i,\ldots,p_j)$ is the minimum size subpath of $d(\{u,v\})$ intersecting $p$. Then, $(p,\{u,v\})$ is a {\em turning point} in $f$ if there exists $\epsilon>0$ such that at least one of the following conditions is satisfied:
\begin{itemize}
	\item $\gi(\mathsf{line}_\epsilon(p,p_i))\cap \gi(f)=\{p\}$ (see $(p_3,\{u_1,u_2\})$ in Figure~\ref{fig:TP}). 
	\item $\gi(\mathsf{line}_\epsilon(p,p_j))\cap \gi(f)=\{p\}$ (see $(p_1,\{u_1,u_2\})$ in Figure~\ref{fig:TP}).
%\item $\gi(p^\epsilon_{p_i,p})\cap \gi(f)=\{p\}$ (see $(p_3,\{u_1,u_2\})$ in Figure~\ref{fig:TP}). 
%\item $\gi(p^\epsilon_{p,p_j})\cap \gi(f)=\{p\}$ (see $(p_1,\{u_1,u_2\})$ in Figure~\ref{fig:TP}).
\end{itemize}
\end{definition}

%Observe that, as $d(\{u,v\})$ is a straight-line path and $p\notin\{p_i,p_j\}$, either $j=i+1$ or $j=i+2$ and $p=p_{i+1}$. Therefore, $p^\epsilon_{p_i,p}$ and $p^\epsilon_{p,p_j}$ are well defined.
 
 Observe that, as $d(\{u,v\})$ is a straight-line path and $p\notin\{p_i,p_j\}$, either $j=i+1$ or $j=i+2$ and $p=p_{i+1}$. Therefore, $\mathsf{line}_\epsilon(p,p_i)$ and $\mathsf{line}_\epsilon(p,p_j)$ are well defined.

Let $G=(V,E)$ be a graph and let $d$ be a polyline grid drawing of $G$. Let $f$ be a frame. We denote by $\tp(f,d)$ the number of turning points of the edges of $d$ in $f$, that is, $\mathsf{\tp}(f,d)=|\{(p,\{u,v\})~|~\{u,v\}\in E, p\in \gi(d(\{u,v\}))\cap \gi(f), (p,\{u,v\})$ is a turning point in $f\}|$.
When $d$ is clear from the context, we refer to $\mathsf{\tp}(f,d)$ as $\mathsf{\tp}(f)$.

%\begin{observation} \label{obs:twoframes}
%Let $f=(p_1,\ldots,p_k=p_1)$ be a frame and let $p,p'\in \gp(f)$ be two points on the frame, such that. Then, there exist two axis-parallel paths $P^1(p_1,p_2)=(p_1,p_{a_1},\ldots p_{a_\ell},p_2)$ and $P^2(p_1,p_2)=(p_1,p_{b_1},\ldots p_{b_q},p_2)$ such that $\gp(P^1(p_1,p_2))\cap \gp(P^2(p_1,p_2))=\{p_1,p_2\}$ and $\gp(P^1(p_1,p_2))\cup \gp(P^2(p_1,p_2))=P(f)$.      
%\end{observation}

%When $p_1$ and $p_2$ are clear from the context, we we refer to $w^1(p_1,p_2)$ and $w^2(p_1,p_2)$ simply as $w^1$ and $w^2$.

%For a frame $F=(p_1,\ldots, p_k)$ and two points in the frame $p_1,p_2\in P(F)$, we denote by $W_1$ and $W_2$ the drawn walks $W^1_{\mathsf{min}}$ and $W^2_{\mathsf{min}}$ correspondingly, where $W^1$ and $W^2$ are the two drawn walks from Observation \ref{obs:twoframes}.     

 \medskip\noindent{\bf Cutters and Associated Vertices.} We now define an axis-parallel path $c$ that splits a frame $f$ into two smaller subframes. We call $c$ a {\em cutter of $f$}.

\begin{definition} [{\bf Cutter of a Frame}] \label{def:Cutter}
Let $f$ be a frame. An axis-parallel path $c=(c_1,\ldots,c_k)$ is a {\em cutter} of $f$ if $c$ is inside $f$ and $\pp(f)\cap\pp(c)=\{c_1,c_k\}$. 
\end{definition}

We say that two axis-parallel paths (or cycles) $P$ and $Q$ are {\em equal} if $\pp(P)=\pp(Q)$.
Let $P$ be an axis-parallel path (cycle). We denote by $P^{\mathsf{min}}$ the axis-parallel path (cycle) that is equal to $P$, and has the minimum number of vertices among the axis-parallel paths (cycles) that are equal to $P$. Observe that $P^{\mathsf{min}}$ is unique.

Observe that, due to the Definition \ref{def:Cutter}, given a frame $f=(p_1,\ldots,p_k)$ and a cutter $c=(c_1,\ldots,c_t)$ of $f$, either $c_1=p_i$ for some $1\leq i\leq k$, or there is no such an $i$ and $c_1$ is on the line segment $\ell(p_j,p_{j+1})$ for some $1\leq j< k$. For the latter case, observe that the frame $f'=(p_1,\ldots,p_j,c_1,p_{j+1},\ldots,p_k)$ is equal to $f$. A symmetric argument holds for $c_t$. Therefore, we denote by $\frc$ the frame equivalent to $f$ that contains exactly the vertices of $f$, $c_1$ and $c_t$.

\begin{definition} [{\bf Frames Associated with a Cutter}] \label{def:assFrameCutter}
Let $f$ be a frame and let $c=(c_1,\ldots,c_t)$ be a cutter of $f$. Let $\frc=(p_1\ldots,p_i=c_1,\ldots,p_j=c_t,\ldots, p_k)$. Let $f_1=(c_1,c_2,\ldots, c_t=p_{j},p_{j+1}\ldots,p_k=p_1,\ldots p_i=c_1)$ and $f_2=(c_1,\ldots,c_t=p_j,p_{j-1},\ldots,p_{i+1},$ $p_i=c_1)$ be the two frames obtained from $f$ and $c$. Let $p$ be the minimum grid point in the set $\gp(f)\setminus \{c_1,c_t\}$. Let $\tilde{f}$ be the frame in $\{f_1,f_2\}$ that contains $p$, and let $\hat{f}$ be the other frame in $\{f_1,f_2\}$. Then, the {\em two frames associated} with $f$ and $c$ are $\mathsf{AssoFr}_1(f,c)=\tilde{f}^{\mathsf{min}}$ and $\mathsf{AssoFr}_2(f,c)=\hat{f}^{\mathsf{min}}$.
\end{definition}

For simplicity, when no confusion arises, we denote $f_1(c)=\mathsf{AssoFr}_1(f,c)$ and $f_2(c)=\mathsf{AssoFr}_2(f,c)$.
Let $G=(V,E)$ be a graph and let $d$ be a polyline grid drawing of $G$. Let $f$ be a frame. For $\{u,v\}\in E$, we say that $f$ {\em separates} $\{u,v\}$ in $d$ (or $\{u,v\}$ is {\em separated} by $f$ in $d$) if one among $u$ and $v$ is drawn strictly outside $f$ in $d$, and the other one is drawn strictly inside $f$ in $d$. 
Next, we define the {\em set of vertices associated with a frame}:

\begin{definition} [{\bf Vertices Associated with a Frame}] \label{def:AssVer}
Let $G=(V,E)$ be a graph and let $d$ be a polyline grid drawing of $G$. Let $f$ be a frame. The set of {\em vertices associated} with $f$ in $d$, denoted by $\av(f)$, is the set of vertices of $G$ that are intersected by $f$ and the endpoints of edges that are separated by $f$ in $d$ (see Figure~\ref{fig:assFr}).   
\end{definition}

Let $G=(V,E)$ be a graph and let $d$ be a polyline grid drawing of $G$. Let $f$ be a frame. Let $c$ be a cutter of $f$ and let $\{u,v\}\in E$. We say that $(c,f)$ {\em separates} $\{u,v\}$ in $d$ (or $\{u,v\}$ is {\em separated} by $(c,f)$ in $d$) if $u$ and $v$ are drawn inside $f$ and there exists $i\in \{1,2\}$ such that one among $u$ and $v$ is drawn strictly inside $f_i(c)$ in $d$ and the other one is drawn strictly outside $f_i(c)$ in $d$. 

\begin{definition} [{\bf Vertices Associated with a Cutter}] \label{def:AssVerCutter}
Let $G=(V,E)$ be a graph and let $d$ be a polyline drawing of $G$. Let $f$ be a frame and let $c$ be a cutter of $f$. The {\em set of vertices associated} with $(c,f)$ in $d$, denoted by $\av(c,f)$, is the set of vertices of $G$ that are intersected by $c$ and of the endpoints of edges that are separated by $(c,f)$ in $d$ (see Figure~\ref{fig:assCu}).   
\end{definition}

\medskip\noindent{\bf Cost and Rectangular.} Now, we define the {\em cost} of a frame in a polyline grid drawing $d$ as the sum of the number of vertices of the frame, the number of vertices of $G$ on the frame, and the total number of turning points of any edge in the frame. This number will be in used to measure the quality of (or how ``complicated'' is) the decomposition defined later in this section.
 
\begin{definition} [{\bf Cost of a Frame}] \label{def:sizeFr}
Let $G=(V,E)$ be a graph, let $d$ be a polyline grid drawing of $G$ and let $f$ be a frame. The {\em cost} of $f$ in $d$, denoted by $\sizef(f)$, is the sum of the number of vertices of $f^{\mathsf{min}}$, the number of vertices of $G$ on $f$ with respect to $d$ and $\tp(f)$.   
\end{definition}

For later use, we would like define the {\em contribution of a vertex to $\sizef(f)$}.

\begin{definition} [{\bf Contribution of a Vertex to the Cost}] \label{def:Contr}
	Let $G=(V,E)$ be a graph, let $d$ be a polyline grid drawing of $G$, let $f$ be a frame and let $v\in V$. The {\em contribution of $v$ to $\sizef(f)$} is $a+\frac{1}{2}\cdot b$ such that:
	\begin{itemize}
		\item $a\in \{0,1\}$, where $a=1$ if and only if $v$ is drawn on $f$ in $d$.
		\item $b\in \mathbb{N}_0$ is the total number of turning points in $f$ of any edge having $v$ as one of its endpoints.
	\end{itemize} 
\end{definition}

Observe that $\sizef(f)$ can be expressed by the contributions of the every vertex in $G$: 

\begin{observation}\label{obs:Contr}
Let $G=(V,E)$ be a graph, let $d$ be a polyline grid drawing of $G$, and let $f$ be a frame. Then, 	$\sizef(f)$ is exactly the sum of contributions of the every vertex in $G$plus the number of vertices of $f^{\mathsf{min}}$.
\end{observation}  

Given a graph $G=(V,E)$ and a polyline grid drawing $d$ of $G$, we would like to construct a simple frame that strictly bounds $d$. For this purpose, we will consider the ``minimum size'' rectangle with this property. First, we denote the following:

\begin{itemize}
\item $r_{\mathrm{min}}=\mathrm{min}\{ r\in \mathbb{N}_0~|~$there exists $c\in \mathbb{N}_0$ such that $d^{-1}((r,c))\neq \emptyset \}$.
\item $c_{\mathrm{min}}=\mathrm{min}\{ c\in \mathbb{N}_0~|~$there exists $r\in \mathbb{N}_0$ such that $d^{-1}((r,c))\neq \emptyset \}$.
\item $r_{\mathrm{max}}=\mathrm{max}\{ r\in \mathbb{N}_0~|~$there exists $c\in \mathbb{N}_0$ such that $d^{-1}((r,c))\neq \emptyset \}$.
\item $c_{\mathrm{max}}=\mathrm{max}\{ c\in \mathbb{N}_0~|~$there exists $r\in \mathbb{N}_0$ such that $d^{-1}((r,c))\neq \emptyset \}$.
\end{itemize}

Without loss of generality, we assume that $r_{\mathrm{min}}=c_{\mathrm{min}}=1$, otherwise we can modify $d$ as follows. For every vertex $u\in V$, we replace $d(u)$ by $d(u)=(\fr(u)-{r}_{\mathrm{min}}+1,\fc(u)-{c}_{\mathrm{min}}+1)$, and change the drawing of the edges accordingly. Intuitively, this operation ``shifts'' the entire drawing $d$ of $G$ in parallel to the axes.
For a graph $G=(V,E)$ and a straight-line grid drawing $d$ of $G$, the {\em rectangular} of $G$ in $d$, denote by $R_d$ is the frame $((0,0),(0,c_{\mathrm{max}}+1),(r_{\mathrm{max}}+1,c_{\mathrm{max}}+1),(r_{\mathrm{max}}+1,0),(0,0))$. Observe that the set of vertices that are drawn strictly inside $R_d$ is $V$, and $\av(R_d)=\emptyset$ (e.g., see Figure~\ref{fig:Rd}). 

 \medskip\noindent{\bf Drawn Tree Decomposition.} We are now ready to define our decomposition of a polyline grid drawing $d$ of a graph $G$, called a {\em drawn tree decomposition}. The decomposition consists of a rooted binary tree ${\cal T}=(V_T,E_T)$ with root $v_r$, and a function $\alpha: V_T\rightarrow \mathsf{Frames}$ that associates each vertex of $\cal{T}$ with a frame, defined by induction. The frame associated with $v_r$ is $R_d$. In addition, every vertex $v\in V_T$ of $\cal{T}$ that is not a leaf, is associated with a cutter $c_v$ of $\alpha(v)$. Then, the frames associated with the two children of $v$ are $f_1(c_v)$ and $f_2(c_v)$, where $f=\alpha(v)$. Moreover, we include a tree decomposition $({\cal T},\beta)$ of $G$ such that : For internal vertex $v\in V_T$, $\beta(v)$ is the set of vertices associated with $\alpha(v)$ or with $(c_v,\alpha(v))$; for $v\in V_T$ that is a leaf, $\beta(v)$ is the set of vertices associated with $\alpha(v)$. Later, in Section \ref{sec:DrawnTreeFrame}, we will prove that the demand that $(\cal{T},\beta)$ is a tree decomposition can be dropped as it is implied by $(\cal{T},\alpha)$.

\begin{definition} [{\bf Drawn Tree Decomposition}] \label{def:DrawnDeco}
Let $G=(V,E)$ be a graph and let $d$ be a polyline grid drawing of $G$. A {\em drawn tree decomposition} of $d$ is a triple $({\cal T}=(V_T,E_T),\beta: V_T\rightarrow 2^V,\alpha: V_T\rightarrow \mathsf{Frames})$ where $({\cal T},\beta)$ is a tree decomposition of $G$, ${\cal T}$ is a binary rooted tree, and:
\begin{enumerate}
\item $\alpha(v_r)=R_d$, where $v_r$ is the root of ${\cal T}$.
%\item Every $v\in V_T$ that is not a leaf, has exactly two children $v_1$ and $v_2$.
\item For every internal vertex $v\in V_T$, there exists a cutter $c_v$ of $f$ such that $\alpha(v_1)=f_1(c_v)$ and $\alpha(v_2)=f_2(c_v)$, where $\alpha(v)=f$ and $v_1, v_2$ are the children of $v$ in ${\cal T}$. We say that $c_v$ is the {\em cutter associated with $v$}.   
\item A vertex $v\in V_T$ is a leaf if and only if there are no grid points in the interior of $\alpha(v)$. 
\item For every internal vertex $v\in V_T$, $\beta(v)=\av(\alpha(v))\cup \av(c_v,\alpha(v))$.  \label{def: DC Con:5} %, where $\alpha(v)=f$ and $c$ is the cutter associated with $v$. 
\item For every leaf $v\in V_T$, $\beta(v)=\av(\alpha(v))$.   \label{def: DC Con:6}%, where $\alpha(v)=f$.
\end{enumerate}  
The {\em width} of a drawn tree decomposition, denoted by $\mathsf{width}({\cal T},\beta,\alpha)$, is the maximum cost of a frame in $d$, i.e., $\mathsf{width}({\cal T},\beta,\alpha)=\max\{\sizef(\alpha(v))~|~v\in V_T\}$.
The {\em drawn treewidth} of $d$, denoted by $\mathsf{dtw}(d)$, is the minimum width of a drawn tree decomposition of $d$, i.e., $\mathsf{dtw}(d)=\min\{\mathsf{width}({\cal T},\beta,\alpha)~|~({\cal T},\beta,\alpha)$ is a drawn tree decomposition of $d\}$.    
\end{definition}

For every $v\in V_T$, we say that $\alpha(v)$ is the {\em frame associated with $v$}.  

Let $G=(V,E)$ be a graph and let $d$ be a polyline grid drawing of $G$, and let $f$ be a frame. We denote by $\mathsf{VerIn}(f)$ the set of vertices drawn strictly inside $f$ in $d$.

\begin{definition} [{\bf Balanced Drawn Tree Decomposition}] \label{def:BalaDrawnDeco}
Let $G=(V,E)$ be a graph and let $d$ be a polyline grid drawing of $G$. A {\em balanced drawn tree decomposition} of $G$ is a drawn tree decomposition $({\cal T},\beta,\alpha)$ where the following condition is satisfied:
For every non-root vertex $v\in V_T$ with parent $v'$, $|\mathsf{VerIn}(f)|\leq \frac{2}{3}|\mathsf{VerIn}(f')|$ where $\alpha(v)=f$ and $\alpha(v')=f'$.            
\end{definition}

%!TEX root =Main-Movement.tex

\section{Frame-Tree}\label{sec:DrawnTreeFrame}

We now present a seemingly simpler definition for our drawn tree decomposition, called {\em frame-tree}, that will be in use in this section. This definition includes the same structure of a binary tree and a frame associated with every vertex, but here we do not demand that the tree decomposition (specifically, $\beta$) is part of the drawn tree decomposition.

\begin{definition} [{\bf Frame-Tree}] \label{def:DrawnFrDeco}
Let $G=(V,E)$ be a graph and let $d$ be a polyline grid drawing of $G$. A {\em frame-tree} of $d$ is a pair $({\cal T}=(V_T,E_T),\alpha: V_T\rightarrow \mathsf{Frames})$ where $\cal T$ is a binary rooted tree, and:
\begin{enumerate}
\item $\alpha(v_r)=R_d$, where $v_r$ is the root of $\cal T$. \label{def:tfCon1}
\item For every internal vertex $v\in V_T$, there exists a cutter $c_v$ of $f$ such that $\alpha(v_1)=f_1(c_v)$ and $\alpha(v_2)=f_2(c_v)$, where $\alpha(v)=f$ and $v_1, v_2$ are the children of $v$ in ${\cal T}$. We say that $c_v$ is the {\em cutter associated with $v$}.   
\item A vertex $v\in V_T$ is a leaf if and only if there are no grid points in the interior of $\alpha(v)$. \label{def:tfCon3}            
\end{enumerate}   
\end{definition}

Observe that the conditions of a frame-tree (given in Definition \ref{def:DrawnFrDeco}), are identical to the first three conditions of a drawn tree decomposition (given in Definition \ref{def:DrawnDeco}). Moreover, the last two conditions of Definition \ref{def:DrawnDeco} concern the function $\beta$, which is not part of a frame-tree. So, it might seem like Definition \ref{def:DrawnDeco} is more restrictive than Definition \ref{def:DrawnFrDeco}. In particular, given a drawn tree decomposition of a polyline grid drawing $d$ of $G$, $({\cal T}=(V_T,E_T),\beta: V_T\rightarrow 2^V,\alpha: V_T\rightarrow \mathsf{Frames})$, it is easy to see that $({\cal T},\alpha)$ is a frame-tree of $d$:  

\begin{observation}\label{obs:DTDIsDTF}
	Let $G$ be a graph, let $d$ be a polyline grid drawing of $G$ and let $({\cal T},\beta,\alpha)$ be a drawn tree decomposition of $d$. Then, $({\cal T},\alpha)$ is a frame-tree of $d$.  
\end{observation}

 In what follows, we show the opposite direction of Observation \ref{obs:DTDIsDTF}. That is, for every frame-tree $({\cal T},\alpha)$ of $d$, let $({\cal T},\beta,\alpha)$ where $\beta(v)=\av(\alpha(v))\cup \av(c_v,\alpha(v))$ if $v$ is not a leaf, and $\beta(v)=\av(\alpha(v))$ if $v$ is a leaf. Then, $({\cal T},\beta)$ is a tree decomposition of $G$, and Conditions \ref{def: DC Con:5} and \ref{def: DC Con:6} of  Definition \ref{def:DrawnDeco} are satisfied. Therefore, $({\cal T},\beta,\alpha)$ is a drawn tree decomposition of $d$.

%For the sake of simplicity, given two drawn walks $w_1$ and $w_2$, we denote the sets of points, $P(w_1\cup w_2)$ and $P(w_1\cap w_2)$ by $w_1\cup w_2$ and $w_1\cap w_2$, respectively.  

In order to prove that $({\cal T},\beta)$ is indeed a tree decomposition of $G$, we first show that for every $\{u,u'\}\in E$, there exists $v\in V_T$ such that $u,u'\in \beta(v)$:

\begin{lemma}\label{lem:DtwIsTw1}
	Let $G=(V,E)$ be a graph and let $d$ be a polyline grid drawing of $G$. Let $({\cal T}=(V_T,E_T),\alpha: V_T\rightarrow \mathsf{Frames})$ be a frame-tree of $d$. Let $\beta: V_T\rightarrow 2^{V}$ be the function defined as follows: $\beta(v)=\av(\alpha(v))\cup \av(c_v,\alpha(v))$ for every $v\in V_T$ that is not a leaf, and $\beta(v)=\av(\alpha(v))$ for every $v\in V_T$ that is a leaf. Then, for every $\{u,u'\}\in E$ there exists $v\in V_T$ such that $u,u'\in \beta(v)$. In addition, for every $u\in V$ there exists $v\in V_T$ such that $u\in \beta(v)$.  
\end{lemma}

\begin{proof}
	First, we show that for every $\{u,u'\}\in E$ there exists $v\in V_T$ such that $u,u'\in \beta(v)$. Let $\{u,u'\}\in E$. Let $v\in V_T$ be a lowest vertex in ${\cal T}$ such that $u$ and $u'$ are drawn inside $\alpha(v)$, that is, for every descendant $\ell\in V_T$ of $v$ in ${\cal T}$, it follows that at least one among $u$ and $u'$ is not drawn inside $\alpha(\ell)$. Observe that $u$ and $u'$ are both drawn inside $R_d$. So, the set of vertices $x\in V_T$ in ${\cal T}$ such that $u$ and $u'$ are drawn inside $\alpha(x)$ is not empty, and hence the choice of $v$ is well defined. If both $u$ and $u'$ are drawn on $\alpha(v)$, then $u,u'\in \av(\alpha(v))\subseteq\beta(v)$. In this case, there exists $v\in V_T$ such that $u,u'\in \beta(v)$, and we are done. So, we next suppose that at least one of $u$ and $u'$ is drawn in the strict interior of $\alpha(v)$. Therefore, $v$ is not a leaf.
	
	 So, there exists a cutter $c_v$ of $\alpha(v)=f$ such that $\alpha(v_1)=f_1(c)$ and $\alpha(v_2)=f_2(c)$, where $v_1$ and $v_2$ are the children of $v$. We assume, without loss of generality, that $u$ is drawn in the strict interior of $\alpha(v)$. If the cutter $c_v$ intersects $u$ or $u'$, then both $u$ and $u'$ are drawn inside $\alpha(v_1)$ or $\alpha(v_2)$, a contradiction to $v$ being lowermost. Therefore, $c_v$ intersects neither $u$ nor $u'$. So, $u$ is drawn strictly inside $\alpha(v_1)$ or $\alpha(v_2)$; without loss of generality, assume that $u$ is drawn strictly inside $\alpha(v_1)$. In addition, due to $v$ being a lowest vertex in $\cal{T}$ such that $u$ and $u'$ are drawn inside $\alpha(v)$, $u'$ is not drawn inside $\alpha(v_1)$. Therefore, $u'$ is drawn inside $\alpha(v_2)$, but not inside $\alpha(v_1)$ (see Figure~\ref{fig:firstPart}). Hence, $(c,f)$ separates $\{u,u'\}$, and so $u,u'\in \av(c_v,\alpha(v))\subseteq \beta(v)$. Thus, we get that there exists $v\in V_T$ such that $u,u'\in \beta(v)$, and hence we are done in this case as well.
	 
	 Now, we show that for every $u\in V$ there exists $v\in V_T$ such that $u\in \beta(v)$. This part of the proof is similar to the first part of the proof. Let $u\in V$. Let $v\in V_T$ be a lowest vertex in ${\cal T}$ such that $u$ is drawn inside $\alpha(v)$, that is, for every descendant $\ell\in V_T$ of $v$ in ${\cal T}$, it follows that $u$ is not drawn inside $\alpha(\ell)$. Observe that $u$ is drawn inside $R_d$. So, the set of vertices $x\in V_T$ in ${\cal T}$ such that $u$ is drawn inside $\alpha(x)$ is not empty, and hence the choice of $v$ is well defined. We first assume that $u$ is drawn strictly inside $\alpha(v)$. In this case, $v$ is not a leaf. Then, there exists a cutter $c_v$ of $\alpha(v)=f$ such that $\alpha(v_1)=f_1(c)$ and $\alpha(v_2)=f_2(c)$, where $v_1$ and $v_2$ are the children of $v$. Observe that if $c$ intersects $u$, then $u$ is drawn inside both $\alpha(v_1)$ and $\alpha(v_2)$, a contradiction to $v$ being lowermost. Otherwise, $c$ does not intersect $u$, so $u$ is drawn inside one of $\alpha(v_1)$ and $\alpha(v_2)$; again, this is a contradiction to $v$ being lowermost. Either way, we get a contradiction when we assume that $u$ is drawn strictly inside $u$. Therefore, $u$ is drawn on $\alpha(v)$, and then $u\in \av(\alpha(v))\subseteq\beta(v)$. Therefore, we get that there exists $v\in V_T$ such that $u,\in \beta(v)$, and we are done. 
\end{proof}

\begin{figure}[t]
\centering
    \includegraphics[page=9, width=0.25\textwidth]{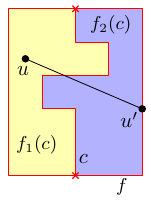}
    \caption{The figure shows a case where $u$ and $u'$ are drawn inside $f$, such that $u$ is strictly inside $f_1(c)$ and $u'$ is strictly outside $f_1(c)$. Therefore, $(c,f)$ separates $\{u,u'\}$.}
  \label{fig:firstPart}
\end{figure}         

\begin{figure}[!t]
	\centering
	\begin{subfigure}{0.3\textwidth}
		\includegraphics[width = \textwidth, page = 15]{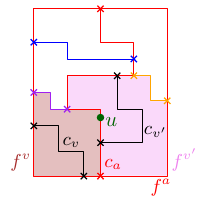}
		\subcaption{}
		\label{case1:a}
	\end{subfigure}
	\hfil
	\begin{subfigure}{0.32\textwidth}
		\includegraphics[width = \textwidth, page = 16]{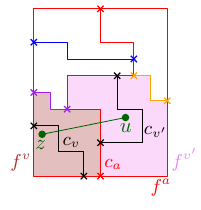}
		\subcaption{}
		\label{case1:b}
	\end{subfigure}
	\begin{subfigure}{0.36\textwidth}
		\includegraphics[width = \textwidth, page = 17]{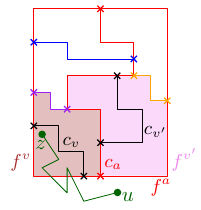}
		\subcaption{}
		\label{case1:c}
	\end{subfigure}

	\caption{Examples of some cases where $a$ is an ancestor of both $v$ and $v'$ in the frame-tree: (a) $u$ is drawn on $(f^v\cup c_v)\cap (f^{v'}\cup c_{v'})$; (b) $u$ is drawn strictly inside $f^a_2(c_a)$ and $z$ is drawn strictly inside $f^v$; (c) $u$ is drawn strictly outside $f^a$ and $z$ is drawn strictly inside $f^v$.} 
	\label{fi:Case1}
\end{figure}

%First, we show that for every $\{u,u'\}\in E$ there exists $v\in V_T$ such that $u,u'\in \beta(v)$. Let $\{u,u'\}\in E$. Let $v\in V_T$ be a minimal vertex in ${\cal T}$ such that $u$ and $u'$ are drawn inside $\alpha(v)$, that is, for every descendant $\ell\in V_T$, of $v$, it follows that $u$ or $u'$ is not drawn inside $\alpha(\ell)$ (observe that $u$ and $u'$ are drawn inside $R_f$, so the set of vertices such that $u$ and $u'$ are drawn inside the frame associated with them, is not empty). If $u$ and $v$ are drawn on $\alpha(v)$, then $u,u'\in S(\alpha(v))=\beta(v)$. If not, then at least one of $u$ and $u'$ is drawn inside the interior of $\alpha(v)$, and therefore, $v$ is not a leaf. So, there exists a cutter $c$ of $\alpha(v)=f$ such that $\alpha(v_1)=f_1(c)$ and $\alpha(v_2)=f_2(c)$, where $v_1$ and $v_2$ are the children of $v$. If the cutter $c$ intersects $u$ or $u'$, then $u$ and $u'$ are drawn inside $f_1(c)$ or $f_2(c)$, a contradiction to the minimality of $v$. Therefore, $c$ does not intersect $u$ or $u'$. Assume, without loss of generality, that $u$ is drawn strictly inside $f_1(c)$. In addition, due to the assumption of the minimality of $v$, $u$ and $u'$ are not drawn inside $f_1(c)$. Therefore, $u'$ is drawn inside $f_2(c)$, but not inside $f_1(c)$ (see Figure~\ref{fig:firstPart}). Therefore, $(c,f)$ separates $\{u,u'\}$, so $u,u'\in S(c_v)\subseteq \beta(v)$. Therefore, we get that there exists $v\in V_T$ such that $u,u'\in \beta(v)$.

Now, we will show that $\beta$ satisfies another property of a tree decomposition. Specifically, we will show that for every $u\in V$, the induced subgraph ${\cal T}[V_T^u]$, where $V_T^u=\{v\in V_T~|~u\in \beta(v)\}$, is connected. Observe that Lemma \ref{lem:DtwIsTw1} already asserts that ${\cal T}[V_T^u]$ is non-empty. 

\begin{lemma}\label{lem:DtwIsTw2} 
		Let $G=(V,E)$ be a graph and let $d$ be a polyline grid drawing of $G$. Let $({\cal T}=(V_T,E_T),\alpha: V_T\rightarrow \mathsf{Frames})$ be a frame-tree of $d$. Let $\beta: V_T\rightarrow 2^{V}$ be the function defined as follows: $\beta(v)=\av(\alpha(v))\cup \av(c_v)$ for every $v\in V_T$ that is not a leaf, and $\beta(v)=\av(\alpha(v))$ for every $v\in V_T$ that is a leaf. Then, for every $u\in V$, the induced subgraph ${\cal T}[V_T^u]$, where $V_T^u=\{v\in V_T~|~u\in \beta(v)\}$, is connected.  
\end{lemma}

\begin{proof}
Let $u\in V$. Assume, towards a contradiction, that ${\cal T}[V_T^u]$ is not connected. Let $v,v'\in V_T^u$ be two vertices belonging to different connected components of ${\cal T}[V_T^u]$. Let $a\in V_T$ be a vertex on the shortest path from $v$ to $v'$ in ${\cal T}$ such that $u\notin \beta(a)$. Then, $a$ is an ancestor of at least one among $v$ and $v'$. Next, we consider two cases.

\smallskip\noindent{\bf Case 1.} Assume first that $a$ is an ancestor of both $v$ and $v'$ (see Figure~\ref{fi:DTD}). Since $a$ is a vertex on the shortest path from $v$ to $v'$ in ${\cal T}$, then $a$ is the lowest common ancestor of $v$ and $v'$. Let $\alpha(a)=f^a$. Now, by the definition of frame-tree, $a$ has exactly two children, $a_1$ and $a_2$, and there exists a cutter $c_a$ of $f^a$ such that $\alpha(a_1)=f^a_1(c_a)$ and $\alpha(a_2)=f^a_2(c_a)$ (see Figure~\ref{d1}). Since $a$ is the lowest common ancestor of $v$ and $v'$, we get that one of $a_1$ or $a_2$ is an ancestor of $v$ (or equal to $v$), and the other one is an ancestor of $v'$ (or equal to $v'$). Without loss of generality, assume that $a_1$ is an ancestor of $v$ (or equal to $v$), and $a_2$ is an ancestor of $v'$ (or equal to $v'$) (see Figure~\ref{fi:DTD}). Observe that $\gp(\alpha(a_1))\cap \gp(\alpha(a_2))\subseteq \gp(\alpha(a))\cup \gp(c_a)$ (see Figure~\ref{d1}), and, in turn, also $(\gp(\alpha(v))\cup \gp(c_v))\cap (\gp(\alpha(v'))\cup \gp(c_{v'}))\subseteq \gp(\alpha(a))\cup \gp(c_a)$ (see Figure~\ref{d5}). Therefore, if $u$ is drawn on both $\alpha(v)\cup c_v$ and $\alpha(v')\cup c_{v'}$, we get that $u$ is drawn also $\alpha(a)\cup c_a$ (see Figure~\ref{case1:a}). Thus, by definition, $u\in \beta(a)$, a contradiction to the assumption that $u\notin \beta(a)$. Therefore, $u$ is not drawn on at least one among $\alpha(v)\cup c_v$ and $\alpha(v')\cup c_{v'}$. Assume, without loss of generality, that $u$ is not drawn $\alpha(v)\cup c_v$. Because $u\in \beta(v)$ and since $u$ is not drawn on $\alpha(v)$, there must be an edge $\{u,z\}\in E$ such that $\{u,z\}$ is separated by $\alpha(v)$ or by $(c_v,\alpha(v))$ in $d$. We have the following cases.
\begin{itemize}
	\item  If $u$ is drawn strictly outside $\alpha(a_1)$, then $\{u,z\}$ is not separated by $(c_v,\alpha(v))$; so, it is separated by $\alpha(v)$. Therefore, $z$ is drawn strictly inside $\alpha(v)$. We have the following two sub-cases.
	\begin{itemize}
		\item If $u$ is drawn inside $\alpha(a_2)$, then $(c_a,\alpha(a))$ separates $\{u,z\}$. So, by definition, $u\in \beta(a)$, a contradiction to the assumption that $u\notin \beta(a)$ (see Figure~\ref{case1:b}).
		\item Otherwise, $u$ is drawn strictly outside $\alpha(a_2)$. Then, it is drawn strictly outside $\alpha(a)$. So, $\alpha(a)$ separates $\{u,z\}$, and then, by definition, $u\in \beta(a)$. This is, again, a contradiction to the assumption that $u\notin \beta(a)$ (see Figure~\ref{case1:c}).
	\end{itemize}
\item Otherwise, $u$ is drawn inside $\alpha(a_1)$. Then, $u$ is drawn strictly inside $\alpha(a_1)$, since $\alpha(a_1)\subseteq \alpha(a)\cup c_a$, and we assumed that $u\notin \beta(a)$. Thus, $u$ is drawn strictly outside $\alpha(a_2)$. So, this case is symmetric to the previous case, where $u$ is drawn strictly outside $\alpha(a_1)$. Therefore, we get a contradiction in this case as well.
\end{itemize}

\begin{figure}[!t]
	\centering
	\includegraphics[page=19, width=0.31\textwidth]{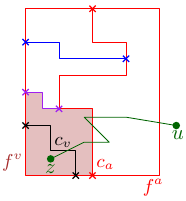}
	\caption{Example of a case where $a$ is an ancestor of $v$, $u$ is drawn strictly outside $f^a$ and $z$ is drawn strictly inside $f^v$. The blue and the purple cutters correspond to vertices on the path from $a$ to $v$ in the corresponding frame-tree.}
	\label{fig:case2}
\end{figure}

\smallskip\noindent{\bf Case 2.}  Now, assume that $a$ is an ancestor of $v$ but not an ancestor of $v'$ (the other case, where $a$ is an ancestor of $v'$ but not an ancestor of $v$, is symmetric). We have the following cases. 

\smallskip\noindent{\bf Case 2(i).} Assume that $u$ is drawn strictly outside $\alpha(a)$. Since $\alpha(v)$ and $c_v$ are bounded by $\alpha(a)$, then $u$ is drawn on neither $\alpha(v)$ nor $c_v$. Since $u\in \beta(v)$, there must be an edge $\{u,z\}\in E$ such that $\{u,z\}$ is separated by $\alpha(v)$ or by $(c_{v},\alpha(v))$ in $d$. Now, $u$ is not drawn inside $\alpha(v)$. So, $\{u,z\}$ is not separated by $(c_{v},\alpha(v))$, which means that $\{u,z\}$ is separated by $\alpha(v)$. Therefore, $z$ is drawn strictly inside $\alpha(v)$; so, $z$ is drawn strictly inside $\alpha(a)$. We get that $\{u,z\}$ is separated by $\alpha(a)$ (see Figure~\ref{fig:case2}), and then, by the definition of $\beta$, $u\in \beta(a)$, a contradiction to the assumption that $a\notin V_T^u$. 

\smallskip\noindent{\bf Case 2(ii).} Assume that $u$ is drawn inside $\alpha(a)$. If $u$ is drawn on $\alpha(a)$, then, by the definition of $\beta$, $u\in \beta(a)$, a contradiction to the assumption that $a\notin V_T^u$. So, $u$ is drawn strictly inside $\alpha(a)$. Now, observe that since $a$ is not an ancestor of $v'$, then either $v'$ is an ancestor of $a$, or $v'$ and $a$ are incomparable. 
\begin{itemize}
	\item If $v'$ is an ancestor of $a$, then, since $u$ is drawn strictly inside $\alpha(a)$, $u$ is drawn on neither $\alpha(v')$ nor $c_{v'}$.
	\item  If $v'$ and $a$ are incomparable, then, again, since $u$ is drawn strictly inside $\alpha(a)$, $u$ is drawn on neither $\alpha(v')$ nor $c_{v'}$. 
\end{itemize}
Therefore, either way, we get that $u$ is drawn on neither $\alpha(v')$ nor $c_{v'}$. 

 Now, since $u\in \beta(v')$, there must be an edge $\{u,z\}\in E$ such that $\{u,z\}$ is separated by $\alpha(v')$ or by $(c_{v'},\alpha(v'))$ in $d$.  We have the following cases.

 \begin{itemize}
 	\item If $z$ is drawn strictly outside $\alpha(a)$, then we get that $\{u,z\}$ is separated by $\alpha(a)$; so, by definition, $u\in \beta(a)$, a contradiction to the assumption that $a\notin V_T^u$ (see Figure~\ref{fi:Case3}).
 	\item Otherwise, $z$ is drawn inside $\alpha(a)$. We have the following two sub-cases.
 	
 	 \begin{figure}[t]
 		\centering
 		\begin{subfigure}{0.3\textwidth}
 			\includegraphics[width = \textwidth, page = 20]{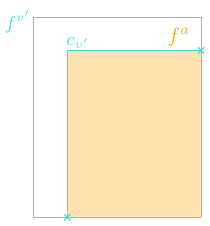}
 			\subcaption{}
 			\label{a}
 		\end{subfigure}
 		\hfil
 		\begin{subfigure}{0.36\textwidth}
 			\includegraphics[width = \textwidth, page = 21]{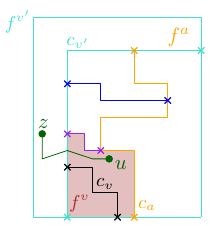}
 			\subcaption{}
 			\label{c}
 		\end{subfigure}
 		
 		\caption{Example of a case where $a$ is an ancestor of $v$, $v'$ is an ancestor of $a$, $u$ is drawn inside $f^a$ and $z$ is drawn strictly outside $f^a$. The blue and purple cutters correspond to vertices on the path from $a$ to $v$ in the corresponding frame-tree.} 
 		\label{fi:Case3}
 	\end{figure}
 	
 	\begin{itemize}
 		\item If $v'$ and $a$ are incomparable, then, since $u$ and $z$ are drawn inside $\alpha(a)$, we get that $\{u,z\}$ is separated by neither $\alpha(v')$ nor $(c_{v'},\alpha(v'))$. This is a contradiction to the assumption that $\{u,z\}$ is separated by either $\alpha(v')$ or $(c_{v'},\alpha(v'))$. 
 		\item Otherwise, $v'$ is an ancestor of $a$. Let $v'_1$ and $v'_2$ be the two children of $v'$. Since $u$ and $z$ are drawn inside $\alpha(v)$, and since $\alpha(v)$ is bounded by $\alpha(v')$, $u$ and $z$ are drawn inside $\alpha(v')$. So, $\{u,z\}$ is not separated by $\alpha(v')$, which means that $\{u,z\}$ is separated by $(c_{v'},\alpha(v'))$. Therefore, by the definition of a cutter being separator of an edge, there exists $i\in\{1,2\}$ such that one of $u$ or $z$ is drawn strictly inside $\alpha(v'_i)$ (and hence it is not drawn on $\alpha(v'_j)$ for $j\in \{1,2\}, j\neq i$), and the other vertex is not drawn inside $\alpha(v'_i)$. So, one among $u$ and $z$ is not drawn inside $\alpha(v'_1)$, and one among $u$ and $z$ is not drawn inside $\alpha(v'_2)$. Now, since $v'$ is an ancestor of $a$, $\alpha(v)$ is bounded by either (and possibly equal to) $\alpha(v'_1)$ or $\alpha(v'_2)$. So, one among $u$ and $z$ is not drawn inside $\alpha(v)$, a contradiction to the assumption that $u$ and $z$ are both drawn inside $\alpha(v)$.                                                                                       
 	\end{itemize}
 \end{itemize}  
This completes the proof.
\end{proof}

Now, from Lemmas \ref{lem:DtwIsTw1} and \ref{lem:DtwIsTw2} we can conclude that $({\cal T},\beta)$ is a tree decomposition of $G$, as we state in the following corollary.

\begin{corollary} \label{col:TWBoundedBYdraTW}
	Let $G=(V,E)$ be a graph and let $d$ be a polyline grid drawing of $G$. Let $({\cal T}=(V_T,E_T),\alpha: V_T\rightarrow \mathsf{Frames})$ be a frame-tree of $d$. Let $\beta: V_T\rightarrow 2^{V}$ be the function defined as follows: $\beta(v)=\av(\alpha(v))\cup \av(c_v,\alpha(v))$ for every $v\in V_T$ that is not a leaf, and $\beta(v)=\av(\alpha(v))$ for every $v\in V_T$ that is a leaf. Then, $(\cal{T},\beta)$ is a tree decomposition of $G$. 
\end{corollary}

\section{Relation to Treewidth}\label{sec:relToTreewidth}

Equipped with Corollary \ref{col:TWBoundedBYdraTW}, we would like to prove a lower bound on the drawn treewidth of a drawing based on the treewidth of the graph. Let $G$ be a graph, let $d$ be a polyline grid drawing of $G$, and let $({\cal T}=(V_T,E_T),\alpha: V_T\rightarrow \mathsf{Frames})$ be a frame-tree of $d$. In Lemma \ref{lem:dtwIsLBoundedbyTw} ahead, we show that for every $v\in V_T$ that is not a leaf, it follows that $|\beta(v)|\leq 2\cdot(\sizef(\alpha(v))+\sizef(\alpha(v_1))+\sizef(\alpha(v_2)))$, where $v_1$ and $v_2$ are the two children of $v$ in $\cal{T}$. To this end, we prove that every $u\in\beta(v)$ contributes at least $\frac{1}{2}$ to the sum $\sizef(\alpha(v))+\sizef(\alpha(v_1))+\sizef(\alpha(v_2))$. In addition, for every $v\in V_T$ that is a leaf, we show that $|\beta(v)|\leq 2\cdot\sizef(\alpha(v))$. We are now ready to prove Lemma \ref{lem:dtwIsLBoundedbyTw}.

\begin{lemma}\label{lem:dtwIsLBoundedbyTw}
	Let $G=(V,E)$ be a graph and let $d$ be a polyline grid drawing of $G$. Let $({\cal T}=(V_T,E_T),\alpha: V_T\rightarrow \mathsf{Frames})$ be a frame-tree of $d$.	Let $(\cal{T},\beta)$ be the tree decomposition of $G$ defined in Corollary \ref{col:TWBoundedBYdraTW}. Let $v\in V_T$. If $v$ is an internal vertex, then $|\beta(v)|\leq 2\cdot(\sizef(\alpha(v))+\sizef(\alpha(v_1))+\sizef(\alpha(v_2)))$, where $v_1$ and $v_2$ are the two children of $v$ in $\cal{T}$. Otherwise, $v$ is a leaf, and then $|\beta(v)|\leq 2\cdot\sizef(\alpha(v))$.
\end{lemma}

\begin{proof}
Assume that $v$ is not a leaf, and let $v_1$ and $v_2$ be the two children of $v$ in $\cal{T}$ (the proof where $v$ is a leaf is similar). By the definition of the function $\beta$, $\beta(v)=\av(\alpha(v))\cup \av(c_v,\alpha(v))$. First, let $u\in \av(\alpha(v))$. Then, by the definition of the vertices associated with a frame (see Definition \ref{def:AssVer}), we have two cases.
	\begin{itemize}
		\item $u$ is drawn on $\alpha(v)$ in $d$. In this case, by the definition of the contribution of a vertex to $\sizef(\alpha(v))$ (see Definition \ref{def:Contr}), $u$ contributes at least $1$ to $\sizef(\alpha(v))$.
		\item There exists an edge $\{u,z\}$ that is separated by $\alpha(v)$ in $d$. Therefore, one among $u$ and $z$ is drawn strictly inside $\alpha(v)$ and the other is drawn strictly outside $\alpha(v)$, in $d$. Observe that, in this case, there must be at least one point $p\in d(\{u,z\})$ such that $(p,\{u,z\})$ is a turning point in $\alpha(v)$. So, in this case, $u$ contributes at least $\frac{1}{2}$ to $\sizef(\alpha(v))$. 
	\end{itemize}

Second, let $u\in \av(c_v,\alpha(v))$. Then, by the definition of the vertices associated with a cutter (see Definition \ref{def:AssVerCutter}), we have two cases. 
\begin{itemize}
		\item $u$ is drawn on $c_v$ in $d$. So, $u$ is drawn on both $\alpha(v_1)$ and $\alpha(v_2)$, in $d$. Therefore, in this case, by the definition contribution, $u$ contributes at least $1$ to $\sizef(\alpha(v_1))$ and at least $1$ to $\sizef(\alpha(v_2))$. 
		\item There exists an edge $\{u,z\}$ that is separated by $(c_v,\alpha(v))$ in $d$. Therefore, by the definition of the separation of an edge by a cutter, there exists $i\in \{1,2\}$ such that one among $u$ and $z$ is drawn strictly inside $\alpha(v_i)$ in $d$, and the other one is drawn strictly outside $\alpha(v_i)$ in $d$. Hence, by the definition of the separation of an edge by a frame, $\alpha(v_i)$ separates $\{u,z\}$ in $d$. In this case, notice that there must be at least one turning point of the edge $\{u,z\}$ in $\alpha(v_i)$. So, by the definition of contribution, in this case, $u$ contributes at least $\frac{1}{2}$ to $\sizef(\alpha(v_i))$. 
	\end{itemize}

In conclusion, we have proved that every $u\in \beta(v)$ that is an internal vertex contributes at least $\frac{1}{2}$ to the sum $\sizef(\alpha(v))+\sizef(\alpha(v_1))+\sizef(\alpha(v_2))$. So, $|\beta(v)|\leq 2\cdot(\sizef(\alpha(v))+\sizef(\alpha(v_1))+\sizef(\alpha(v_2)))$.
\end{proof}

Recall that the width of a frame-tree $({\cal T},\alpha)$ is $\mathsf{width}({\cal T},\alpha)=\max\{\sizef(\alpha(v))~|~v\in V_T\}$. In addition, the width of a tree decomposition $(\cal {T},\beta)$ is $\mathsf{width}({\cal T},\beta)=\max_{v\in V_T}|\beta(v)|-1$.  From Lemma \ref{lem:dtwIsLBoundedbyTw}, we conclude that, for every $v\in V_T$, $|\beta(v)|\leq 6\cdot \mathsf{width}({\cal T},\alpha)$. So, $\mathsf{width}({\cal T},\beta)\leq 6\cdot \mathsf{width}({\cal T},\alpha)$ (for any frame-tree $({\cal T},\alpha)$ and its corresponding tree decomposition $({\cal T},\beta)$). Thus, $\frac{1}{6}\cdot\mathsf{tw}(G)\leq \mathsf{dtw}(d)$. In particular, we have the following corollary.  

\begin{corollary}
	Let $G=(V,E)$ be a graph and let $d$ be a polyline grid drawing of $G$. Then, $\mathsf{tw}(G)\leq \OO(\mathsf{dtw}(d))$. 
\end{corollary}

%	Let $G=(V,E)$ be a graph and let $d$ be a segmented grid drawing of $G$. Let $({\cal T}=(V_T,E_T),\alpha: V_T\rightarrow \mathsf{Frames})$ be a drawn tree-frame of $d$. Let $\beta: V_T\rightarrow 2^{V}$ be the function defined as follows: $\beta(v)=\av(\alpha(v))\cup \av(c_v,\alpha(v))$ for every $v\in V_T$ that is not a leaf, and $\beta(v)=\av(\alpha(v))$ for every $v\in V_T$ that is a leaf. Then, for every $\{u,u'\}\in E$ there exists $v\in V_T$ such that $u,u'\in \beta(v)$. In addition, for every $u\in V$ there exists $v\in V_T$ such that $u\in \beta(v)$. We show that $\mathsf{width}({\cal T},\beta)\leq 4\cdot\mathsf{width}({\cal T},\alpha)$. Let $v\in V_T$. Assume that $v$ is not a leaf. Then, $\beta(v)=\av(\alpha(v))\cup \av(c_v,\alpha(v))$. Now, let $u\in \av(\alpha(v))$. Then, by the definition of the vertices associated with a frame, we have two cases.
%\begin{itemize}
%	\item $u$ is drawn on $\alpha(v)$. In this case, by the definition of the size of a frame, $u$ contributes $1$ to $\mathsf{size}(\alpha(v))$.
	%\item There exists an edge $\{u,z\}$ that is separated by $\alpha(v)$ in $d$. Therefore, one among $u$ and $z$ is drawn strictly inside $\alpha(v)$ and the other is drawn strictly outside $\alpha(v)$, in $d$. Observe, that in this case, there must be at least one turning point of the edge $\{u,z\}$ in $\alpha(v)$. Therefore, $u$ and $z$ contribute at least $1$ to $\mathsf{size}(\alpha(v))$.  
%\end{itemize}

\section{Problem Solving Scheme Using Drawn Tree Decompositions} \label{sec:ProbSche}

In this section, we present our scheme for solving problems, based on our notion of a drawn tree decomposition. This scheme will be useful for various graph drawing and recognition problems. %More preciously, we present a general algorithm, that gets from the user a problem specific ``information'', and solve the problem according to this ``information''. 

Briefly, let $G=(V,E)$ be a graph and let $f_{\mathsf{init}}$ be a frame. Assume that we would like to construct a polyline grid drawing $d$ of $G$ such that $d$ is strictly bounded by $f_{\mathsf{init}}$ and has some specific properties, or return ``no'' if no such a polyline grid drawing exists. 
Our scheme works as follows. Observe that if such a drawing $d$ exists, it has a frame-tree of width $\mathsf{dtw}(d)$. We iterate over the potential values $k$ for $\mathsf{dtw}(d)$, and aim to ``guess'' a drawing $d'$ of $G$ and a frame-tree $({\cal T},\alpha)$ of $d'$ with $\mathsf{width}(d',({\cal T},\alpha))\leq k$.

For a given $k$, we try to ``build'' the drawing $d'$, using dynamic programming. First, we begin with ``leaf frames'', that is, frames with no grid points in their interior. Then, we continue with frames in order by increasing ``size'', where ``size'' refers to the number of grid points in their interior.  We ``solve'' each one of these frames. That is, for every frame $f$, we guess some specific information: we guess exactly how the intersection of the frame $f$ with the drawing $d'$ ``looks like'' and what vertices and edges or ``parts'' of edges are drawn inside $f$. Our guess should provide that the width of $f$ in any frame-tree of any drawing with the same specific information be bounded by $k$. We would like to know, for each guess, if there exists a ``partial'' drawing with the specific information we had guessed. Assume that we guessed a frame $f$ that is not a leaf in any drawn frame decomposition of $d'$, that is, a frame that has a grid point in its interior. Then, we can guess the cutter of $f$, $c$, in the frame-tree, and the specific information regarding the two frames associated with $c$, $f_1(c)$ and $f_2(c)$. Then, because in the interior of each of $f_1(c)$ and $f_2(c)$ there are less points than in their interior of $f$, we have already ``solved'' $f_1(c)$ and $f_2(c)$ with the ``specific information'' guess. So, by looking up an already computed entry of the dynamic programming table, we try to ``glue'' the ``partial drawings'' for the specific guess of $f_1(c)$ and $f_2(c)$, in order to construct a partial drawing for our guess for $f$, or conclude that there is no such partial drawing.

%!TEX root =Main-Movement.tex

\subsection{$G^*$-Drawings}

In order to present and prove the correctness of our scheme, we introduce a few definitions. As we explained in the beginning of Section \ref{sec:ProbSche}, we aim to ``solve'' small parts of potential drawings. Assume that $d$ is a polyline grid drawing of a graph $G=(V,E)$, and let $f$ be a frame. Observe that the part of $d$ that is bounded by $f$, that is, the {\em sub-drawing} of $d$ that is bounded by $f$, might contain ``parts'' of edges. In order to define such sub-drawings as drawings, we define the term {\em $G^*$-drawing}. Intuitively, a $G^*$-drawing is a drawing that might contain ``parts'' of edges of $G$, represented by new ``dummy'' vertices, where the set of dummy vertices is denoted by $V^*$. In particular, for every edge $\{u,v\}$, vertices of the form $uv_i\in V^*$, $i\in \mathbb{N}$, will represent the turning points (see Definition \ref{def:IntPoint}) of the edge $\{u,v\}$ in $f$ in the sub-drawing. By ``making'' the turning points vertices, we will be able to define this kind of sub-drawings as drawings. 
Now, assume that the choices for $i\in \mathbb{N}$ represent the order of appearance of the turning points of the edge $\{u,v\}$, starting from $u$. Then, $uv_i$ and $uv_{i+1}$ are connected by a ``part'' of the edge $\{u,v\}$ in $d$; so, we will connect them with a dummy edge (e.g., see the edge $\{uv_2,uv_3\}$ in Figure~\ref{fig:GStarDrawing}). We will denote the set of these dummy edges by $E^*$. 

Let $G=(V,E)$ be a graph. We assume that the set of vertices $V$ is an ordered set, that is, for every $u,v\in V$, $v\neq u$, either $u<v$ or $v<u$. First, for an edge $\{u,v\}\in E$ such that $u<v$, we define the set of dummy vertices as $V^*_{\{u,v\}}=\{uv_i~|~ i\in \mathbb{N}\}$. In addition, we define the set of dummy edges of $\{u,v\}$ as $E^*_{\{u,v\}}=\{\{uv_i,uv_{i+1}\}~|~ i\in \mathbb{N}\}$. Now, we are ready to define the sets $V^*$ and $E^*$.

\begin{definition}[{\bf $V^*$ and $E^*$}] 
	Let $G=(V,E)$ be a graph. Then:
	\begin{itemize}
		\item $V^*=\bigcup_{\{u,v\}\in E, u<v} V^*_{\{u,v\}}$.
		\item $E^*=\bigcup_{\{u,v\}\in E, u<v} E^*_{\{u,v\}}$.
	\end{itemize} 
\end{definition}

Recall that our intention is to define ``parts'' of drawings as drawings. For this purpose, we need to define the sub-graph that is drawn in a sub-drawing. We denote such a sub-graph by a pair $(U,F)$, where $U$ is the set of vertices and $F$ is the set of edges. In what follows, we aim to define which $(U,F)$ is a {\em valid} pair, that is, which $(U,F)$ can represent a sub-graph that we may encounter. As the sub-drawings that we encounter contain vertices and edges from the sets $V\cup V^*$ and $E\cup E^*$, respectively, we demand that $U\subset V\cup V^*$, and $F\subset E\cup E^*$, where the endpoints of the edges in $F$ belong to $U$ (see Condition \ref{con:ValidPair1} in Definition \ref{def:ValidPair}). In Condition \ref{con:ValidPair2}, for every $u,v \in U$ such that $\{u,v\} \in E$, we consider two different cases. The first one is when the edge $\{u,v\}$ is not intersected by the frame $f$, except maybe at its endpoints (e.g., see the edge $\{x,w\}$ in Figure~\ref{fig:GStarDrawing}), or $\{u,v\}$ is drawn on $f$. In this case, $\{u,v\}$ has no turning points other than, possibly, $(d(u),\{u,v\})$ or $(d(v),\{u,v\})$; so, we do not need to use any of the dummy vertices of $\{u,v\}$, and hence $U\cap V^*_{\{u,v\}}=\emptyset$ (Condition \ref{con:ValidPair21}). In the other case, $\{u,v\}$ intersects $f$ in a point (or points) that is not its endpoints and $\{u,v\}$ is not drawn on $f$. In this case, $\{u,v\}$ has turning points other than $(d(u),\{u,v\})$ or $(d(v),\{u,v\})$; so, we have a dummy vertex of $\{u,v\}$ for every such turning point (e.g., see the edges $\{u,v\}$ and $\{s,t\}$ in Figure~\ref{fig:GStarDrawing}). Specifically, we have that $U\cap V^*_{\{u,v\}} = \{uv_1,uv_2,\ldots ,uv_{\mathsf{index}(u,v)}\}$ (Condition \ref{con:ValidPair221}). Now, we demand that the numbering of the dummy vertices corresponds the order of appearance of the turning points that they represent, from $u$ to $v$. Then, $u$ should be connected by an edge to $uv_1$, $v$ should be connected by an edge to $uv_{\mathsf{index}(u,v)}$, and for every $1\leq j\leq \mathsf{index}(u,v)-1$, $uv_j$ should be connected by an edge to $uv_{j+1}$(Condition \ref{con:ValidPair222}). Here, we have $\subseteq$ rather than $=$ since some of the aforementioned edges can be outside $f$ (e.g., see the edges $\{s,st_1\}$ and $\{st_1,st_2\}$ in Figure~\ref{fig:GStarDrawing}). In addition, in this case, the edge $\{u,v\}$ does not appear in the drawing (but only parts of this edge appear), and so $\{u,v\}\notin F$ (Condition \ref{con:ValidPair223}); e.g., see the edge $\{s,t\}$ in Figure~\ref{fig:GStarDrawing}.

% $F\cap E^*_{\{u,v\}}\subseteq \{\{u,uv_1\}\}\cup \{\{uv_j,uv_{j+1}\}~|~1\leq j\leq \mathsf{index}(u,v)-1\}\cup \{\{uv_{\mathsf{index}(u,v)},v\}\}$.   the turning points of $\{u,v\}$, is continuous.    . Let $d$ be a $G^*$- drawing. We denote by $V(d)$ and $E(d)$ the sets of the vertices and edges of $d$, respectively. In Conditions \ref{G*drawcon3} and \ref{G*drawcon4}, we demand that for every edge $\{u,v\}$, the numbering of the dummy vertices that represent the turning points of $\{u,v\}$, is continuous. That is, $uv_1,\ldots, uv_{\mathsf{index}(u,v)}$ appear in the drawing $d$, where ${\mathsf{index}(u,v)}\in \mathbb{N}$ is the number of the dummy vertices of $\{u,v\}$ in $d$. In addition, sometimes we denote $u$ by $uv_0$ and $v$ by $uv_{{\mathsf{index}(u,v)}+1}$. In Condition \ref{G*drawcon5}, we demand that if an edge $\{u,v\}$ appears in $d$, then we do not need any dummy vertices for this edge, so $uv_1\notin V(d)$ (e.g., see the edge $\{x,z\}$ in Figure~\ref{fig:InfFDr}). Similarly, in Condition \ref{G*drawcon6}, we demand that if we do have at least one dummy vertex for an edge $\{u,v\}$, then the edge cannot appear in $d$, but only ``parts'' of this edge can appear in $d$, and therefore $\{u,v\}\notin E(d)$

\begin{definition}[{\bf Valid $(U,F)$}] \label{def:ValidPair}
	Let $G=(V,E)$ be a graph. A pair $(U,F)$ is {\em valid} if the following conditions are satisfied.
	\begin{enumerate}
		\item $U\subset V\cup V^*$, and $F\subset E\cup E^*$ where the endpoints of the edges in $F$ belong to $U$. \label{con:ValidPair1}
		\item For every $u,v \in U$ such that $\{u,v\} \in E$, exactly one of the following conditions is satisfied:\label{con:ValidPair2}
		\begin{enumerate}
			\item $U\cap U^*_{\{u,v\}}=\emptyset$. \label{con:ValidPair21}
			\item There exists $\mathsf{index}(u,v)\in \mathbb{N}$ such that: \label{con:ValidPair22}
			\begin{enumerate}
				\item $U\cap V^*_{\{u,v\}} = \{uv_1,uv_2,\ldots ,uv_{\mathsf{index}(u,v)}\}$.  \label{con:ValidPair221}
				\item $F\cap E^*_{\{u,v\}}\subseteq \{\{u,uv_1\}\}\cup \{\{uv_j,uv_{j+1}\}~|~1\leq j\leq \mathsf{index}(u,v)-1\}\cup \{\{uv_{\mathsf{index}(u,v)},$ $v\}\}$.\label{con:ValidPair222} 
				\item $\{u,v\}\notin F$.\label{con:ValidPair223}
			\end{enumerate}   
		\end{enumerate}
	\end{enumerate} 
\end{definition}

Let $d$ be a drawing. We denote by $V(d)$ and $E(d)$ the vertex set and the edge set associated with $d$, respectively. Sometimes we denote $u$ by $uv_0$ and $v$ by $uv_{{\mathsf{index}(u,v)}+1}$. Let $f$ be a frame. We denote the set of all $P\in {\cal P}$ such that $P$ is strictly inside $f$ by  ${\cal P}(f)$ (e.g., see the blue path in Figure~\ref{fig:pathPandPStar}). We define an {\em almost straight-line path in $f$} as a plane path, where i) the endpoints of the path are mapped to points in $\gis(f)$, ii) the internal vertices are mapped to grid points strictly inside $f$, and iii) every edge is mapped to the line segment $s$ connecting the images of their endpoints, and there exist $a,b\in \gps(\fin)$ such that $s$ is on $\ell(a,b)$ (e.g., see the red path in Figure~\ref{fig:pathPandPStar}). We denote by ${\cal P}^*(f)$ the set of all almost straight line paths in $f$. 

\begin{figure}[!t]
	\centering
	\includegraphics[width = 0.4\textwidth, page = 40]{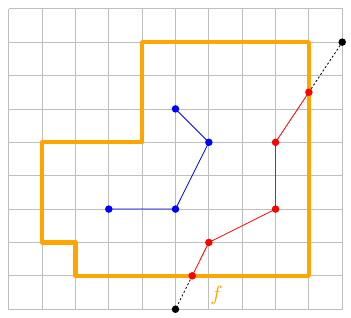}
	\caption{Examples of a path in ${\cal P}(f)$, drawn in blue, and a path in ${\cal P}^*(f)$, drawn in red.} 
	\label{fig:pathPandPStar}
\end{figure}

Now, we are ready to define a {\em $G^*$-drawing}. First, a $G^*$-drawing is a drawing of a graph that satisfies the conditions of Definition \ref{def:ValidPair}, that is, $(V(d),E(d))$ is valid (Condition \ref{G*drawcon10} in Definition \ref{def:gstdr}). Second, we present the conditions needed in order to verify that sub-drawings are sub-drawings of a potential polyline grid drawing (see Definition \ref{def:StrightD}). Every vertex from the set $V$ should be drawn on a grid point, as in a polyline grid drawing. On the other hand, a vertex $uv_i$ from the set $V^*$ represents a turning point of the edge $\{u,v\}$ in a frame $f$, so $uv_i$ is drawn on points from the set $\gis(\fin)$ (Condition \ref{G*drawcon2}). Similarly, every edge $\{u,v\}\in E$ should be drawn as a path from the set ${\cal P}(f_{\mathsf{init}})$, as in a polyline grid drawing. Now, consider an edge $\{uv_i,uv_{i+1}\}\in E(d)\cap E^*$, that is, $\{uv_i,uv_{i+1}\}$ represents a part of the edge $\{u,v\}$. Observe that, from Condition \ref{G*drawcon2}, vertices from $V^*$ can be drawn on a point in $\gis(\fin)$. Thus, the endpoints of the drawn of the edge  $\{uv_i,uv_{i+1}\}$ are drawn on points from the set $\gis(\fin)$. In addition, every part of a drawing of an edge, that is, a part of a $p\in {\cal P}(\fin)$, satisfies that every edge of $p$ is mapped to the line segment connecting the images of their endpoints, which are from the set $\gps(\fin)$. Thus, $d(\{uv_i,uv_{i+1}\})\in {\cal P}^*(f_{\mathsf{init}})$ (Condition \ref{G*drawcon3}). Now, recall that in the definition of a polyline grid drawing (Definition \ref{def:StrightD}), we demand that an edge in the drawing does not intersect itself. Therefore, in Conditions \ref{G*drawcon7}, \ref{G*drawcon8} and \ref{G*drawcon9} we demand that the ``parts'' of the edge $\{u,v\}$ are non-intersecting (e.g., see the edges $\{st_2,st_3\}$ and $\{st_3,t\}$ in Figure~\ref{fig:GStarDrawing}).

\begin{definition}[{\bf $G^*$-Drawing}] \label{def:gstdr}
Let $G=(V,E)$ be a graph. A drawing $d$ is a {\em $G^*$-drawing} if the following conditions are satisfied.
\begin{enumerate}
%\item $V(d)\subseteq V\cup V^*$. \label{G*drawcon1}
%\item $E(d)\subseteq E\cup E^*$.  \label{G*drawcon2}
%\item For every $uv_i\in V(d)$, there exists a unique $\mathsf{index}(u,v)\in \mathbb{N}$ such that for every $1\leq j\leq \mathsf{index}(u,v)$, $uv_j\in V(d)$, and for every $j>\mathsf{index}(u,v)$, $uv_j\notin V(d)$. \label{G*drawcon3}
%\item For every $uv_i\in V(d)$, if $\{uv_i,v\}\in d(V)$ then $i=\mathsf{index}(u,v)$. \label{G*drawcon4}
%\item For every $\{u,v\}\in E$, if $\{u,v\}\in E(d)$, then $uv_1\notin V(d)$. \label{G*drawcon5}
%\item For every $\{u,v\}\in E$, if $u,v\in V(d)$ and $uv_1\notin V(d)$, then $\{u,v\}\in E(d)$.  \label{G*drawcon6}
\item $(V(d),E(d))$ is valid. \label{G*drawcon10} 
\item For every $u\in V(d)\cap V$, $d(u)\in \gps(f_{\mathsf{init}})$, and for every $u\in V(d)\cap V^*$, $d(u)\in \gis(f_{\mathsf{init}})$. \label{G*drawcon2} 
\item For every $\{u,v\}\in E(d)\cap E$, $d(\{u,v\})\in {\cal P}(f_{\mathsf{init}})$, and for every $\{uv_i,uv_{i+1}\}\in E(d)\cap E^*$, $d(\{uv_i,uv_{i+1}\})\in {\cal P}^*(f_{\mathsf{init}})$. \label{G*drawcon3} 
\item For every $uv_i,uv_j\in V^*$ such that $i\neq j$, $d(uv_i)\neq d(uv_j)$. \label{G*drawcon7} 
\item For every $uv_j\in V(d)$ and $\{uv_i,uv_{i+1}\}\in E(d)$, $j\neq i,i+1$, $d$ does not draw $uv_j$ on $d(\{uv_i,uv_{i+1}\})$.  \label{G*drawcon8}  
\item Every two  edges $\{uv_i,uv_{i+1}\}, \{uv_j,uv_{j+1}\}\in E(d)$, $i\neq j$, do not intersect internally.\label{G*drawcon9}    
\end{enumerate}
\end{definition}  

In the rest of this subsection, we aim to extend some of the definitions in Section \ref{sec:drawnSep}, to be compatible with the definition of a $G^*$-drawing.

%\begin{figure}[!t]
%	\centering
%	\begin{subfigure}{0.496\textwidth}
%		\includegraphics[width = \textwidth, page = 38]{figures/drawnTreewidth}
%		\subcaption{}
%		\label{fig:GStarDrawing2}
%	\end{subfigure}
%	\hfil
%	\begin{subfigure}{0.496\textwidth}
%		\includegraphics[width = \textwidth, page = 39]{figures/drawnTreewidth}
%		\subcaption{}
%		\label{fig:GStarDrawing2}
%	\end{subfigure}
%	
%	\caption{An example of a $G^*$ drawing.} 
%	\label{fig:GStarDrawing}
%\end{figure}

\begin{figure}[!t]
	\centering
	\includegraphics[width = 0.5\textwidth, page = 39]{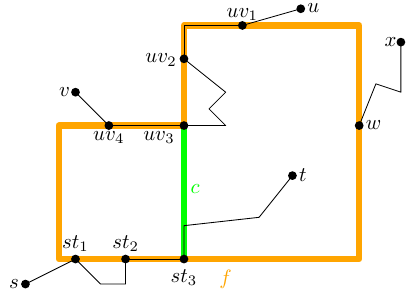}
	\caption{An example of a $G^*$-drawing. A frame $f$ is drawn in orange and a cutter $c$ of $f$ is drawn in green.} 
	\label{fig:GStarDrawing}
\end{figure}

First, we extend Definition \ref{def:IntPoint}, the definition of turning points. Recall that we use $V^*$ to define sub-drawings as drawings, but the vertices in $V^*$ do not appear in the original drawing. Therefore, we would not like to refer to a point $d(uv_i)$, where $uv_i\in V^*$, as a turning point in the partial drawing in a frame $f$, unless it is a turning point also in the original drawing. For example, see the point $d(st_3)$ in Figure~\ref{fig:GStarDrawing}. Observe that $(d(st_3),\{s,t\})$ is a turning point of $f$, but $(d(st_3),\{s,t\})$ is not a turning point of the left subframe created by the cutter $c$. Technically, we refer to a dummy edge $\{uv_i,uv_{i+1}\}\in E^*$ not as to a standalone edge, but as part of the edge $\{u,v\}$, and we refer to a dummy vertex $uv_i\in V^*$ as a ``point'' on the edge $\{u,v\}$.  Recall that a turning point of an edge $\{u,v\}$ in $f$ is an intersection point of $f$ and $d(\{u,v\})$, such that ``right after'' or ``right before'' $p$  there is no intersection between $d(\{u,v\})$ and $f$. We split the definition into three parts corresponding to three cases regarding the point $p\in \gi(d(\{uv_i,uv_{i+1}\}))\cap \gi(f)$. If $p\notin \{d(uv_i),d(uv_{i+1})\}$, the terms right after or right before $p$ are well defined on $d(\{uv_i,uv_{i+1}\})$. In this case, if $(p,\{u,v\})$ is indeed a turning point, we call it a {\em middle turning point} (see Definition \ref{def:IntPointG*3}), e.g. see the red point $p$ in Figure~\ref{fig:lrmTurningPoints}. Now, assume that $p\in \{d(uv_i),d(uv_{i+1})\}$, and, without loss of generality, that $p=d(uv_i)$. Then, the term right after $p$ is well defined on $d(\{uv_i,uv_{i+1}\})$, but the term right before is not, since $d(uv_i)$ is the ``left'' endpoint of $d(\{uv_i,uv_{i+1}\})$. So, we can conclude that the drawings of $\{u,v\}$ and $f$ do not intersect right before $p$ if one of two conditions satisfied:
\begin{itemize}
\item $\{uv_{i-1},uv_i\}\notin E(d)$.
\item $\{uv_{i-1},uv_i\}\in E(d)$ and right before $d(uv_i)$, $d(\{uv_{i-1},uv_i\})$ and $f$ do not intersect.
\end{itemize} 
In this case, if $(p,\{u,v\})$ is a turning point, we call it a {\em left turning point} (see Definition \ref{def:IntPointG*1}); e.g., see the red points $d(uv_2)$ and $d(st_2)$ in Figure~\ref{fig:lrmTurningPoints}. Observe that if $i=0$, that is, $\{uv_i,uv_{i+1}\}=\{u,uv_1\}$, then trivially $\{uv_{i-1},uv_i\}=\{uv_{-1},v\}\notin E(d)$, so $(p,\{u,v\})$ is a turning point as expected. The other case, where $p=d(uv_{i+1})$, is symmetric, and we call this kind of turning points {\em right turning points} (see Definition \ref{def:IntPointG*2}); e.g., see the red points $d(st_2)$ and $d(st_3)$ in Figure~\ref{fig:lrmTurningPoints}. At last, $(p,\{u,v\})$ is a turning point in $f$ if $(p,\{u,v\})$ is a left, right or middle turning point (see Definition \ref{def:IntPointG*}).

%\begin{itemize}
%\item If $p=d(uv_i)$, then $(p,\{u,v\})$ is a turning point if one of the following condition is satisfied:
%\begin{itemize}
%\item 
%\end{itemize}
%\end{itemize}
In the following definitions, we use the following notation: let $a$ and $b$ be two points in $\mathbb{R}^2$, and let $\epsilon > 0$; then, recall (from Section \ref{sec:drawnSep}) that we denote $\ell(a,a_\epsilon)$ by $\mathsf{line}_\epsilon(a,b)$, where $a_\epsilon$ is the point on the line $\ell(a,b)$ at distance $\epsilon$ from $a$ if it exists.

%Let $a$ and $b$ be two points in $\mathbb{R}^2$ and let $\epsilon > 0$. We denote $\ell(a,a_\epsilon)$ by $\mathsf{line}_\epsilon(a,b)$, where $a_\epsilon$ is the point on the line $\ell(a,b)$ at distance $\epsilon$ from $a$ if it exists.

%\begin{definition} [{\bf Turning Points of a Drawn Edge in a Frame}] \label{def:IntPoint}
%Let $G=(V,E)$ be a graph and let $d$ be a segmented grid drawing of $G$. Let $f$ be a frame and let $\{u,v\}$ be an edge of $G$. First, we say that $(d(u),\{u,v\})$ (and similarly $(d(v),\{u,v\})$) is a {\em turning point} in $f$ if $d(u)\in \gp(f)$ ($d(v)\in \gp(f))$ (see $(d(u_3),\{u_3,u_4\})$ in Figure~\ref{fig:TP}). In addition, let $p\in \gi(d(\{u,v\}))\cap \gi(f)$ be a point, such that $p\notin \{d(u),d(v)\}$. Let $p_i$ and $p_j$ be the two vertices of the path $d(\{u,v\})=(p_1,p_2,\ldots,p_k)$ such that $j>i$, $p\notin \{p_i,p_j\}$ and $(p_i,\ldots,p_j)$ is the minimum size subpath containing $p$. We say that $(p,\{u,v\})$ is a {\em turning point} in $f$, if there exists $\epsilon>0$ such that at least one of the following conditions is satisfied:
%\begin{itemize}
%	\item $\gi(p^\epsilon_{p_i,p})\cap \gi(f)=\{p\}$ (see point $(p_3,\{u_1,u_2\})$ in Figure~\ref{fig:TP}). 
%	\item $\gi(p^\epsilon_{p,p_j})\cap \gi(f)=\{p\}$ (see point $(p_1,\{u_1,u_2\})$ in Figure~\ref{fig:TP}).
%\end{itemize}
%\end{definition}

\begin{definition} [{\bf Left Turning Point in a $G^*$-Drawing}] \label{def:IntPointG*1}
Let $G=(V,E)$ be a graph and let $d$ be a $G^*$-drawing. Let $f$ be a frame and let $\{uv_i,uv_{i+1}\}\in E(d)$. Then, $(d(uv_i),\{u,v\})$ is a {\em left turning point} in $f$ if there exists $\epsilon>0$ such that at least one of the following conditions is satisfied:
\begin{itemize}
	\item $\gi(\mathsf{line}_\epsilon(p,p'))\cap \gi(f)=\{p\}$, where $d(\{uv_{i},uv_{i+1}\})=(p,p',\ldots,d(uv_{i+1}))$.
	\item $\{uv_{i-1},uv_{i}\}\notin E(d)$.
	\item $\{uv_{i-1},uv_{i}\}\in E(d)$ and $\gi(\mathsf{line}_\epsilon(p,p''))\cap \gi(f)=\{p\}$, where $d(\{uv_{i-1},uv_{i}\})=(d(uv_{i+1}),\ldots,p'',p)$.
%\item $\gi(p^\epsilon_{p,p'})\cap \gi(f)=\{p\}$, where $d(\{uv_{i},uv_{i+1}\})=(p,p',\ldots,d(uv_{i+1}))$.
%\item $\{uv_{i-1},uv_{i}\}\notin E(d)$.
%\item $\{uv_{i-1},uv_{i}\}\in E(d)$ and $\gi(p^\epsilon_{p'',p})\cap \gi(f)=\{p\}$, where $d(\{uv_{i-1},uv_{i}\})=(d(uv_{i+1}),\ldots,p'',p)$.
\end{itemize}
\end{definition}

\begin{figure}[!t]
	\centering
	\includegraphics[width = 0.5\textwidth, page = 43]{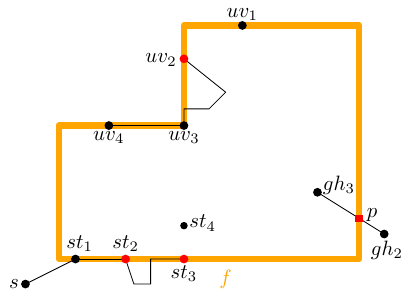}
	\caption{An illustration of left, right and middle turning points (drawn in red) in a $G^*$ drawing in a frame. Specifically, $(d(uv_2),\{u,v\})$ is a left turning point in $f$; $(d(st_2),\{s,t\})$ is both a left and a right turning point in $f$; $(d(st_3),\{s,t\})$ is a right turning point in $f$; $p$ is a middle turning point in $f$.} 
	\label{fig:lrmTurningPoints}
\end{figure}

\begin{definition} [{\bf Right Turning Point in a $G^*$-Drawing}] \label{def:IntPointG*2}
Let $G=(V,E)$ be a graph and let $d$ be a $G^*$-drawing. Let $f$ be a frame and let $\{uv_i,uv_{i+1}\}\in E(d)$. Then, $(d(uv_{i+1}),\{u,v\})$ is a {\em right turning point} in $f$ if there exists $\epsilon>0$ such that at least one of the following conditions is satisfied:
\begin{itemize}
\item $\gi(\mathsf{line}_\epsilon(p,p'))\cap \gi(f)=\{p\}$, where $d(\{uv_{i},uv_{i+1}\})=(d(uv_i)\ldots,p',p)$.
\item $\{uv_{i+1},uv_{i+2}\}\notin E(d)$.
\item $\{uv_{i+1},uv_{i+2}\}\in E(d)$ and $\gi(\mathsf{line}_\epsilon(p,p''))\cap \gi(f)=\{p\}$, where $d(\{uv_{i+1},uv_{i+2}\})=(p,p'',\ldots d(uv_{i+2}))$.
%\item $\gi(p^\epsilon_{p,p'})\cap \gi(f)=\{p\}$, where $d(\{uv_{i},uv_{i+1}\})=(d(uv_i)\ldots,p',p)$.
%\item $\{uv_{i+1},uv_{i+2}\}\notin E(d)$.
%\item $\{uv_{i+1},uv_{i+2}\}\in E(d)$ and $\gi(p^\epsilon_{p,p''})\cap \gi(f)=\{p\}$, where $d(\{uv_{i+1},uv_{i+2}\})=(p,p'',\ldots d(uv_{i+2}))$.
\end{itemize}
\end{definition}

\begin{definition} [{\bf Middle Turning Point in a $G^*$-Drawing}] \label{def:IntPointG*3}
Let $G=(V,E)$ be a graph and let $d$ be a $G^*$-drawing. Let $f$ be a frame and let $\{uv_i,uv_{i+1}\}\in E(d)$. Let $p\in (\gi$ $(d(\{uv_i,uv_{i+1}\}))\cap \gi(f))\setminus \{d(uv_i),d(uv_{i+1})\}$ be a point. Let $p_i$ and $p_j$ be the two vertices of the path $d(\{uv_i,uv_i+1\})=(p_1,p_2,\ldots,p_k)$ such that $j>i$, $p\notin \{p_i,p_j\}$ and $(p_i,\ldots,p_j)$ is the minimum size subpath containing $p$\footnote{Observe that $j\in \{i+1,i+2\}$.}. Then, $(p,\{u,v\})$ is a {\em middle turning point} in $f$ if there exists $\epsilon>0$ such that at least one of the following conditions is satisfied:
\begin{itemize}
\item $\gi(\mathsf{line}_\epsilon(p,p_i))\cap \gi(f)=\{p\}$. 
\item $\gi(\mathsf{line}_\epsilon(p,p_j))\cap \gi(f)=\{p\}$.
%\item $\gi(p^\epsilon_{p_i,p})\cap \gi(f)=\{p\}$. 
%\item $\gi(p^\epsilon_{p,p_j})\cap \gi(f)=\{p\}$.
\end{itemize}
\end{definition}

\begin{definition} [{\bf Turning Point in a $G^*$-Drawing}] \label{def:IntPointG*}
Let $G$ be a graph and let $d$ be a $G^*$-drawing. Let $f$ be a frame and let $\{uv_i,uv_{i+1}\}\in E(d)$. Let $p\in \gi(d(\{uv_i,$ $uv_{i+1}\}))\cap \gi(f)$ be a point. Then, $(p,\{u,v\})$ is a {\em turning point} if one of the following conditions is satisfied:
\begin{itemize}
\item $(p,\{u,v\})$ is a left turning point.
\item $(p,\{u,v\})$ is a right turning point.
\item $(p,\{u,v\})$ is a middle turning point.
\end{itemize}
\end{definition}

The definition for the number of turning points in $f$, $\mathsf{TurPoints}(f)$, remains unchanged, that is, $\mathsf{TurPoints}(f)=|\{(p,\{u,v\})~|~p\in P(f),\{u,v\}\in E, (p,\{u,v\})$ is a turning point in $f\}|$.

%\begin{observation}
%Let $G=(V,E)$ be a graph and let $d$ be a straight-line $G^*$ drawing. Let $f$ be a frame, let $\{u,v\}\in E$ be an edge and let $p\in P(f)$ be a point. Then, $(p,\{u,v\})$is a turning point if and only if there exists $\epsilon>0$ such that at least one of the following conditions is satisfied:
%\begin{itemize}
%\item $P([p,p+\epsilon])\cap (\bigcup_{i=0}^t P(d(\{uv_i,uv_{i+1}\}))=\{p\}$.
%\item $P([p-\epsilon,p])\cap (\bigcup_{i=0}^t P(d(\{uv_i,uv_{i+1}\}))=\{p\}$. 
%\end{itemize}   
%\end{observation}

Moreover, the definition for the cost of a frame (Definition \ref{def:sizeFr}) remains unchanged for $G^*$-drawings, as stated next: 

%\begin{definition} [{\bf Cost of a Frame}] \label{def:sizeFr}
%Let $G=(V,E)$ be a graph, let $d$ be a segmented grid drawing of $G$ and let $f$ be a frame. The {\em cost} of $f$ in $d$, denote by $\sizef(f)$, is the sum of the number of vertices of $f^{\mathsf{min}}$, the number of vertices of $G$ on $f$ and $\tp(f)$.   
%\end{definition}

\begin{definition} [{\bf Cost of a Frame in a $G^*$-Drawing}] \label{def:sizeFrG*}
Let $G=(V,E)$ be a graph, let $d$ be a $G^*$-drawing and let $f$ be a frame. The {\em cost} of $f$ in $d$, denote by $\sizef(f)$, is the sum of the number of vertices of $f^{\mathsf{min}}$, the number of vertices of $G$ on $f$ with respect to $d$ and $\tp(f)$.     
\end{definition}

%Let $G=(V,E)$ be a graph, let $d$ be a segmented grid drawing of $G$ and let $f$ be a frame. The {\em cost} of $f$ in $d$, denote by $\sizef(f)$, is the sum of the number of vertices of $f^{\mathsf{min}}$, the number of vertices of $G$ on $f$ with respect to $d$ and $\tp(f)$.   
%When $d$ is clear from the context, we refer to  $\mathsf{Width}(w,d)$ simply as $\mathsf{Width}(w)$.

%For a $G^*$-drawing $d$ and a frame $f$, we say that $d$ is {\em bounded} by $f$, or $f$ is a  {\em boundary frame} of $d$, if no vertex or edge are drawn in the outer face of $f$. We say that $d$ is {\em strictly bounded} by $f$ if $d$ is bounded by $f$ and also the intersection of the embedding of all the vertices and edges of $d$, and $f$ is empty. 

Lastly, we aim to extend Definition \ref{def:DrawnFrDeco}. Observe that in Condition \ref{def:tfCon1} of this definition, we demand that the frame associated with $v_r$ is $R_d$. As we will see later, for every other frame $f$ that strictly bounds $d$, we will get the same value for $\mathsf{dtw(d)}$ if we change Definition \ref{def:DrawnFrDeco} in a way that the frame that is associated with $v_r$, the ``initial frame'', is $f$. Nevertheless, we might get a different value for $\mathsf{dtw(d)}$ if we take two frames $f$ and $f'$ that bound (but not necessarily strictly bound) $d$ as the initial frames.  Moreover, recall that in our scheme, we aim to build (partial) drawing, with a bounded $\mathsf{dtw}$, by gluing two smaller parts of that (partial) drawing.  To control $\mathsf{dtw(d)}$, we would like to build a frame-tree of the glued drawing, based on the frame-trees we should already have for the smaller parts. So, we extend  Definition \ref{def:DrawnFrDeco} to be compatible with both a $G^*$-drawing and an initial frame other than $R_d$. 

% \begin{definition} [{\bf Drawn Tree-Frame}] \label{def:DrawnFrDeco}
%	Let $G=(V,E)$ be a graph and let $d$ be a segmented grid drawing of $G$. A {\em drawn tree-frame} of $d$ is a pair $({\cal T}=(V_T,E_T),\alpha: V_T\rightarrow \mathsf{Frames})$ where $\cal T$ is a rooted tree with root $v_r\in V_T$ and the following conditions are satisfied:
%	\begin{enumerate}
%	\item $\alpha(v_r)=R_d$. \label{def:tfCon1}
%\item Every $v\in V_T$ that is not a leaf, has exactly two children $v_1$ and $v_2$.
%\item For every $v\in V_T$ that is not a leaf, there exists a cutter $c_v$ of $f$, such that $\alpha(v_1)=f_1(c_v)$ and $\alpha(v_2)=f_2(c_v)$, where $\alpha(v)=f$. We say that $c_v$ is the cutter associated with $v$.   
%\item A vertex $v\in V_T$ is a leaf if and only if there are no grid points in the interior of $\alpha(v)$.             
%	\end{enumerate}   
%\end{definition}                  

\begin{definition} [{\bf $f$-Frame-Tree}] \label{def:DrawnDecof}
Let $G=(V,E)$ be a graph, let $d$ be a $G^*$-drawing and let $f$ be a boundary frame of $d$. An {\em $f$-frame-tree} of $d$ is a pair $({\cal T}=(V_T,E_T),\alpha: V_T\rightarrow \mathsf{Frames})$ where $\cal T$ is a binary rooted tree, and: 
\begin{enumerate}
	\item $\alpha(v_r)=f$, where $v_r$ is the root of $\cal T$.
%	\item Every $v\in V_T$ that is an internal vertex has exactly two children $v_1$ and $v_2$.
	\item For every internal vertex $v\in V_T$, there exists a cutter $c_v$ of $f$ such that $\alpha(v_1)=f_1(c_v)$ and $\alpha(v_2)=f_2(c_v)$, where $\alpha(v)=f$ and $v_1, v_2$ are the children of $v$ in ${\cal T}$. We say that $c_v$ is the {\em cutter associated with $v$}.    
	\item A vertex $v\in V_T$ is a leaf if and only if there are no grid points in the interior of $\alpha(v)$.             
\end{enumerate}   
\end{definition}

The {\em width} of an $f$-drawn tree frame, denoted by $\mathsf{width}({\cal T},\alpha)$, is the maximum cost of a frame in $d$, i.e., $\mathsf{width}({\cal T},\alpha)=\max\{\sizef(\alpha(v))~|~v\in V_T\}$.
The {\em $f$-drawn treewidth} of $d$, denoted by $\mathsf{dtw}(d,f)$, is minimum width of an $f$-drawn tree frame of $d$, i.e., $\mathsf{dtw}(d,f)=\min\{\mathsf{width}({\cal T},\alpha)~|~({\cal T},\alpha)$ is an $f$-drawn tree frame of $d\}$.    

%\begin{lemma}
%Let $G=(V,E)$ be a graph and let $f$ be a rectilinear drawing of $G$. Let $({\cal T}=(V_T,E_T),\alpha: V_T\rightarrow \mathsf{Frames})$ be a drawn tree-frame of $f$. Then, $({\cal T}=(V_T,E_T),\beta: V_T\rightarrow 2^{V})$, where $\beta(v)=S(\alpha(v))$ for every $v\in V_T$, is a tree decomposition of $G$.  
%\end{lemma}

%$G*(F)$-operation. Let $F=(U,U',I,f)$ be a frame. Then $G*=(V*,E*)$ is the graph defined by $V*=U\cup V(I)$ and $E*=\{\{u,v\}~|~\{u,v\}\in E$ and $u,v\in U\} \cup E(I)$.

%!TEX root =Main-Movement.tex

\subsection{Info-Frames}\label{sec:infoFra}

In this subsection we introduce the term {\em info-frame}. An info-frame is a way for us to describe a minimal amount of information we need to store for small parts of a drawing in order to glue two parts together. We begin with the first property of an info-frame. Assume that we have a part of a drawing $d$, that is, a $G^*$-drawing, bounded by a frame $f$. First, we would like to know exactly how the boundary of the drawing looks like, that is, the drawing that is ``on'' $f$, denoted by $d_f$. Second, we would like to know the set of vertices that are drawn strictly inside $f$, denoted by $U_f$. Also, we would like to know the set of edges with at least one endpoint in $V(d_f)$ that we have in the partial drawing, denoted by $E_f$.  Now, recall that in polyline grid drawings edges are drawn as paths from the set $\cal{P}$. So, for every $\{u,v\}\in E$, $d(\{u,v\})$ is a drawn straight line path, whose vertices are on grid points. In particular, bends in $d(\{u,v\})$ can be only on grid points. In light of this, consider a vertex $uv_i$ drawn on a point from the set $\gis(\fin)\setminus \gps(\fin)$. We need to store the ``direction'' of  $\{uv_i,uv_{i+1}\}\in E^*$ (or $\{uv_{i-1},uv_{i}\}\in E^*$), which represents a part of the edge $\{u,v\}$, in order to verify that the other part of the edge has the same direction, so there is no bending at the point $d(\{uv_i,uv_{i+1}\})$. 

In Definition \ref{def:infFr} ahead, we state the conditions that a tuple  $F=(f,d_f,E_f,U_f,\mathsf{V^*Dir}_f)$ should satisfy in order to have the first property an info-frame. The intuition behind these conditions is as follows.

\smallskip\noindent{\bf Conditions \ref{definfFraCon1}, \ref{definfFraCon2} and \ref{definfFraCon33}.} First, we state, as explained, that $f$ is a frame, $d_f$ is a $G^*$-drawing on $f$, $U_f\subseteq V$, and $E_f\subseteq \{\{u,v\}\in E^*\cup E~|~u\in V(d_f),v\in V(d_f)\cup U_f\}$. In addition, we expect that the vertices in $V^*$ represent the turning points in $f$ (e.g., see vertices $uv_1, uv_2, uv_3, uv_4$ in Figure~\ref{fig:infFandD}). 

\smallskip\noindent{\bf Condition \ref{definfFraCon4}.} As explained earlier, for every $uv_i\in V(d_f)\cap V^*$ such that $d_f(uv_i)\in \gis($ $\fin)\setminus \gps(\fin)$, we would like to store the direction of the edge corresponding to $uv_i$. For this purpose, we would like to store a point $a\in \gis(\fin)$, $a\neq d_f(uv_i)$, such that $\ell(\mathsf{V^*Dir}_f(uv_i),d_f(uv_i))$ is inside $f$. The line $\ell(\mathsf{V^*Dir}_f(uv_i),d_f(uv_i))$ represents the direction of the edge attached to $uv_i$ from the point $d_f(uv_i)$ (e.g., see the pink lines $\ell(c, d_f(uv_2)$, $\ell(c', d_f(gh_1)$ and $\ell(c'', d_f(mn_1)$ in Figure~\ref{fig:infFandD}). So, later, when we glue two ``sides'' of the edge corresponding to $uv_i$, we will verify that the directions of both sided of $uv_i$ are on the same line, and therefore there is no bend at the point $d_f(uv_i)$.

% straight line ``around'' $d_f(uv_i)$, such that i) one part of the line, from one endpoint of the line to $d_f(uv_i)$, is inside $f$ ii) the other part of the line is outside $f$ (e.g., see the red line passing through $uv_2$ and $gh_1$ in Figure~\ref{fig:InfF}). The side of the line that is inside $f$, is on the part of the edge, attached to $uv_i$, that is inside $f$ (e.g., see the red line passing through $uv_2$ (resp., $gh_1$) and the purple edge $\{uv_1,uv_2\}$ (resp., $\{g,gh_1\}$) in Figure~\ref{fig:InfFDr}). So, later, when we glue two sides of the edge attached to $uv_i$, we get that there is no bending at the point $d_f(uv_i)$.

%We now define the set $\giss(\fin)=\gis\cup \{$the intersection point of $\ell(a,a')$ and $\ell(b,b')~|~a,a'\in \gps(\fin), b,b'\in \gis(\fin)$ and $\ell(a,a')$ and $\ell(b,b')$ are not parallel$\}$ and we define the {\em direction of $uv_i$ in $F$}, denoted by $\mathsf{V^*Dir}_f(uv_i)$ to be $(a,b)$, $a,b\in \giss(\fin)$, such that i) $a,b\neq d_f(uv_i)$ ii) $d_f(uv_i)$ is on $\ell(a,b)$ iii) $\ell(a,d_f(uv_i))$ is inside $f$ and $\ell(b,d_f(uv_i))$ is outside $f$. The reason to take points from the set $\giss(\fin)$ will be clear later in Section \ref{sec:infoFra}.

%Let $a,b\in \gps(f_\mathsf{init})$, let $c\in \gp$ be a point on $\ell(a,b)$ and let $\epsilon\in \mathbb{R}$. We denote $c_{a,b}^{\epsilon,a} =c_{a,c}^{\epsilon}$ and $c_{a,b}^{\epsilon,b} =c_{c,b}^{\epsilon}$.

%Let $f$ be a frame. We denote $\epsilon(f)=\frac{1}{2}\mathsf{min}\{|\ell(a,b)|~|~a,b\in\gis(f)\}$.

\begin{definition}[{\bf Template for an Info-Frame}] \label{def:infFr}
	Let $G=(V,E)$ be a graph. A tuple  $F=(f,d_f,E_f,U_f,\mathsf{V^*Dir}_f)$ is an {\em \bf Info-Frame Template} if the following conditions are satisfied.
	\begin{enumerate}
		\item $f$ is a frame. \label{definfFraCon1}
		\item $d_f$ is a $G^*$-drawing on $f$ such that for every vertex in $V^*$, $uv_i\in V(d_f)$, it follows that $(d_f(uv_i),\{u,v\})$ is a turning point in $f$.  \label{definfFraCon2}
		\item $U_f\subseteq V$ and $E_f\subseteq \{\{u,v\}\in E^*\cup E~|~u\in V(d_f),v\in V(d_f)\cup U_f\}$. \label{definfFraCon33}
		\item For every $uv_i\in V(d_f)\cap V^*$ such that \\ $d_f(uv_i)\in \gis(\fin)\setminus\gps( \fin)$, \\ $\mathsf{V^*Dir}_f(uv_i)\in \gis(\fin)$ and the following conditions hold:\label{definfFraCon4}
		\begin{enumerate}
			\item $\mathsf{V^*Dir}_f(uv_i)\neq d_f(uv_i)$.
			\item $\ell(\mathsf{V^*Dir}_f(uv_i),d_f(uv_i))$ is inside $f$.
		\end{enumerate}
	\end{enumerate}
\end{definition}

%For a frame $f$ and a $G^*$-drawing, $d_f$, we say that {\em $d_f$ is on $f$} if the vertices of $d_f$ and the edges of $d_f$ are drawn into points in $P(f)$ (see Figure~\ref{fig:InfF}). 

We continue with the next property of an info-frame (Definition \ref{definfFraCon3}). Here, we give the intuition behind the conditions.  

\smallskip\noindent{\bf Conditions \ref{definfFraCon3a} and \ref {definfFraCon3b}.} In these conditions, we state additional requirements expected to be satisfied when $F$ is meant to describe a part of a drawing of the graph $G$ that is bounded by $f$. Recall that $U_f$ represents the set of vertices that are drawn strictly inside $f$, and $E_f$ represents the set of edges that are drawn strictly inside $f$, except at least one endpoint on $f$. Thus, we demand that $V(d_f)\cap U_f=\emptyset$ and $E_f\cap E(d_f)=\emptyset$ (e.g., see Figure~\ref{fig:infFandD}). 

\smallskip\noindent{\bf Condition \ref{definfFraCon3c}.} Now, assume that there is an edge $\{u,v\}\in E$ such that $u$ is drawn strictly inside $f$, and $v$ is drawn on $f$. There are exactly two cases we consider in Condition \ref{definfFraCon3c}. The first case, stated in Condition \ref{definfFraCon3c1},  is that the drawing of $\{u,v\}$ intersects $f$ only at $d(v)$ (e.g., see the edge $\{a,b\}$ in Figure~\ref{fig:InfFDr}). In this case, there are no turning points of $\{u,v\}$ in $f$ other than $(d(u),\{u,v\})$, so $U^*_{\{u,v\}}\cap V(d_f)=\emptyset$. In addition, $\{u,v\}$ is drawn strictly inside $f$ except at $d(u)$, therefore $\{u,v\}\in E_f$. The other case, stated in Condition \ref{definfFraCon3c2}, is where the drawing of $\{u,v\}$ intersects $f$ in points other than $d(u)$. Then, there must be at least one turning point other than $(d(u),\{u,v\})$, so $U^*_{\{u,v\}}\cap V(d_f)\neq \emptyset$. In addition, in this case, the edge $\{u,v\}$ is ``divided'' into the edges $\{u,uv_1\},\{uv_1,uv_2\},\ldots \{uv_{\mathsf{index}(u,v)},v\}$, so $\{u,v\}\notin E_f$. Further, $\{uv_{\mathsf{index}(u,v)},v\}$ is an edge that is strictly drawn inside $f$ except at the endpoint $d(uv_{\mathsf{index}(u,v)})$, so $\{uv_{\mathsf{index}(u,v)},v\}\in E_f$. 

\smallskip\noindent{\bf Condition \ref{definfFraCon3d}.} This condition is the symmetric version of Condition \ref{definfFraCon3c}.

\smallskip\noindent{\bf Condition \ref{definfFraCon3e}.} Now, assume that there is an edge $\{u,v\}$ such that $u$ is drawn strictly inside $f$, and $v$ is drawn strictly outside $f$, that is, $u\in U_f$ and $v\notin U_f\cup V(d_f)$. In this case, there must be a turning point of $\{u,v\}$ in $f$, therefore $uv_1\in V(d_f)$. In addition, $uv_1$ is connected with an edge to $u$ (that is strictly inside $f$), so $\{u,uv_1\}\in E_f$ (e.g., see the edge $\{u,v\}$ in Figure~\ref{fig:infFandD}).

\smallskip\noindent{\bf Condition \ref{definfFraCon3f}.} This condition is the symmetric version of Condition \ref{definfFraCon3e}.

\smallskip\noindent{\bf Condition \ref{definfFraCon3g}.} Now, assume that there is an edge $\{u,v\}$ such that $u$ and $v$ are drawn strictly inside $f$, that is, $u,v\in U_f$. If $\{u,v\}$ intersects $f$, then there are $\mathsf{index}(u,v)$ turning points of $\{u,v\}$ in $f$, so $uv_1,\ldots,uv_{\mathsf{index}(u,v)}\in V(d_f)$ (in particular $U^*_{\{u,v\}}\cap V(d_f)\neq \emptyset$). In addition, $uv_1$ and $uv_{\mathsf{index}(u,v)}$ are connected with an edge that is strictly inside $f$ to $u$ and $v$, respectively. So $\{u,uv_1\},\{uv_{\mathsf{index}(u,v)},v\}\in E_f$, and Condition \ref{definfFraCon3g1} is satisfied. Otherwise, $\{u,v\}$ does not intersect $f$, so there are no intersection points of $\{u,v\}$ in $f$. Thus $U^*_{\{u,v\}}\cap V(d_f)=\emptyset$, and Condition \ref{definfFraCon3g2} is satisfied (e.g., see the edge $\{x,w\}$ in Figure~\ref{fig:InfFDr}). 

\smallskip\noindent{\bf Condition \ref{definfFraCon3h}.} Observe that for every $\{u,v\}\in E$, $1\leq i\leq \mathsf{index}(u,v)$, the edge $\{uv_i,uv_{i+1}\}$ might be either drawn strictly inside $f$, except for necessarily at the endpoints, or drawn on $f$, but these two cases cannot hold simultaneously. Therefore, if $\{uv_i,uv_{i+1}\}\in E(d_f)$, then $\{uv_i,uv_{i+1}\}\notin E_f$, so Condition \ref{definfFraCon3h} is satisfied (e.g., see the edge $\{uv_3,uv_4\}$ in Figure~\ref{fig:infFandD}).

\begin{definition}[{\bf Edges and Vertices of an Info-Frame Property}] \label{definfFraCon3}
	Let $G=(V,E)$ be a graph. A tuple  $F=(f,d_f,E_f,U_f,\mathsf{V^*Dir}_f)$ exhibits the {\em \bf Edges and Vertices of an Info-Frame Property} if $F$ is an {\bf Info-Frame Template} and the following conditions are satisfied:
	
	\begin{enumerate}
		\item $U_f\cap V(d_f)=\emptyset$ and $E_f\cap E(d_f)=\emptyset$. \label{definfFraCon3a}  \label{definfFraCon3b}
	%	\item $E_f\cap E(d_f)=\emptyset$.  \label{definfFraCon3b} %r every $\{u,v\}\in E(d_f)$, $\{u,v\}\notin E_f$.  \label{definfFraCon1}  
		\item For every $\{u,v\}\in E$ such that $u\in U_f$ and $v\in V(d_f)$, exactly one of the following conditions holds. \label{definfFraCon3c} 
		\begin{enumerate}
			\item $U^*_{\{u,v\}}\cap V(d_f)=\emptyset$ and $\{u,v\}\in E_f$.  \label{definfFraCon3c1}
			\item  $U^*_{\{u,v\}}\cap V(d_f)\neq \emptyset$, $\{u,v\}\notin E_f$ and $\{uv_{\mathsf{index}(u,v)},v\}\in E_f$.  \label{definfFraCon3c2}
		\end{enumerate}
		\item For every $\{u,v\}\in E$ such that $v\in U_f$ and $u\in V(d_f)$, exactly one of the following condition holds.  \label{definfFraCon3d}
		\begin{enumerate}
			\item $uv_1\notin V(d_f)$ and $\{u,v\}\in E_f$.  \label{definfFraCon3d1}
			\item  $\{uv_{\mathsf{index}(u,v)},v\}\in E_f$, $\{u,v\}\notin E_f$ and $\{u,uv_1\}\in E_f$.  \label{definfFraCon3d2}
		\end{enumerate}
		\item For every $\{u,v\}\in E$ such that $u\in U_f$ and $v\notin U_f\cup V(d_f)$, $\{u,uv_1\}\in E_f$.  \label{definfFraCon3e}
		\item For every $\{u,v\}\in E$ such that $v\in U_f$ and $u\notin U_f\cup V(d_f)$, $\{uv_{\mathsf{index}(u,v)},v\}\in E_f$.  \label{definfFraCon3f} 
		\item For every $\{u,v\}\in E$ such that $u,v\in U_f$, exactly one of the following conditions holds.\label{definfFraCon3g} 
		\begin{enumerate}
			\item  $U^*_{\{u,v\}}\cap V(d_f)\neq \emptyset$ and $\{uv_{\mathsf{index}(u,v)},v\}\in E_f$.\label{definfFraCon3g1} 
			\item $U^*_{\{u,v\}}\cap V(d_f)=\emptyset$. \label{definfFraCon3g2} 
		\end{enumerate}
		%\item For every $\{u,v\}\in E$, $0\leq i\leq t$, if $\{uv_i,uv_{i+1}\}\in E(d_f)$, then $\{uv_i,uv_{i+1}\}\notin E_f$.   
		\item  For every $\{u,v\}\in E$ and $1\leq i\leq \mathsf{index}(u,v)$ such that $\{uv_i,uv_{i+1}\}\in E(d_f)$, $\{uv_i,uv_{i+1}\}$ $\notin E_f$.\label{definfFraCon3h}
	\end{enumerate} 
	\end{definition}

%\begin{observation}
	%Let $f\in \mathsf{Frames}(\fin)$, let $a,b,c\in \gis(\fin)$ such that $c$ is on $f$, $a$ is inside $f$ and $b$ is strictly outside $f$ and $c$ is on $\ell(a,b)$. Then, $\ell(c,c_{a,b}^{\epsilon(\fin),a})$ is strictly inside $f$ except for the point $c$.
%\end{observation} 
The next property concerns the vertices $uv_i\in V(d_f)\cap V^*$ such that\\ $d_f(uv_i)\in\gis(\fin)\setminus \mathsf{GridPointSet}(f_{\mathsf{init}})$. We show that for such a vertex there must be exactly one neighbor in $E_f$, and no edges on $f$. Recall that $uv_i\in V(d_f)\cap V^*$ represents a turning point of the edge $\{u,v\}$ in $f$ and $\{u,v\}$ is drawn as a path in $\cal{P}$. So, if $uv_i$ has an edge on $f$, then one among $d(\{u,v\})$ and $f$ bends at the point $d_f(uv_i)$, a contradiction (Condition \ref{def:validDircon1}). Now, since $(d_f(uv_i),\{u,v\})$ is a turning point in $f$, then there must be part of $\{u,v\}$ that intersects $f$ exactly at $d_f(uv_i)$. So, $uv_i$ has exactly one neighbor in $E_f$ (Condition \ref{def:validDircon2}). 

\begin{definition}[{\bf Vertices on Grid Intersection Property}]\label{def:validDir}
Let $G=(V,E)$ be a graph. A tuple  $F=(f,d_f,E_f,U_f,\mathsf{V^*Dir}_f)$ exhibits the {\em \bf Vertices on Grid Intersection Property} if $F$ is an {\bf Info-Frame Template} and for every $uv_i\in V(d_f)\cap V^*$ such that $d_f(uv_i)\in \gis(\fin)\setminus \mathsf{GridPointSet}(f_{\mathsf{init}})$ the following conditions are satisfied:
	\begin{enumerate}
	\item $\{uv_{i-1},uv_{i}\},\{uv_i,uv_{i+1}\}\notin E(d_f)$.\label{def:validDircon1}
	\item Exactly one among  of the following conditions is satisfied:\label{def:validDircon2}
	\begin{enumerate}
		\item $\{uv_{i-1},uv_{i}\}\in E_f$.
		\item $\{uv_{i},uv_{i+1}\}\in E_f$.
	\end{enumerate}
	\end{enumerate} 
	
\end{definition}

Now, we are ready to define the term info-frame:

\begin{definition}[{\bf Info-Frame}] \label{def:infFr3}
	Let $G=(V,E)$ be a graph. A tuple  $F=(f,d_f,E_f,U_f,$ $\mathsf{V^*Dir}_f)$ is an {\em info-frame} if $F$ is an {\bf Info-Frame Template} and $F$ exhibits the {\bf Edges and Vertices of an Info-Frame Property} and the  {\bf Vertices on Grid Intersection Property}.
\end{definition}

\begin{figure}[!t]
	\centering
	\begin{subfigure}{0.45\textwidth}
		\includegraphics[width = \textwidth, page = 44]{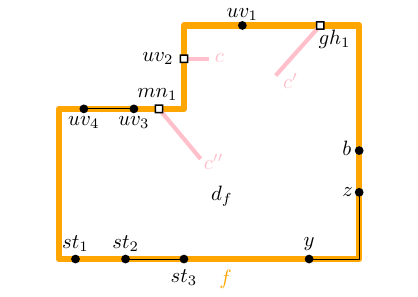}
		\subcaption{}
		\label{fig:InfF}
	\end{subfigure}
	\hfil
	\begin{subfigure}{0.45\textwidth}
		\includegraphics[width = \textwidth, page = 45]{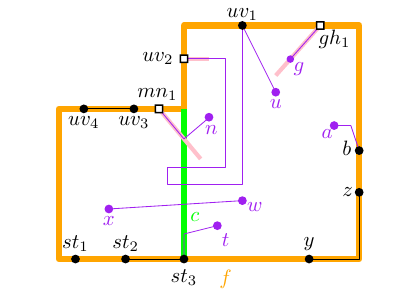}
		\subcaption{}
		\label{fig:InfFDr}
	\end{subfigure}
	
	\caption{(a) An example of an info-frame $F=(f,d_f,E_f,U_f, \mathsf{V^*Dir}_f)$ of a graph $G$ with $\{\{a,b\},\{u,v\},\{s,t\},\{y,z\},\{x,w\}\, \{g,h\}, \{i,j\}\}\subseteq E$, where the frame $f$ is drawn in orange, $d_f$ is drawn in black, $E_f=\{\{u,uv_1\},\{uv_1,uv_2\},\{g,gh_1\}, \{mn_1, n\},  \{st_3,t\}\}$ and $U_f=\{a, g, n, t,u,w,x\}$. A cutter $c$ of $f$ is drawn in green. The vertices which are mapped to points in $\gis(f) \setminus \mathsf{GridPointSet}(f)$ are shown by hollow squares. The pairs in $\mathsf{V^*Dir}_f$ and the lines connecting these pairs (representing directions) are shown in pink. (b) A drawing of the info-frame $F$. The drawings of the edges $E_f$ and the vertices in $U_f$ are shown in purple. } 
	\label{fig:infFandD}
\end{figure}

Now, for later use, define a few additional notations.  Let $d$ and $d'$ be two $G^*$-drawings and let $P\subset \mathbb{R}\times \mathbb{R}$  be a set of points. We denote by $V(d,P)$ the set of vertices in $V(d)$ that are drawn on points in $P$, that is $V(d,P)=\{v\in V(d)~|~$ there exists $p\in P$ such that $d(v)=p\}$. Next, we define a notation for a case where $d$ and $d'$ ``agree'' on the points of $P$. That is, for every vertex $u\in V(d,P)$, $d(u)=d'(u)$ and vice versa. In addition, for every point $p\in P$ and an edge $\{u,v\}$, $p$ is on the drawing of $\{u,v\}$ in $d$ if and only if $p$ is on the drawn of $\{u,v\}$ in $d'$. In this case, we say that  $d$ and $d'$ are {\em equal in $P$}.

\begin{definition}[{\bf Drawings Equality in a Set of Points}]\label{def:DrawindEq}
	For two $G^*$-drawings, $d$ and $d'$, and a set of points $P\subset \mathbb{R}\times \mathbb{R}$, we say that $d$ and $d'$ are {\em equal in $P$}, denoted by $d(P)=d'(P)$, if the following conditions are satisfied:
	\begin{enumerate}
		\item For every $u\in V(d,P)$, $d(u)=d'(u)$.
		\item For every $u\in V(d',P)$, $d(u)=d'(u)$.
		\item For every $p\in P$ and $\{u,v\}\in E$, $p\in P(d(\{uv_i,uv_{i+1}\}))$ if and only if $p\in P(d'(\{uv_i,$ $uv_{i+1}\}))$.
	\end{enumerate}         
\end{definition}

%For a $G^*$-drawing $d$, and a frame $f$, we denote by $d[f]$, the $G^*$-drawing induced by $d$ that is on $f$, if the part of $d$ that is drawn on $f$ is a $G^*$ drawing; Otherwise $d[f]=\emptyset$.

Now, we define the concept of a {\em drawing of an info-frame}. Briefly, a drawing $d$ of an info-frame $F$ is a drawing exemplifying the ``information'' encoded by $F$. Since we would like that the drawing would be part of a drawing of the graph $G$, we would like that this drawing would ``make sense'' as in the following example. Let $\{u,v\}\in E$ be an edge, and assume that $u,v\in U_f$, that is, the vertices $u$ and $v$ are to be drawn strictly inside $f$. In addition, assume that $U^*_{\{u,v\}}\cap V(d_f)=\emptyset$, that is, the edge $\{u,v\}$ and the frame $f$ do not intersect. In this case, we expect that $\{u,v\}$ is drawn strictly inside $f$. Therefore, $\{u,v\}\in E(d)$, as stated in Condition \ref{infFramcon5} (e.g., see the edge $\{x,w\}$ in Figure~\ref{fig:InfFDr}). In addition, recall that in Condition \ref{definfFraCon4} of the definition of {\bf Template for an Info-Frame} (Definition \ref{def:infFr}), for every $uv_i$ drawn on a point from $\gis(\fin)\setminus \gps(\fin)$, we store a direction in $F$. In Definition \ref{def:validDir}, we verify that every such vertex corresponding to exactly one edge from the set $E_f$, and has no edges on $f$. Thus, the direction is for this unique edge $e$. In particular, we ask that $\ell(\mathsf{V^*Dir}_f(uv_i),d_f(uv_i))$ is on the first edge of $d(e)$ (e.g., see the direction line $\ell(c, d_f(uv_2))$ and the drawing of the edge $\{uv_2, uv_1\}$ in Figure~\ref{fig:InfFDr}), or the this first edge is on $\ell(\mathsf{V^*Dir}_f(uv_i),d_f(uv_i))$ (e.g., see the direction line $\ell(c', d_f(gh_1))$ (resp., $\ell(c'', d_f(mn_1))$), and the drawing of the edge $\{uv_2, uv_1\}$ (resp., $\{mn_1, j\}$) in Figure~\ref{fig:InfFDr}). This way, we verify that $d(e)$ indeed goes to the same direction as $\mathsf{V^*Dir}_f(uv_i)$, as we meant (see Condition \ref{infFramcon6}). We state these conditions in the next definition.

\begin{definition}[{\bf Drawing of an Info-Frame}] \label{def:infFrDr}
	Let $F=(f,d_f,E_f,U_f,\mathsf{V^*Dir}_f)$ be an info-frame and let $d$ be a $G^*$-drawing. We say that $d$ is a {\em drawing of $F$} if the following conditions are satisfied.
	\begin{enumerate}
		\item $d$ is bounded by $f$ and all the vertices in $V^*$ of $d$ are drawn on $f$. \label{infFramcon1}
		\item $d_f(\pp(f))=d(\pp(f))$. \label{infFramcon2}
		\item $U_f$ is the set of vertices that are drawn strictly inside $f$ in $d$. \label{infFramcon3}
		\item $E_f\subseteq E(d)$ is the set of each edge $e$ that is drawn strictly inside $f$, except maybe at the endpoints, and at least one endpoint of $e$ is drawn on $f$. \label{infFramcon4}
		%\item For every $\{u,v\}\in E$ such that $u<v$, $u\in U_f$ and $v\notin U_f$, $\{uv_1,u\}\in E(d)$.
		%\item For every $\{u,v\}\in E$ such that $u<v$, $u\notin U_f$ and $v\in U_f$, $\{uv_{\mathsf{index}(u,v)},v\}\in E(d)$.
		%\item For every $uv_1\in V(d_f)$, if $u\in U_f$ then $\{uv_1,u\}\in E(d)$.
		%\item For every $uv_1\in V(d_f)$, if $v\in U_f$ then $\{uv_{\mathsf{index}(u,v)},v\}\in E(d)$.
		\item For every $\{u,v\}\in E$ such that $u,v\in U_f$, if $U^*_{\{u,v\}}\cap V(d_f)=\emptyset$, then $\{u,v\}\in E(d)$. \label{infFramcon5}
		\item For every $uv_i\in V(d_f)\cap V^*$ such that $d_f(uv_i)\in \gis(\fin)\setminus \mathsf{GridPointSet}($ $f_{\mathsf{init}})$, $\ell(\mathsf{V^*Dir}_f(uv_i),d_f(uv_i))$ is on $\ell(d_f(uv_i),p_1)$, or $\ell(d_f(uv_i),p_1)$ is on $\ell(\mathsf{V^*Dir}_f(uv_i),$ $d_f(uv_i))$, where $z\in \{uv_{i-1},uv_{i+1}\}$ such that $\{uv_i,z\}\in E_f$ and $d(\{uv_i,z\})=(d(uv_i),p_1,$ $\ldots d(z))$.\label{infFramcon6}
		%\item For every $\{u,v\}\in E$ such that $u,v\in U_f$ and $u<v$, exactly one of the following condition is satisfied: \label{infFramcon5}
		%\begin{itemize}
		%\item $\{u,v\}\in E(d)$. \label{infFramcon51}
		%\item $\{uv_1,u\},\{uv_{\mathsf{index}(u,v)},v\}\in E_d$. \label{infFramcon52}
		%\end{itemize}   
	\end{enumerate}
\end{definition}

%\subsection{Information Function}
%In the previous subsection, we discussed on the ``information'' we want to store with respect to a ``part'' of a drawing, that is, the info-frame that associates with this part. Now, it is natural to define a function, that given a drawing $d$ of a graph $G=(V,E)$, and a frame $f$, the function returns the info-frame that ``describes'' the information of the part of the drawing, that is bounded by $f$.

%We call this function   

%!TEX root =Main-Movement.tex

\subsection{Info-Cutter of an Info-Frame}\label{sec:infoCut}

Recall that our algorithm will build parts of drawings from smaller parts of these drawings. Such parts of drawings are described by the info-frames defined in Section \ref{sec:infoFra}. Given an info-frame $F=(f,d_f,E_f,U_f,\mathsf{V^*Dir}_f)$, we would like to know if there exists a drawing of $F$, having already solved this question for every smaller info-frame, that is, every info-frame $F'=(f',d_{f'},E_{f'},U_{f'},\mathsf{V^*Dir}_{f'})$ where $f'$ has fewer grid points in its interior than $f$. Assume that there exists a drawing  $d$ of $F$. Observe that, if we take a cutter $c$ of $f$, we derive two smaller drawings of $d$ inside $f$: one inside $f_1(c)$ and the other inside $f_2(c)$ (e.g., see Figure~\ref{fig:infCutter2}). These two drawings can be described by two info-frames, $F_1$ and $F_2$, respectively. We say that these info-frames are {\em induced} by $d$ and $c$. We will show that every pair of drawings of these two info-frames can be glued together in order to get a drawing of $F$. In light of this, we would like to ``guess'' (i.e, iterate over) every (small enough) cutter $c$ of $f$ possible for every drawing that can be described by $F$ as well as the two info-frames induced by $d$ and $c$. We denote any such guess by $C=(c,F_1,F_2)$. We will say that $C=(c,F_1,F_2)$ is an {\em info-cutter} of $F$; this term is defined formally later in this subsection.

Now, we define some notations that will be useful later. First, we define a notation weaker than that of equality in Definition \ref{def:DrawindEq}. Towards that, consider a drawing $d$ of an info-frame $F$, and the two sub-drawings of $d$ inside $f_1(c)$ and $f_2(c)$, denoted by $d_1$ and $d_2$, respectively. Observe that the intersection of $f_1(c)$ and $f_2(c)$ is exactly $c$. Therefore, we expect that $d_1$ and $d_2$ ``agree'' on the part of the drawing that intersects both $f_1(c)$ and $f_2(c)$. Still, observe that there might be a point $p\in \gi(c)$ and an edge $\{u,v\}\in E$ such that $(p,\{u,v\})$ is a turning point in $f_1$, but not in $f_2$. In such a case, we will find a vertex in $V^*_{\{u,v\}}$ at the point $p$ in $d_1$, but not in $d_2$. That is, for $d_1$ there exists $i\in\mathbb{N}$ such that $d_1(uv_i)=p$, but for $d_2$ there is no such $i$. In addition, for $d_2$ there is a part of the edge $\{u,v\}$ whose drawing includes the point $p$ in $d_2$. That is, there exists an edge $\{uv_j,uv_{j+1}\}\in E(d_2)$ such that $p\in \gi(d_2(\{uv_j,uv_{j+1}\}))$. Keeping this in mind, we would like to be able to say that the drawings $d_1$ and $d_2$ are equal on the intersection between them, up to the aforementioned difference. Roughly speaking, we say that $d_1(c)=d_2(c)$ {\em up to renaming} (see Figure~\ref{fig:renaming}). This means that on $c$, the drawings are identical, except maybe with respect to the vertices in $V^*$. For the formal definition of this term, recall that we denote by $V(d,P)$ the set of vertices in $V(d)$ that are drawn on points in $P$, that is $V(d,P)=\{v\in V(d)~|~$ there exists $p\in P$ such that $d(v)=p\}$.

\begin{definition}[{\bf Equality Up To Renaming}]\label{def:DERN}
	For a set of points $P$, two $G^*$-drawings, $d$ and $d'$, are {\em equal up to renaming in $P$}, denoted as $d=_{\mathsf{rename}}^Pd'$, if the following conditions hold:
	\begin{enumerate}
		\item For every $u\in (V(d,P)\cup V(d',P))\cap V$, $d(u)=d'(u)$. \label{definition:DERNCon1}
		\item For every $p\in P$ and $\{u,v\}\in E$, there exists $i\in\mathbb{N}$ such that $d(uv_i)=p$ or $p\in \pp(d(\{uv_i,uv_{i+1}\}))$ if and only if there exists $j\in\mathbb{N}$ such that $d(uv_j)=p$ or $p\in \pp(d'(\{uv_j,uv_{j+1}\}))$.\label{def:DERN:Con3}
	\end{enumerate}         
\end{definition}

Now, assume that we have two $G^*$-drawings, $d$ and $d'$, such that $d=_{\mathsf{rename}}^Pd'$ for some set of points $P$. Observe that due to Condition \ref{def:DERN:Con3} of Definition \ref{def:DERN}, some of vertices in $V^*\cap V(d,P)$ might not appear in $d'$ as vertices from the set $V^*$, but as part of an edge, and vice versa (e.g., see vertex $st_1$ in $d_{f_2}$ in Figure~\ref{fig:infCutter2}). We would like to handle this issue by making the drawings more ``similar''. 
To this end, we present the operation $\mathsf{MakeVer}(d,\{u,v\},p)$. Given a $G^*$-drawing $d$, an edge $\{u,v\}\in E$, and a point $p$ such that there exists $i\in \mathbb{N}$ for which $p\in \gi(d(\{uv_i,uv_{i+1}\}))$, the operation ``turns'' this point $p$ into a vertex from the set $V^*$ (e.g., see vertex $st_3$ drawn in green in Figure~\ref{fig:renaming1}). Recall that $\ell(c_j,c_{j+1})$ is the line segment joining the points $c_j$ and $c_{j+1}$. Formally, we define $\mathsf{MakeVer}$ as follows. 

\begin{figure}[!t]
	\centering
	\begin{subfigure}{0.45\textwidth}
		\includegraphics[width = \textwidth, page = 47]{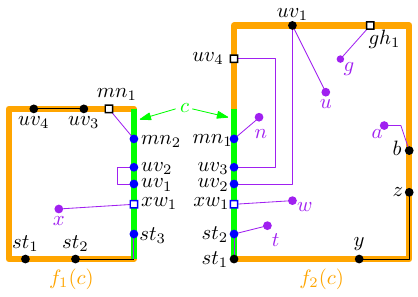}
		\subcaption{}
		\label{fig:infCutter1}
	\end{subfigure}
	\hfil
	\begin{subfigure}{0.45\textwidth}
		\includegraphics[width = \textwidth, page = 48]{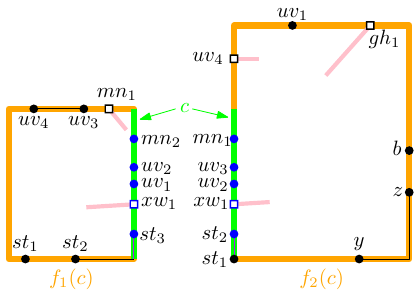}
		\subcaption{}
		\label{fig:infCutter2}
	\end{subfigure}
	
	\caption{Consider the cutter $c$ of $f$, shown in green, and the drawing of $F$ described in Figure~\ref{fig:InfFDr}. Here, we illustrate an example of drawings of $f_1(c)$ and $f_2(c)$ (shown in (a)), and their corresponding info-frames (shown in (b)). The vertices mapped to points in $\gis(f_1) \setminus \mathsf{GridPointSet}(f_1)$ and $\gis(f_2) \setminus \mathsf{GridPointSet}(f_2)$ are marked by hollow squares.} 
	\label{fig:infCutter}
\end{figure}

\begin{definition}[{\bf $\mathsf{MakeVer}$}]\label{def:OpMakeVer}
	Let $d$ be a $G^*$-drawing, let $\{u,v\}\in E$, and \\ let $p\in \gi($ $d)$ such that there exists $i\in \mathbb{N}$ for which \\ $p\in \gi(d(\{uv_i,uv_{i+1}\}))$. Then, $\mathsf{MakeVer}(d,$ $\{u,v\},p)$ performs the following steps on $d$:
	\begin{itemize}
		\item For every $i<\ell\leq \mathsf{index}(u,v)$, rename $uv_\ell$ to $uv_{\ell+1}$. 
		\item Add the vertex $uv_{i+1}$ to $V(d)$, and let $d(uv_{i+1})=p$.
		\item Let $j\in \mathbb{N}$ such that $p\in \ell(c_j,c_{j+1})$ where $d(\{uv_i,uv_{i+1}\})=(c_1,\ldots,c_r)$. Then:
		\begin{itemize}
			\item  Replace the edge $\{uv_i,uv_{i+2}\}$ by the edges $\{uv_i,uv_{i+1}\}$ and $\{uv_{i+1},uv_{i+2}\}$ in $E(d)$.
			\item Update $d(\{uv_i,uv_{i+1}\})=(c_1,\ldots,c_j,p)$ and $d(\{uv_{i+1},uv_{i+2}\})=(p,c_{j+1}\ldots,c_r)$ in $d$.
		\end{itemize} 
	\end{itemize}
\end{definition}

It is easy to see that after using this operation, the resulting drawing $d$ satisfies the conditions of the definition of a $G^*$-drawing (Definition \ref{def:gstdr}): 

\begin{observation} \label{obs:makever}
	Let $d$ be a $G^*$-drawing, let $\{u,v\}\in E$, and let $p\in \gi(d)$ such that there exists $i\in \mathbb{N}$ for which $p\in \gi(d(\{uv_i,uv_{i+1}\}))$. Then, the drawing obtained by $\mathsf{MakeVer}(d,\{u,v\},p)$ is a $G^*$-drawing.
\end{observation}

Recall that when we examine $d$ and $d'$ such that $d=_{\mathsf{rename}}^Pd'$, we would like to make them as similar as possible. To this end, we will use the operation $\mathsf{MakeVer}(d,\{u,v\},p)$ for every edge $\{u,v\}$ and a point $p\in P$ such that there exists $i\in \mathbb{N}$ for which $d'(uv_i)=p$, and there is no $j\in \mathbb{N}$ such that $d(uv_j)=p$, or vice versa. We refer to this overall operation as $\mathsf{MakeAllVer}$:

\begin{definition}[{\bf$\mathsf{MakeAllVer}$}]\label{def:OpMakeVerFA}
	Let $P\subset \mathbb{R}\times \mathbb{R}$ be a set of points and let $d$ and $d'$ be $G^*$-drawings such that $d=_{\mathsf{rename}}^Pd'$. Then, $\mathsf{MakeAllVer}(d,d',P)$ performs the following steps:
	\begin{itemize}
		\item For every edge $\{u,v\}\in E$ and a point $p\in P$, such that there exists $i\in \mathbb{N}$ such that $d'(uv_i)=p$, and there is no $j\in \mathbb{N}$ such that $d(uv_j)=p$, activate $\mathsf{MakeVer}(d,\{u,v\},p)$.
		\item For every edge $\{u,v\}\in E$ and a point $p\in P$, such that there exists $i\in \mathbb{N}$ such that $d(uv_i)=p$, and there is no $j\in \mathbb{N}$ such that $d'(uv_j)=p$, activate $\mathsf{MakeVer}(d',\{u,v\},p)$.
	\end{itemize}
	Let $\hat{d}$ and $\hat{d'}$ the result of the application of $\mathsf{MakeAllVer}(d,d',P)$.
\end{definition}

\begin{figure}[!t]
	\centering
	\begin{subfigure}{0.48\textwidth}
		\includegraphics[height = 0.55\textwidth, page = 49]{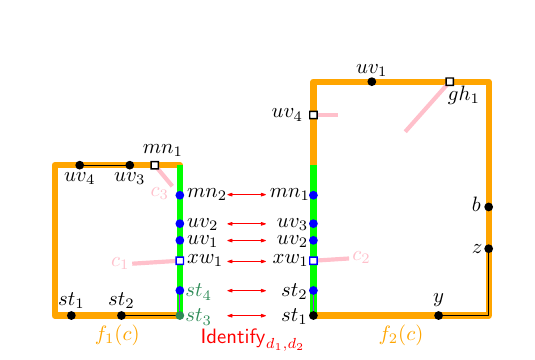}
		\subcaption{}
		\label{fig:renaming1}
	\end{subfigure}
	\hspace{-1.8cm}
	\begin{subfigure}{0.48\textwidth}
		\includegraphics[height = 0.55\textwidth, page = 50]{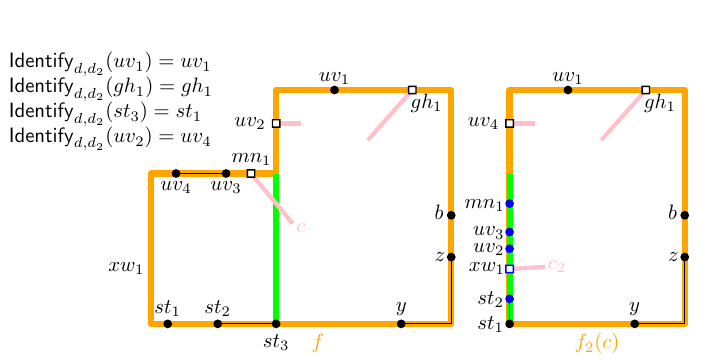}
		\subcaption{}
		\label{fig:renaming2}
	\end{subfigure}
	
	\caption{An example of renaming. The vertices mapped to points in $\gis(f) \setminus \mathsf{GridPointSet}(f)$, $\gis(f_1) \setminus \mathsf{GridPointSet}(f_1)$ and $\gis(f_2) \setminus \mathsf{GridPointSet}(f_2)$ are marked by hollow squares. (a) Observe that $d_{f_1}=d_{f_2}$ up to renaming, where $d_{f_1}$ and $d_{f_2}$ are the drawings in Figure~\ref{fig:infCutter2}. Moreover, $\mathsf{Identify}_{d_1,d_2}$ is the corresponding vertex identification function. Note that, the identification function is only shown for the vertices which are not mapped to $\mathsf{Null}$. Observe that in $d_{f_1}$, the vertex $st_3$ of $d_{f_1}$ in Figure~\ref{fig:infCutter2} is renamed $st_4$ (shown in green here) and a new vertex, shown here in green as well and labeled $st_3$, is added, by the $\mathsf{MakeVer}$ function. (b) Observe that $d_f=d_{f_2}$ up to renaming, where $d_f$ and $d_{f_2}$ are the drawings in Figures~\ref{fig:InfF} and~\ref{fig:infCutter2}, respectively. Moreover, $\mathsf{Identify}_{d,d_2}$ is the corresponding vertex identification function. Note that, the identification function is only shown for the vertices which are not mapped to $\mathsf{Null}$.} 
	\label{fig:renaming}
\end{figure}

%$P_{1,2},P_{1,3},\ldots P_{1,k},P_{2,3},P_{2,4},\ldots,P_{2,k},\ldot,P_{i,i+1},\ldots,P_{i,k},\ldots,P_{k-1,k}$ 

%Again, let $P\subset\mathbb{R}\times\mathbb{R}$ be a set of points, let $d$ and $d'$ be two $G^*$-drawings such that $d=_{\mathsf{rename}}^Pd'$. Assume that for every $uv_i\in U^*_{\{u,v\}}\cap V(d,P)$ there exists $uv_j\in U^*_{\{u,v\}}\cap V(d',P)$ such that $d(uv_i)=d'(uv_j)$, and vice versa. Then, we define a function $\mathsf{Identify}:V(d,P)\cap V^*\rightarrow V(d',P)\cap V^*$ that, given $uv_i\in V(d,P)\cap V^*$, returns $uv_j\in V(d,P)\cap V^*$ such that $d(uv_i)=d'(uv_j)$. We refer to this function as the {\em vertex identification function} of $d$ and $d'$ in $P$. Observe that if such a function exists, then it is unique. Formally, we define this function as follows.

Now, let $d$ and $d'$ be $G^*$-drawings. For later use, we would like to get for every $uv_i\in V(d)\cap V^*$, its corresponding vertex in $d'$, if such vertex exists. That is, a vertex $uv_j\in V(d')\cap V^*$ such that $d(uv_i)=d'(uv_j)$. To this end, we define the function $\mathsf{Identify}_{d,d'}:V(d)\cap V^*\rightarrow (V(d')\cap V^*)\cup \{\mathsf{Null}\}$ as follows. Given $uv_i\in V(d)\cap V^*$, the function returns $uv_j\in V(d)\cap V^*$ such that $d(uv_i)=d'(uv_j)$, if such a vertex exists; otherwise, it returns $\mathsf{Null}$. Observe that if such a vertex $uv_j$ exists, then it is unique. We refer to this function as the {\em vertex identification function} of $d$ and $d'$. Formally, we define this function as follows.

\begin{definition}[{\bf Vertex Identification Function}]\label{def:VNF}
	Let $d$ and $d'$ be $G^*$-drawings. The {\em vertex identification function} of $d$ and $d'$ is the function $\mathsf{Identify}_{d,d'}:V(d)\cap V^*\rightarrow (V(d')\cap V^*)\cup \{\mathsf{Null}\}$ such that for every $uv_i\in V(d)\cap V^*$, $\mathsf{Identify}_{d,d'}(uv_i)=uv_j$ for which $d(uv_i)=d'(uv_j)$, if such $uv_j$ exists; otherwise, $\mathsf{Identify}_{d,d'}(uv_i)=\mathsf{Null}$.
\end{definition}

Observe that the only modifications to the drawings $d$ and $d'$ made by $\mathsf{MakeAllVer}(d,d',P)$ are done by activating $\mathsf{MakeVer}$. So, due to Observation \ref{obs:makever}, $\hat{d}$ and $\hat{d'}$ (defined in Definition \ref{def:OpMakeVerFA}) are $G^*$-drawings. In addition, observe that for every $uv_i\in V(\hat{d},P)\cap V^*$ there exists $uv_j\in V(\hat{d'},P)\cap V^*$ such that $\hat{d}(uv_i)=\hat{d'}(uv_j)$, and vice versa. Thus, we have the following observation:

\begin{observation} \label{obs:makever2}
	Let $P\subset\mathbb{R}\times\mathbb{R}$ be a set of points, let $d$ and $d'$ be two $G^*$-drawings such that $d=_{\mathsf{rename}}^Pd'$. Then, $\hat{d}$ and $\hat{d'}$ are $G^*$-drawings. In addition, for every $uv_i\in V(\hat{d},P)\cap V^*$, $\mathsf{Identify}_{\hat{d},\hat{d'}}(uv_i)\neq \mathsf{Null}$ and for every $uv_i\in V(\hat{d'},P)\cap V^*$, $\mathsf{Identify}_{\hat{d'},\hat{d}}(uv_i)\neq \mathsf{Null}$ 
\end{observation}

We now generalize the function $\mathsf{MakeAllVer}$ (from Definition \ref{def:OpMakeVerFA}) to $k$ $G^*$-drawings, $d_1,d_2,$ $\ldots,d_k$ such that for every $1\leq i<j\leq k$, $d_i=_{\mathsf{rename}}^{P_{i,j}}d_{j}$ where $P_{i,j}\subset \mathbb{R}\times \mathbb{R}$. Our goal is to obtain $k$ $G^*$-drawings $d_1^*,\ldots d_k^*$, such that the following property holds: for every $1\leq i,j\leq k$ and $uv_\ell\in V(d_i^*,P_{i,j})\cap V^*$, it holds that $\mathsf{Identify}_{d_i^*,d_j^*}(uv_\ell)\neq \mathsf{Null}$. To this end, we define the function $\mathsf{MakeAllVerSet}$ (defined in Algorithm \ref{alg:MAVS}). For every $1\leq i<j\leq k$ iteratively, this function evokes $\mathsf{MakeAllVer}(d_i^*,d_j^*,P_{i,j})$. Initially, $d_i^*=d_i$ for each $i\in [k]$, and afterwards $d^*_i$ is updated $k-1$ times. By repeatedly updating $d_i^*$, for each $i\in[k]\setminus \{i\}$, when $\mathsf{MakeAllVerSet}$ terminates, we conclude that the aforementioned property holds.

\begin{definition}[{\bf$\mathsf{MakeAllVerSet}$}]\label{def:OpMakeVerFAS}
	Let $d_1,d_2,\ldots,d_k$ be $G^*$-drawings, and, for every $1\leq i<j\leq k$, let $P_{i,j}\subset \mathbb{R}\times \mathbb{R}$ such that $d_i=_{\mathsf{rename}}^{P_{i,j}}d_{j}$. Then, $\mathsf{MakeAllVerSet}(d_1,d_2,\ldots,d_k,P_{1,2},$ $P_{1,3},\ldots P_{1,k},$ $\ldots,P_{i,i+1},\ldots,P_{i,k},\ldots,P_{k-1,k})$ is described by Algorithm \ref{alg:MAVS}.
\end{definition}

\begin{algorithm}[!t]
	\SetKwInOut{Input}{Input}
	\SetKwInOut{Output}{Output}
	\medskip
	\For{$1\leq i\leq k$}
	{
		$d^*_i\gets d_i$\;
	}
	\For{$1\leq i<$k}
	{\For{$i<j\leq k$}{
			Let $\hat{d^*_i}$ and $\hat{d^*_j}$ the result of the application of $\mathsf{MakeAllVer}(d^*_i,d^*_j,P_{i,j})$\;
			$d^*_i\gets \hat{d^*_i}$\;
			$d^*_j\gets \hat{d^*_j}$\;
		}
	}
	\Return$(d_1^*,\ldots, d_k^*)$\;
	
	\caption{$\mathsf{MakeAllVerSet}(\langle d_1,d_2,\ldots,d_k,P_{1,2},P_{1,3},\ldots P_{1,k},\ldots,P_{i,i+1},\ldots,P_{i,k},\ldots,P_{k-1,k} \rangle)$}
	\label{alg:MAVS}
\end{algorithm}

From Observation \ref{obs:makever2}, after every iteration of Algorithm \ref{alg:MAVS}, for every $i\in[k]$, $d_i^*$ is a $G^*$-drawing. So, for every $i\in[k]$, $d_i^*$ returned by $\mathsf{MakeAllVerSet}$ is a $G^*$-drawing. In addition, from Observation \ref{obs:makever2}, for every $1\leq i,j\leq k$ and $uv_\ell\in V(d_i^*,P_{i,j})\cap V^*$,  $\mathsf{Identify}_{d_i^*,d_j^*}(uv_\ell)\neq \mathsf{Null}$, where $d_i^*$ and $d_j^*$ are the drawings returned by $\mathsf{MakeAllVerSet}$. for each $i\in [k]$. Thus, we have the following observation:

\begin{observation} \label{obs:makever3}
	Let $d_1,d_2,\ldots,d_k$ be $G^*$-drawings, and, for every $1\leq i<j\leq k$, let $P_{i,j}\subset \mathbb{R}\times \mathbb{R}$ such that $d_i=_{\mathsf{rename}}^{P_{i,j}}d_{j}$. Let $d_i^*$, for each $i\in [k]$, be the result of the application $\mathsf{MakeAllVerSet}(d_1,d_2,\ldots,d_k,P_{1,2},P_{1,3},\ldots P_{1,k},$ $\ldots,P_{i,i+1},\ldots,P_{i,k},\ldots,P_{k-1,k})$. Then, for each $i\in[k]$, $d_i^*$ is a $G^*$-drawing, and for every $1\leq i,j\leq k$ and $uv_\ell\in V(d_i^*,P_{i,j})\cap V^*$, it holds that $\mathsf{Identify}_{d_i^*,d_j^*}(uv_\ell)\neq \mathsf{Null}$.
\end{observation}

%Now, let $d$ and $d'$ be two $G^*$-drawings such that  $d=_{\mathsf{rename}}^Pd'$. We say that $d=_{\mathsf{rename}}^Pd'$ with the identification function $\mathsf{Identify}$, where $\mathsf{Identify}$ is the identification function of $\hat{d}$ and $\hat{d'}$ in $P$. Observe that due to Observation \ref{obs:makever2}, there exists such an identification function, so it is well defined.

%\begin{definition}[{\bf Drawings Equality in a Set of Points Up To Renaming $\mathsf{Identify}$}]\label{def:DERNR}
%Let $P\subset\mathbb{R}\times\mathbb{R}$ be a set of points, let $d$ and $d'$ be two $G^*$-drawings such that $d$ and $d'$ are equal up to renaming in $P$. We say that $d$ and $d'$ are {\em equal up to renaming $\mathsf{Identify}$ in $P$} or  {\em $d(P)=d'(P)$ up to renaming $\mathsf{Identify}$} if $d$ and $d'$ are equal up to renaming in $P$. Then, we preform $\mathsf{MakeAllVer}(d,d',P)$, and $\mathsf{Identify}$ is the vertex identification function of $d$ and $d'$ in $P$.
%\end{definition}

%!TEX root =Main-Movement.tex
\subsection{Properties of Info-Cutter} \label{subsec:prpInoCu}

Now, we proceed to define the term {\em info-cutter}. Towards that, we begin with the definition of its first property. Let $F=(f,d_f,E_f,U_f,\mathsf{V^*Dir}_f)$ be an info-frame, and let $C=(c,F_1=(f_1(c),d_{f_1},E_{f_1},U_{f_1},\mathsf{V^*Dir}_{f_1}),F_2=(f_2(c),d_{f_2},E_{f_2},U_{f_2},\mathsf{V^*Dir}_{f_2}))$. Recall (from the discussion at the beginning of Section \ref{sec:infoCut}) that $C$ should ``reflect'' a guess of a cutter $c$ of $f$, and the two info-frames induced by a drawing $d$ of $F$ and $c$. So, we first demand that $c$ is indeed a cutter of $f$, and $F_1$ and $F_2$ are info-frames: 

\begin{definition}[{\bf Template for an Info-Cutter}] \label{def:infCTurPoints0}
	Let $F=(f,d_f,E_f,U_f,\mathsf{V^*Dir}_f)$ be an info-frame. Then, $C=(c,F_1=(f_1(c),d_{f_1},E_{f_1},U_{f_1},\mathsf{V^*Dir}_{f_1}),F_2=(f_2(c),d_{f_2},E_{f_2},U_{f_2},$ $\mathsf{V^*Dir}_{f_2}))$ is an {\em \bf Info-Cutter Template} with respect to $F$ if $c$ is a cutter of $f$, and $F_1$ and $F_2$ are info-frames.
	%\label{def:infCutCon1}
	% \label{def:infCutCon2}
\end{definition}

Now, recall that since $c$ is a cutter of $f$, which splits a drawing of $F$ into two drawings, we expect that the ``common part'' of the drawings of $F_1$ and $F$ will be identical up to renaming. This yields Condition \ref{def:infCutCon3} in Definition \ref{def:infCTurPoints1} ahead. Similarly, we expect that the common part of the drawings of $F_2$ and $F$, and also of the drawings of $F_1$ and $F_2$, will be identical up to renaming (see the identify functions shown in Figures~\ref{fig:renaming2} and~\ref{fig:renaming1}, respectively). So, we have Conditions \ref{def:infCutCon4} and \ref{def:infCutCon5}.

\begin{definition}[{\bf Equality of Common Parts}] \label{def:infCTurPoints1}
	Let $F=(f,d_f,E_f,U_f,\mathsf{V^*Dir}_f)$ be an info-frame. Then, $C=(c,F_1=(f_1(c),d_{f_1},E_{f_1},U_{f_1},\mathsf{V^*Dir}_{f_1}),F_2=(f_2(c),d_{f_2},E_{f_2},U_{f_2},$ $\mathsf{V^*Dir}_{f_2}))$ exhibits {\em \bf Equality of Common Parts Property} with respect to $F$ if the following conditions are satisfied:
	\begin{enumerate}
		\item $d_f=_{\mathsf{rename}}^{\mathsf{Common}(f,f_1)}d_{f_1}$, where $\mathsf{Common}(f,f_1)=\gi(f)\cap \gi($ $f_1(c))$.\label{def:infCutCon3}
		\item $d_f=_{\mathsf{rename}}^{\mathsf{Common}(f,f_2)}d_{f_2}$, where $\mathsf{Common}(f,f_2)=\gi(f)\cap \gi($ $f_2(c))$. \label{def:infCutCon4}
		\item $d_{f_1}=_{\mathsf{rename}}^{\gi(c)}d_{f_2}$. \label{def:infCutCon5}
	\end{enumerate}
\end{definition}

In the following property, we consider the directions of vertices drawn on points from $\mathsf{GridPointSet}(f_{\mathsf{init}})\setminus \gp(\fin)$. First, consider such vertices that are drawn on $c$. Observe that these vertices are in the common part of $f_1(c)$ and $f_2(c)$. As explained in Section \ref{sec:infoFra}, we would like the directions of these vertices to be identical. That is, for a vertex $uv_i\in V(d_{f_1},\gi(c))\cap V^*$, the direction in $F_1$ is identical to that of the vertex $\mathsf{Identify}_{d_f,d_{f_1}}(uv_i)$ in $F_2$. Now, recall that the direction $\mathsf{V^*Dir}_{f_1}(uv_i)=a$ represents the line $\ell(a,d_{f_1}(uv_i))$, which is inside $f_1(c)$, and hence outside $f_2(c)$ (e.g., see $\mathsf{V^*Dir}_{f_1}(xw_1)=c_1$ in Figure~\ref{fig:renaming1}). Similarly, $\mathsf{V^*Dir}_{f_2}(\mathsf{Identify}_{d_{f_1},d_{f_2}}(uv_i))=b$ represents the line $\ell(b,d_{f_2}(\mathsf{Identify}_{d_{f_1},d_{f_2}}($ $uv_i)))$, which is inside $f_2(c)$, and hence outside $f_1(c)$ (e.g., see $\mathsf{V^*Dir}_{f_1}(xw_1)=c_2$ in Figure~\ref{fig:renaming1}). We expect that if we attach these two lines together, we get a straight line, so there is no bending at the point $d_{f_1}(uv_i)$ (e.g., $c_1$, $xw_1$ and $c_2$ should be collinear in Figure~\ref{fig:renaming1}). Equivalently, we verify that $d_{f_1}(uv_i)$ is on $\ell(a,b)$. Now, consider a vertex $uv_i$ drawn on a point from $\mathsf{GridPointSet}(f_{\mathsf{init}})\setminus \gp(\fin)$ and drawn on $f$. That is, a vertex representing a turning point of $\{u,v\}$ in $f$, drawn not on a grid point. Assume that $uv_i$ is drawn also on $f_1(c)$ (e.g., see the vertex $mn_1$ in Figure~\ref{fig:renaming1} and~\ref{fig:renaming2}). This vertex has directions in both $F$ and $F_1$, representing the direction of the edge attached to $uv_i$ from $d_f(uv_i)$ to the interior of $f$ (e.g., see $c$ in Figure~\ref{fig:renaming2}), and also to the interior of $f_1(c)$ (e.g., see $c_3$ in Figure~\ref{fig:renaming1}). Therefore, we expect these two directions to be equal. Nevertheless, since $f_1(c)$ is smaller and contained inside $f$, the direction of $uv_i$ in $F$ might be too large and it gets out of $f_1(c)$ (e.g., see the direction line $\ell(c, d_f(mn_1))$ in Figure~\ref{fig:renaming2}). Therefore, we ask the direction of $uv_i$ in $F_1$ to be on its direction in $F$ (e.g., the line $\ell(c_3, d_{f_1}(mn_1))$ shown in Figure~\ref{fig:renaming1} should be contained in the line $\ell(c, d_f(mn_1))$ shown in Figure~\ref{fig:renaming2}). Observe that, for our purpose it is enough, since we only aim to make sure that there is no bending of the edge $\{u,v\}$ at the point $d_f(uv_i)$.  

\begin{definition}[{\bf Equality of $V^*$ Directions}] \label{def:infCTurPoints2}
		Let $F=(f,d_f,E_f,U_f,\mathsf{V^*Dir}_f)$ be an info-frame. Then, $C=(c,F_1=(f_1(c),d_{f_1},E_{f_1},U_{f_1},\mathsf{V^*Dir}_{f_1}),F_2=(f_2(c),d_{f_2},E_{f_2},U_{f_2},\mathsf{V^*Dir}_{f_2}))$ exhibits {\em \bf Equality of $V^*$ Directions} with respect to $F$ if the following conditions are satisfied: 
	\begin{enumerate}
		\item $C$ exhibits {\bf Equality of Common Parts Property}.
		\item For every $uv_i\in V(d_{f_1},\gi(c))\cap V^*$ such that\\ $d_{f_1}(uv_i)\in \gis($ $\fin)\setminus \mathsf{GridPointSet}(f_{\mathsf{init}})$, $d_{f_1}(uv_i)$ is on $\ell(a,b)$, where $\mathsf{V^*Dir}_{f_1}(uv_i)=a$ and $\mathsf{V^*Dir}_{f_2}($ $\mathsf{Identify}_{d_{f_1},d_{f_2}}$ $(uv_i))=b$.\label{con2DefDir}
		\item For every $uv_i\in V(d_{f_1},\mathsf{Common}(f,f_1))\cap V^*$, where $\mathsf{Common}(f,f_1)=\gi($ $f)\cap \gi(f_1(c))$ such that\\ $d_{f_1}(uv_i)\in \gis(\fin)\setminus \mathsf{GridPointSet}$ $(f_{\mathsf{init}})$, $\ell(d_{f_1}(uv_i),\mathsf{V^*Dir}_{f_1}(uv_i))$ is on $\ell(d_f(\mathsf{Identify}_{d_{f_1},d_{f}}(uv_i)),\mathsf{V^*Dir}_{f}(\mathsf{Identify}_{d_{f_1},d_{f}}(uv_i)))$ or $\ell(d_f$ $(\mathsf{Identify}_{d_{f_1},d_{f}}$ $(uv_i)),\mathsf{V^*Dir}_{f}$ $(\mathsf{Identify}_{d_{f_1},d_{f}}(uv_i)))$ is on $\ell(d_{f_1}(uv_i),\mathsf{V^*Dir}_{f_1}(uv_i))$.
		\item For every $uv_i\in V(d_{f_2},\mathsf{Common}(f,f_2))\cap V^*$, where $\mathsf{Common}(f,f_2)=\gi($ $f)\cap \gi(f_1(c))$ such that\\ $d_{f_2}(uv_i)\in \gis(\fin)\setminus \mathsf{GridPointSet}$ $(f_{\mathsf{init}})$, $\ell(d_{f_2}(uv_i),\mathsf{V^*Dir}_{f_2}(uv_i))$ is on $\ell(d_f(\mathsf{Identify}_{d_{f_2},d_{f}}(uv_i)),\mathsf{V^*Dir}_{f}(\mathsf{Identify}_{d_{f_2},d_{f}}(uv_i)))$ or $\ell(d_f$ $(\mathsf{Identify}_{d_{f_2},d_{f}}(uv_i)),\mathsf{V^*Dir}_{f}($ $\mathsf{Identify}_{d_{f_2},d_{f}}(uv_i)))$ is on $\ell(d_{f_2}(uv_i),\mathsf{V^*Dir}_{f_2}(uv_i))$.
	\end{enumerate}
\end{definition}

Next, recall that $U_f$ is the set of vertices drawn strictly inside $f$ in every drawing of $F$. So, for every drawing $d$ of $F$ and for every $u\in U_f\subseteq V$, we expect that exactly one of the following cases holds:

\smallskip\noindent{\bf Case 1.} $u$ is drawn strictly inside $f_1(c)$. Then, as $U_{f_1}$ is the set of vertices of $V$ that are drawn strictly inside the sub-drawing of $d$ bounded by $f_1(c)$, $u\in U_{f_1}$.

\smallskip\noindent{\bf Case 2.} $u$ is drawn strictly inside $f_2(c)$. Then, as $U_{f_2}$ is the set of vertices of $V$ that are drawn strictly inside the sub-drawing of $d$ bounded by $f_2(c)$, $u\in U_{f_2}$.

\smallskip\noindent{\bf Case 3.} $u$ is drawn on $c$. Observe that, since the first and last vertices of $c$ are drawn on $f$, $\gi(c)\cap \gi(f)\neq \emptyset$. So, since $u$ is drawn strictly inside $f$ and also on $c$, $d(u)=p$ for some $p\in \gi(c)\setminus \gi(f)$. Therefore, since $d_{f_1}$ is the drawing on $f_1(c)$, $u\in V(d_{f_1},\gi(c)\setminus \gi(f))$. Symmetrically, $u\in V(d_{f_2},\gi(c)\setminus \gi(f))$. However, since $d_{f_1}=_{\mathsf{rename}}^{\gi(c)}d_{f_2}$ (Condition \ref{def:infCutCon5} of Definition \ref{def:infCTurPoints1}), $V(d_{f_1},\gi(c))\cap V= V(d_{f_2},\gi($ $c))\cap V$. 

So $u\in V(d_{f_1},\gi(c)\setminus \gi(f))$ implies \\ $u\in V(d_{f_2},\gi(c)\setminus \gi(f))$, and vice versa (Condition \ref{definition:DERNCon1} of Definition \ref{def:DERN}). 

Thus, we conclude that $U_f=(V(d_{f_1},\gi(c)\setminus \gi(f))) \cup U_{f_1}\cup U_{f_2}$ (Condition \ref{def:infCutCon6} of Definition \ref{def:infCTurPoints13}). Moreover, we expect that these three sets will be distinct. Observe that $V(d_{f_1},
\gi(c)\setminus \gi(f))\cap U_{f_1}=\emptyset$ and $V(d_{f_1},$ $\gi(c)\setminus \gi(f))\cap U_{f_2}=\emptyset$ trivially, since $F_1$ and $F_2$ are info-frames (Condition \ref{definfFraCon3a} of Definition \ref{def:infFr}). So, we only need to demand that $U_{f_1}\cap U_{f_2}=\emptyset$ (Condition \ref{def:infCutCon62} of Definition \ref{def:infCTurPoints13}).

\begin{definition}[{\bf Partition of $U_f$}] \label{def:infCTurPoints13}
	Let $F=(f,d_f,E_f,U_f,\mathsf{V^*Dir}_f)$ be an info-frame. Then, $C=(c,F_1=(f_1(c),d_{f_1},E_{f_1},U_{f_1},\mathsf{V^*Dir}_{f_1}),F_2=(f_2(c),d_{f_2},E_{f_2},U_{f_2},\mathsf{V^*Dir}_{f_2}))$ is {\em \bf $U_f$ - Partitioned} with respect to $F$ if: 
	\begin{enumerate}
		\item $U_f=V(d_{f_1},\gi(c)\setminus \gi(f)) \cup U_{f_1}\cup U_{f_2}$.\label{def:infCutCon6}
		\item $U_{f_1}\cap U_{f_2}=\emptyset$. \label{def:infCutCon62}
	\end{enumerate}
\end{definition}

% For the sake of simplicity, we use the renaming functions $R$, $R_1$ and $R_2$ as follows. The function $R_1$ ($R_2$), gets a $V^*$ vertex and returns

We continue with another property of info-cutters, called {\bf Partition of $E_f$}. Recall that an edge in $E_f$ is an edge drawn strictly inside $f$ except for at least one of its endpoints, which is drawn on $f$, in every drawing of $F$. Here, we consider the way we expect edges from $E_f$ to ``influence'' the info-frames $F_1$ and $F_2$. For the sake of readability, we split this  property into three similar properties. We begin with Part (a) of {\bf Partition of $E_f$}: the partition of edges in $E_f$ with both endpoints on $f$. Let $\{uv_i,uv_{i+1}\}\in E_f$ such that $uv_i,uv_{i+1}\in V(d_f)$, that is, $\{uv_i,uv_{i+1}\}$ is an edge drawn strictly inside $f$ except at the endpoints, which are drawn on $f$, in every drawing $d$ of $F$ (e.g., see the edge $\{uv_1,uv_2\}$ in Figure~\ref{fig:InfFDr}). Recall that for a drawing $d$, $\pp(d)$ is the set of points in the plane that intersect $d$. We denote by $\mathsf{Common}(f,f_1)$, $\mathsf{Common}(f,f_2)$ and $\mathsf{Common}(f_1,f_2)$ the sets $\gi(f)\cap \gi(f_1(c))$, $\gi(f)\cap \gi(f_2(c))$ and\\ $\gi(c)$, respectively. Let $d_f^*,d_{f_1}^*,d_{f_2}^*$ be the result of the application of $\mathsf{MakeAllVerSet}(d_f,d_{f_1},d_{f_2},\mathsf{Common}(d_f,d_{f_1}),$\\ $\mathsf{Common}(d_f,d_{f_2}),\mathsf{Common}(d_{f_1},d_{f_2}))$ (see Definition \ref{def:OpMakeVerFAS}). Observe that $d_f^*,d_{f_1}^*,d_{f_2}^*$ are well defined since we assume the satisfaction of the {\bf Equality of Common Parts} property (see Definition \ref{def:infCTurPoints1}).   Considering the two frames $f_1(c)$ and $f_2(c)$, we expect that exactly one of the following conditions is satisfied.

%In addition, due to Observation \ref{obs:makever3}, the corresponding identification functions are not empty sets.

\smallskip\noindent{\bf Condition \ref{infcutedgeCon712}.} The edge $\{uv_i,uv_{i+1}\}$ is drawn strictly inside $f_1(c)$ except for both its endpoints, and therefore $\{\mathsf{Identify}_{d_f,d_{f_1}}(uv_i),\mathsf{Identify}_{d_f,d_{f_1}}(uv_{i+1})\}\in E_{f_1}$. Observe that in this case, since the edge is drawn strictly inside $f_1(c)$ except for both its endpoints, then the endpoints are turning points of $\{u,v\}$ in $f_1(c)$. Thus, there are vertices from the set $V^*_{\{u,v\}}$ drawn in these points in $d_{f_1}$, so $\mathsf{Identify}_{d_f,d_{f_1}}(uv_i)\neq \mathsf{Null}$ and $\mathsf{Identify}_{d_f,d_{f_1}}(uv_{i+1})\neq \mathsf{Null}$.

%   and therefore  We give now intuition for why $\mathsf{Identify}_{d_f,d_{f_1}}(uv_i)\neq \mathsf{Null}$ (and similarly $\mathsf{Identify}_{d_f,d_{f_1}}(uv_{i+1})\neq \mathsf{Null}$). Observe that since Observe, that since we assume the {\bf Equality of Common Parts Property}, and also so one of the following conditions holds:
%\begin{itemize}
%\item There exists $uv_j\in V(d_{f_1})\cap V^*$ such that $uv_i=uv_j$, and therefore $\{\mathsf{Identify}_{d_f,d_{f_1}}(uv_i)=uv_j$, and we are done.
%\item The first condition does not hold, and there exists an edge $\{uv_j,uv_{j+1}\}\in E(d_{f_1})\cap E^*$ such that $f(uv_i)\in \pp(d_{f_1}(\{uv_j,uv_{j+1}\}))$

%\end{itemize}

\smallskip\noindent{\bf Condition \ref{infcutedgeCon713}.} The edge $\{uv_i,uv_{i+1}\}$ is drawn strictly inside $f_2(c)$ except for both its endpoints, and therefore $\{\mathsf{Identify}_{d_f,d_{f_2}}(uv_i),\mathsf{Identify}_{d_f,d_{f_2}}(uv_{i+1})\}\in E_{f_2}$. Observe that, it can be shown that $\mathsf{Identify}_{d_f,d_{f_2}}(uv_i)\neq \mathsf{Null}$ and $\mathsf{Identify}_{d_f,d_{f_2}}(uv_{i+1})\neq \mathsf{Null}$ similarly to the previous condition.

\smallskip\noindent{\bf Condition \ref{infcutedgeCon714}.} The edge $\{uv_i,uv_{i+1}\}$ is drawn on $c$, that is,  $\{\mathsf{Identify}_{d_f,d_{f_1}^*}(uv_i),\mathsf{Identify}_{d_f,d_{f_1}^*}($ $uv_{i+1})\} \in E(d_{f_1}^*,\pp(c))$. Observe that, from Observation \ref{obs:makever3}, for every $uv_q\in V(d_f^*,$ $\mathsf{Common}(f,$ $f_1))\cap V^*$, $\mathsf{Identify}_{d_f^*,d_{f_1}^*}(uv_{q})\neq \mathsf{Null}$. Now, since $\mathsf{Identify}_{d_f,d_f^*}(uv_{q})\neq \mathsf{Null}$ and $uv_i,uv_{i+1}\in V(f,\mathsf{Common}(d_f,d_{f_1}))\cap V^*$, we get that $\mathsf{Identify}_{d_f,d_{f_1}^*}(uv_{i})\neq \mathsf{Null}$ and $\mathsf{Identify}_{d_f,d_{f_1}^*}($ $uv_{i+1})\neq \mathsf{Null}$ 

\smallskip\noindent{\bf Condition \ref{infcutedgeCon715}.} The edge $\{uv_i,uv_{i+1}\}$ has $\ell\in \mathbb{N}$ turning points in $f_1(c)$ other than $(d(uv_i),\{u,v\})$ and $(d(uv_{i+1}),\{u,v\})$ (which may not be turning points in $f_1(c)$). We further discuss additional constraints satisfied in this case later.

%The vertex identification functions $\mathsf{Identify}$, $\mathsf{Identify}_1$ and $\mathsf{Identify}_2$ we use in the following definition, are defined in Conditions \ref{def:infCutCon3}, \ref{def:infCutCon4} and \ref{def:infCutCon5} of the definition of equality of common parts property (Definition \ref{def:infCTurPoints1}).

\begin{definition}[{\bf Partition of $E_f$ With Both Endpoints on $f$}] \label{def:infCTurPoints2a}
	Let $F=(f,d_f,E_f,U_f,$ $\mathsf{V^*Dir}_f)$ be an info-frame. Let $C=(c,F_1=(f_1(c),d_{f_1},E_{f_1},U_{f_1},\mathsf{V^*Dir}_{f_1}),F_2=(f_2(c),d_{f_2},$ $E_{f_2},U_{f_2},\mathsf{V^*Dir}_{f_2}))$ exhibits {\bf Equality of Common Parts} with respect to $F$. Let $d_f^*,d_{f_1}^*,d_{f_2}^*$ be the result of the application of $\mathsf{MakeAllVerSet}(d_f,d_{f_1},d_{f_2},\mathsf{Common}(d_f,d_{f_1}),\mathsf{Common}(d_f,$ $d_{f_2}),\mathsf{Common}($ $d_{f_1},d_{f_2}))$ (see Definition \ref{def:OpMakeVerFAS}). Then, $C$ exhibits {\em \bf Partition of $E_f$ With Both Endpoints on $f$} with respect to $F$ if for every $\{uv_i,uv_{i+1}\}\in E_f$ such that $uv_i,uv_{i+1}\in V(d_f)$, exactly one of the following four conditions holds: \label{infcutedgeCon7}
	\begin{enumerate}
		\item $\{\mathsf{Identify}_{d_f,d_{f_1}}(uv_i),\mathsf{Identify}_{d_f,d_{f_1}}(uv_{i+1})\}\in E_{f_1}$. \label{infcutedgeCon712}
		\item $\{\mathsf{Identify}_{d_f,d_{f_2}}(uv_i),\mathsf{Identify}_{d_f,d_{f_2}}(uv_{i+1})\}\in E_{f_2}$. \label{infcutedgeCon713}
		\item $\{\mathsf{Identify}_{d_f,d_{f_1}^*}(uv_i),\mathsf{Identify}_{d_f,d_{f_1}^*}(uv_{i+1})\}\in E(d_{f_1}^*,\pp(c))$. \label{infcutedgeCon714}
		\item $\{uv_i,uv_{i+1}\}$ partly intersects $c$ (see Definition \ref{def:partIntCBoth}). \label{infcutedgeCon715}
	\end{enumerate}
\end{definition}

We now consider the case of Condition \ref{infcutedgeCon715} of Definition \ref{def:infCTurPoints2a}.
The edge $\{uv_i,uv_{i+1}\}$ has $\ell\in \mathbb{N}$ turning points in $f_1(c)$ other than $(d(uv_i),\{u,v\})$ and $(d(uv_{i+1}),\{u,v\})$ (which may not be turning points in $f_1(c)$).
These $\ell$ turning points, $(p,\{u,v\})$, must be such that $p\in \gi(c)\setminus \gi(f)$, since the edge $\{uv_i,uv_{i+1}\}$ is drawn strictly inside $f$ except for both its endpoints (e.g., see the edge $\{uv_1,uv_2\}$ in Figure~\ref{fig:InfFDr}). Therefore, $\ell$ vertices from $V^*_{\{u,v\}}$ belong to $\gi(c)\setminus \gi(f)$, and they are connected by edges from $E^*_{\{u,v\}}$ in a ``path-like manner''. Further, the numbering of these vertices should match the order of appearance of the points corresponding to them from $uv_i$ to $uv_{i+1}$. Let $r\in [2]$ be such that $uv_i$ is drawn on $f_r(c)$. We expect that the vertex $uv_i$ is mapped to $uv_j$, for some $j\in \mathbb{N}$, in $d_r$, and the remaining $\ell$ vertices from $V^*$ on $c\setminus f$ are numbered $j+1,\ldots,j+\ell$ in $d_r$. Formally, $uv_{j+1},\ldots,uv_{j+\ell}\in V(d_r,\gi(c)\setminus \gi(f))$ (e.g., see the edge $\{uv_1,uv_2\}$ in Figure~\ref{fig:InfFDr} and the addition of its turning points and their labeling in Figure~\ref{fig:infCutter1}). Furthermore, for every $0\leq t<\ell$ we expect that the vertices $uv_{j+t}$ and $uv_{j+t+1}$ are connected with an edge. So, we demand that exactly one of the following condition is satisfied. 

% Here, since $r\in [2]$ is unknown, we use the vertex identification function $\mathsf{Identify}$ as follows. The notation $\mathsf{Identify}_1(\mathsf{Identify}_r(uv_j))$ (or similarly $\mathsf{Identify}_2(\mathsf{Identify}_r(uv_j))$) returns $uv_j$, if $r=1$ (if $r=2$), else, it returns $\mathsf{Identify}(\mathsf{Identify}_r(uv_j))$.

\smallskip\noindent{\bf Condition \ref{infcutedgeCon7151}.} The edge between the vertices $uv_{j+t}$ and $uv_{j+t+1}$ is drawn strictly inside $f_1(c)$ except for both its endpoints, and therefore this edge is in $E_{f_1}$. Now, if $r=2$, the labeling of the endpoints of this edge in $d_{f_1}$ might be other than $uv_{j+t}$ and $uv_{j+t+1}$. In this case, we need to use the vertex identification function of $d_{f_2}$ and $d_{f_1}$ on $uv_{j+t}$ and $uv_{j+t+1}$. Recall that the aforementioned $\ell$ vertices represent turning points of $\{u,v\}$ (in $(f_1(c)$ or $f_2(c)$). Since they are drawn in $\gi(c)\setminus \gi(f)$, it is easy to see that a point in this set is a turning point of $\{u,v\}$ in $f_1(c)$ if and only if it is a turning point of $\{u,v\}$ in $f_2(c)$. So, we get that there exists a vertex from $V^*_{\{u,v\}}$ drawn on a point $p\in \gi(c)\setminus \gi(f)$ of $d_{f_1}$ if and only if there exists a vertex from $V^*_{\{u,v\}}$ drawn on a point $p\in \gi(c)\setminus \gi(f)$ of $d_{f_2}$. Now, since $r\in [2]$ is unknown, we mention the two possible cases.
\begin{itemize}
	\item If $r=1$, then $\{uv_{j+t},uv_{j+t+1}\}\in E_{f_1}$.
	\item if $r=2$, then $\{\mathsf{Identify}_{d_{f_2},d_{f_1}}(uv_{j+t}),\mathsf{Identify}_{d_{f_2},d_{f_1}}(uv_{j+t+1})\}\in E_{f_1}$ (e.g., see the edge $\{uv_1,uv_2\}$ in $f_1(c)$ in Figure~\ref{fig:infCutter1}).
\end{itemize} 
%This issue is also relevant for the rest of the conditions of this property. 

\smallskip\noindent{\bf Condition \ref{infcutedgeCon7152}.} This condition is symmetric to the first condition. Here, the edge connecting the vertices $uv_{j+t}$ and $uv_{j+t+1}$ is drawn strictly inside $f_2(c)$ except for both its endpoints, and therefore this edge is in $E_{f_2}$. There are two cases:
\begin{itemize}
	\item If $r=2$, then $\{uv_{j+t},uv_{j+t+1}\}\in E_{f_2}$.
	\item if $r=1$, then $\{\mathsf{Identify}_{d_{f_1},d_{f_2}}(uv_{j+t}),\mathsf{Identify}_{d_{f_1},d_{f_2}}(uv_{j+t+1})\}\in E_{f_2}$ (e.g., see the edge $\{uv_3,uv_4\}$ in $f_2(c)$ in Figure~\ref{fig:infCutter1}).
\end{itemize}

\smallskip\noindent{\bf Condition \ref{infcutedgeCon7153}.} The edge between the vertices $uv_{j+t}$ and $uv_{j+t+1}$ is drawn on $c$, that is, $\{uv_{j+t},uv_{j+t+1}\}\in E(d_r,\pp(c))$ (e.g., see the edge $\{st_1,st_2\}$ in $f_2(c)$ in Figure~\ref{fig:infCutter1}).

\smallskip\noindent{\bf Conditions \ref{infcutedgeCon7155}, \ref{infcutedgeCon7156} and \ref{infcutedgeCon7157}.} Similarly, we have an edge between the vertices $uv_{j+\ell}$ and $uv_{i+1}$. We state the main differences between the cases corresponding to this edge. Observe that $uv_{i+1}$ might not be drawn on $f_r(c)$, so we cannot assume that $uv_{i+1}$ is labeled $uv_{j+\ell+1}$, as expected (since there is an edge between $uv_{j+\ell}$ and $uv_{i+1}$). So, assume without loss of generality that $r=1$ and $uv_{i+1}$ is not drawn on $f_1(c)$. Then, if the edge between $uv_{j+\ell}$ and $uv_{i+1}$ is drawn strictly inside $f_2(c)$, except at its endpoints, then the edge will appear in $E_{f_2}$. Now, we wish to get the labeling of the endpoints of this edge. For the vertex $uv_{j+\ell}$, which is drawn on $c$, we can get its labeling from the drawing $d_{f_1}$, that is, $\mathsf{Identify}_{d_{f_1},d_{f_2}}(uv_{j+\ell})$. As for the vertex $uv_{i+1}$, we can get its labeling from the drawing $f$, that is, $\mathsf{Identify}_{d_f,d_{f_2}}(uv_{i+1})$. Now, assume that $r=1$ and  the edge between $uv_{j+\ell}$ and $uv_{i+1}$ is drawn strictly inside $f_1(c)$. Then, we expect that $uv_{i+1}$ is labeled $uv_{j+\ell+1}$, as explained, and so $\{uv_{j+\ell},uv_{j+\ell+1}\}\in E_{f_1}$. 

%We state these conditions in Conditions \ref{infcutedgeCon7155}, \ref{infcutedgeCon7156} and \ref{infcutedgeCon7157}, respectively. 

% in  and \ref{infcutedgeCon71513} andConditions \ref{infcutedgeCon71511}, \ref{infcutedgeCon71512}

\begin{definition}[{\bf Partial Intersection of $\{uv_i,uv_{i+1}\}\in E_f$ and $c$}] \label{def:partIntCBoth}
	Let $F=(f,d_f,E_f,$ $U_f,\mathsf{V^*Dir}_f)$ be an info-frame. Let $C=(c,F_1=(f_1(c),d_{f_1},E_{f_1},U_{f_1},\mathsf{V^*Dir}_{f_1}),F_2=(f_2(c),$ $d_{f_2},E_{f_2},U_{f_2},\mathsf{V^*Dir}_{f_2}))$ exhibits {\bf Equality of Common Parts} with respect to $F$. Let $d_f^*,d_{f_1}^*,d_{f_2}^*$ be the result of the application of $\mathsf{MakeAllVerSet}(d_f,d_{f_1},d_{f_2},\mathsf{Common}(d_f,d_{f_1}),$ $\mathsf{Common}(d_f,d_{f_2}),\mathsf{Common}($ $d_{f_1},d_{f_2}))$ (see Definition \ref{def:OpMakeVerFAS}). Let $\{uv_i,uv_{i+1}\}\in E_f$ such that $uv_i,uv_{i+1}\in V(d_f)$. Then, $\{uv_i,uv_{i+1}\}$ {\em partly intersects} $c$ if the following condition holds.
	There exist $\ell \in \mathbb{N}$ and $r\in[2]$ such that
	 $\mathsf{Identify}_{d_f,d^*_{f_r}}(uv_i)=uv_j$, $uv_{j+1},\ldots,uv_{j+\ell}\in V(d_{f_r},\gi(c)\setminus\gi(f))$ and for every $0\leq t<\ell$, exactly one of Conditions \ref{infcutedgeCon7151}-\ref{infcutedgeCon7153} and exactly one of Conditions \ref{infcutedgeCon7155}-\ref{infcutedgeCon7157} are satisfied:
	\begin{enumerate}
		\item
		If $r=1$, then $\{uv_{j+t},uv_{j+t+1}\}\in E_{f_1}$; else, $r=2$, and then $\{\mathsf{Identify}_{d_{f_2},d_{f_1}}(uv_{j+t}),$ $\mathsf{Identify}_{d_{f_2},d_{f_1}}$ $(uv_{j+t+1})\}\in E_{f_1}$.\label{infcutedgeCon7151}
		
		\item If $r=2$, then $\{uv_{j+t},uv_{j+t+1}\}\in E_{f_2}$; else, $r=1$, and then $\{\mathsf{Identify}_{d_{f_1},d_{f_2}}(uv_{j+t}),$ $\mathsf{Identify}_{d_{f_1},d_{f_2}}$ $(uv_{j+t+1})\}\in E_{f_2}$. \label{infcutedgeCon7152} 
		
		\item $\{uv_{j+t},uv_{j+t+1}\}\in E(d_{f_r}^*,\pp(c))$. \label{infcutedgeCon7153} 
		%		\end{enumerate} 
	
	%In addition,  exactly one of the following condition holds: 
	%	\begin{enumerate}
		%	\item If $r=1$ then $\{uv_j,uv_{j+1})\}\in E_{f_1}$. Else, $r=2$, and then $\{\mathsf{Identify}_{d_{f_2},d_{f_1}}(uv_{j+t}),\mathsf{Identify}_{d_{f_2},d_{f_1}}$ $(uv_{j+t+1})\}\in E_{f_1}$.\label{infcutedgeCon71511}
		%	\item If $r=2$ then $\{\mathsf{Identify}_2(uv_i)),\mathsf{Identify}_2(uv_i)+1)\}\in E_{f_2}$. Else, $r=1$, and then $\{\mathsf{Identify}(\mathsf{Identify}_1(uv_i)),\mathsf{Identify}(\mathsf{Identify}_1(uv_i)+1)\}\in E_{f_2}$.\label{infcutedgeCon71512}
		%	\item $\{\mathsf{Identify}_r(uv_i),\mathsf{Identify}_r(uv_i)+1\}\in E(d_r,\pp(c))$. \label{infcutedgeCon71513}
		%	\end{enumerate}
	%	\begin{enumerate}
		\item If $r=1$, then $\mathsf{Identify}_{d_f,d_{f_1}}(uv_{i+1})=uv_{j+\ell+1}$ and $\{uv_{j+\ell},uv_{j+\ell+1}\}\in E_{f_1}$; else, $r=2$, and then $\{\mathsf{Identify}_{d_{f_2},d_{f_1}}(uv_{j+\ell}),\mathsf{Identify}_{d_f,d_{f_1}}$ $(uv_{i+1})\}\in E_{f_1}$. \label{infcutedgeCon7155}
		
		\item If $r=2$, then $\mathsf{Identify}_{d_f,d_{f_2}}(uv_{i+1})=uv_{j+\ell+1}$ and $\{uv_{j+\ell},uv_{j+\ell+1}\}\in E_{f_2}$; else, $r=1$, and then $\{\mathsf{Identify}_{d_{f_1},d_{f_2}}(uv_{j+\ell}),\mathsf{Identify}_{d_f,d_{f_2}}$ $(uv_{i+1})\}\in E_{f_2}$.\label{infcutedgeCon7156}
		
		\item $\mathsf{Identify}_{d_f,d_{f_r}}(uv_{i+1})=uv_{j+\ell+1}$ and $\{uv_{j+\ell},uv_{j+\ell+1}\}\in E(d_{f_r}^*,\pp(c))$. \label{infcutedgeCon7157}
	\end{enumerate}
	
\end{definition}

We proceed with other properties related to {\bf partition of $E_f$}. Specifically, the next two properties are similar to {\bf Partition of $E_f$ With Both Endpoints on $f$}, but with a few differences. Let $\{uv_i,uv_{i+1}\}\in E_f$ such that  $\{uv_i,uv_{i+1}\}$ is an edge that is drawn strictly inside $f$, except at exactly one endpoint. Recall that since $F$ is an info-frame, every vertex that is drawn strictly inside $f$ is a vertex from the set $V$. So, in this case, any vertex that is not drawn on $f$, but drawn strictly inside $f$, must be a vertex from $V$ and not from $V^*$. Therefore, we have two cases: either $\{uv_i,uv_{i+1}\}=\{u,uv_1\}$ (e.g., see the edge $\{u,uv_1\}$ in Figure~\ref{fig:InfFDr}), or $\{uv_i,uv_{i+1}\}=\{uv_{\mathsf{index}(u,v)},v\}$ (e.g., see the edge $\{st_3,t\}$ in Figure~\ref{fig:InfFDr}). We handle the first case in the {\bf Partition of $\{u,uv_1\}\in E_f$ With One Endpoint on $f$} property (see Definition \ref{def:infCTurPoints2b}), and the second case in the {\bf Partition of $\{uv_{\mathsf{index}(u,v)},v\}\in E_f$ With One Endpoint on $f$} property (see Definition \ref{def:infCTurPoints2c}). In what follows, we explain the differences between {\bf Partition of $E_f$ With Both Endpoints on $f$} and {\bf Partition of $\{u,uv_1\}\in E_f$ With One Endpoint on $f$}. (The differences between {\bf Partition of $E_f$ With Both Endpoints on $f$} and {\bf Partition of $\{uv_{\mathsf{index}(u,v)},v\}\in E_f$ With One Endpoint on $f$} are similar). We first compare Definition \ref{def:infCTurPoints2a} to Definition \ref{def:infCTurPoints2b}.

\smallskip\noindent{\bf Conditions \ref{infcutedgeCon2b1}, \ref{infcutedgeCon2b2} and \ref{infcutedgeCon2b3}.} These conditions are similar to Conditions \ref{infcutedgeCon712}, \ref{infcutedgeCon713} and \ref{infcutedgeCon714} of the definition of {\bf Partition of $E_f$ With Both Endpoints on $f$} (Definition \ref{def:infCTurPoints2a}). 

\smallskip\noindent{\bf Condition \ref{infcutedgeCon2b4}.} 
This condition is similar to Condition \ref{infcutedgeCon715} of Definition \ref{def:infCTurPoints2a}. Here, we consider the case where there is a turning point $(p,\{u,v\})$ for $p\in \gi(c)\setminus \gi(f)$, other than the endpoints of $\{u,uv_1\}$, in $f_1(c)$ (and $f_2(c)$). We further discuss additional constraints satisfied in this case later.

\begin{definition}[{\bf Partition of $\{u,uv_1\}\in E_f$ With One Endpoint on $f$}] \label{def:infCTurPoints2b}
	Let $F=(f,d_f,E_f,U_f,\mathsf{V^*Dir}_f)$ be an info-frame. Let $C=(c,F_1=(f_1(c),d_{f_1},E_{f_1},U_{f_1},\mathsf{V^*Dir}_{f_1}),F_2=(f_2(c),d_{f_2},E_{f_2},U_{f_2},\mathsf{V^*Dir}_{f_2}))$ exhibits {\bf Equality of Common Parts} with respect to $F$. Let $d_f^*,d_{f_1}^*,d_{f_2}^*$ be the result of the application of $\mathsf{MakeAllVerSet}(d_f,d_{f_1},d_{f_2},\mathsf{Common}$ $(d_f,d_{f_1}),$\\$\mathsf{Common}$ $(d_f,d_{f_2}),\mathsf{Common}(d_{f_1},d_{f_2}))$ (see Definition \ref{def:OpMakeVerFAS}). Then, $C$ exhibits {\em \bf Partition of $\{u,uv_1\}\in E_f$ With One Endpoint on $f$} with respect to $F$ if for every $\{u,uv_{1}\}\in E_f$ such that $u\in U_f$, exactly one of the following four conditions holds: 
	\begin{enumerate}
		\item $\{u,uv_1\}\in E_{f_1}$. \label{infcutedgeCon2b1}
		\item $\{u,uv_1\}\in E_{f_2}$. \label{infcutedgeCon2b2}
		\item $\{u,uv_1\}\in E(d_{f_1}^*,\pp(c))$. \label{infcutedgeCon2b3}
		\item $\{u,uv_{1}\}$ partly intersects $c$ (see Definition \ref{def:partIntCU}). \label{infcutedgeCon2b4}
	\end{enumerate}
\end{definition}

We now consider the case of Condition \ref{infcutedgeCon2b4} of Definition \ref{def:infCTurPoints2b}. Let $\ell\in \mathbb{N}$ such that there are $\ell$ turning points $(p,\{u,v\})$ for $p\in \gi(c)\setminus \gi(f)$, other than the endpoints of $\{u,uv_1\}$, in each of $f_1(c)$ and $f_2(c)$. Recall that the purpose of the numbering of vertices from the set $V^*_{\{u,v\}}$ is to state the order of the appearance of the vertices, from $u$ to $v$. Now, these vertices (that are from the set $V^*_{\{u,v\}}$) that are on $c$, appear before the intersection of the edge $\{u,v\}$ and $f$ (represented by the vertex $uv_1$ in $d_f$). Therefore, we expect that the numbering of these vertices will be $1$ to $\ell$, where $\ell$ is the number of vertices from $V^*_{\{u,v\}}$ in $f_1(c)$ (observe that $\ell$ is also the number of vertices from $V^*_{\{u,v\}}$ in $f_2(c)$). So, as opposed to Condition \ref{infcutedgeCon715} of the definition of {\bf Partition of $E_f$ With Both Endpoints on $f$} (Definition \ref{def:infCTurPoints2a}), the numbering of vertices from $V^*_{\{u,v\}}$ is known in advance.

\smallskip\noindent{\bf Conditions \ref{infcutedgeCon4ba1}, \ref{infcutedgeCon4ba2} and \ref{infcutedgeCon4ba3}.} First, we consider the edges $\{uv_t,uv_{t+1}\}$ for every $1\leq t\leq \ell$. Here, the cases are similar to Conditions \ref{infcutedgeCon712}, \ref{infcutedgeCon713} and \ref{infcutedgeCon714} of the definition of {\bf Partition of $E_f$ With Both Endpoints on $f$} (Definition \ref{def:infCTurPoints2a}). 

Next, we consider the edge between $uv_1$ and $u$. We have that exactly one of Conditions \ref{infcutedgeCon4ba}, \ref{infcutedgeCon4bb}, \ref{infcutedgeCon4bc}, \ref{infcutedgeCon4bd} and \ref{infcutedgeCon4be} is satisfied.

%\smallskip\noindent{\bf Conditions \ref{infcutedgeCon4ba1}, \ref{infcutedgeCon4ba2} and \ref{infcutedgeCon4ba3}.} Now, Conditions \ref{infcutedgeCon4ba}, \ref{infcutedgeCon4bb}, \ref{infcutedgeCon4bc}, \ref{infcutedgeCon4bd} and \ref{infcutedgeCon4be} deal with the edge between $u$ and $uv_1$ (observe that $uv_1$ here is a turning point of $\{u,v\}$ on $\gi(c)\setminus \gi(f)$). Here, we have that exactly one of the following conditions is satisfied. 

\smallskip\noindent{\bf Condition \ref{infcutedgeCon4ba}.} The vertex $u$ is drawn strictly inside $f_1(c)$, and therefore we have that $u\in U_{f_1}$. In addition, since $uv_1$ is the vertex from the set $V^*_{\{u,v\}}$ that is the closest to $u$, it follows that the edge $\{u,uv_1\}$ is drawn strictly inside $f_1(c)$ except at the endpoint $d(uv_1)$. So, $\{u,uv_1\}\in E_{f_1}$.

\smallskip\noindent{\bf Condition \ref{infcutedgeCon4bb}.} Similarly, the vertex $u$ is drawn strictly inside $f_2(c)$, and therefore we have that $u\in U_{f_2}$. In addition, since $uv_1$ is the vertex from the set $V^*_{\{u,v\}}$ that is the closest to $u$, then, it follows that the edge $\{u,uv_1\}$ is drawn strictly inside $f_2(c)$ except at the endpoint $d(uv_1)$. So, $\{u,uv_1\}\in E_{f_2}$ (e.g., see the vertex $t$ and the edge $\{st_3,t\}$ in Figure~\ref{fig:infCutter1}). 

\smallskip\noindent{\bf Condition \ref{infcutedgeCon4bc}.} The vertex $u$ is drawn on $c$, so $u\in V(d_{f_1}, \gi(c))$, and the edge $\{u,uv_1\}$ is drawn strictly inside $f_1(c)$ except at its endpoints. Then, $\{u,uv_1\}\in E_{f_1}$.

\smallskip\noindent{\bf Condition \ref{infcutedgeCon4bd}.} The vertex $u$ is drawn on $c$, so $u\in V(d_{f_1}, \gi(c))$, and the edge $\{u,uv_1\}$ is drawn strictly inside $f_2(c)$ except at its endpoints. Then, $\{u,uv_1\}\in E_{f_2}$.

\smallskip\noindent{\bf Condition \ref{infcutedgeCon4be}.} The vertex $u$ and the edge $\{u,uv_1\}$ are drawn on $c$. So, $\{u,uv_{1}\}\in E(d_{f_1}^*,$ $\pp(c))$. 

%	In addition, we have that exactly one of the following conditions is satisfied:
%	\begin{itemize}
	%	\item The edge $\{uv_\ell,\mathsf{Identify}_1(uv_1)\}$ is drawn strictly inside $f_1(c)$ except at the endpoints, so $\{uv_\ell,\mathsf{Identify}_1(uv_1)\}\in E_{f_1}$, as stated in Condition \ref{infcutedgeCon4bc}. 
	%\item The edge $\{uv_\ell,\mathsf{Identify}_2(uv_1)\}$ is drawn strictly inside $f_2(c)$ except at the endpoints, so $\{uv_\ell,\mathsf{Identify}_2(uv_1)\}\in E_{f_2}$, as stated in Condition \ref{infcutedgeCon4bd}. 
	%	\item The edge $\{uv_\ell,\mathsf{Identify}_1(uv_1)\}$ is drawn on $c$, so $\{uv_\ell,\mathsf{Identify}_1(uv_1)\}\in E(d_{f_1},\pp(c))$, as we state in Condition \ref{infcutedgeCon4be}. 
	%\end{itemize}
	%Conditions \ref{infcutedgeCon4ca}, \ref{infcutedgeCon4cb} and \ref{infcutedgeCon4cc} are similar to Conditions \ref{infcutedgeCon7155}, \ref{infcutedgeCon7156} and \ref{infcutedgeCon7157} of Definition \ref{def:infCTurPoints2a}. 

	\begin{definition}[{\bf Partial Intersection of $\{u,uv_{1}\}\in E_f$ and $c$}] \label{def:partIntCU}
		Let $F=(f,d_f,E_f,U_f,$ $\mathsf{V^*Dir}_f)$ be an info-frame. Let $C=(c,F_1=(f_1(c),d_{f_1},E_{f_1},U_{f_1},\mathsf{V^*Dir}_{f_1}),F_2=(f_2(c),d_{f_2},$ $E_{f_2},U_{f_2},\mathsf{V^*Dir}_{f_2}))$ exhibits {\bf Equality of Common Parts} with respect to $F$. Let $d_f^*,d_{f_1}^*,d_{f_2}^*$ be the result of the application of $\mathsf{MakeAllVerSet}(d_f,d_{f_1},d_{f_2},\mathsf{Common}(d_f,d_{f_1}),\mathsf{Common}(d_f,$ $d_{f_2}),\mathsf{Common}(d_{f_1},d_{f_2}))$ (see Definition \ref{def:OpMakeVerFAS}). Then, $\{u,uv_{1}\}\in E_f$ such that $u\in U_f$, {\em partly intersects} $c$ if the following condition holds.
		There exist $\ell \in \mathbb{N}$ such that $uv_{1},\ldots,uv_{\ell}\in V(d_{f_1},\gi(c)\setminus\gi(f))$, and for every $1\leq t\leq \ell$, exactly one of Conditions \ref{infcutedgeCon7151}-\ref{infcutedgeCon7153} and exactly one of Conditions \ref{infcutedgeCon7155}-\ref{infcutedgeCon7157} are satisfied:
		\begin{enumerate}
			\item $\{uv_t,uv_{t+1}\}\in E_{f_1}$. \label{infcutedgeCon4ba1} 
			
			\item $\{uv_t,uv_{t+1}\}\in E_{f_1}$. \label{infcutedgeCon4ba2} 
			
			\item $\{uv_{t},uv_{t+1}\}\in E(d_{f_1}^*,\pp(c))$. \label{infcutedgeCon4ba3} 
			
			\item $u\in U_{f_1}$ and $\{u,uv_1\}\in E_{f_1}$. \label{infcutedgeCon4ba} 
			
			\item $u\in U_{f_2}$ and $\{u,uv_1\}\in E_{f_2}$. \label{infcutedgeCon4bb} 
			
			\item $u\in V(d_{f_1},\gi(c))$ and $\{u,uv_1\}\in E_{f_1}$. \label{infcutedgeCon4bc} 
			
			\item $u\in V(d_{f_1},\gi(c))$ and$\{u,uv_1\}\in E_{f_2}$. \label{infcutedgeCon4bd} 
			
			\item $\{u,uv_{1}\}\in E(d_{f_1}^*,\pp(c))$. \label{infcutedgeCon4be} 
		\end{enumerate}
	\end{definition}

			We now present a similar case to Definitions \ref{def:infCTurPoints2b} and \ref{def:partIntCU}, where the edge we examine is $\{uv_{\mathsf{index}(u,v)},v\}$.

			\begin{definition}[{\bf Partition of $\{uv_{\mathsf{index}(u,v)},v\}\in E_f$ With One Endpoint on $f$}] \label{def:infCTurPoints2c}
				Let $F=(f,d_f,E_f,U_f,\mathsf{V^*Dir}_f)$ be an info-frame. Let $C=(c,F_1=(f_1(c),d_{f_1},E_{f_1},U_{f_1},\mathsf{V^*Dir}_{f_1}),F_2=(f_2(c),d_{f_2},E_{f_2},U_{f_2},\mathsf{V^*Dir}_{f_2}))$ exhibits {\bf Equality of Common Parts} with respect to $F$. Let $d_f^*,d_{f_1}^*,d_{f_2}^*$ be the result of the application of $\mathsf{MakeAllVerSet}(d_f,d_{f_1},d_{f_2},\mathsf{Common} (d_f,d_{f_1}),$\\ $\mathsf{Common}(d_f,d_{f_2}),\mathsf{Common}(d_{f_1},d_{f_2}))$ (see Definition \ref{def:OpMakeVerFAS}). Then, $C$ exhibits {\em \bf Partition of}\\ $\{uv_{\mathsf{index}(u,v)},$ $v\}\in E_f$ {\em \bf With One Endpoint on $f$} with respect to $F$ if for every $\{uv_{\mathsf{index}(u,v)},$ $v\}\in E_f$ such that $v\in U_f$, exactly one of the following four conditions holds: 
				\begin{enumerate}
					\item $\{\mathsf{Identify}_{d_f,d_{f_1}}(uv_{\mathsf{index}(u,v)}),v\}\in E_{f_1}$. \label{infcutedgeCon8712}
					\item $\{\mathsf{Identify}_{d_f,d_{f_2}}(uv_{\mathsf{index}(u,v)}),v\}\in E_{f_2}$. \label{infcutedgeCon8713}
					\item $\{\mathsf{Identify}_{d_f,d_{f_1}^*}(uv_{\mathsf{index}(u,v)}),v)\}\in E(d_{f_1}^*,\pp(c))$. \label{infcutedgeCon8714}
					\item $\{uv_{\mathsf{index}(u,v)},v\}$ partly intersects $c$ (see Definition \ref{def:partIntV}).
				\end{enumerate}
			\end{definition}
			
			\begin{definition}[{\bf Partial Intersection of $\{uv_{\mathsf{index}(u,v)},v\}\in E_f$ and $c$}] \label{def:partIntV}
				Let $F=(f,d_f,E_f,$ $U_f,\mathsf{V^*Dir}_f)$ be an info-frame. Let $C=(c,F_1=(f_1(c),d_{f_1},E_{f_1},U_{f_1},\mathsf{V^*Dir}_{f_1}),F_2=(f_2(c),d_{f_2},$ $E_{f_2},U_{f_2},\mathsf{V^*Dir}_{f_2}))$ exhibits {\bf Equality of Common Parts} with respect to $F$. Let $d_f^*,d_{f_1}^*,d_{f_2}^*$ be the result of the application of $\mathsf{MakeAllVerSet}(d_f,d_{f_1},d_{f_2},$ $\mathsf{Common}(d_f,d_{f_1}),$ $mathsf{Common}(d_f,d_{f_2}),\mathsf{Common}($ $d_{f_1},d_{f_2}))$ (see Definition \ref{def:OpMakeVerFAS}). Let $\{uv_{\mathsf{index}(u,v)},v\}\in E_f$ such that $v\in U_f$. Then, $\{uv_{\mathsf{index}(u,v)},v\}$ {\em partly intersects} $c$ if the following condition holds.
				There exist $\ell \in \mathbb{N}$ and $r\in[2]$ such that $\mathsf{Identify}_{d_f,d^*_{f_r}}(uv_{\mathsf{index}(u,v)})=uv_j, uv_{j+1},\ldots,uv_{j+\ell}\in V(d_{f_r},\gi(c)\setminus\gi(f))$, and for every $0\leq t<\ell$, exactly one of Conditions \ref{infcutedgeCon1}-\ref{infcutedgeCon3} and exactly one of Conditions \ref{infcutedgeCon4}-\ref{infcutedgeCon8} are satisfied:
				\begin{enumerate}
					\item
					If $r=1$, then $\{uv_{j+t},uv_{j+t+1}\}\in E_{f_1}$; else, $r=2$, and then $\{\mathsf{Identify}_{d_{f_2},d_{f_1}}(uv_{j+t}),$ $\mathsf{Identify}_{d_{f_2},d_{f_1}}$ $(uv_{j+t+1})\}\in E_{f_1}$.\label{infcutedgeCon1}
					
					\item If $r=2$, then $\{uv_{j+t},uv_{j+t+1}\}\in E_{f_2}$; else, $r=1$, and then $\{\mathsf{Identify}_{d_{f_1},d_{f_2}}(uv_{j+t}),$ $\mathsf{Identify}_{d_{f_1},d_{f_2}}$ $(uv_{j+t+1})\}\in E_{f_2}$. \label{infcutedgeCon2} 
					
					\item $\{uv_{j+t},uv_{j+t+1}\}\in E(d_{f_r}^*,\pp(c))$. \label{infcutedgeCon3} 
					%		\end{enumerate} 
				
				\item $v\in U_{f_1}$. If $r=1$, then $\{uv_{j+\ell},v\}\in E_{f_1}$; else, $r=2$, and then  $\{\mathsf{Identify}_{d_{f_2},d_{f_1}}(uv_{j+\ell}),v\}$ $\in E_{f_1}$. \label{infcutedgeCon4} 
				
				\item $v\in U_{f_2}$.If $r=2$, then $\{uv_{j+\ell},v\}\in E_{f_2}$; else, $r=1$, and then $\{\mathsf{Identify}_{d_{f_1},d_{f_2}}(uv_{j+\ell}),v\}\in E_{f_2}$. \label{infcutedgeCon5} 
				
				\item $v\in V(d_{f_1},\gi(c))$. If $r=1$, then $\{uv_{j+\ell},v\}\in E_{f_1}$; else, $r=2$, and then $\{\mathsf{Identify}_{d_{f_2},d_{f_1}}(uv_{j+\ell}),v\}\in E_{f_1}$. \label{infcutedgeCon6} 
				
				\item $v\in V(d_{f_1},\gi(c))$. If $r=2$, then $\{uv_{j+\ell},v\}\in E_{f_2}$; else, $r=1$, and then $\{\mathsf{Identify}_{d_{f_1},d_{f_2}}(uv_{j+\ell}),v\}\in E_{f_2}$. \label{infcutedgeCon77} 
				
				\item $\{uv_{j+\ell},v\}\in E(d_{f_r}^*,\pp(c))$. \label{infcutedgeCon8} 
			\end{enumerate}
			
		\end{definition}

				We continue with the last property. Here, we consider the edges drawn strictly inside $f$. That is, edges of the form $\{u,v\}\in E$ where $u,v\in U_f$ and $V(d_f)\cap V^*_{\{u,v\}}=\emptyset$ (e.g., see the edge $\{x,w\}$ in Figure~\ref{fig:InfFDr}). Recall, that $u,v\in U_f$ means that $u$ and $v$ are drawn strictly inside $f$, and considering that also $V(d_f)\cap V^*_{\{u,v\}}=\emptyset$, we have that the edge $\{u,v\}\in E$ does not intersect $f$; therefore, the edge $\{u,v\}$ is drawn strictly inside $f$.
				So, let $\{u,v\}\in E$ be such an edge. We expect that exactly one of the following conditions is satisfied.
				
				\smallskip\noindent{\bf Condition \ref{infCutcon101}.} The edge $\{u,v\}$ is drawn strictly inside $f_1(c)$. Therefore, the vertices $u$ and $v$ are drawn strictly inside $f_1(c)$, so $u,v\in U_{f_1}$. In addition, in this case, since $\{u,v\}$ does not intersect $f_1(c)$, and, in particular, it does not intersect $c$, it follows that $V(d_{f_1},\gi(c))\cap V^*_{\{u,v\}}=\emptyset$. 
				
				\smallskip\noindent{\bf Condition \ref{infCutcon102}.} Similarly to the previous condition, the edge $\{u,v\}$ is drawn strictly inside $f_2(c)$. Therefore, the vertices $u$ and $v$ are drawn strictly inside $f_2(c)$, so $u,v\in U_{f_2}$. In addition, in this case, since $\{u,v\}$ does not intersect $f_2(c)$, and, in particular, it does not intersect $c$, it follows that $V(d_{f_2},\gi(c))\cap V^*_{\{u,v\}}=\emptyset$. 
				
				\smallskip\noindent{\bf Condition \ref{infCutcon103}.} The edge $\{u,v\}$ is drawn on $c$. Observe that in this case, there are no turning points of $\{u,v\}$ and $f_1(c)$ except at the endpoints, so $V(d_{f_1},\gi(c))\cap V^*_{\{u,v\}}=\emptyset$. 
				
				\smallskip\noindent{\bf Condition \ref{infCutcon104}.} 
				This condition is similar to Condition \ref{infcutedgeCon715} of Definition \ref{def:infCTurPoints2a}. Here, we consider the case where there is a turning point $(p,\{u,v\})$ for $p\in \gi(c)\setminus \gi(f)$, other than the endpoints of $\{u,v\}$, in both $f_1(c)$ and $f_2(c)$. We further discuss additional constraints satisfied in this case later.

				\begin{definition}[{\bf Partition of $\{u,v\}\in E$ Drawn Strictly Inside $f$}] \label{def:infCTurPoints3}
					Let $F=(f,d_f,E_f,U_f,$ $\mathsf{V^*Dir}_f)$ be an info-frame. Then, $C=(c,F_1=(f_1(c),d_{f_1},E_{f_1},U_{f_1},\mathsf{V^*Dir}_{f_1}),$ $F_2=(f_2(c),d_{f_2},$ $E_{f_2},$ $U_{f_2},\mathsf{V^*Dir}_{f_2}))$ exhibits {\em \bf Partition of $\{u,v\}\in E$ Drawn Strictly Inside $f$} with respect to $F$ if for every $u,v\in U_f$ such that $\{u,v\}\in E$ and $V(d_f)\cap V^*_{\{u,v\}}=\emptyset$, exactly one of the following four conditions holds: 
					\begin{enumerate}
						\item $V(d_{f_1},\gi(c))\cap V^*_{\{u,v\}}=\emptyset$ and $u,v\in U_{f_1}$. \label{infCutcon101}
						\item $V(d_{f_1},\gi(c))\cap V^*_{\{u,v\}}=\emptyset$ and $u,v\in U_{f_2}$. \label{infCutcon102}
						\item  $V(d_{f_1},\gi(c))\cap V^*_{\{u,v\}}=\emptyset$ and $\{u,v\}\in E(d_{f_1},\pp(c))$. \label{infCutcon103}
						\item $\{u,v\}$ partly intersects $c$ (see Definition \ref{def:partIntUV}). \label{infCutcon104}
					\end{enumerate}
				\end{definition}

				We now consider the case of Condition \ref{infCutcon104} of Definition \ref{def:infCTurPoints3}.  Observe that the previous three conditions of Definition \ref{def:infCTurPoints3} deal with the case where no vertices from the set $V^*_{\{u,v\}}$ are created. That is, the edge $\{u,v\}$ has no turning points in $f_1(c)$ or $f_2(c)$, except maybe at the endpoints. Condition \ref{infCutcon104} deals with the other case, where we have at least one such a turning point. Observe that since $V(d_f)\cap V^*_{\{u,v\}}=\emptyset$, the vertices from $V^*_{\{u,v\}}$ (that represent turning points of $\{u,v\}$ in $f_1(c)$ or $f_2(c)$) belong to $V(d_{f_1},\gi(c)\setminus \gi(f))$ (e.g., see the edge $\{x,w\}$ in Figure~\ref{fig:infCutter1}). Therefore, there exist $uv_1,\ldots uv_\ell \in V(d_{f_1},\gi(c)\setminus \gi(f))$ such that, for every $1\leq t<\ell$, exactly one of the following conditions is satisfied:

				\smallskip\noindent{\bf Condition \ref{infcutedgeCon51}.} The edge between the vertices $uv_j$ and $uv_{j+1}$ is drawn strictly inside $f_1(c)$, except at the endpoints, and therefore, $\{uv_j,uv_{j+1}\}\in E_{f_1}$.
				
				\smallskip\noindent{\bf Condition \ref{infcutedgeCon52}.} The edge between the vertices $uv_j$ and $uv_{j+1}$ is drawn strictly inside $f_2(c)$, except at the endpoints, and therefore, $\{uv_j,uv_{j+1}\}\in E_{f_2}$.
				
				\smallskip\noindent{\bf Condition \ref{infcutedgeCon53}.} The edge between the vertices $uv_j$ and $uv_{j+1}$ is drawn on $c$, so $\{uv_j,uv_{j+1}\}\in E(d_{f_1},\pp(c))$.

				\smallskip\noindent{\bf Conditions \ref{infcutedgeCon1041}-\ref{infcutedgeCon1045}.}
				These conditions deal with the part of the edge that appears first, that is, $\{u,uv_1\}$, and are similar to Conditions \ref{infcutedgeCon4ba}-\ref{infcutedgeCon4be} of Definition \ref{def:infCTurPoints2b}.
				
				\smallskip\noindent{\bf Conditions \ref{infcutedgeCon1051}-\ref{infcutedgeCon1055}.}
				These conditions deal with the part of the edge that appears last, that is, $\{uv_\ell,v\}$, and are similar to Conditions \ref{infcutedgeCon1041}-\ref{infcutedgeCon1045} of this definition.

				\begin{definition}[{\bf Partial Intersection of $\{u,v\}\in E$ and $c$}] \label{def:partIntUV}
					Let $F=(f,d_f,E_f,U_f,$ $\mathsf{V^*Dir}_f)$ be an info-frame. Let $C=(c,F_1=(f_1(c),d_{f_1},E_{f_1},U_{f_1},\mathsf{V^*Dir}_{f_1}),F_2=(f_2(c),d_{f_2},$ $E_{f_2},U_{f_2},\mathsf{V^*Dir}_{f_2}))$. $u,v\in U_f$ such that $\{u,v\}\in E$ and $V(d_f)\cap V^*_{\{u,v\}}=\emptyset$. Then, $\{u,v\}$ {\em partly intersects} $c$ if the following condition holds.
					There exist $\ell \in \mathbb{N}$ such that $uv_{1},\ldots,uv_{\ell}\in V(d_{f_1},\gi(c)\setminus\gi(f))$, and for every $1\leq t<\ell$, exactly one of Conditions \ref{infcutedgeCon51}-\ref{infcutedgeCon53}, exactly one of conditions \ref{infcutedgeCon1041}-\ref{infcutedgeCon1045} and exactly one of Conditions \ref{infcutedgeCon1051}-\ref{infcutedgeCon1055} are satisfied:
					\begin{enumerate}
						\item $\{uv_t,uv_{t+1}\}\in E_{f_1}$. \label{infcutedgeCon51}
						\item $\{uv_t,uv_{t+1}\}\in E_{f_2}$. \label{infcutedgeCon52} 
						\item $\{uv_t,uv_{t+1}\}\in E(d_{f_1},\pp(c))$. \label{infcutedgeCon53} 
						
						\item $u\in U_{f_1}$ and $\{u,uv_{1}\}\in E_{f_1}$.  \label{infcutedgeCon1041}
						\item $u\in U_{f_2}$ and $\{u,uv_{1}\}\in E_{f_2}$. \label{infcutedgeCon1042}
						\item $u\in V(d_{f_1},\gi(c))$ and $\{u,uv_{1}\}\in E(d_{f_1},\pp(c))$. \label{infcutedgeCon1043}
						\item $u\in V(d_{f_1},\gi(c))$ and $\{u,uv_{1}\}\in E_{f_1}$. \label{infcutedgeCon1044}
						\item $u\in V(d_{f_1},\gi(c))$ and $\{u,uv_{1}\}\in E_{f_2}$. \label{infcutedgeCon1045}
						
						\item $v\in U_{f_1}$ and $\{uv_\ell,v\}\in E_{f_1}$. \label{infcutedgeCon1051}
						\item $v\in U_{f_2}$ and $\{uv_\ell,v\}\in E_{f_2}$. \label{infcutedgeCon1052}
						\item $v\in V(d_{f_1},\gi(c))$ and $\{uv_\ell,v\}\in E(d_{f_1},\pp(c))$.\label{infcutedgeCon1053}
						\item $v\in V(d_{f_1},\gi(c))$ and $\{uv_\ell,v\}\in E_{f_1}$. \label{infcutedgeCon1054}
						\item $v\in V(d_{f_1},\gi(c))$ and $\{uv_\ell,v\}\in E_{f_2}$. \label{infcutedgeCon1055}
					\end{enumerate}
					
				\end{definition}

						Now, we are ready to define the term {\em info-cutter} of an info frame. A triple $C=(c,F_1=(f_1(c),d_{f_1},E_{f_1},U_{f_1},\mathsf{V^*Dir}_{f_1}),F_2=(f_2(c),d_{f_2},E_{f_2},U_{f_2},\mathsf{V^*Dir}_{f_2}))$ is an info-cutter of an info-frame $F=(f,d_f,E_f,U_f,\mathsf{V^*Dir}_f)$ if $C$ exhibits all the properties defined in this subsection with respect to $F$. Formally, we have the following definition.
						
						\begin{definition}[{\bf Info-Cutter}] \label{def:infCt}
							Let $F=(f,d_f,E_f,U_f,\mathsf{V^*Dir}_f)$ be an info-frame. Then $C=(c,F_1=(f_1(c),d_{f_1},E_{f_1},U_{f_1},\mathsf{V^*Dir}_{f_1}),F_2=(f_2(c),d_{f_2},E_{f_2},U_{f_2},\mathsf{V^*Dir}_{f_2}))$ is an {\em info-cutter} of $F$ if $C$ exhibits the following properties with respect to $F$:
							{\bf Info-Cutter Template},
								  {\bf $U_f$-Partitioned}, 
							{\bf Equality of Common Parts}, 
							{\bf Equality of $V^*$ Directions},
							{\bf Partition of $E_f$ With Both Endpoints on $f$}, 
								 {\bf Partition of $\{u,uv_1\}\in E_f$ With One Endpoint on $f$}, 
							 {\bf Partition of $\{uv_{\mathsf{index}(u,v)},v\}\in E_f$ With One Endpoint on $f$} and 
								{\bf Partition of $\{u,v\}\in E$ Drawn Strictly Inside $f$}.
						\end{definition}
						
						%For an info-cutter $C=(c,F_1,F_2)$ of $F$, we denote by $F_1(C)$ and $F_2(C)$ the info-frames $F_1$ and  respectivelly

%!TEX root =Main-Movement.tex

\subsection{The function $\mathsf{Splitter}$}\label{sec:splitter}

In this subsection, we introduce the $\mathsf{Splitter}$ function. Let $F$ be an info-frame, let $d$ be a drawing of $F$ and let $c$ be a cutter of $f$. The function $\mathsf{Splitter}$ on input $(F,d,c)$ returns a triple $(C,d_1,d_2)$ where $C=(c,F_1,F_2)$ is an info-cutter of $F$, $d_1$ is a drawing of $F_1$ and $d_2$ is a drawing of $F_2$ (e.g., see Figures~\ref{fig:splitter4} and~\ref{fig:splitter5}). We denote by $\mathsf{Splitter}(F,d,c)_{C}$, $\mathsf{Splitter}(F,d,c)_{d_1}$ and $\mathsf{Splitter}(F,d,c)_{d_2}$, the output $C$, $d_1$ and $d_2$, respectively, of the $\mathsf{Splitter}$ function, on input $(F,d,c)$. Later, we show that when we ``glue'' $d_1$ and $d_2$ with a function called $\mathsf{Glue}$, we reconstruct the drawing $d$.

\begin{algorithm}[!t]
	\SetKwInOut{Input}{Input}
	\SetKwInOut{Output}{Output}
	\medskip
	{\textbf{function} $\mathsf{SplitterPart}1$}$(\langle F,d,c \rangle)$\;
	$d_1\gets d$\;
	\For{every $\{u,v\}\in E$	\label{Alg:SpiltPartMakeTurn1}}
	{
		
		\For{every $0\leq i\leq  \mathsf{index}(u,v)$ such that $\{uv_i,uv_{i+1}\}\in E(d)$ }
		{
				\For {every turning point $(p,\{u,v\})$ in $f_1(c)$ such that $p\in d(\{uv_i,uv_{i+1}\})\setminus \{d(u_i),d(uv_{i+1})\}$}
			{
				$\mathsf{MakeVer}(d_1,\{u,v\},p)$\;
			}
		}
	}\label{Alg:SpiltPartMakeTurn2}
	Delete every vertex drawn strictly outside $f_1(c)$ in $d_1$\;\label{algo1line39}
	Delete every edge drawn strictly outside $f_1(c)$ in $d_1$\;
	Delete every edge drawn strictly outside $f_1(c)$ except for at least one of its endpoints, which is drawn on $f_1(c)$, in $d_1$\; 	 \label{algo1line392}
	
		\For{every $\{u,v\}\in E$\label{Alg:FixingTurn1}}
		{
	$a_1\gets 1$, $i\gets 1$\;
	
			\While{$i\leq  \mathsf{index}(u,v)$}
			{
				\If{$(d_1(uv_i),\{u,v\})$ is a turning point in $f_1(c)$}
				{
					Rename $uv_i$ to $uv_{a_1}$ in $d_1$\;
					$a_1\gets a_1+1$,
					$i\gets i+1$;
				}
			}
			\Else{$\mathsf{DeleteVer}(d_1,uv_i)$ \label{Alg:FixingTurnDelete}\;} 
			}\label{Alg:FixingTurn2}
		$F_1\gets \mathsf{IndInfFra}(d_1,f_1(c))$\; \label{algoIndFram}
	\Return$(d_1, F_1)$\; \label{algo1line40}
	\caption{$\mathsf{SplitterPart}1$}
	\label{alg:SplitterPart}
\end{algorithm}

\begin{algorithm}[!t]
	\SetKwInOut{Input}{Input}
	\SetKwInOut{Output}{Output}
	\medskip
	{\textbf{function} $\mathsf{Splitter}$}$(\langle F,d,c \rangle)$\;
	$(d_1,F_1)\gets \mathsf{SplitterPart}1(F,d,c)$\;
	$(d_2,F_2)\gets \mathsf{SplitterPart}2(F,d,c)$\;
	$C\gets (c,F_1,F_2)$\; 
	\Return$(C, d_1, d_2)$\;
	
	\caption{$\mathsf{Splitter}$}
	\label{alg:splitter}
\end{algorithm}

\smallskip\noindent{\bf The Function $\mathsf{SplitterPart}1$.} In Algorithm \ref{alg:splitter}, we define the function $\mathsf{Splitter}$ and its output $C=(c,F_1,F_2)$, $d_1$ and $d_2$.
We show that $C$ is an info-cutter of $F$, $d_1$ is a drawing of $F_1$ and $d_2$ is a drawing of $F_2$.
For this purpose, we define another function, $\mathsf{SplitterPart}1(F,d,c)$ in Algorithm \ref{alg:SplitterPart} (used by the function $\mathsf{Splitter}$) that returns $(d'_1, F(d'_1))$, such that $F(d'_1)$ is an info-frame, and $d'_1$ is a drawing of $F(d'_1)$. The function $\mathsf{Splitter}$ also uses the function $\mathsf{SplitterPart}2$, which is identical to $\mathsf{SplitterPart}1$ up to the obvious changes. Later, we will explain the steps of Algorithm \ref{alg:SplitterPart}. First, we define the two functions used by $\mathsf{SplitterPart}1$.

The first function we define, called $\mathsf{DeleteVer}$, is the ``reverse'' operation of $\mathsf{MakeVer}$ (see Definition \ref{def:OpMakeVer} in Section \ref{sec:infoCut}). Intuitively, this function delete an ``unnecessary'' vertex from $V^*$ in a $G^*$-drawing $d$. That is, a vertex $uv_i\in V(d)\cap V^*$ such that $(d(uv_i),\{u,v\})$ is not a turning point in $f_1(c)$ (e.g., see the vertex $st_3$ in Figure~\ref{fig:splitter3}, which is deleted in Figure~\ref{fig:splitter4}). Observe that in this case, $\{uv_{i-1},uv_i\},\{uv_i,uv_{i+1}\}\in E(d)$. We now define formally the function $\mathsf{DeleteVer}$, whose input is a $G^*$-drawing $d$ and a vertex $uv_i\in V(d)\cap V^*$ such that $\{uv_{i-1},uv_i\},\{uv_i,uv_{i+1}\}\in E(d)$, and delete the vertex from $d$. 
For two sequences of points $P=(p_1,\ldots,p_\ell)$, and $Q=(q_1,\ldots,q_k)$, where $p_\ell=q_1$, we denote by $P_1\cdot P_2$ the sequence of points $(p_1,\ldots,p_\ell,q_2,\ldots,q_k)$.

\begin{definition}[{\bf $\mathsf{DeleteVer}$}]\label{def:OpDeleteVer}
	Let $d$ be a $G^*$-drawing and let $uv_i\in V(d)\cap V^*$ such that $\{uv_{i-1},uv_i\},\{uv_i,uv_{i+1}\}\in E(d)$. Then, $\mathsf{DeleteVer}(d,uv_i)$ performs the following steps on $d$:
	%	 let $\{u,v\}\in E$, and let $p\in \gi(d)$ such that there exists $i\in \mathbb{N}$ for which $p\in \gi(d(\{uv_i,uv_{i+1}\}))$. Then, $\mathsf{MakeVer}(d,\{u,v\},p)$ performs the following steps on $d$:
	\begin{itemize}
		\item Add the edge $\{uv_{i-1},uv_{i+1}\}$ to $E(d)$, and update $d(\{uv_{i-1},uv_{i+1}\})=d(\{uv_{i-1},uv_{i}\})\cdot d(\{uv_i,uv_{i+1}\})$.
		\item  Delete the edges $\{uv_{i-1},uv_{i}\}$ and $\{uv_{i},uv_{i+1}\}$ from $E(d)$.
		\item For every $i<\ell\leq \mathsf{index}(u,v)$, rename $uv_\ell$ to $uv_{\ell-1}$. 
	\end{itemize}
\end{definition}

\begin{figure}[!t]
	\centering
	\begin{subfigure}{0.48\textwidth}
		\includegraphics[width = \textwidth, page = 58]{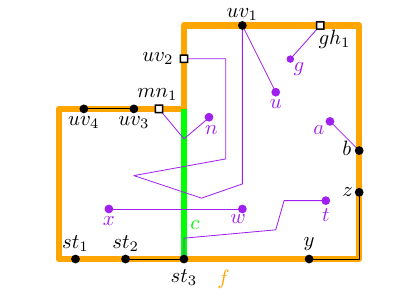}
		\subcaption{}
		\label{fig:splitter1}
	\end{subfigure}
	\hfil
	\begin{subfigure}{0.48\textwidth}
		\includegraphics[width = \textwidth, page = 59]{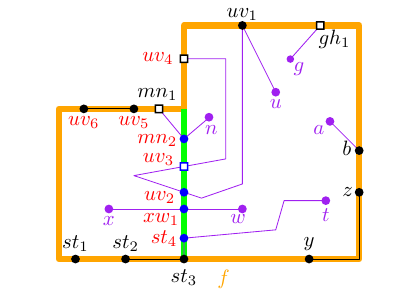}
		\subcaption{}
		\label{fig:splitter2}
	\end{subfigure}
	\hfil
	\begin{subfigure}{0.48\textwidth}
		\includegraphics[width = \textwidth, page = 60]{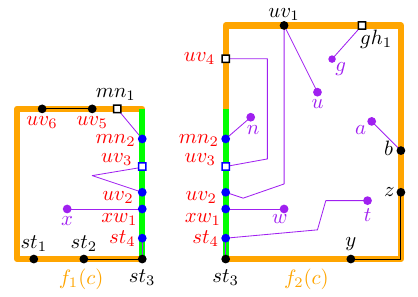}
		\subcaption{}
		\label{fig:splitter3}
	\end{subfigure}
	\hfil
	\begin{subfigure}{0.48\textwidth}
		\includegraphics[width = \textwidth, page = 61]{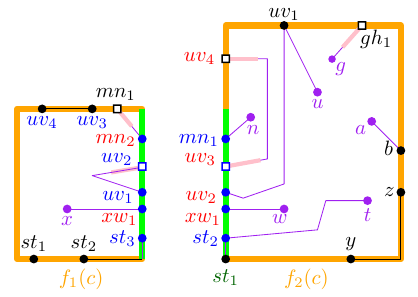}
		\subcaption{}
		\label{fig:splitter4}
	\end{subfigure}
		\hfil
		\begin{subfigure}{0.48\textwidth}
			\includegraphics[width = \textwidth, page = 62]{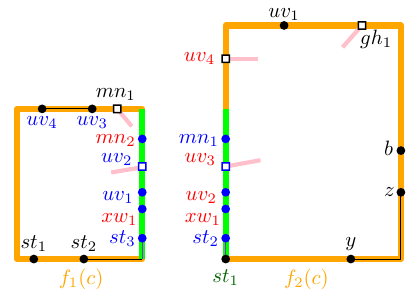}
			\subcaption{}
			\label{fig:splitter5}
		\end{subfigure}
	\caption{Example of the $\mathsf{Splitter}$ function. The vertices mapped to points in $\gis(f) \setminus \mathsf{GridPointSet}(f)$, $\gis(f_1) \setminus \mathsf{GridPointSet}(f_1)$ and $\gis(f_2) \setminus \mathsf{GridPointSet}(f_2)$ are denoted by hollow squares. (a) A drawing $d$ of the info-frame $F$ described in Figure~\ref{fig:InfF}. A cutter $c$ of $f$ is colored green. (b) Adding turning points (colored red) by the function. (c) Deleting vertices outside $f_1(c)$ and $f_2(c)$ by the function. (d) Deleting unnecessary vertices and renaming of the remaining vertices (colored blue) by the function. (e) The info-frames for $f_1(c)$ and $f_2(c)$ returned by the function. Directions are shown by thick pink lines.}
	\label{fig:splitter}
\end{figure}

Next, we define the term {\em info-frame induced} by a $G^*$-drawing. Intuitively, given a $G^*$-drawing $d$ and a frame $f$, we would like to compute the info-frame $F=(f,d_f,E_f,U_f,\mathsf{V^*Dir}_f)$ such that $d$ is a drawing of $F$ (e.g., see Figures~\ref{fig:splitter4} and~\ref{fig:splitter5}). Observe that not for every $G^*$-drawing $d$ and a frame $f$ there exists such an info-frame, but if there exists, then it is unique. If $d$ is not bounded by $f$, or if there is a vertex $uv_i\in V(d)\cap V^*$ drawn strictly inside $f$, then there is no aforementioned info-frame. In addition, recall that we demand $d=^{\pp(f)}d_f$ (see Definition \ref{def:DrawindEq} in Section \ref{sec:infoFra}). Now, observe that for every drawing $d_f$ on $f$, and for every turning point $(p,\{u,v\})$ in $f$ in $d_f$, there exists $uv_i\in V(d_f)\cap V^*_{\{u,v\}}$ such that $d_f(uv_i)=p$. So, since $d=^{\pp(f)}d_f$, we get that $d(uv_i)=p$. We conclude that for every turning point $(p,\{u,v\})$ in $f$ in $d$, there exists $uv_i\in V(d)\cap V^*_{\{u,v\}}$ such that $d(uv_i)=p$. Now, recall that for every $uv_i\in V(d_{f_1})\cap V^*$ such that $d_{f_1}(uv_i)\in \gis(\fin)\setminus \mathsf{GridPointSet}(f_{\mathsf{init}})$ there are no edges on $f$ attached to it, and $uv_i$ has exactly one neighbor.
Formally, we define the term  info-frame induced by a $G^*$-drawing as follows.

\begin{definition}[{\bf Frame Induced By a $G^*$-Drawing}]\label{def:inducedInfoFrame}
Let $d$ be a $G^*$-drawing, and let $f$ be a frame such that the following conditions are satisfied:
\begin{enumerate}
	\item $d$ is bounded by $f$.
	\item For every turning point $(p,\{u,v\})$ in $f$ in $d$, there exists $uv_i\in V(d)\cap V^*_{\{u,v\}}$ such that $d(uv_i)=p$.
	\item For every $uv_i\in V(d)\cap V^*$, $d(uv_i)\in \gi(f)$.
	\item For every $uv_i\in V(d_{f_1})\cap V^*$ such that $d_{f_1}(uv_i)\in \gis(\fin)\setminus \mathsf{GridPointSet}($ $f_{\mathsf{init}})$ there exists exactly one $z\in \{uv_{i-1},uv_{i+1}\}$ such that $\{z,uv_i\}\in E(d)$. In addition, $\pp$ $(d(\{z,uv_i\}))\cap \pp(f)\subseteq \{d(uv_i),d(z)\}$.
\end{enumerate} 
Then, the {\em info-frame induced by $d$ and $f$} is $\mathsf{IndInfFra}(d,f)=(f,d_f,E_f,U_f,\mathsf{V^*Dir}_f)$, where:
\begin{enumerate}
	\item $d_f$ is defined as follows: $V(d_f)=V(d,\gi(f)$, $E(d_f)=E(d,\pp($ $f))$. For every $u\in V(d_f)$, $d_f(u)=d(u)$, and for every $e\in E(d_f)$, $d_f(e)=d(e)$. 
	\item $E_f\subseteq E(d)$ is the set of edges drawn strictly inside $f$ in $d$ except for at least one of their endpoints, which is drawn on $f$.
	\item $U_f$ is the set of vertices drawn strictly inside $f$ in $d$.
	\item For every $uv_i\in V(d_{f_1})\cap V^*$ such that $d_{f_1}(uv_i)\in \gis(\fin)\setminus \mathsf{GridPointSet}($ $f_{\mathsf{init}})$, let $z\in \{uv_{i-1},uv_{i+1}\}$ such that $\{z,uv_i\}\in E(d)$. Let $E(d)=(d(uv_i),p_1,\ldots,d(z))$. Then, $\mathsf{V^*Dir}_f(uv_i)=p_1$.\label{def:inducedInfoFrameC4}
\end{enumerate}
\end{definition}

It is easy to see that $\mathsf{IndInfFra}(d,f)$ satisfies the conditions of Definition \ref{def:infFr3}. Therefore, $\mathsf{IndInfFra}(d,f)$ is an info-frame:

\begin{observation}\label{obs:induceInfo}
Let $d$ be a $G^*$-drawing, and let $f$ be a frame such that the conditions of Definition \ref{def:inducedInfoFrame} are satisfied. Then, $\mathsf{IndInfFra}(d,f)$ is an info-frame.
\end{observation}

Later, we show that $d$ is a drawing of $F$.
Now we are ready to describe the steps of Algorithm~\ref{alg:SplitterPart}.

	 \smallskip\noindent{\bf Lines \ref{Alg:SpiltPartMakeTurn1}-\ref{Alg:SpiltPartMakeTurn2}: Turning Points in $f_1(c)$ Turn to Vertices From $V^*$.} In the first part of the algorithm, we turn every turning point of every edge $\{u,v\}$ in $f_1(c)$ to a vertex from $V^*_{\{u,v\}}$. In particular, for every $\{u,v\}\in E$ we iterate over $0\leq i \leq \mathsf{index}(u,v)$ such that $\{uv_i,uv_{i+1}\}\in E(d)$. Then, for every turning point $p$ drawn on this edge, except at its endpoints, we add a vertex from the set $V^*_{\{u,v\}}$ by activating $\mathsf{MakeVer}(d_1,\{u,v\},p)$ (see Definition \ref{def:OpMakeVer} in Section \ref{sec:infoCut}) (e.g., see the red vertices $st_4$, $xw_1$, $uv_2$, $uv_3$ and $mn_2$ in Figure~\ref{fig:splitter2}). Recall that after activating $\mathsf{MakeVer}$ $d_1$ remains $G^*$-drawing, due to 
	 Observation \ref{obs:makever}. 
	 
	  \smallskip\noindent{\bf Lines \ref{algo1line39}-\ref{algo1line392}: Deleting Edges and Vertices Drawn Strictly Outside $f_1(c)$.} In these steps, we aim to obtain from $d_1$ exactly the part of $d_1$ that is bounded by $f_1(c)$. Observe that, after preforming Lines \ref{Alg:SpiltPartMakeTurn1}-\ref{Alg:SpiltPartMakeTurn2}, for each $e\in E(d_1)$, exactly one of the following conditions holds:
	  \begin{itemize}
	  	\item $e$ is drawn inside $f_1(c)$.
	  	\item $e$ is drawn strictly outside $f_1(c)$.
	  	\item $e$ is drawn strictly outside $f_1(c)$ except at least one of its endpoints, which is drawn on $f_1(c)$.
	  \end{itemize}
  
  We delete every edge that satisfies the second or third conditions above. In addition, we delete every vertex drawn strictly outside $f_1(c)$ in $d_1$  (e.g., see Figures~\ref{fig:splitter2} and ~\ref{fig:splitter3}).
  
	  \smallskip\noindent{\bf Lines \ref{Alg:FixingTurn1}-\ref{Alg:FixingTurn2}: Deleting Unnecessary Vertices From $V^*$ and Fixing Labelings.} Observe that at this stage of the algorithm, we have two issues with $d_1$: First, there might be some vertices from $V^*_{\{u,v\}}$ that are not drawn on a turning point of $\{u,v\}$ in $f_1(c)$. Second, since we deleted vertices drawn strictly outside $f_1(c)$, the labeling of the vertices in $V^*_{\{u,v\}}$ might not be continuous, for some $\{u,v\}\in E$. We fix these two issues simultaneously: We iterate over $V^*_{\{u,v\}}$, for each $\{u,v\}\in E$. If $uv_i$ is indeed drawn on a turning point of $\{u,v\}$ in $f_1(c)$ , we fix its labeling (e.g., see the blue vertices in Figure~\ref{fig:splitter4}); otherwise, we delete it, using $\mathsf{DeleteVer}$ defined in Definition \ref{def:OpDeleteVer} (e.g., see the vertex $st_3$ in Figure~\ref{fig:splitter3}, which is deleted in Figure~\ref{fig:splitter4}). Observe that every vertex in $V(d_1)\cap V^*$ at the beginning of the algorithm is drawn on $f$, since $d_1=d$ and $d$ is a drawing of $F$. After Lines  \ref{Alg:SpiltPartMakeTurn1}-\ref{Alg:SpiltPartMakeTurn2}, we might add vertices to $V(d_1)\cap V^*$ drawn on $f_1(c)$. So, after these lines, each vertex in $V(d_1)\cap V^*$ is drawn on $f_1(c)$ or on $f$. After Lines \ref{algo1line39}-\ref{algo1line392}, we delete every vertex drawn strictly outside $f_1(c)$. So, at this stage, every vertex in $V(d_1)\cap V^*$ is drawn on $f_1(c)$. Thus, if for some $uv_i\in V(d_1)\cap V^*$, $(d_1(uv_i),\{u,v\})$ is not a turning point in $f_1(c)$, then $\{uv_{i-1},uv_i\},\{uv_{i-1},uv_i\}\in E(d_1)\cap E^*$. Therefore, $\mathsf{DeleteVer}(d_1,uv_i)$ in Line \ref{Alg:FixingTurnDelete} is well defined. 
	 
	\smallskip\noindent{\bf Line \ref{algoIndFram}: Computing the Induced Frame of $d_1$ and $f_1(c)$.} Observe that at this stage of the algorithm we have the following conditions satisfied (where for the last condition, we present a proof below):
	\begin{enumerate}
		\item $d_1$ is a $G^*$-drawing (see Lemma \ref{lem:dIsGStar}) bounded by $f_1(c)$.
		\item Every turning point of every edge $\{u,v\}$ in $f_1(c)$ is a vertex from $V^*_{\{u,v\}}$.
		\item Every vertex form $V(d)\cap V^*$ is drawn on $f_1(c)$.
		\item  For every $uv_i\in V(d_1)\cap V^*$ such that $d_1(uv_i)\in \gis(\fin)\setminus \mathsf{GridPointSet}($ $f_{\mathsf{init}})$ there exists exactly one $z\in \{uv_{i-1},uv_{i+1}\}$ such that $\{z,uv_i\}\in E(d_1)$. In addition, $\pp(d_1(\{z,uv_i\}))\cap \pp(f_1(c))\subseteq \{d_1(uv_i),d_1(z)\}$ (see Lemma \ref{lemmaCon4}).\label{cond4}
	\end{enumerate}

We now show that Condition \ref{cond4} is satisfied. 

 \begin{lemma}\label{lemmaCon4}
	For every $uv_i\in V(d_1)\cap V^*$ such that $d_1(uv_i)\in \gis(\fin)\setminus \mathsf{GridPointSet}(f_{\mathsf{init}})$ there exists exactly one $z\in \{uv_{i-1},uv_{i+1}\}$ such that $\{z,uv_i\}\in E(d_1)$. In addition, $\pp(d_1(\{z,uv_i\}))\cap \pp(f_1(c))\subseteq \{d_1(uv_i),d_1(z)\}$.
\end{lemma}

\begin{proof}
	Let $uv_i\in V(d_1)\cap V^*$ such that $d_{1}(uv_i)\in \gis(\fin)\setminus \mathsf{GridPointSet}($ $f_{\mathsf{init}})$. First, assume that $d_{1}(uv_i)$ is on $f$. Observe that $d_{1}(uv_i)$ is not on $c$, so $d_{1}(uv_i)$ is not on $f_2(c)$. So, if there exists an edge with $uv_i$ as one of its endpoints, on $f_1(c)$, then there exists one also on $f$, a contradiction to $F$ being an info-frame (see Condition \ref {def:validDircon1} of Definition \ref{def:validDir}). Now, let $u_j=\mathsf{Identify}_{d_{1},d}(uv_i)$. Since $F$ is an info-frame, for exactly one $z'$ among $uv_{j-1}$ and $uv_{j+1}$ we get that $\{uv_j,z'\}\in d$ (see Condition \ref {def:validDircon1} of Definition \ref{def:validDir}). Now, if there are no turning points of $\{uv_j,z'\}$ in $f$ in $d$, then $z=\mathsf{Identify}_{d_{1},d_{f_2}}(z')$ is the only neighbor of $uv_i$ in $d_1$, and $z\in \{uv_i,uv_{i+1}\}$ since $d_1$ is a $G^*$-drawing. Otherwise, there exists at least one turning point of $\{uv_j,z'\}$ in $f$ in $d$. Let $p$ be the closest point to $uv_j$ in the path $d(\{uv_j,z'\})$ such that $(p,\{uv_j,z'\})$ is a turning point in $f_1(c)$ in $d$. So, there exists $z\in V(d_1)\cap U^*_{\{u,v\}}$ such that $d_1(z)=p$, $\{uv_i,z\}\in E(d_1)$ and $z\in \{uv_i,uv_{i+1}\}$, since $d_1$ is a $G^*$-drawing. Observe that $z$ is the only neighbor of $uv_i$ in $d_1$. In addition, observe that $\{uv_i,z\}$ is drawn strictly inside $f_1(c)$ except for its endpoints, which are drawn inside $f$, so  $\pp(d_1(\{z,uv_i\}))\cap \pp(f_1(c))\subseteq \{d_1(uv_i),d_1(z)\}$.
	
	Second, assume that $d_{f_1}(uv_i)$ is not on $f$. Then, it belongs to $\gi(c)\setminus \gi(f)$. So, there exists $\{uv_j,uv_{j+1}\}\in E(d)$ such that\\ $d_{f_1}(uv_i)\in \gi$ $(d(\{uv_j,uv_{j+1}\}))$ and $(d_{f_1}(uv_i),\{u,v\})$ is a turning point in $f_1(c)$ in $d$. Let $c=(c_1,\ldots,c_t)$ and let $d(\{uv_j,uv_{j+1}\})=(p_1,\ldots p_q)$. Let $a,b\in \mathbb{N}$ such that $d_{f_1}(uv_i)$ is on $\ell(c_a,c_{a+1})$ and on $\ell(p_b,p_{b+1})$. Since $d_{f_1}(uv_i)\in \gis(\fin)\setminus \mathsf{GridPointSet}(f_{\mathsf{init}})$, then $d_{f_1}(uv_i)\neq c_a,c_{a+1},p_b,p_{b+1}$. So, exactly one among $\ell(p_b,d_{f_1}(uv_i))$ and $\ell(d_{f_1}(uv_i),p_{b+1})$ is inside $f_1(c)$. Thus, $uv_i$ has exactly one neighbor in $d_1$, and no edges on $f_1(c)$ are attached to it. So, we get that there exists exactly one $z\in \{uv_{i-1},uv_{i+1}\}$ such that $\{z,uv_i\}\in E(d_1)$, and $\pp(d_1(\{z,uv_i\}))\cap \pp(f_1(c))\subseteq \{d_1(uv_i),d_1(z)\}$. This completes the proof.
\end{proof}

	 So, the conditions of Definition \ref{def:inducedInfoFrame} are satisfied, and hence $\mathsf{IndInfFra}(d_1,f_1(c))$ is well defined. Moreover, by Observation \ref{obs:induceInfo}, $\mathsf{IndInfFra}(d_1,f_1(c))$ is an info-frame (e.g., see Figure~\ref{fig:splitter5}).
	 
	The function $\mathsf{SplitterPart}_1(F,d,c)$ returns $d_1$ and $\mathsf{IndInfFra}(d_1,f_1(c))$.
	 
	 We now prove that $d_1$ returned by Algorithm \ref{alg:SplitterPart}, is indeed a drawing of $\mathsf{IndInfFra}(d_1,f_1(c))$. To this end, we first prove that $d_1$ is a $G^*$-drawing:

	 \begin{lemma} \label{lem:dIsGStar}
	 	Let $F=(f,d_f,E_f,U_f,\mathsf{V^*Dir}_f)$ be an info-frame, let $d$ be a drawing of $F$ and let $c$ be a cutter of $f$. Let $d_1$ and $F_1$ be the output of $\mathsf{SplitterPart}_1(F,d,c)$. Then, $d_1$ is a $G^*$-drawing.
	 \end{lemma}
  \begin{proof}
 %In Observation \ref{obs:induceInfo} we saw that $F_1$ is an info-frame. 
 %We show that $d_1$ is a drawing of $F_1$. 
 %Observe that, every turning point $((r,k),\{uv_i,uv_{i+1}\})$ of $f_1(c)$ turns into a $V^*$ vertex during the course of the algorithm, therefore, there is no edge that is drawn into points on $f_1(c)$ and also into points strictly outside $f_1(c)$ (see Figure~\ref{fig:InfSplitt}). Thus, the operation of deleting edges that are drawn strictly outside $f_1(c)$ (maybe except at the endpoints), in Line \ref{algo1line39}, is well defined. 
 We prove that $d_1$ is a $G^*$-drawing by showing that the conditions of Definition \ref{def:gstdr} are satisfied.  We first show that Condition \ref{G*drawcon10}  is satisfied, that is, $(V(d_1),E(d_1))$ is valid, by showing that the Conditions of Definition \ref{def:ValidPair} are satisfied:  
 \begin{enumerate}
 	\item It is easy to see that $V(d_1)\subseteq V\cup V^*$ and $E(d_1)\subseteq E\cup E^*$, so Condition \ref{con:ValidPair1} is satisfied. 
 	\item Let $\{u,v\}\in E$, and assume that $V(d_1)\cap V^*_{\{u,v\}}\neq \emptyset$.
 	\begin{enumerate}
 		\item  In Lines \ref{Alg:FixingTurn1}-\ref{Alg:FixingTurn2} the algorithm labels the vertices in $V(d_1)\cap V^*_{\{u,v\}}$ from $1$ to $\mathsf{index}(u,v)=|V(d_1)\cap V^*_{\{u,v\}}|$, so Condition \ref{con:ValidPair221} is satisfied. 
 		\item We aim to prove that $E(d_1)\cap E^*_{\{u,v\}}\subseteq \{\{u,uv_1\}\}\cup \{\{uv_j,uv_{j+1}\}~|~1\leq j\leq \mathsf{index}(u,v)-1\}\cup \{\{uv_{\mathsf{index}(u,v)},v\}\}$. We saw that $d_1$ obtained by the end of Line \ref{Alg:SpiltPartMakeTurn2} is a $G^*$-drawing, so Condition \ref{con:ValidPair222} holds at this stage. In Lines \ref{algo1line39}-\ref{algo1line392}, we delete edges, so the condition still holds by the end of these steps. Now, in Lines \ref{Alg:FixingTurn1}-\ref{Alg:FixingTurn2} observe that we iterate over the vertices in $V^*_{\{u,v\}}$ by the order of their labeling at the beginning of this stage. In every iteration, we either rename or delete the current vertex from $V^*_{\{u,v\}}$. For every vertex we delete, we connect its two neighbors with an edge and we do not change the labeling of vertices we have already labeled. So, when we rename a vertex $v_{a_1}$, then the label of the previous vertex we renamed is $v_{a_1-1}$ and it is the only vertex that may be connected with an edge to $v_{a_1}$ among the vertices we have already labeled. Thus, by the end of the algorithm, Condition \ref{con:ValidPair222} is satisfied.
 		\item We aim to show that Condition \ref{con:ValidPair223} holds, that is, $\{u,v\}\notin E(d_1)$. Assume towards a contradiction that $\{u,v\}\in E(d_1)$. We saw that $d_1$ obtained by the end of Line \ref{Alg:SpiltPartMakeTurn2} is a $G^*$-drawing, so Condition \ref{con:ValidPair223} holds in this stage. So, $\{u,v\}\notin E(d_1)$ by the end of Line \ref{Alg:SpiltPartMakeTurn2}. Now, since $\{u,v\}\in E(d_1)$ by the end of the algorithm, this means that in Lines \ref{algo1line39}-\ref{algo1line392} we deleted every vertex in $V(d)\cap V_{\{u,v\}}^*$, a contradiction to the assumption that $V(d)\cap V_{\{u,v\}}^*\neq \emptyset$.  %in Lines  Then since $d$ is a $G^*$-drawing, then $\{u,uv_1\}\notin E(d)$. Now, since $d'_1$ is obtained from $d$ by renaming points on edges, then, it follows that if  $\{u,v\}\in E(d'_1)$ then, $\{u,uv_1\}\notin E(d'_1)$, and if $\{u,uv_1\}\notin E(d'_1)$, then $\{u,v\}\in E(d'_1)$. Assume that $\{u,v\}\notin E(d)$ and  $\{u,v\}\in E(d'_1)$. Again, since $d'_1$ is obtained from $d$ by renaming points on edges, it follows that $\{u,uv_1\},\{uv_1,uv_2\},\ldots,\{uv_{t},v\}\in E(d)$, and the vertices $uv_1,\ldots, uv_{\mathsf{index}(u,v)}$ were deleted by the $\mathsf{SplitterPart}1$ function, and therefore  $uv_1\notin V(d'_1)$. Assume that $\{u,v\}\notin E(d)$ and  $\{u,v\}\notin E(d'_1)$. Since $d$ is a $G^*$-drawing, then $uv_1\in V(d)$. Moreover, since $\{u,v\}\notin E(d'_1)$, and  $d'_1$ is obtained from $d$ by renaming points on edges, then $uv_1\in V(d'_1)$. Therefore, Conditions \ref{G*drawcon5} and 
 	\end{enumerate} 
  \end{enumerate}
 We saw that the conditions of Definition \ref{def:ValidPair} are satisfied, so $(V(d_1),E(d_1))$ is valid, and hence Condition \ref{G*drawcon10}  is satisfied.
 
 Now, we aim to prove that Condition \ref{G*drawcon2} is satisfied. Let $u\in V(d_1)\cap V$. Observe that $d_1(u)=d(u)$, and since $d$ is a $G^*$-drawing, $d(u)\in\gps(f_{\mathsf{init}})$. Let $uv_i\in V(d_1)\cap V^*$. If $uv_i$ is on $f$, then there exists $uv_j\in V(d)\cap V^*$ such that\\ $d_1(uv_i)=d(uv_j)$. So, since $d$ is a $G^*$-drawing, $d(uv_j)\in {\cal P}(f_{\mathsf{init}})$. If $d_1(uv_i)$ is not on $f$, then it belongs to $\gi(c)\setminus \gi(f)$. So, there exists $\{uv_j,uv_{j+1}\}\in E(d)$ such that $d_{f_1}(uv_i)\in \gi$ $(d(\{uv_j,uv_{j+1}\}))$ and $(d_{f_1}(uv_i),\{u,v\})$ is a turning point in $f_1(c)$ in $d$. Let $c=(c_1,\ldots,c_t)$ and let $d(\{uv_j,uv_{j+1}\})=(p_1,\ldots p_q)$. Let $a,b\in \mathbb{N}$ such that $d_{f_1}(uv_i)$ is on $\ell(c_a,c_{a+1})$ and on $\ell(p_b,p_{b+1})$. Now, observe that $d_1(uv_i)\neq p_b,p_{b+1}$, and since $d$ is a $G^*$-drawing, $d(\{uv_j,uv_{j+1}\})\in {\cal P}^*(f_{\mathsf{init}})$. Thus, there exist $x,y\in \gps(\fin)$ such that $\ell(p_b,p_{b+1})$ is on $\ell(x,y)$. Therefore, $d_1(uv_i)$ is the intersection point of $\ell(x,y)$ and $\ell(c_a,c_{a+1})$, so $d_1(uv_i)\in \gis(\fin)$. Thus, we get that Condition \ref{G*drawcon2}  is satisfied.
 
 We now show that Condition \ref{G*drawcon3} is satisfied.
 Let $\{u,v\}\in E(d_1)\cap E$. Observe that, in this case, there are no turning points of $\{u,v\}$ in $f_1(c)$ in $d$. Therefore, $\{u,v\}\in E(d)\cap E$, $d_1(\{u,v\})=d(\{u,v\})$, and since $d$ is a $G^*$-drawing, $d(\{u,v\})\in {\cal P}(f_{\mathsf{init}})$. Now, let $\{uv_i,uv_{i+1}\}\in E(d)\cap E^*$. There exists $\{uv_j,uv_{j+1}\}\in V(d)$, where $d(\{uv_j,uv_{j+1}\})=(p_1,\ldots p_q)$, and $1\leq a_1\leq a_2<a_3\leq a_4\leq q$ such that $d_1(uv_i)$ is on $\ell(p_{a_1},p_{a_2})$, $d_1(uv_{i+1})$ is on $\ell(p_{a_3},p_{a_4})$ and $d_1(\{uv_i,uv_{i+1}\})=(d_1(uv_i),p_{a_2},\ldots p_{a_3},d_1(uv_{i+1}))$ (if $d_1(uv_i)=p_{a_2}$ or $d_1(uv_{i+1})=p_{a_3}$, then omit $p_{a_2}$ or $p_{a_3}$ from $d_1(\{uv_i,uv_{i+1}\})$, respectively). Now, since $d$ is a $G^*$-drawing, $d(\{u,v\})\in {\cal P}^*(f_{\mathsf{init}})$, and it is easy to see that $d_1(uv_{i+1})\in {\cal P}^*(f_{\mathsf{init}})$. So, Condition \ref{G*drawcon3}  is satisfied.
 
 Since $d_1$ obtained by the end of Line \ref{Alg:SpiltPartMakeTurn2} is a $G^*$-drawing, and $d_1$ in the end of the algorithm is obtained from $d$ by deleting edges and vertices, and renaming points on edges, then it follows that Conditions \ref{G*drawcon7}, \ref{G*drawcon8} and \ref{G*drawcon9} are satisfied.

 In conclusion, we proved that the conditions of Definition \ref{def:gstdr} are satisfied, so $d_1$ is a $G^*$-drawing.
\end{proof}

 Next, we prove that $d_1$ returned by Algorithm \ref{alg:SplitterPart} is a drawing of $\mathsf{IndInfFra}(d_1,f_1(c))$.

 \begin{lemma} \label{lem:dIsDraOfF1}
 	Let $F=(f,d_f,E_f,U_f,\mathsf{V^*Dir}_f)$ be an info-frame, let $d$ be a drawing of $F$ and let $c$ be a cutter of $f$. Let $d_1$ and $F_1=(f_1(c),d_{f_1},E_{f_1},U_{f_1},\mathsf{V^*Dir}_{f_1})$ be the output of $\mathsf{SplitterPart}_1(F,d,$ $c)$. Then, $d_1$ is a drawing of $F_1$. 
 \end{lemma}
 
\begin{proof}
We need to show that the conditions of Definition \ref{def:infFrDr} are satisfied. Observe that Conditions \ref{infFramcon1}-\ref{infFramcon4} are trivially satisfied due to the definition of a frame induced by a $G^*$-drawing (Definition \ref{def:inducedInfoFrame}). We prove now that Condition \ref{infFramcon5} holds. Let $\{u,v\}\in E$ such that $u,v\in U_{f_1}$ and $V(d_{f_1})\cap V^*_{\{u,v\}}=\emptyset$. We aim to prove that $\{u,v\}\in E(d_1)$. This can be shown similarly to the proof of Condition \ref{con:ValidPair223} of Definition \ref{def:ValidPair} in Lemma \ref{lem:dIsGStar}. Now, we show that Condition \ref{infFramcon6} holds. Let $uv_i\in V(d_{f_1})\cap V^*$ such that $d_{f_1}(uv_i)\in \gis(\fin)\setminus \mathsf{GridPointSet}(f_{\mathsf{init}})$ and let $z\in \{uv_{i-1},uv_{i+1}\}$ such that $\{z,uv_i\}\in E(d_1)$. Observe that $\mathsf{V^*Dir}_{f_1}(uv_i)=p_1$, where $E(d_1)=(d(uv_i),p_1,\ldots,d(z))$, so $\ell(\mathsf{V^*Dir}_{f_1}(uv_i),d_f(uv_i))$ is on $d_1(\{uv_i,z\})$. Thus, Condition \ref{infFramcon6} holds. Therefore, the conditions of Definition \ref{def:infFrDr} are satisfied, so $d_1$ is a drawing of $F_1$.
\end{proof}

Now, we show that $C=(c,F_1,F_2)$ returned by the $\mathsf{Splitter}$ function is an info-cutter of $F$, by showing that the conditions of Definition \ref{def:infCt} are satisfied. We begin by proving that $C$ is an {\bf Info-Cutter Template}, {\bf $U_f$-Partitioned} and exhibits the {\bf Equality of Common Parts} property with respect to $F$ in Lemmas \ref{lem:CisInfoCut1} and \ref{lem:CisInfoCut2} and Observation \ref{lem:CisInfoCut3}, respectively.  

\begin{lemma}\label{lem:CisInfoCut1}
	Let $F=(f,d_f,E_f,U_f,\mathsf{V^*Dir}_f)$ be an info-frame, let $d$ be a drawing of $F$ and let $c$ be a cutter of $f$. Let $C=\mathsf{Splitter}(F,d,c)_C$. Then, $C$ is an {\bf Info-Cutter Template} with respect to $F$.  
\end{lemma}

\begin{proof}
	Let $F_1=(f_1(c),d_{f_1},E_{f_1},U_{f_1},\mathsf{V^*Dir}_{f_1})$ and $F_2=(f_2(c),d_{f_2},E_{f_2},U_{f_2},\mathsf{V^*Dir}_{f_2})$.
	First, $c$ is a cutter of $f$ by our assumption. In addition, in Observation \ref{obs:induceInfo} we saw that $F_1$ and $F_2$ are info-frames. So, by Definition \ref{def:infCTurPoints0}, $C$ is an {\bf Info-Cutter Template} with respect to $C$.
\end{proof}

\begin{lemma}\label{lem:CisInfoCut2}
	Let $F=(f,d_f,E_f,U_f,\mathsf{V^*Dir}_f)$ be an info-frame, let $d$ be a drawing of $F$ and let $c$ be a cutter of $f$. Let $C=\mathsf{Splitter}(F,d,c)_C$. Then, $C$ is {\bf $U_f$-Partitioned} with respect to $F$.  
\end{lemma}
\begin{proof}
	Let $F_1=(f_1(c),d_{f_1},E_{f_1},U_{f_1},\mathsf{V^*Dir}_{f_1})$ and $F_2=(f_2(c),d_{f_2},E_{f_2},U_{f_2},\mathsf{V^*Dir}_{f_2})$.
We show that $C$ is {\bf $U_f$-Partitioned} with respect to $F$ (see Definition \ref{def:infCTurPoints13}). Observe that during the course of $\mathsf{SplitterPart}1$ (see Algorithm \ref{alg:SplitterPart}) we do not change the drawing of vertices from the set $V$ in $d_1$ (and, respectively, $d_2$ with respect to $\mathsf{SplitterPart}2$). Therefore, for every vertex $u\in U_f$, that is, a vertex from the set $V$ drawn strictly inside $f$, we have that exactly one of the following conditions is satisfied:
\begin{itemize}
	\item $u$ is drawn strictly inside $f_1(c)$. Then, $u\in U_{f_1}$.
	\item $u$ is drawn strictly inside $f_2(c)$. Then, $u\in U_{f_2}$.
	\item $u$ is drawn on $c$. Observe that, since the first and last vertices of $c$ are drawn on $f$, $\gi(c)\cap \gi(f)\neq \emptyset$. So, since $u$ is drawn strictly inside $f$ and also on $c$, $d(u)=p$ for some $p\in \gi(c)\setminus \gi(f)$. Therefore, since $d_{f_1}$ is a drawing on $f_1(c)$, $u\in V(d_{f_1},\gi(c)\setminus \gi($ $f))$. 
\end{itemize}
Thus, we conclude that $U_f=(V(d_1,\gi(c)\setminus \gi(f))) \cup U_{f_1}\cup U_{f_2}$ (Condition \ref{def:infCutCon6} of Definition \ref{def:infCTurPoints13}). In addition, it is easy to see that $U_{f_1}\cap U_{f_2}=\emptyset$ (Condition \ref{def:infCutCon62} of Definition \ref{def:infCTurPoints13}). So, $C$ is {\bf $U_f$-Partitioned} with respect to $F$.
\end{proof}

Observe that, by the definition of $\mathsf{SplitterPart}1$ (see Algorithm \ref{alg:SplitterPart}), for every $u\in (V(d)\cap V) \cup (V(d_1)\cap V)$ drawn inside $f_1(c)$, $d(u)=d_1(u)$. Moreover, for every $\{uv_i,uv_{i+1}\}\in E(d)$ and $p\in d(\{uv_i,uv_{i+1}\})$ inside $d_1(c)$, there exists $\{uv_j,uv_{j+1}\}\in E(d_1)$ such that $p\in d_1(\{uv_j,uv_{j+1}\})$, and vice versa. In particular, $d_f=_{\mathsf{rename}}^{\mathsf{Common}}( f,f_1)d_{f_1}$, where\\ $\mathsf{Common}(f,f_1)=\gi(f)\cap \gi(f_1(c))$.

Similarly, $d_f=_{\mathsf{rename}}^{\mathsf{Common}(f,f_2)}d_{f_2}$, where\\ $\mathsf{Common}$ $(f,f_2)=\gi(f)\cap \gi(f_2(c))$, and\\ $d_{f_1}=_{\mathsf{rename}}^{\gi(c)}d_{f_2}$. Therefore, the conditions of Definition \ref{def:infCTurPoints1} are satisfied. So, we have the following observation:

\begin{observation}\label{lem:CisInfoCut3}
	Let $F=(f,d_f,E_f,U_f,\mathsf{V^*Dir}_f)$ be an info-frame, let $d$ be a drawing of $F$ and let $c$ be a cutter of $f$. Let $C=\mathsf{Splitter}(F,d,c)_C$. Then, $C$  exhibits {\bf Equality of Common Parts} with respect to $F$.  
\end{observation}

Now, we prove that $C$  exhibits {\bf Equality of $V^*$ Directions} with respect to $F$ (Definition \ref{def:infCTurPoints2}):

\begin{lemma}\label{lem:CisInfoCut31}
	Let $F=(f,d_f,E_f,U_f,\mathsf{V^*Dir}_f)$ be an info-frame, let $d$ be a drawing of $F$ and let $c$ be a cutter of $f$. Let $C=\mathsf{Splitter}(F,d,c)_C$. Then, $C$  exhibits {\bf Equality of $V^*$ Directions} with respect to $F$.  
\end{lemma}

\begin{proof}
We show that the conditions of Definition \ref{def:infCTurPoints2} are satisfied.\\
Let $uv_i\in V(d_{f_1},\gi(c))\cap V^*$ such that $d_{f_1}(uv_i)\in \gis(\fin)\setminus \mathsf{GridPointSet}(f_{\mathsf{init}})$. Observe that, similarly to the proof of Lemma \ref{lem:dIsGStar}, it can be shown that there exists $\{uv_j,uv_{j+1}\}\in E(d)$ such that $d_{f_1}(uv_i)\in \gi(d(\{uv_j,uv_{j+1}\}))$ and $(d_{f_1}(uv_i),\{u,v\})$ is a turning point in both $f_1(c)$ and $f_2(c)$ in $d$. Moreover, let $d(\{uv_j,uv_{j+1}\})=(p_1,\ldots p_q)$. Let $a,b\in \mathbb{N}$ such that $d_{f_1}(uv_i)$ is on $\ell(p_b,p_{b+1})$. \\
Since $d_{f_1}(uv_i)\in \gis$ $(\fin)\setminus \mathsf{GridPointSet}(f_{\mathsf{init}})$, then $d_{f_1}(uv_i)\neq p_b,p_{b+1}$. In addition, exactly one among $\ell(p_b,d_{f_1}(uv_i))$ and $\ell(d_{f_1}(uv_i),p_{b+1})$ is inside $f_1(c)$ and the other one is inside $f_2(c)$. Therefore, one among $\mathsf{V^*Dir}_{f_1})(uv_i)$ and $\mathsf{V^*Dir}_{f_2}(\mathsf{Identify}_{d_{f_1},d_{f_2}}(uv_i))$ is $p_b$ and the other one is $p_{b+1}$. It is easy to see that $d_{f_1}(uv_i)$ is on $\ell(p_b,p_{b+1})$. Now, let $uv_i\in V(d_{f_1},\mathsf{Common}(f,f_1))\cap V^*$, where $\mathsf{Common}(f,f_1)=\gi(f)\cap \gi($ $f_1(c))$ such that $d_{f_1}(uv_i)\in \gis(\fin)\setminus \mathsf{GridPointSet}(f_{\mathsf{init}})$. Observe that, $\ell(d_{f_1}($ $uv_i),\mathsf{V^*Dir}_{f_1}(uv_i)$ is on the first edge of $d_1(uv_i,z)$ or the first edge of $d_1(uv_i,z)$ is on $\ell(d_{f_1}(uv_i),$ $\mathsf{V^*Dir}_{f_1}(uv_i)$, where $z\in\{uv_{i-1},uv_{i+1}\}$ is the neighbor of $uv_i$ in $d_1$. Now, the first edge of $d_1(uv_i,z)$ is on the first edge of $d(\mathsf{Identify}_{d_{f_1},d_{f}}(uv_i)),z')$, where $z'$ is the neighbor of\\ $d(\mathsf{Identify}_{d_{f_1},d_{f}}(uv_i)$ in $d$. In addition, since $d$ is a drawing of $F$, $\ell(d_f(\mathsf{Identify}_{d_{f_1},d_{f}}(uv_i)),$ $\mathsf{V^*Dir}_{f}(\mathsf{Identify}_{d_{f_1},d_{f}}(uv_i)))$ is on the first edge of $d(\mathsf{Identify}_{d_{f_1},d_{f}}(uv_i)),z')$ or vice versa. Therefore, we get that $\ell(d_{f_1}(uv_i),\mathsf{V^*Dir}_{f_1}(uv_i))$ is on\\ $\ell(d_f(\mathsf{Identify}_{d_{f_1},d_{f}}(uv_i)),$  $\mathsf{V^*Dir}_{f}(\mathsf{Identify}_{d_{f_1},d_{f}}(uv_i)))$ or vice versa. Similarly, we can show that this condition holds also for every $uv_i\in V(d_{f_2},\mathsf{Common}(f,f_2))\cap V^*$, where $\mathsf{Common}(f,$ $f_2)=\gi(f)\cap \gi(f_2(c))$ such that\\ $d_{f_2}(uv_i)\in \gis(\fin)\setminus \mathsf{GridPointSet}(f_{\mathsf{init}})$. Thus, we get that $C$ exhibits {\bf Equality of $V^*$ Directions}.
\end{proof}

We continue with the proof that $C$ is an info-cutter of $F$. Now, we prove that $C$ exhibits {\bf Equality of Common Parts} with respect to $F$. 
 Towards this, we prove that $C$ exhibits {\bf Partition of $E_f$ With Both Endpoints on $f$} with respect to $F$.

% $C$ exhibits {\em \bf Partition of $\{u,uv_1\}\in E_f$ With One Endpoint on $f$} property with respect to $F$.
%$C$ exhibits {\em \bf Partition of $\{uv_{\mathsf{index}(u,v)},v\}\in E_f$ With One Endpoint on $f$} property

\begin{lemma}\label{lem:CisInfoCut4}
	Let $F=(f,d_f,E_f,U_f,\mathsf{V^*Dir}_f)$ be an info-frame, let $d$ be a drawing of $F$ and let $c$ be a cutter of $f$. Let $C=\mathsf{Splitter}(F,d,c)_C$. Then, $C$ exhibits the {\bf Partition of $E_f$ With Both Endpoints on $f$} property with respect to $F$. 
\end{lemma}

\begin{proof}
We prove that $C$ exhibits {\bf Partition of $E_f$ With Both Endpoints on $f$} property with respect to $F$ by showing that the conditions of Definition \ref{def:infCTurPoints2a} hold.
Let $\{uv_i,uv_{i+1}\}\in E_f$ such that $uv_i,uv_{i+1}\in V(d_f)$. Since $d$ is a drawing of $F$, $\{uv_i,uv_{i+1}\}$ is an edge that drawn strictly inside $d$ except at its endpoints. Then, exactly one of the following conditions is satisfied:
\begin{enumerate}
	\item The edge $\{uv_i,uv_{i+1}\}$ is drawn strictly inside $f_1(c)$ except at its endpoints (e.g., see $\{uv_3,uv_4\}$ in Figure~\ref{fig:splitter1}). Then, $\{\mathsf{Identify}_{d_f,d_{f_1}}(uv_i),$ $\mathsf{Identify}_{d_f,d_{f_1}}(uv_{i+1})\}\in E_1$, so Condition  \ref{infcutedgeCon712} is satisfied. Observe that $(d_{f_1}(uv_i),\{u,v\})$ and $(d_{f_1}(uv_{i+1}),\{u,v\})$ are turning points in $f_1(c)$. Thus, there exists $uv_j\in V(d_{f_1})\cap V^*$ such that $d_f(uv_i)=d_{f_1}(uv_j)$, so $\mathsf{Identify}_{d_f,d_{f_1}}(uv_i)\neq \mathsf{Null}$. Similarly, we have that $\mathsf{Identify}_{d_f,d_{f_1}}(uv_{i+1})\neq  \mathsf{Null}$.
	\item The edge $\{uv_i,uv_{i+1}\}$ is drawn strictly inside $f_2(c)$ except at its endpoints (e.g., see $\{a,b\}$ in Figure~\ref{fig:splitter1}). Then, $\{\mathsf{Identify}_{d_f,d_{f_2}(uv_i)},$ $\mathsf{Identify}_{d_f,d_{f_2}}(uv_{i+1})\}\in E_2$, so Condition  \ref{infcutedgeCon713} is satisfied. Similarly to the previous case, we have that $\mathsf{Identify}_{d_f,d_{f_2}}(uv_i)\neq \mathsf{Null}$ and $\mathsf{Identify}_{d_f,d_{f_2}}(uv_{i+1})\neq  \mathsf{Null}$.
	\item The edge $\{uv_i,uv_{i+1}\}$ is drawn on $c$, so $\{\mathsf{Identify}_{f,d_1^*}(uv_i),\mathsf{Identify}_{f,d_1^*}(uv_{i+1})\}\in E(d_1^*,$ $\pp(c))$, and hence Condition \ref{infcutedgeCon714} is satisfied.
	\item Otherwise, we show that $\{uv_i,uv_{i+1}\}$ partly intersects $c$ (see Definition \ref{def:partIntCBoth}) (e.g., see $\{uv_1,uv_2\}$ in Figure~\ref{fig:splitter1}). Let $r\in [2]$
	 such that $uv_i$ is drawn on $f_r(c)$ in $d$. Observe that there must be $\ell \in \mathbb{N}$ points $p_i\in \gi(c)\setminus \gi(f)$ such that $p_i$ is on $d(\{uv_i,uv_{i+1}\})$, and $(p_i,\{u,v\})$ is a turning point of $f_r(c)$ in $d$. Now, in $d_r$ (returned by $\mathsf{SplitterPart}_1$), on each of these points are drawn vertices from $V^*_{\{u,v\}}$, enumerated in an increasing order from $uv_i$ to $uv_{i+1}$. Let $\mathsf{Identify}_{d_f,d_{f_r}}(uv_i)=uv_j$. Now, for every $0\leq t<\ell$, since $(d_r(uv_{j+t}),\{u,v\})$ is a turning point in $f_r(c)$, exactly one of the following conditions is satisfied. \label{lem:inCuCon4}
	\begin{itemize}
		\item The part of the edge $\{uv_{i},uv_{i+1}\}$ from the point $d_r(uv_{j+t})$ to the point $d_r(uv_{j+t+1})$ is drawn strictly inside $f_1(c)$ except at the points $d_r(uv_{j+t})$ and $d_r(uv_{j+t+1})$.
		Now, if $r=1$, then $\{uv_{j+t},uv_{j+t+1}\}\in E_1$; else,  $r=2$, and then $\{\mathsf{Identify}_{f_{d_2},f_{d_1}}(uv_{j+t}),$ $\mathsf{Identify}_{f_{d_2},f_{d_1}}$ $(uv_{j+t+1})\}\in E_1$. So, Condition \ref{infcutedgeCon7151} is satisfied.
		
			\item Similarly, the part of the edge $\{uv_{i},uv_{i+1}\}$ from the point $d_r(uv_{j+t})$ to the point $d_r(uv_{j+t+1})$ is drawn strictly inside $f_2(c)$ except at the points $d_r(uv_{j+t})$ and\\ $d_r(uv_{j+t+1})$.
	If $r=2$, then $\{uv_{j+t},uv_{j+t+1}\}\in E_2$; else, $r=1$, and then $\{\mathsf{Identify}_{f_{d_1},f_{d_2}}$ $(uv_{j+t}),$ $\mathsf{Identify}_{f_{d_1},f_{d_2}}$ $(uv_{j+t+1})\}\in E_2$. So, Condition \ref{infcutedgeCon7152} is satisfied. 
		
		\item The part of the edge $\{uv_i,uv_{i+1}\}$ from the point $d_r(uv_{j+t})$ to the point $d_r(uv_{j+t+1})$ is drawn on $c$. In this case, it follows that $\{uv_{j+t},uv_{j+t+1}\}\in E(d_r^*,\pp(c))$. So, Condition \ref{infcutedgeCon7153} is satisfied.
	\end{itemize}
 
 Similarly, it can be shown that exactly one of Conditions \ref{infcutedgeCon7155}-\ref{infcutedgeCon7157} is satisfied.
Then, the conditions of Definition \ref{def:partIntCBoth} hold, so $\{uv_i,uv_{i+1}\}$ partly intersects $c$.
\end{enumerate}
In conclusion, we proved that all conditions of Definition \ref{def:infCTurPoints2a} hold, so $C$ exhibits the {\bf Partition of $E_f$ With Both Endpoints on $f$} property with respect to $F$. 
\end{proof}

Similarly to Lemma \ref{lem:CisInfoCut4}, it can be proved that $C$ exhibits the {\bf Partition of $\{u,uv_1\}\in E_f$ With One Endpoint on $f$} and {\bf Partition of $\{uv_{\mathsf{index}(u,v)},v\}\in E_f$ With One Endpoint on $f$} properties with respect to $F$. We state this in the next lemma (without providing a proof):

\begin{lemma}\label{lem:CisInfoCut5}
	Let $F=(f,d_f,E_f,U_f,\mathsf{V^*Dir}_f)$ be an info-frame, let $d$ be a drawing of $F$ and let $c$ be a cutter of $f$. Let $C=\mathsf{Splitter}(F,d,c)_C$. Then, $C$ exhibits the {\bf Partition of $\{u,uv_1\}\in E_f$ With One Endpoint on $f$} and {\bf Partition of $\{uv_{\mathsf{index}(u,v)},v\}\in E_f$ With One Endpoint on $f$} properties with respect to $F$.
\end{lemma}

Now, we prove that $C$ exhibits the {\bf Partition of $\{u,v\}\in E$ Drawn Strictly Inside $f$}  property with respect to $F$.

\begin{lemma}\label{lem:CisInfoCut6}
	Let $F=(f,d_f,E_f,U_f,\mathsf{V^*Dir}_f)$ be an info-frame, let $d$ be a drawing of $F$ and let $c$ be a cutter of $f$. Let $C=\mathsf{Splitter}(F,d,c)_C$. Then, $C$ exhibits {\bf Partition of $\{u,v\}\in E$ Drawn Strictly Inside $f$} property with respect to $F$.
\end{lemma}

\begin{proof}
	We show that the conditions of Definition \ref{def:infCTurPoints3} holds. Let $u,v\in U_f$ such that $\{u,v\}\in E$ and $V(d_f)\cap V^*_{\{u,v\}}=\emptyset$. Since $d$ is a drawing of $F$, it follows that $\{u,v\}$ is drawn strictly inside $f$ in $d$. Then, exactly one of the following conditions holds: 
	\begin{enumerate}
		\item $\{u,v\}$ is drawn strictly inside $f_1(c)$. Therefore, $V(f_{d_1},\gi(c))\cap V^*_{\{u,v\}}=\emptyset$ and $u,v\in U_1$; so, Condition \ref{infCutcon101} is satisfied.
			\item $\{u,v\}$ is drawn strictly inside $f_2(c)$. Therefore, $V(f_{d_2},\gi(c))\cap V^*_{\{u,v\}}=\emptyset$ and $u,v\in U_2$; so, Condition \ref{infCutcon102} is satisfied.
				\item $\{u,v\}$ is drawn on $c$. Therefore $V(f_{d_1},\gi(c))\cap V^*_{\{u,v\}}=\emptyset$ and $\{u,v\}\in E(f_{d_1},\pp(c))$; so, Condition \ref{infCutcon103} is satisfied.
				\item Otherwise, we aim to prove that $\{u,v\}$ partly intersects $c$ (see Definition \ref{def:partIntUV}). There are $\ell \in \mathbb{N}$ points $p_i\in \gi(c)\setminus \gi(f)$, such that $(p_i,\{u,v\})$ is a turning point in $f_1(c)$. It can be shown, similarly to Condition \ref{lem:inCuCon4} in the proof of Lemma \ref{lem:CisInfoCut4}, that for every $1\leq t<\ell$, exactly one of Conditions \ref{infcutedgeCon51}-\ref{infcutedgeCon53}, exactly one of Conditions \ref{infcutedgeCon1041}-\ref{infcutedgeCon1045} and exactly one of Conditions \ref{infcutedgeCon1051}-\ref{infcutedgeCon1055} are satisfied. So, Condition \ref{infCutcon104} is satisfied.
	\end{enumerate}
Therefore, $C$ exhibits the {\bf Partition of $\{u,v\}\in E$ Drawn Strictly Inside $f$} property with respect to $F$.
\end{proof}

In conclusion, in Lemmas \ref{lem:CisInfoCut1}-\ref{lem:CisInfoCut6} and in Observation \ref{lem:CisInfoCut3} we saw that $\mathsf{Splitter}(F,d,c)_C$ is an {\bf Info-Cutter Template} and {\bf $U_f$-Partitioned} with respect to $F$. In addition, $\mathsf{Splitter}(F,d,c)_C$ exhibits the {\bf Equality of $V^*$ Directions}, {\bf Equality of Common Parts}, {\bf Partition of $E_f$ With Both Endpoints on $f$}, {\bf Partition of $\{u,uv_1\}\in E_f$ With One Endpoint on $f$}, {\bf Partition of $\{uv_{\mathsf{index}(u,v)},v\}\in E_f$ With One Endpoint on $f$} and {\bf Partition of $\{u,v\}\in E$ Drawn Strictly Inside $f$} properties with respect to $F$. Therefore, by Definition \ref{def:infCt}, $\mathsf{Splitter}(F,d,c)_C$ is info-cutter of $F$. Moreover, in Lemma \ref{lem:dIsDraOfF1}, we saw that $\mathsf{Splitter}(F,d,c)_{d_1}$ and $\mathsf{Splitter}(F,d,c)_{d_2}$ are drawings of $F_1$ and $F_2$, respectively, where $\mathsf{Splitter}(F,d,c)_C=(c,F_1,F_2)$.
So, we have the following lemma:

\begin{lemma}\label{splitter}
	Let $F=(f,d_f,E_f,U_f,\mathsf{V^*Dir}_f)$ be an info-frame, let $d$ be a drawing of $F$ and let $c$ be a cutter of $f$. Then, $\mathsf{Splitter}(F,d,c)_C$ is info-cutter of $F$. In addition, $\mathsf{Splitter}(F,d,c)_{d_1}$ and $\mathsf{Splitter}(F,d,c)_{d_2}$ are drawings of $F_1$ and $F_2$, respectively, where $\mathsf{Splitter}(F,d,c)_C=(c,F_1,F_2)$.
\end{lemma}

Now, for a later use (in Section \ref{sec:proofScheme}), we present another observation. Let $f'$ be a frame that is bounded by $f_1(c)$. Observe that the part of the drawing $d$ intersected by $f'$ is equal to the part of the drawing $d_1$ intersected by $f'$, up to renaming. Thus, we have the following claims. The subset of vertices from $V$ that intersect $f'$ in $d$ equals the subset of vertices from the set $V$ that intersect $f'$ in $d_1$. Moreover, every turning point $(p,\{u,v\})$ in $f'$ in $d$ is a turning point in $d_1$, and vice versa. Therefore, we get that  $\sizef(f',d)=\sizef(f',d_1)$, as we state in the next observation:

\begin{observation}\label{obs:turningEq2}
	Let $F=(f,d_f,E_f,U_f,\mathsf{V^*Dir}_f)$ be an info-frame, let $C=(c,F_1,F_2)$ be an info-cutter of $F$, let $d$ be a drawing of $F$, and let $f'$ be a frame that is bounded by $f_1(c)$. Then, $\sizef(f',d)=\sizef(f',d_1)$, where $d_1=\mathsf{Splitter}(F,d,C)_{d_1}$.
\end{observation}

%!TEX root =Main-Movement.tex

\subsection{The Function $\mathsf{Glue}$} 

Now, we introduce the ``inverse function'' of the $\mathsf{Splitter}$ function from Section \ref{sec:splitter}, called the $\mathsf{Glue}$ function. This function gets as input an info-frame $F$, an info-cutter $C=(c,F_1,F_2)$ of $F$, a drawing $d_1$ of $F_1$ and a drawing $d_2$ of $F_2$. The purpose of this function is to ``glue'' the drawings $d_1$ and $d_2$, and return as output a drawing of $F$ (see Figure~\ref{fig:glue}). The way we glue the two drawings is very intuitive, due to the definition of an info-cutter explained in Section \ref{sec:infoCut}.

\begin{figure}[!t]
	\centering
	\begin{subfigure}{0.48\textwidth}
		\includegraphics[width = \textwidth, page = 63]{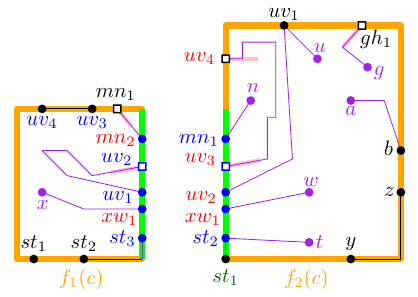}
		\subcaption{}
		\label{fig:glue1}
	\end{subfigure}
	\hfil
	\begin{subfigure}{0.48\textwidth}
		\includegraphics[width = \textwidth, page = 64]{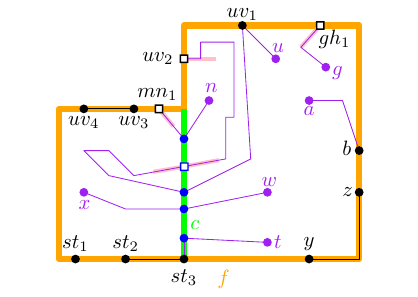}
		\subcaption{}
		\label{fig:glue2}
	\end{subfigure}
	\hfil
	\begin{subfigure}{0.48\textwidth}
		\includegraphics[width = \textwidth, page = 65]{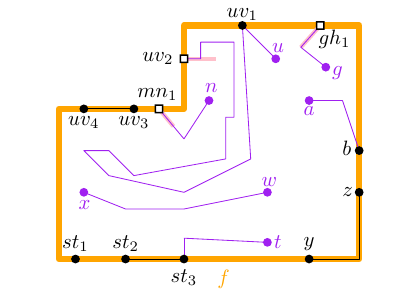}
		\subcaption{}
		\label{fig:glue3}
	\end{subfigure}
	\caption{An illustration of the $\mathsf{Glue}$ function. The vertices mapped to points in $\gis(f) \setminus \mathsf{GridPointSet}(f)$, $\gis(f_1) \setminus \mathsf{GridPointSet}(f_1)$ and $\gis(f_2) \setminus \mathsf{GridPointSet}(f_2)$ are denoted by hollow squares. (a) Drawings of the info-frames of $f_1(c)$ and $f_2(c)$, described in Figure~\ref{fig:splitter5}. A cutter $c$ of $f$ is colored green. (b) The figure showing the intuition behind the Glue function. The drawing shown is obtained by merging the common vertices of $f_1(c)$ and $f_2(c)$ together. (c) The drawing $d$ of $F$ obtained by $\mathsf{Glue}$ from the figures in (a).}
	\label{fig:glue}
\end{figure}

Let $F=(f,d_f,E_f,U_f,\mathsf{V^*Dir}_f)$ be an info-frame, let $C=(c,F_1=(f_1(c),d_{f_1},E_{f_1},U_{f_1},$ $\mathsf{V^*Dir}_{f_1}),$ $F_2=(f_2(c),d_{f_2},E_{f_2},U_{f_2},\mathsf{V^*Dir}_{f_2}))$ be an info-cutter of $F$, and let $d_1$ and $d_2$ be drawings of $F_1$ and $F_2$, respectively. $\mathsf{Glue}(F,C,d_1,d_2)$ will be a pair of functions $d=(d_{\mathsf{V}},d_{\mathsf{E}})$ (defined in Definition \ref{def:glue}). First, we define $d_{\mathsf{V}}$ by $\mathsf{GlueVer}(F,C,d_1,d_2)$ as follows:

\begin{definition}  [{\bf $\mathsf{GlueVer}$}] \label{def:glueVer}
	Let $F=(f,d_f,E_f,U_f,\mathsf{V^*Dir}_f)$ be an info-frame, let $C=(c,F_1=(f_1(c),d_{f_1},E_{f_1},U_{f_1},\mathsf{V^*Dir}_{f_1})$, $F_2=(f_2(c),d_{f_2},E_{f_2},U_{f_2},\mathsf{V^*Dir}_{f_2}))$ be an info-cutter of $F$, and let $d_1$ and $d_2$ be drawings of $F_1$ and $F_2$, respectively. Then, $\mathsf{GlueVer}(F,C,d_1,d_2)=d_{\mathsf{V}}$ is the function defined as follows:
	\begin{enumerate}
		\item For every $v\in V(d_f)$, $d_{\mathsf{V}}(v)=d_f(v)$. \label{def:glueVer1}
		\item For every $v\in U_f$, if $v\in V(d_1)$, then $d_{\mathsf{V}}(v)=d_1(v)$; else, $d_{\mathsf{V}}(v)=d_2(v)$. \label{def:glueVer2}
		\end{enumerate}
\end{definition}

 Next, we define the function $d_{\mathsf{E}}$ by $\mathsf{GlueEdg}(F,C,d_1,d_2)$. Here, the definitions for the drawings of the edges correlate with
 the properties {\bf Partition of $E_f$ With Both Endpoints on $f$}, {\bf Partition of $\{u,uv_1\}\in E_f$ With One Endpoint on $f$}, 
{\bf Partition of $\{uv_{\mathsf{index}(u,v)},v\}\in E_f$ With One Endpoint on $f$} and {\bf Partition of $\{u,v\}\in E$ Drawn Strictly Inside $f$} (see Definitions \ref{def:infCTurPoints2a}, \ref{def:infCTurPoints2b}, \ref{def:infCTurPoints2c} and \ref{def:infCTurPoints3}, respectively). So, for the sake of readability, we split the definition for the drawings of the edges into separate definitions, according to these properties. 

%in Definition \ref{def:glue}. 

%For the sake of readability, we define $d_{\mathsf{E}}$ separately in few mo .

\begin{definition}  [{\bf $\mathsf{GlueEdg}$}] \label{def:glueEdg}
Let $F=(f,d_f,E_f,U_f,\mathsf{V^*Dir}_f)$ be an info-frame, let $C=(c,F_1=(f_1(c),d_{f_1},E_{f_1},U_{f_1},\mathsf{V^*Dir}_{f_1})$ $F_2=(f_2(c),d_{f_2},E_{f_2},U_{f_2},\mathsf{V^*Dir}_{f_2}))$ be an info-cutter of $F$, and let $d_1$ and $d_2$ be drawings of $F_1$ and $F_2$, respectively. Then, $\mathsf{GlueEdg}(F,C,d_1,d_2)=d_{\mathsf{E}}$ is the function defined as follows:
\begin{enumerate}
	\item For every $e\in E(d_f)$, $d_{\mathsf{E}}(e)=d_f(e)$ (e.g., see the edges $\{uv_3, uv_4\}$, $\{st_2,st_3\}$ and $\{y,z\}$ in Figure~\ref{fig:glue3}). \label{def:gleEdgCon1}
 		\item For every $\{uv_i,uv_{i+1}\}\in E_f$, we define $d_{\mathsf{E}}(\{uv_i,uv_{i+1}\})$ as follows.
 		\begin{enumerate}
 		\item If both endpoints of $\{uv_i,uv_{i+1}\}$ are drawn on $f$, that is, $uv_i,uv_{i+1}\in V(d_f)$, then $d_{\mathsf{E}}(\{uv_i,uv_{i+1}\})$ is  {\bf$\mathsf{Glue}$ of $\{uv_i,uv_{i+1}\}$ With Both Endpoints on $f$} (Definition \ref{def:glueEf2}) (e.g., see the edge $\{uv_1, uv_2\}$ in Figure~\ref{fig:glue3}). 
 		\item If exactly one endpoint of $\{uv_i,uv_{i+1}\}$ is drawn on $f$ and $\{uv_i,uv_{i+1}\}=\{u,uv_{1}\}$, then $d_{\mathsf{E}}(\{uv_i,uv_{i+1}\})$ is {\bf$\mathsf{Glue}$ of $\{u,uv_{1}\}$ With One Endpoint on $f$} (e.g., see the edges $\{u, uv_1\}$, $\{a,b\}$ and $\{g,gh_1\}$ in Figure~\ref{fig:glue3}). \label{def:glue4b}
 		\item If exactly one endpoint of $\{uv_i,uv_{i+1}\}$ is drawn on $f$ and $\{uv_i,uv_{i+1}\}=\{uv_{\mathsf{index}(u,v)},$ $v\}$, then $d_{\mathsf{E}}(\{uv_i,uv_{i+1}\})$ is {\bf$\mathsf{Glue}$ of $\{uv_{\mathsf{index}(u,v)},v\}$ With One Endpoint on $f$} (e.g., see the edges $\{mn_1, n\}$ and $\{st_3, t\}$ in Figure~\ref{fig:glue3}). \label{def:glue4c}
 		\end{enumerate}	
 \item For every $u,v\in U_f$ such that $\{u,v\}\in E$ and $V(d_f)\cap V^*_{\{u,v\}}=\emptyset$, $d_{\mathsf{E}}(\{u,v\})$ is {\bf$\mathsf{Glue}$ of $\{u,v\}$ Drawn Strictly Inside $f$} (e.g., see the edge $\{x, w\}$ in Figure~\ref{fig:glue3}). \label{def:gluecon5}
 
	\end{enumerate}
	
\end{definition}

%We show how we define $d_{\mathsf{E}}(\{uv_i,uv_{i+1}\})$ according to the first case. The other cases are defined similarly.

%\bf Partition of $\{u,uv_1\}\in E_f$ With One Endpoint on $f$} property (Definition \ref{def:infCTurPoints2b}).
%\item Exactly one endpoint of $\{uv_i,uv_{i+1}\}$ is drawn on $f$, and $\{uv_i,uv_{i+1}\}=\{uv_{\mathsf{index}(u,v)},v\}$. This case correlates with {\bf Partition of $\{uv_{\mathsf{index}(u,v)},v\}\in E_f$ With One Endpoint on $f$} property (Definition \ref{def:infCTurPoints2c}).
Next, in Definition \ref{def:glueEf2}, we present the definition of {\bf $\mathsf{Glue}$ of $\{uv_i,uv_{i+1}\}$ With Both Endpoints on $f$}. In this definition, we follow the definition of the {\bf Partition of $E_f$ With Both Endpoints on $f$} property (see Definition \ref{def:infCTurPoints2a}) in order to define a drawing of $\{uv_i,uv_{i+1}\}$. Recall that this property describes four possible conditions regarding $\{uv_i,uv_{i+1}\}$, where exactly one of them is satisfied. We define $d_{\mathsf{E}}(\{uv_i,uv_{i+1}\})$ according to each condition of Definition \ref{def:infCTurPoints2a}:
 
 \smallskip\noindent{\bf Condition \ref{infcutedgeCon712}.}
 In this case, $\{\mathsf{Identify}_{d_f,d_{f_1}}(uv_i),\mathsf{Identify}_{d_f,d_{f_1}}(uv_{i+1})\}\in E_{f_1}$. This means that the edge $\{uv_i,uv_{i+1}\}$ (up to renaming its endpoints) is drawn strictly inside $f_1(c)$ except for both of its endpoints, which are drawn on $f$. Based on this drawing, we obtain a drawing of $\{uv_i,uv_{i+1}\}$ strictly inside $f_1(c)$, and therefore strictly inside $f$, except for both of its endpoints, which are drawn on $f$.
  
   \smallskip\noindent{\bf Condition \ref{infcutedgeCon713}.} In this case, $\{uv_i,uv_{i+1}\}$ (up to renaming its endpoints) is drawn strictly inside $f_2(c)$, and this is symmetric to Condition \ref{infcutedgeCon712}.
   
   \smallskip\noindent{\bf Condition \ref{infcutedgeCon714}.} In this case, $\{uv_i,uv_{i+1}\}$ is drawn on $c$, so we can obtain its drawing from both $d_1$ and $d_2$ (here, we arbitrary choose $d_1$). 
   
  \smallskip\noindent{\bf Condition \ref{infcutedgeCon715}.} In this case, some parts of $\{uv_i,uv_{i+1}\}$ can be contained in $f_1(c)$ and some parts of $\{uv_i,uv_{i+1}\}$ can be contained in $f_2(c)$. Here, we define the drawing of each part according to each condition of Definition \ref{def:partIntCBoth}. In particular, for every such part, we have three cases, similar to Conditions \ref{infcutedgeCon712}--\ref{infcutedgeCon714}:
  \begin{itemize}
  	\item The part is drawn strictly inside $f_1(c)$ except for both of its endpoints (Condition \ref{infcutedgeCon7151} of Definition \ref{def:partIntCBoth}).
  	\item The part is drawn strictly inside $f_2(c)$ except for both of its endpoints (Condition \ref{infcutedgeCon7152} of Definition \ref{def:partIntCBoth}).
  	\item  The part is drawn on $c$ (Condition \ref{infcutedgeCon7153} Definition \ref{def:partIntCBoth}).
  	\end{itemize} 
   We obtain a drawing for each part similarly to Conditions \ref{infcutedgeCon712}--\ref{infcutedgeCon714}  in Definition \ref{def:glueEf2}, and then connect these parts together, from $uv_i$ to $uv_{i+1}$. Conditions \ref{def:edggluconD}--\ref{def:edggluconF} in Definition \ref{def:glueEf2} consider the last part, which is connected to $uv_{i+1}$. These conditions are similar to Conditions \ref{def:edggluconA}--\ref{def:edggluconC} and are needed for technical reasons (see the discussion before Definition \ref{def:partIntCBoth} in Section \ref{subsec:prpInoCu}). Remind, that we denote by ${\cal P}^*(\fin)$ the set of all almost straight line paths \footnote{Remind, that an almost straight-line path in $\fin$, is plane path, where i) the endpoints of the path are mapped to points in $\gis(\fin)$, ii) the internal vertices are mapped to grid points strictly inside $\fin$, and iii) every edge is mapped to the line segment $s$ connecting the images of their endpoints, and there exist $a,b\in \gps(\fin)$ such that $s$ is on $\ell(a,b)$. } in $\fin$.  Lastly, we delete any point from $\gis(\fin)\setminus \gps(\fin)$, accept for the first and the last points, in the path we obtained, in order to get a path from the set ${\cal P}^*(\fin)$. Observe that, every point $p_i$ we deleted is on $c$, and since $C$ is an info-cutter of $F$, $C$ exhibits {\bf Equality of $V^*$ Directions} with respect to $F$. Therefore, Condition \ref{con2DefDir} of Definition \ref{def:infCTurPoints2} is satisfied, so $p_{i-1},p_i$ and $p_{i+1}$ are on a straight line, and hence the path did not change. That is, $\pp(P)=\pp(P')$, where $P$ and $P'$ are the paths we obtained before and after deleting the points from $\gis(\fin)\setminus \gps(\fin)$, respectively.  E.g, consider the edge $\{uv_1,uv_2\}$ in  Figure~\ref{fig:glue3}. The drawing of $\{uv_1,uv_2\}$ obtained from $d_1$ and $d_2$ in Figure~\ref{fig:glue1} by gluing: (i) the drawing of $\{uv_1,uv_2\}$ in $d_2$ (ii) the drawing of $\{uv_1,uv_2\}$ in $d_1$ (iii) and the drawing of $\{uv_3,uv_4\}$ in $d_2$ (see Figure~\ref{fig:glue2}). Then, we deleted the points from $\gis(\fin)\setminus \gps(\fin)$, accept for the first and the last points, in the path we obtained (see Figure~\ref{fig:glue3}). We remind that for two sequences of points $P=(p_1,\ldots,p_\ell)$ and $Q=(q_1,\ldots,q_k)$, where $p_\ell=q_1$, we denote by $P_1\cdot P_2$ the sequence of points $(p_1,\ldots,p_\ell,q_2,\ldots,q_k)$.

%Let $P$ be an axis-parallel path (cycle). We denote by $P^{\mathsf{min}}$ the axis-parallel path (cycle) that is equal to $P$, and has the minimum number of vertices among the axis-parallel paths (cycles) that are equal to $P$. Observe that $P^{\mathsf{min}}$ is unique.

%\ref{def:infCTurPoints2a}

\begin{definition} [{\bf $\mathsf{Glue}$ of $\{uv_i,uv_{i+1}\}$ With Both Endpoints on $f$}] \label{def:glueEf2}
	Let $F=(f,d_f,E_f,$ $U_f,\mathsf{V^*Dir}_f)$ be an info-frame, let $C=(c,F_1=(f_1(c),d_{f_1},E_{f_1},U_{f_1},\mathsf{V^*Dir}_{f_1})$ $F_2=(f_2(c),d_{f_2},$ $E_{f_2},$ $U_{f_2},\mathsf{V^*Dir}_{f_2}))$ be an info-cutter of $F$, and let $d_1$ and $d_2$ be drawings of $F_1$ and $F_2$, respectively. Let $\{uv_i,uv_{i+1}\}\in E_f$ such that $uv_i,uv_{i+1}\in V(d_f)$. Then, {\em \bf$\mathsf{Glue}$ of $\{uv_i,uv_{i+1}\}$}, denoted by $d_{\mathsf{E}}(\{uv_i,uv_{i+1}\})$, is defined according the following cases.
	
	\begin{enumerate}
		\item If $\{\mathsf{Identify}_{d_f,d_{f_1}}(uv_i),\mathsf{Identify}_{d_f,d_{f_1}}(uv_{i+1})\}\in E_{f_1}$, then\\ $d_{\mathsf{E}}(\{uv_i,uv_{i+1}\})=$ $d_1(\{\mathsf{Identify}_{d_f,d_{f_1}}(uv_i),\mathsf{Identify}_{d_f,d_{f_1}}(uv_{i+1})\})$. \label{con1defglueedge}
		\item If $\{\mathsf{Identify}_{d_f,d_{f_2}}(uv_i),\mathsf{Identify}_{d_f,d_{f_2}}(uv_{i+1})\}\in E_{f_2}$, then\\ $d_{\mathsf{E}}(\{uv_i,uv_{i+1}\})=d_2(\{\mathsf{Identify}_{d_f,d_{f_2}}$ $(uv_i),\mathsf{Identify}_{d_f,d_{f_2}}(uv_{i+1})\})$.
		\item If $\{\mathsf{Identify}_{d_f,d_{f_1}^*}(uv_i),\mathsf{Identify}_{d_f,d_{f_1}^*}(uv_{i+1})\}\in E(d_{f_1}^*,\pp(c))$, then\\ $d_{\mathsf{E}}(\{uv_i,uv_{i+1}\})$ $=d_{f_1}^*(\{\mathsf{Identify}_{d_f,d_{f_1}^*})$.\label{con3defglueedge}
		\item Otherwise, let $\ell,j \in \mathbb{N}$ and $r\in[2]$ be defined as in Definition \ref{def:partIntCBoth}. We define $\mathsf{PartEdge}_t$, for every $0\leq t<\ell$, according to the following cases:\label{con4defglueedge}
		\begin{enumerate}
			\item If $r=1$ and $\{uv_{j+t},uv_{j+t+1}\}\in E_{f_1}$, then $\mathsf{PartEdge}_t=d_1(\{uv_{j+t},uv_{j+t+1}\})$; if $r=2$ and $\{\mathsf{Identify}_{d_{f_2},d_{f_1}}(uv_{j+t}),\mathsf{Identify}_{d_{f_2},d_{f_1}}(uv_{j+t+1})\}\in E_1$, then $\mathsf{PartEdge}_t=d_1(\{\mathsf{Identify}_{d_{f_2},d_{f_1}}(uv_{j+t}),\mathsf{Identify}_{d_{f_2},d_{f_1}}(uv_{j+t+1})\})$.\label{def:edggluconA} 
			\item If $r=2$ and $\{uv_{j+t},uv_{j+t+1}\}\in E_{f_2}$, then $\mathsf{PartEdge}_t=d_2(\{uv_{j+t},uv_{j+t+1}\})$; if $r=1$ and $\{\mathsf{Identify}_{d_{f_1},d_{f_2}}(uv_{j+t}),\mathsf{Identify}_{d_{f_1},d_{f_2}}(uv_{j+t+1})\}\in E_2$, then $\mathsf{PartEdge}_t=d_2(\{\mathsf{Identify}_{d_{f_1},d_{f_2}}(uv_{j+t}),\mathsf{Identify}_{d_{f_1},d_{f_2}}(uv_{j+t+1})\})$.\label{def:edggluconB} 
			\item If $\{uv_{j+t},uv_{j+t+1}\}\in E(d_{f_r}^*,\pp(c))$, then $\mathsf{PartEdge}_t=d_{f_r}^*(\{uv_{j+t},uv_{j+t+1}\})$.\label{def:edggluconC} 
	
%	\indent
	
	\indent
		In addition, we define $\mathsf{PartEdge}_\ell$ according to the following cases:

			\item If $r=1$ and $\mathsf{Identify}_{d_f,d_{f_1}}(uv_{i+1})=uv_{j+\ell+1}$, then  $\mathsf{PartEdge}_\ell=d_1(\{uv_{j+\ell},uv_{j+\ell+1}\})$; else, $r=2$, and then $\mathsf{PartEdge}_\ell=d_1(\{\mathsf{Identify}_{d_{f_2},d_{f_1}}(uv_{j+\ell}),\mathsf{Identify}_{d_f,d_{f_1}}(uv_{i+1})\})$.\label{def:edggluconD} 
			
			\item If $r=2$ and $\mathsf{Identify}_{d_f,d_{f_2}}(uv_{i+1})=uv_{j+\ell+1}$, then $\mathsf{PartEdge}_\ell=d_2(\{uv_{j+\ell},uv_{j+\ell+1}\})$; else, $r=1$, and then $\mathsf{PartEdge}_\ell=d_2(\{\mathsf{Identify}_{d_{f_1},d_{f_2}}(uv_{j+\ell}),\mathsf{Identify}_{d_f,d_{f_2}}(uv_{i+1})\})$.\label{def:edggluconE}  
			
			\item Otherwise, $\mathsf{Identify}_{d_f,d_{f_r}}(uv_{i+1})=uv_{j+\ell+1}$, and then\\  $\mathsf{PartEdge}_\ell=d_{f_r}^*(\{uv_{j+\ell},uv_{j+\ell+1}\})$.\label{def:edggluconF} 
		\end{enumerate}
		
		Now, let $P=\mathsf{PartEdge}_0\cdot \mathsf{PartEdge}_1\cdot \ldots \cdot \mathsf{PartEdge}_\ell=(p_1,\ldots,p_q)$.
		Let $d_{\mathsf{E}}(\{uv_i,uv_{i+1}\})$ be the path obtained from $P$ by deleting every $p_i\in \gis(\fin)\setminus \gps$ $(\fin)$ where $i\neq 1,q$.
	\end{enumerate}	
	
\end{definition}

The other definitions, for the drawings of the edges in Conditions \ref{def:glue4b}, \ref{def:glue4c} and \ref{def:gluecon5} of Definition \ref{def:glueEdg}, are similar: each one of them naturally correlates with its parallel info-cutter property (see Section \ref{subsec:prpInoCu}). In particular, we refer to the discussions before Definitions \ref{def:infCTurPoints2b}, \ref{def:partIntCU} and \ref{def:infCTurPoints3}, where we consider cases regarding the drawing of an edge $\{uv_i,uv_{i+1}\}$, which fit to the cases corresponding to Conditions \ref{def:glue4b}, \ref{def:glue4c} and \ref{def:gluecon5}, respectively. So, we do not state these definitions explicitly.

Now, we are ready to define the drawing $\mathsf{Glue}$.

\begin{definition}  [{\bf $\mathsf{Glue}$}] \label{def:glue}
	Let $F=(f,d_f,E_f,U_f,\mathsf{V^*Dir}_f)$ be an info-frame, let $C=(c,F_1=(f_1(c),d_{f_1},E_{f_1},U_{f_1},\mathsf{V^*Dir}_{f_1})$, $F_2=(f_2(c),d_{f_2},E_{f_2},U_{f_2},\mathsf{V^*Dir}_{f_2}))$ be an info-cutter of $F$, and let $d_1$ and $d_2$ be drawings of $F_1$ and $F_2$, respectively. Then, $\mathsf{Glue}(F,C,d_1,d_2)=(\mathsf{GlueVer}(F,C,d_1,$ $d_2),\mathsf{GlueEdg}(F,C,d_1,d_2))$. 
\end{definition}

Recall that our goal is to define a drawing $d$ of $F$. But first, observe that we use $d_1$ and $d_2$ in order to define $d=(d_{\mathsf{V}},d_{\mathsf{E}})$ (see Definitions \ref{def:glueVer} and \ref{def:glueEdg}). Each time we use $d_1$ or $d_2$ for a certain vertex or edge, we assume that this vertex or edge belongs to the domain of $d_1$ or $d_2$, respectively. These assumptions are valid due to the info-cutter properties (see Section \ref{subsec:prpInoCu}), as we will immediately show. Specifically, we say that definitions are {\em well defined} if the aforementioned assumptions are valid, and we begin by proving that $d_{\mathsf{V}}$ is well defined: 

\begin{lemma}\label{lem:welldefV}
	Let $F=(f,d_f,E_f,U_f,\mathsf{V^*Dir}_f)$ be an info-frame, let $C=(c,F_1=(f_1(c),d_{f_1},$ $E_{f_1},$ $U_{f_1},\mathsf{V^*Dir}_{f_1})$ $F_2=(f_2(c),d_{f_2},E_{f_2},U_{f_2},\mathsf{V^*Dir}_{f_2}))$ be an info-cutter of $F$, and let $d_1$ and $d_2$ be drawings of $F_1$ and $F_2$, respectively. Then, $\mathsf{GlueVer}(F,C,d_1,d_2)=d_{\mathsf{V}}$ (see Definition \ref{def:glueVer}) is well defined.
\end{lemma}

\begin{proof}
In Condition \ref{def:glueVer1}, for every $u\in V(d_f)$, $d_{\mathsf{V}}(u)=d_f(u)$, so $u$ is clearly in the domain of $d_f$.

In Condition \ref{def:glueVer2}, for every $u\in U_f$, if $u\in V(d_1)$, then $d_{\mathsf{V}}(u)=d_1(u)$, so $d_{\mathsf{V}}(u)$ is well defined in this case.

 If $u\notin V(d_1)$, then $d_{\mathsf{V}}(u)=d_2(u)$. We show that $d_{\mathsf{V}}(u)$ is well defined in this case as well.
 Observe that since $C$ is an info-cutter of $F$, then $C$ is {\bf $U_f$-Partitioned} with respect to $F$. Therefore, by Condition \ref{def:infCutCon6} of Definition \ref{def:infCTurPoints13}, $U_f=V(d_1,\gi(c)\setminus \gi(f)) \cup U_1\cup U_2$. Now, since $d_1$ is a drawing of $F_1$, then, by Condition \ref{infFramcon3} of Definition \ref{def:infFrDr}, $U_1$ is the set of vertices drawn strictly inside $f_1(c)$. So, $u\notin V(d_1)$ implies $u\notin U_1$. In addition, $u\notin V(d_1)$ implies $u\notin V(d_1,\gi(c)\setminus \gi(f))$, and then, $u\in U_2$. Now, since $d_2$ is a drawing of $F_2$, again by Condition \ref{infFramcon3} of Definition \ref{def:infFrDr}, $U_2$ is the set of vertices drawn strictly inside $f_1(c)$ in $d_2$, so $u\in V(d_2)$.  

Therefore, $d_{\mathsf{V}}$ is well defined.
\end{proof}

Now, we aim to prove that $d_{\mathsf{E}}$ is well defined. In particular, we show that the drawings of the edges in Definition \ref{def:glueEf2} are well defined; the proof for the other edges (corresponding to Conditions \ref{def:glue4b}, \ref{def:glue4c} and \ref{def:gluecon5} in Definition \ref{def:glueEdg}) is similar. We also prove (in the next lemma) three additional properties that will come in handy when we prove that $d$ is a drawing of $F$.

\begin{lemma}\label{lem:welldef1}
		Let $F=(f,d_f,E_f,U_f,\mathsf{V^*Dir}_f)$ be an info-frame, let $C=(c,F_1=(f_1(c),d_{f_1},$ $E_{f_1},$ $U_{f_1},\mathsf{V^*Dir}_{f_1})$, $F_2=(f_2(c),d_{f_2},E_{f_2},U_{f_2},\mathsf{V^*Dir}_{f_2}))$ be an info-cutter of $F$, and let $d_1$ and $d_2$ drawings of $F_1$ and $F_2$, respectively. Let $\{uv_i,uv_{i+1}\}\in E_f$ such that $uv_i,uv_{i+1}\in V(d_f)$. Then, the following conditions are satisfied:
		\begin{enumerate}
			\item {\bf$\mathsf{Glue}$ of $\{uv_i,uv_{i+1}\}$} is well defined. \label{gluedecon1}
				\item $d_{\mathsf{E}}(uv_i,uv_{i+1})\in {\cal P}^*(\fin)$. \label{gluedecon4}
			\item $d_{\mathsf{E}}(\{uv_i,uv_{i+1}\})$ is drawn strictly inside $f$ except for its endpoints. \label{gluedecon3}
			\item The endpoints of $d_{\mathsf{E}}(\{uv_i,uv_{i+1}\})$ are $d_{\mathsf{V}}(uv_i)$ and $d_{\mathsf{V}}(uv_{i+1})$ defined in Definition \ref{def:glueVer}. \label{gluedecon2}
		
		\end{enumerate} 
\end{lemma}

%	\item If $\{\mathsf{Identify}_{d_f,d_{f_1}}(uv_i),\mathsf{Identify}_{d_f,d_{f_1}}(uv_{i+1})\}\in E_{f_1}$, then $d_{\mathsf{E}}(\{uv_i,uv_{i+1}\})=d_1(\{\mathsf{Identify}_{d_f,d_{f_1}}$ $(uv_i),\mathsf{Identify}_{d_f,d_{f_1}}(uv_{i+1})\})$.
%\item If $\{\mathsf{Identify}_{d_f,d_{f_2}}(uv_i),\mathsf{Identify}_{d_f,d_{f_2}}(uv_{i+1})\}\in E_{f_2}$, then $d_{\mathsf{E}}(\{uv_i,uv_{i+1}\})=d_2(\{\mathsf{Identify}_{d_f,d_{f_2}}$ $(uv_i),\mathsf{Identify}_{d_f,d_{f_2}}(uv_{i+1})\})$.
%\item If $\{\mathsf{Identify}_{d_f,d_{f_1}^*}(uv_i),\mathsf{Identify}_{d_f,d_{f_1}^*}(uv_{i+1})\}\in E(d_{f_1}^*,\pp(c))$, then $d_{\mathsf{E}}(\{uv_i,uv_{i+1}\})=d_{f_1}^*(\{\mathsf{Identify}_{d_f,d_{f_1}^*})$.
%\item Otherwise, let $\ell,j \in \mathbb{N}$ and $r\in[2]$ defined in Definition \ref{def:partIntCBoth}. We define $\mathsf{PartEdge}_t$, for every $0\leq t<\ell$ according to the following cases:

\begin{proof}
	We prove the satisfaction of the first three conditions. It easy to see, in each of the following cases, that Condition \ref{gluedecon2} is satisfied as well, due to the definition of $d_{\mathsf{V}}$ in Definition \ref{def:glueVer}. Now, since $C$ is an info-cutter of $F$, $C$ exhibits the {\bf Partition of $E_f$ With Both Endpoints on $f$} property with respect to $F$. Therefore, exactly one of the following conditions is satisfied.
		\begin{enumerate}
		\item $\{\mathsf{Identify}_{d_f,d_{f_1}}(uv_i),\mathsf{Identify}_{d_f,d_{f_1}}(uv_{i+1})\}\in E_{f_1}$. In this case, since $d_1$ is a drawing of $F_1$, $\{\mathsf{Identify}_{d_f,d_{f_1}}(uv_i),\mathsf{Identify}_{d_f,d_{f_1}}(uv_{i+1})\}\in E(d_1)$. Thus, $d_{\mathsf{E}}(\{uv_i,uv_{i+1}\})$ is well defined, and $d_{\mathsf{E}}(uv_i,uv_{i+1})\in {\cal P}^*(\fin)$, so Conditions \ref{gluedecon1} and \ref{gluedecon4} are satisfied. In addition, again since $d_1$ is a drawing of $F_1$ and $\{\mathsf{Identify}_{d_f,d_{f_1}}(uv_i),$ $\mathsf{Identify}_{d_f,d_{f_1}}(uv_{i+1})\}\in E_{f_1}$, we have that $\{\mathsf{Identify}_{d_f,d_{f_1}}(uv_i),\mathsf{Identify}_{d_f,d_{f_1}}$ $(uv_{i+1})\}$ is drawn strictly inside $f_1$ except for its endpoints, which are drawn on $\gi(f)\cap \gi(f_1(c))$ in $d_1$. So, $\{uv_i,uv_{i+1}\}$ is drawn strictly inside $f$ except for its endpoints, which are drawn on $f$. Thus, Condition \ref{gluedecon3} is satisfied. \label{proofGlu1}
			\item $\{\mathsf{Identify}_{d_f,d_{f_2}}(uv_i),\mathsf{Identify}_{d_f,d_{f_2}}(uv_{i+1})\}\in E_{f_2}$. In this case, since $d_2$ is a drawing of $F_2$, $\{\mathsf{Identify}_{d_f,d_{f_2}}(uv_i),\mathsf{Identify}_{d_f,d_{f_2}}(uv_{i+1})\}\in E(d_2)$. Thus, $d_{\mathsf{E}}(\{uv_i,uv_{i+1}\})$ is well defined, and $d_{\mathsf{E}}(uv_i,uv_{i+1})\in {\cal P}^*(\fin)$, so Conditions \ref{gluedecon1} and \ref{gluedecon4} are satisfied. In addition, again since $d_2$ is a drawing of $F_2$ and $\{\mathsf{Identify}_{d_f,d_{f_2}}(uv_i),$ $\mathsf{Identify}_{d_f,d_{f_2}}(uv_{i+1})\}\in E_{f_2}$, we have that $\{\mathsf{Identify}_{d_f,d_{f_2}}(uv_i),\mathsf{Identify}_{d_f,d_{f_2}}$ $(uv_{i+1})\}$ is drawn strictly inside $f_2$ except for its endpoints, which are drawn on $\gi(f)\cap \gi(f_2(c))$ in $d_2$. So, $\{uv_i,uv_{i+1}\}$ is drawn strictly inside $f$ except for its endpoints, which are drawn on $f$. Thus, Condition \ref{gluedecon3} is satisfied.
	\label{proofGlu2}
		\item Since $\{\mathsf{Identify}_{f,d_{f_1}^*}(uv_i),\mathsf{Identify}_{f,d_{f_1}^*}(uv_{i+1})\}\in E(d_{f_1}^*,\pp(c))$, $d_{\mathsf{E}}(\{uv_i,uv_{i+1}\})$ $=d_{f_1}^*(\{\mathsf{Identify}_{d_f,d_{f_1}^*})$ is well defined and $d_{\mathsf{E}}(uv_i,uv_{i+1})\in {\cal P}^*(\fin)$. In addition,\\ $\{\mathsf{Identify}_{f,d_{f_1}^*}(uv_i),$ $\mathsf{Identify}_{f,d_{f_1}^*}(uv_{i+1})\}\in E(d_{f_1}^*,\pp(c))$ implies that $d_{\mathsf{E}}(\{uv_i,$ $uv_{i+1}\})$ is drawn strictly inside $f$ except for its endpoints. Thus, Conditions \ref{gluedecon1}, \ref{gluedecon4} and \ref{gluedecon3} are satisfied. \label{proofGlu3}
		\item Otherwise, $\{uv_i,uv_{i+1}\}$ partly intersects $c$. Then, let $\ell,j \in \mathbb{N}$ and $r\in[2]$ be as defined in Definition \ref{def:partIntCBoth}. We show that $\mathsf{PartEdge}_t$ is well defined, for every $0\leq t<\ell$, according to the following different cases:
		\begin{enumerate}
			\item If $r=1$ and $\{uv_{j+t},uv_{j+t+1}\}\in E_{f_1}$, then $\mathsf{PartEdge}_t=d_1(\{uv_{j+t},uv_{j+t+1}\})$; if $r=2$ and $\{\mathsf{Identify}_{d_{f_2},d_{f_1}}(uv_{j+t}),\mathsf{Identify}_{d_{f_2},d_{f_1}}(uv_{j+t+1})\}\in E_1$, then $\mathsf{PartEdge}_t=d_1(\{\mathsf{Identify}_{d_{f_2},d_{f_1}}(uv_{j+t}),\mathsf{Identify}_{d_{f_2},d_{f_1}}(uv_{j+t+1})\})$. In both cases, similarly to the case shown in Condition \ref{proofGlu1}, it follows that $\mathsf{PartEdge}_t$ is well defined.
			\item If $r=2$ and $\{uv_{j+t},uv_{j+t+1}\}\in E_{f_2}$, then $\mathsf{PartEdge}_t=d_2(\{uv_{j+t},uv_{j+t+1}\})$; if $r=1$ and $\{\mathsf{Identify}_{d_{f_1},d_{f_2}}(uv_{j+t}),\mathsf{Identify}_{d_{f_1},d_{f_2}}(uv_{j+t+1})\}\in E_2$, then $\mathsf{PartEdge}_t=d_2(\{\mathsf{Identify}_{d_{f_1},d_{f_2}}(uv_{j+t}),\mathsf{Identify}_{d_{f_1},d_{f_2}}(uv_{j+t+1})\})$. In both cases, similarly to the case shown in Condition \ref{proofGlu2}, it follows that $\mathsf{PartEdge}_t$ is well defined.
			\item If $\{uv_{j+t},uv_{j+t+1}\}\in E(d_{f_r}^*,\pp(c))$, then $\mathsf{PartEdge}_t=d_{f_r}^*(\{uv_{j+t},uv_{j+t+1}\})$, so it is easy to see that $\mathsf{PartEdge}_t$ is well defined.
			
			%	\indent
			
			\indent
			Now, we examine $\mathsf{PartEdge}_\ell$ according to the following different cases:
			
			\item If $r=1$ and $\mathsf{Identify}_{d_f,d_{f_1}}(uv_{i+1})=uv_{j+\ell+1}$, then $\mathsf{PartEdge}_\ell=d_1(\{uv_{j+\ell},uv_{j+\ell+1}\})$; else, $r=2$ and then $\mathsf{PartEdge}_\ell=d_1(\{\mathsf{Identify}_{d_{f_2},d_{f_1}}(uv_{j+\ell}),\mathsf{Identify}_{d_f,d_{f_1}}(uv_{i+1})\})$. In both cases, similarly to the case shown in Condition \ref{proofGlu1}, it follows that $\mathsf{PartEdge}_\ell$ is well defined.
			
			\item If $r=2$ and $\mathsf{Identify}_{d_f,d_{f_2}}(uv_{i+1})=uv_{j+\ell+1}$, then $\mathsf{PartEdge}_\ell=d_2(\{uv_{j+\ell},uv_{j+\ell+1}\})$; else, $r=1$ and then $\mathsf{PartEdge}_\ell=d_2(\{\mathsf{Identify}_{d_{f_1},d_{f_2}}(uv_{j+\ell}),\mathsf{Identify}_{d_f,d_{f_2}}(uv_{i+1})\})$. In both cases, similarly to the case shown in Condition \ref{proofGlu2}, it follows that $\mathsf{PartEdge}_\ell$ is well defined.
			
			\item Otherwise, $\mathsf{Identify}_{d_f,d_{f_r}}(uv_{i+1})=uv_{j+\ell+1}$ and then  $\mathsf{PartEdge}_\ell=d_{f_r}^*(\{uv_{j+\ell},uv_{j+\ell+1}\}$ $)$, so it is easy to see that $\mathsf{PartEdge}_\ell$ is well defined.
		\end{enumerate}
	
	Now, it is easy to see that for every for every $0\leq t<\ell$, if $\mathsf{PartEdge}_t=(p_1,\ldots p_k)$ and $\mathsf{PartEdge}_{t+1}=(q_1,\ldots q_s)$, then $p_k=q_1$. So, $P=\mathsf{PartEdge}_0\cdot \mathsf{PartEdge}_1\cdot \ldots \cdot \mathsf{PartEdge}_\ell$ is well defined, and thus $d_{\mathsf{E}}(\{uv_i,uv_{i+1}\})$ is well defined. In conclusion, we get that Condition \ref{gluedecon1} is satisfied.
	
	Moreover, recall that $d_{\mathsf{E}}(\{uv_i,uv_{i+1}\})$ is the path obtained from $P=\mathsf{PartEdge}_0\cdot \mathsf{PartEdge}_1\cdot \ldots \cdot \mathsf{PartEdge}_\ell=(p_1,\ldots,p_q)$ by deleting every\\ $p_i\in \gis(\fin)\setminus \gps(\fin)$ where $i\neq 1,q$. Therefore, it is easy to see that $d_{\mathsf{E}}(uv_i,uv_{i+1})\in {\cal P}^*(\fin)$, so Condition \ref{gluedecon4} is satisfied.

	 In addition, consider the following observations:
	 \begin{itemize}
	 	\item $\mathsf{PartEdge}_0$, similarly to Conditions \ref{proofGlu1}, \ref{proofGlu2} and \ref{proofGlu3}, is drawn strictly inside $f$, except at one of the endpoints: one of the endpoints is drawn on\\ $\gi(f_r)\cap \gi(f)$ (and therefore drawn on $f$), and the other is drawn on $\gi(c)\setminus \gi(f)$ (and therefore strictly inside $f$).
	 	\item Similarly, for every $0< t<\ell$,  $\mathsf{PartEdge}_t$ is drawn strictly inside $f$, as both of its endpoints are drawn on $\gi(c)\setminus \gi(f)$ (and therefore strictly inside $f$).
	 	\item $\mathsf{PartEdge}_\ell$, similarly to $\mathsf{PartEdge}_0$, is drawn strictly inside $f$, except at one of the endpoints: one of the endpoints is drawn on is drawn on $\gi(c)\setminus \gi(f)$ (and therefore strictly inside $f$), and the other is drawn on $\gi(f_r)\cap \gi(f)$ (and therefore drawn on $f$).
	 \end{itemize}   
		
		Finally, we get that $P=\mathsf{PartEdge}_0\cdot \mathsf{PartEdge}_1\cdot \ldots \cdot \mathsf{PartEdge}_\ell$ is strictly inside $f$, except for both of its endpoints, and therefore also $d_{\mathsf{E}}(\{uv_i,uv_{i+1}\})$, so Condition \ref{gluedecon3} is satisfied. 
	\end{enumerate}
This completes the proof of the lemma.
\end{proof}

The following lemma, Lemma \ref{lem:welldef2}, is the lemma analogous to Lemma \ref{lem:welldef1}, corresponding to the drawings of edges defined in Conditions  \ref{def:glue4b} and \ref{def:glue4c} of Definition \ref{def:glueEdg}. Observe that, in Condition \ref{gluedecon23} of Lemma \ref{lem:welldef2}, only one endpoint of the edge is drawn on $d_f$, as opposed to Condition \ref{gluedecon3} of Lemma \ref{lem:welldef1}. Since Lemma \ref{lem:welldef2} can be proved similarly to Lemma \ref{lem:welldef1}, we do not provide a proof.

\begin{lemma} \label{lem:welldef2}
	Let $F=(f,d_f,E_f,U_f,\mathsf{V^*Dir}_f)$ be an info-frame, let $C=(c,F_1=(f_1(c),d_{f_1},$ $E_{f_1},$ $U_{f_1},$ $\mathsf{V^*Dir}_{f_1}), F_2=(f_2(c),d_{f_2},E_{f_2},U_{f_2},\mathsf{V^*Dir}_{f_2}))$ be an info-cutter of $F$, and let $d_1$ and $d_2$ be drawings of $F_1$ and $F_2$, respectively. Let $\{uv_i,uv_{i+1}\}\in E_f$ such that exactly one among $uv_i$ and $uv_{i+1}$ is drawn on $d_f$. Then, the following conditions are satisfied:
	\begin{enumerate}
		\item {\bf$\mathsf{Glue}$ of $\{uv_i,uv_{i+1}\}$} is well defined. \label{gluedecon21}
		\item $d_{\mathsf{E}}(uv_i,uv_{i+1})\in {\cal P}^*(\fin)$. \label{gluedecon24}
		\item $d_{\mathsf{E}}(\{uv_i,uv_{i+1}\})$ is drawn strictly inside $f$ except for exactly one of its endpoints. \label{gluedecon23}
		\item The endpoints of $d_{\mathsf{E}}(\{uv_i,uv_{i+1}\})$ are $d_{\mathsf{V}}(uv_i)$ and $d_{\mathsf{V}}(uv_{i+1})$, defined in Definition \ref{def:glueVer}. \label{gluedecon22}
	\end{enumerate} 
\end{lemma}

Similarly, the following lemma is analogous to Lemma \ref{lem:welldef1}, corresponding to the drawings of edges defined in Condition \ref{def:gluecon5} of Definition \ref{def:glueEdg}. In Condition \ref{gluedecon33} of this lemma, the edge $\{u,v\}$ is drawn strictly inside $f$.

\begin{lemma} \label{lem:welldef3}
	Let $F=(f,d_f,E_f,U_f,\mathsf{V^*Dir}_f)$ be an info-frame, let $C=(c,F_1=(f_1(c),d_{f_1},$ $E_{f_1},$ $U_{f_1},$ $\mathsf{V^*Dir}_{f_1}), F_2=(f_2(c),d_{f_2},E_{f_2},U_{f_2},\mathsf{V^*Dir}_{f_2}))$ be an info-cutter of $F$, and let $d_1$ and $d_2$ be drawings of $F_1$ and $F_2$, respectively. Let $u,v\in U_f$ such that $\{u,v\}\in E$ and $V(d_f)\cap V^*_{\{u,v\}}=\emptyset$. Then, the following conditions are satisfied:
	\begin{enumerate}
		\item{\bf$\mathsf{Glue}$ of $\{u,v\}$} is well defined. \label{gluedecon31}
			\item $d_{\mathsf{E}}(uv_i,uv_{i+1})\in {\cal P}^*(\fin)$. \label{gluedecon34}
	\item $d_{\mathsf{E}}(\{u,v\})$ is drawn strictly inside $f$. \label{gluedecon33}
		\item The endpoints of $d_{\mathsf{E}}(\{u,v\})$ are $d_{\mathsf{V}}(u)$ and $d_{\mathsf{V}}(v)$, defined in Definition \ref{def:glueVer}. \label{gluedecon32}
	\end{enumerate} 
\end{lemma}

Now, towards achieving our goal proving that $d$, defined in Definition \ref{def:glue}, is a drawing of $F$, we first prove that $d$ is a $G^*$-drawing: 

\begin{lemma}
	Let $F=(f,d_f,E_f,U_f,\mathsf{V^*Dir}_f)$ be an info-frame, let $C=(c,F_1=(f_1(c),d_{f_1},$ $E_{f_1},$ $U_{f_1},\mathsf{V^*Dir}_{f_1}), F_2=(f_2(c),d_{f_2},E_{f_2},U_{f_2},\mathsf{V^*Dir}_{f_2}))$ be an info-cutter of $F$, and let $d_1$ and $d_2$ be drawings of $F_1$ and $F_2$, respectively. Then, $\mathsf{Glue}(F,C,d_1,d_2)=d$ is a $G^*$-drawing.
\end{lemma}

\begin{proof}
	First, in Lemma \ref{lem:welldefV}, we proved that $d_{\mathsf{V}}$ is well defined. In addition, due to Lemmas \ref{lem:welldef1}--\ref{lem:welldef3}, $d_{\mathsf{E}}$ is well defined. Thus, $d=(d_{\mathsf{V}},d_{\mathsf{E}})$ is well defined. Moreover, due to the forth conditions of Lemmas \ref{lem:welldef1}--\ref{lem:welldef3}, the drawing defined by $d_{\mathsf{E}}$, for each edge $e$, connects the two endpoints of $e$; so, $d$ is a drawing. Now, we show that $d$ is a $G^*$-drawing (according to Definition \ref{def:gstdr}). First, we show that $(V(d),E(d))$ is valid by showing that the conditions of Definition \ref{def:ValidPair} hold:
	\begin{enumerate}
		\item It is easy to see that $V(d)\subset V\cup V^*$ and $E(d)\subset E\cup E^*$ where the endpoints of the edges in $E(d)$ belong to $V(d)$, so Condition \ref{con:ValidPair1} is satisfied.
		\item Let $u,v \in V(d)$ such that $\{u,v\} \in E$, and assume that $U\cap V^*_{\{u,v\}}\neq\emptyset$. We aim to prove that Conditions \ref{con:ValidPair221}, \ref{con:ValidPair222} and \ref{con:ValidPair223} are satisfied.
		\begin{enumerate}
			\item Observe that $V^*_{\{u,v\}}\cap V(d)=V^*_{\{u,v\}}\cap V(d_f)$. So, since $d_f$ is a $G^*$-drawing, $V(d_f)\cap V^*_{\{u,v\}} = \{uv_1,uv_2,\ldots ,uv_{\mathsf{index}(u,v)}\}$ for some $\mathsf{index}(u,v)\in \mathbb{N}$, and hence $V(d)\cap V^*_{\{u,v\}} = \{uv_1,uv_2,\ldots ,uv_{\mathsf{index}(u,v)}\}$. Thus, Condition \ref{con:ValidPair221} holds.
			\item Similarly,  $E(d)\cap E^*_{\{u,v\}}=E(d_f)\cap E^*_{\{u,v\}}$. Therefore, since $d_f$ is a $G^*$-drawing, $E(d_f)\cap E^*_{\{u,v\}}\subseteq \{\{u,uv_1\}\}\cup \{\{uv_j,uv_{j+1}\}~|~1\leq j\leq \mathsf{index}(u,v)-1\}\cup \{\{uv_{\mathsf{index}(u,v)},v\}\}$. So, $E(d)\cap E^*_{\{u,v\}}\subseteq \{\{u,uv_1\}\}\cup \{\{uv_j,uv_{j+1}\}~|~1\leq j\leq \mathsf{index}(u,v)-1\}\cup \{\{uv_{\mathsf{index}(u,v)},v\}\}$. Thus, Condition \ref{con:ValidPair222} holds.
			\item Since $d_f$ is a $G^*$-drawing, $\{u,v\}\notin E(d_f)$. Since $F$ is an info-frame, by Definition \ref{def:infFr}, $\{u,v\}\notin E_f$. Now, since we assume that $U\cap V^*_{\{u,v\}}\neq\emptyset$, we do not add $\{u,v\}$ to $d$ by the definition of $\mathsf{Glue}$ (Definition \ref{def:glue}) in Condition \ref{def:gluecon5}. Therefore, $\{u,v\}\notin E(d)$, so Condition \ref{con:ValidPair223} holds.
		\end{enumerate}

	\end{enumerate}
	Thus, $(V(d),E(d))$ is valid, so Condition \ref{G*drawcon10}  of Definition \ref{def:glue} is satisfied.
	Now, observe that for every $u\in V(d)\cap V$, $d(u)=d_1(d)$ or $d(u)=d_1(d)$. So, since $d_1$ and $d_2$ are $G^*$-drawings, $d(u)\in \gps(f_{\mathsf{init}})$. Similarly, for every $u\in V(d)\cap V^*$, $d(u)\in \gis(f_{\mathsf{init}})$, so Condition \ref{G*drawcon2} is satisfied. In addition, due to the second conditions of Lemmas \ref{lem:welldef1}--\ref{lem:welldef3}, for every $\{uv_i,uv_{i+1}\}\in E(d)\cap E^*$, $d(\{uv_i,uv_{i+1}\})\in {\cal P}^*(f_{\mathsf{init}})$. Now, let $\{u,v\}\in E(d)\cap E$. Again, $d(\{u,v\})\in {\cal P}^*(f_{\mathsf{init}})$, due to the second conditions of Lemmas \ref{lem:welldef1}--\ref{lem:welldef3}. Furthermore, due to the forth conditions of Lemmas \ref{lem:welldef1}--\ref{lem:welldef3}, the first and the last points of $E(\{u,v\})$ are $d(u)$ and $d(v)$, and as we saw, $d(u),d(v)\in \gps(\fin)$. Thus, $d(\{u,v\})\in {\cal P}(f_{\mathsf{init}})$, and so Condition \ref{G*drawcon3}  is satisfied. 
	
	Now, notice that $V^*_{\{u,v\}}\cap V(d)=V^*_{\{u,v\}}\cap V(d_f)$, $d_f$ is a $G^*$-drawing, and for every $u\in V^*_{\{u,v\}}\cap V(d)$, $d(u)=d_f(u)$. Therefore, for every $uv_i,uv_j\in V^*$ such that $i\neq j$, $d(uv_i)\neq d(uv_j)$, so Condition \ref{G*drawcon7}  is satisfied. Similarly, observe that $d=^{\pp(f)}d_f$, and in addition, $d_f$ is a $G^*$-drawing and $V^*_{\{u,v\}}\cap V(d)=V^*_{\{u,v\}}\cap V(d_f)$. Therefore, for every $uv_j\in V(d)$ and $\{uv_i,uv_{i+1}\}\in E(d)$ such that $j\neq i,i+1$, $uv_j$ is not drawn on $d(\{uv_i,uv_{i+1}\})$, so Condition \ref{G*drawcon8}  is satisfied. Now, observe that the drawing of an edge $\{uv_i,uv_{i+1}\}\in E^*$, defined by the $\mathsf{Glue}$ function, is a concatenation of drawings of different edges form $E^*_{\{u,v\}}$, drawn in $d_1$ or $d_2$. In addition, $d_1$ and $d_2$ are $G^*$-drawings, $\pp(d_1)\cap \pp(d_2)\subseteq \pp(c)$ and $d_1=^{\gi(c)}_\mathsf{rename} d_2$. Thus, every two different edges $\{uv_i,uv_{i+1}\}, \{uv_j,uv_{j+1}\}$,  are non-intersecting in $d$, so Condition \ref{G*drawcon9} is satisfied.  
	
	We have proved that the conditions of Definition \ref{def:gstdr} are satisfied, so $d$ is a $G^*$-drawing.
\end{proof}

Now, we are ready to prove that $d$ is a drawing of $F$:

\begin{lemma}\label{lem:glue}
	Let $F=(f,d_f,E_f,U_f,\mathsf{V^*Dir}_f)$ be an info-frame, let $C=(c,F_1=(f_1(c),d_{f_1},$ $E_{f_1},$ $U_{f_1},\mathsf{V^*Dir}_{f_1}), F_2=(f_2(c),d_{f_2},E_{f_2},U_{f_2},\mathsf{V^*Dir}_{f_2}))$ be an info-cutter of $F$, and let $d_1$ and $d_2$ be drawings of $F_1$ and $F_2$, respectively. Then, $\mathsf{Glue}(F,C,d_1,d_2)=d$ is a drawing of $F$.
\end{lemma}

\begin{proof}
	We show that the conditions of Definition \ref{def:infFrDr} hold:
	\begin{enumerate}
		\item Since $d_1$ and $d_2$ are drawings of $F_1$ and $F_2$, respectively, they are bounded by $f_1(c)$ and $f_2(c)$, respectively, so they are bounded by $f$. In addition, $d_f$ is bounded by $f$ as well. Therefore, by the construction of $\mathsf{Glue}$ (see Definition \ref{def:glue}), $d$ is bounded by $f$. Furthermore,  $V^*_{\{u,v\}}\cap V(d)=V^*_{\{u,v\}}\cap V(d_f)$, $d_f$ is on $f$, and for every $u\in V^*_{\{u,v\}}\cap V(d)$, $d(u)=d_f(u)$. Therefore, all the vertices in $V^*$ of $d$ are drawn on $f$. So, Condition \ref{infFramcon1} holds.
		\item From Condition \ref{def:glueVer1} of Definition \ref{def:glueVer} and Condition \ref{def:gleEdgCon1} of Definition \ref{def:glueEdg}, it is clear that every element drawn in $d_f$ has the same drawing in $d$. Now, due to the third conditions of Lemmas \ref{lem:welldef1}--\ref{lem:welldef3}, no additional elements are drawn on $f$ in $d$, so $d=^{\pp(f)}d_f$. Thus, Condition \ref{infFramcon2} holds.
		\item Due to the third conditions of Lemmas \ref{lem:welldef1}--\ref{lem:welldef3}, $U_f$ is the set of vertices of $d$ drawn strictly inside $f$, so Condition \ref{infFramcon3} holds.
		\item From the second conditions of Lemmas \ref{lem:welldef1} and \ref{lem:welldef2}, every $e\in E_f$ is drawn strictly inside $f$, except maybe at the endpoints, and at least one endpoint of $e$ is drawn on $f$. In addition, every edge in $E(d)\setminus E_f$ is drawn either on $f$ (defined in Condition \ref{def:gleEdgCon1} of Definition \ref{def:glueEdg}), or strictly inside $f$ (proved in Condition \ref{gluedecon33} of Lemma \ref{lem:welldef3}).
		Therefore, $E_f\subseteq E(d)$ is the set of each edge $e$ that is drawn strictly inside $f$, except maybe at the endpoints, and at least one endpoint of $e$ is drawn on $f$. So, Condition \ref{infFramcon4} is satisfied.
		\item For every $u,v\in U_f$ such that $\{u,v\}\in E$ and $V(d_f)\cap V^*_{\{u,v\}}=\emptyset$, from Condition \ref{def:gluecon5} of Definition \ref{def:glue}, $\{u,v\}\in E(d)$, so Condition \ref{infFramcon5} holds.
		\item Let $uv_i\in V(d_f)\cap V^*$ such that $d_f(uv_i)\in \gis(\fin)\setminus \mathsf{GridPointSet}(f_{\mathsf{init}})$. Assume, without loss of generality, that $d(uv_i)\in f_1(c)$. In addition, assume that $d(uv_i)$ was defined in Condition \ref{con1defglueedge} of Definition \ref{def:glueEf2} (the other cases are similar). Since $C$ is an info-cutter of $F$, we get that $\ell(\mathsf{V^*Dir}_f(uv_i),d_f(uv_i))$ is on $\ell(\mathsf{V^*Dir}_{f_1}( \mathsf{Identify}_{d_{f},d_{f_1}}(uv_i)),d_f($ $uv_i))$ or vice versa. Since $d_1$ is a drawing of $F_1$, $\ell(\mathsf{V^*Dir}_{f_1}(\mathsf{Identify}_{d_{f},d_{f_1}}(uv_i)),d_f(uv_i))$ is on $\ell(d_f(uv_i),p_1)$, or vice versa, where $p_1$ is the second point of $d_1(\mathsf{Identify}_{d_{f},d_{f_1}}(uv_i)),$ $\mathsf{Identify}_{d_{f},d_{f_1}}$ $(uv_{i+1}))$. Now, since $d_{\mathsf{E}}(\{uv_i,uv_{i+1}\})=d_1(\{\mathsf{Identify}_{d_f,d_{f_1}}(uv_i),\mathsf{Identify}_{d_f,d_{f_1}}$ $(uv_{i+1})\})$, we get that $\ell(\mathsf{V^*Dir}_f(uv_i),d_f(uv_i))$ is on $\ell(d_f(uv_i),p_1)$ or vice versa, where $p_1$ is the second point of $d(uv_i,uv_{i+1})$. Therefore, Condition \ref{infFramcon6} is satisfied.
		\end{enumerate}
		We have thus proved that the conditions of Definition \ref{def:infFrDr} are satisfied, so $d$ is a drawing of $F$.
\end{proof}

%	\item For every $uv_i\in V(d_f)\cap V^*$ such that $d_f(uv_i)\in \gis(\fin)\setminus \mathsf{GridPointSet}(f_{\mathsf{init}})$, $\ell(\mathsf{V^*Dir}_f(uv_i),d_f(uv_i))$ is on $\ell(d_f(uv_i),p_1)$, or $\ell(d_f(uv_i),p_1)$ is on $\ell(\mathsf{V^*Dir}_f(uv_i),d_f(uv_i))$, where $z\in \{uv_{i-1},uv_{i+1}\}$ such that $\{uv_i,z\}\in E_f$ and $d(\{uv_i,z\})=(d(uv_i),p_1\ldots,d(z))$.\label{infFramcon6}

Now, for a later use, we have the following observation. Let $f'$ be a frame bounded by $f_1(c)$. Observe that the part of the drawing $d$ that is intersected by $f'$ is equal to the part of the drawing $d_1$ that is intersected by $f'$, up to renaming. Thus, the subset of vertices from $V$ that are drawn on $f'$ in $d$ equals the subset of vertices from $V$ that are drawn on $f'$ in $d_1$. Moreover, every turning point $(p,\{u,v\})$ of $f'$ in $d$ is a turning point of $f'$ in $d_1$, and vice versa. Therefore, we get that $\sizef(f',d)=\sizef(f',d_1)$, as we state in the next observation:

\begin{observation}\label{obs:turningEq1}
	Let $F=(f,d_f,E_f,U_f,\mathsf{V^*Dir}_f)$ be an info-frame, let $C=(c,F_1,F_2)$ be an info-cutter of $F$, let $d_1$ and $d_2$ be drawings of $F_1$ and $F_2$, respectively, and let $f'$ be a frame that is bounded by $f_1(c)$. Then, $\sizef(f',d)=\sizef(f',d_1)$, where $d=\mathsf{Glue}(F,C,d_1,d_2)$.
\end{observation}

In the following lemma, we show that when we glue $\mathsf{Splitter}(F,d,c)_{d_1}$ and $\mathsf{Splitter}(F,d,c)_{d_2}$, we construct the drawing $d$. 

\begin{lemma} \label{lem:Glueto}
	Let $F$ be an info-frame, let $d$ be a drawing of $F$, and let $c$ be a cutter of $f$. Let $C$, $d_1$ and $d_2$ be $\mathsf{Splitter}(F,d,c)_{C}$, $\mathsf{Splitter}(F,d,c)_{d_1}$ and $\mathsf{Splitter}(F,d,c)_{d_2}$, respectively. Then, $\mathsf{Glue}(F,C,d_1,d_2)=d$.
\end{lemma}

\begin{proof}
	We denote $\mathsf{Glue}(F,C,d_1,d_2)$ by $d'$, and aim to show that $d=d'$.  Since $d$ and $d'$ are drawings of $F$, we get that $d$ and $d'$ are bounded by $f$ and also $d(f)=d(f')$. 
	
	Now, let $u\in V(d)$ be a vertex drawn strictly inside $d$. If $u$ is drawn on $f_1(c)$ in $d$, then $d(u)=d_1(u)=d'(u)$. Similarly, if $u$ is drawn on $f_2(c)$ in $d$, then $d(u)=d_2(u)=d'(u)$. Since $d$ and $d'$ are drawings of $F$, the set of vertices that are strictly drawn inside $d$ is equal to the set of vertices that are drawn strictly inside $d'$. Therefore, we get that $d_{\mathsf{V}}=d'_{\mathsf{V}}$.
	
	Now, let $\{u,v\}\in E(d)$ be an edge. We have thee following cases:
	\begin{itemize}
		\item $\{u,v\}$ is drawn on $f$.
		\item $\{u,v\}$ is drawn strictly inside $f$ except at the endpoints, which are drawn on $f$.
		\item $\{u,v\}$ is drawn strictly inside $f$ except at exactly one endpoint, which is drawn on $f$.
		\item $\{u,v\}$ is drawn strictly inside $f$.
	\end{itemize}
	As for the first case, since $d=^{\pp(f)}d'$, it follows that $d(\{u,v\})=d'(\{u,v\})$. The rest of the cases are similar, so we show here only the proof for the last case.
	Thus, assume that $\{u,v\}$ is drawn strictly inside $d$. If $\{u,v\}$ is drawn strictly inside $f_1(c)$ in $d$, then $d(\{u,v\})=d_1(\{u,v\})=d'(\{u,v\})$. Similarly, if $\{u,v\}$ is drawn strictly inside $f_2(c)$ in $d$, then $d(\{u,v\})=d_2(\{u,v\})=d'(\{u,v\})$. If $\{u,v\}$ is drawn on $c$ in $d$, then $d(\{u,v\})=d_1(\{u,v\})=d'(\{u,v\})$.  Otherwise, assume that $u$ is drawn strictly inside $f_1(c)$ and $v$ is drawn strictly inside $f_2(c)$ in $d$ (the other cases are similar). Let $uv_1,\ldots uv_\ell$ the vertices created by the function $\mathsf{Splitter}$ in $d_1$. Observe that in this case, $uv_1,\ldots uv_\ell$ are the vertices created by the function $\mathsf{Splitter}$, with respect to the edge $\{u,v\}$, in $d_2$ as well, such that $d_1(uv_i)=d_2(uv_i)$ for every $1\leq i\leq \ell$. We get that $\{u,uv_1\}\in E(d_1)$, $\{uv_\ell,v\}\in E(d_2)$, and for every $1\leq i\leq \ell-1$, exactly one of the following conditions holds:
	\begin{itemize}
		\item $\{uv_i,uv_{i+1}\}\in E(d_1)$.
		\item $\{uv_i,uv_{i+1}\}\in E(d_2)$.
		\item $\{uv_i,uv_{i+1}\}\in E(d_1(c))$.
	\end{itemize} 
For every $1\leq i< \ell$, let $t_i\in[2]$ be such that the edge connecting $uv_i$ and $uv_{i+1}$ is drawn in $d_{t_i}$.
Observe that $d_{t_i}(\{uv_i,uv_{i+1}\})$ is the part of the drawing of the edge $\{u,v\}$ in $d$, from the point $d_{t_i}(uv_i)$ to the point $d_{t_i}(uv_{i+1})$. Observe that $P=d_1(\{u,uv_1\})\cdot d_{t_1}(\{uv_1,uv_2\})\cdot\ldots\cdot d_{t_{\ell-1}}(uv_{\ell-1},uv_\ell)\cdot d_2(\{uv_\ell,v\})=d(\{u,v\})$. Now, by the $\mathsf{Glue}$ function, we get that $d'(\{u,v\})$ is the path obtained from $P=(p_1,\ldots,p_q)$ by deleting every $p_i\in \gis(\fin)\setminus \gps(\fin)$ where $i\neq 1,q$. As we saw, $P$ is the same path as $d'(\{u,v\})$ (see the discussion before Definition \ref{def:glueEf2}). Therefore, we get that $d'(\{u,v\})=d(\{u,v\})$. Since $d$ and $d'$ are drawings of $F$, the set of edges that are drawn strictly inside $d$ is equal to the set of edges that are drawn strictly inside $d'$. Similarly, the set of edges in $d$ that are drawn strictly inside $f$, except maybe at the endpoints, is equal to the set of edges in $d'$ that are drawn strictly inside $f$, except maybe at the endpoints,. Therefore, we get that $d_{\mathsf{E}}=d'_{\mathsf{E}}$.
	
In conclusion, we get that $d=d'$.
\end{proof}

\subsection{Problem Information}

Up until this subsection, we focused on the ``information'' we need to store in order to get any polyline grid drawing of $G$. However, we are often interested in a specific kind of drawing and not just in any polyline grid drawing. To this end, our algorithm also gets from the user the specific ``information'' we need for her or his use. Observe that when we use the $\mathsf{Glue}$ function in order to glue two sub-drawings, we do not distinguish between the possible drawings that can be the sub-drawings. That is, given an info-frame $F$ and an info-cutter $C=(c,F_1,F_2)$ of $F$, the $\mathsf{Glue}$ function returns a drawing $d$ of $F$ with the input of any drawings $d_1$ and $d_2$ of $F_1$ and $F_2$, respectively. Therefore, given an info-frame $F$, we would like to have a partition (corresponding to an equivalence relation) of the set of all drawings of $F$ based on the user's problem. 

In particular, the user provides a function, denoted by $\mathsf{Classifier}$, that given an info-frame $F$ and a drawing $d$ of $F$, returns a value $I'\in \mathsf{INF}$, where $\mathsf{INF}$ is a universe chosen by the user.

We now define formally the term {\em classifier}:

\begin{definition}[{\bf $\mathsf{INF}$-Classifier}] \label{def:probInf}
	Let $\mathsf{INF}$ be a universe. An {\em $\mathsf{INF}$-classifier} is a function $\mathsf{Classifier}$ that given an info-frame $F$ and a drawing $d$ of $F$, returns a value $I\in \mathsf{INF}$.
\end{definition}

When $\mathsf{INF}$ is clear from the context, we refer to $\mathsf{INF}$-classifier as classifier. 

 Now, for an info-frame $F$ and two drawings $d$ and $d'$ of $F$, we say that $d$ and $d'$ are {\em equivalent} if and only if $\mathsf{Classifier}(F,d)=\mathsf{Classifier}(F,d')$. We thus obtain an equivalent relation on the set of drawings of $F$, and, in turn, also a partition of this set.  

%Formally, we  define $\mathsf{Classifier}$ for every info-frame $F$, denoted by $\mathsf{Classifier}_F$. Each $\mathsf{Classifier}_F$ is a function $\mathsf{Classifier}_F:\mathsf{Drawing}(F)\rightarrow I$, where $\mathsf{Drawing}(F)$ is the set of all drawings of $F$. We denote $\mathsf{Classifier}_F(d)$ by $\mathsf{Classifier}(F,d)$ for every $F$, and a drawing $d$ of $F$. That is, $\mathsf{Classifier}$ is a function that gets an info-frame $F$ and a drawing $d$ of $F$, and returns a value $I'\in I$. We call these functions, {\em problem information} functions. 

In addition to the $\mathsf{Classifier}$ function, the user needs to provide an algorithm that computes the following value. Given an info-frame $F$, an info-cutter $C=(c,F_1,F_2)$, of $F$, and two values $I_1,I_2\in \mathsf{INF}$, the algorithm returns a value $I'\in \mathsf{INF}$ such that the following condition is satisfied: For every two drawings $d_1$ and $d_2$ of $F_1$ and $F_2$, respectively, such that $\mathsf{Classifier}(F_1,d_1)=I_1$ and $\mathsf{Classifier}(F_2,d_2)=I_2$, it holds that $\mathsf{Classifier}(F,d)=I'$ where $d=\mathsf{Glue}(F,C=(c,F_1,F_2),d_1,d_2)$. Observe that $\mathsf{Classifier}(F,d)$ is well defined, since by Lemma \ref{lem:glue}, $\mathsf{Glue}(F,C,d_1,d_2)$ is a drawing of $F$. This way, in our scheme, we do not need to ``distinguish'' between any two drawings in the same equivalence class. We refer to the user's algorithm as {\em $\mathsf{classifierAlg}$}, defined as follows:      

\begin{definition}[{\bf $\mathsf{Classifier}$-Algorithm}] \label{def:probsolv}
	Let $\mathsf{INF}$ be a universe and let $\mathsf{Classifier}$ be a classifier. An algorithm $A$ is a {\em $\mathsf{Classifier}$-algorithm} if the following conditions are satisfied.
	\begin{enumerate}
		\item $A$ gets as input $(F,C,I_1,I_2)$ where $F$ is an info-frame, $C$ is an info-cutter of $F$, and $I_1,I_2\in \mathsf{INF}$. \label{def:probSolvCon1}
		\item $A$ outputs $I'\in \mathsf{INF}$ such that for every two drawings of $d_1$ and $d_2$ of $F_1$ and $F_2$, respectively, where $\mathsf{Classifier}(F_1,d_1)=I_1$ and $\mathsf{Classifier}(F_2,d_2)=I_2$, we have that $\mathsf{Classifier}(F,\mathsf{Glue}(F,C,$ $d_1,d_2))=I'$. \label{def:probSolvCon2}
	\end{enumerate} 
\end{definition}

%!TEX root =Main-Movement.tex

\subsection{Proof of the Scheme}\label{sec:proofScheme}

Now, we are ready to prove the main lemma of our scheme, which is the proof of the inductive step of the algorithm that is presented later. 

For an info-frame $F=(f,d_f,E_f,U_f,\mathsf{V^*Dir}_f)$, we say that $F$ is a {\em leaf} if there are no grid points strictly inside $f$. Recall Condition \ref{def:tfCon3} of the definition of a frame-tree (Definition \ref{def:DrawnFrDeco}): a vertex $v$ is a leaf if and only if the frame associated with $v$ has no grid points strictly inside. Therefore, drawings of leave info-frames will be the parts of the drawing bounded by a frame which is associated with a leaf vertex in a tree decomposition.

We now present a few definitions that will be in use later.

\begin{definition}[{\bf $\mathsf{InfoFrame}(G,k,h,w)$}] \label{def:FAtMoK}
	Let $k,h,w\in \mathbb{N}$. We denote by $\mathsf{InfoFrame}(G,k,h,$ $w)$ the set of info-frames $F=(f,d_f,E_f,U_f,\mathsf{V^*Dir}_f)$ of $G$ such that $f$ is bounded by $R_{h,w}$, and $\sizef(d_f,f)\leq k$. For an info-frame $F\in \mathsf{InfoFrame}(G,k,h,w)$, we say that {\em $F$ costs at most $k$}.
\end{definition}

\begin{definition}[{\bf Info-Cutter Costs At Most $k$}] \label{def:CAtMostK}
Let $C=(c,F_1,F_2)$ be an info-cutter of an info-frame $F=(f,d_f,E_f,U_f,\mathsf{V^*Dir}_f)$ and let $k\in \mathbb{N}$. We say that {\em $C$ costs at most $k$} if $F_1,F_2\in \mathsf{InfoFrame}(G,k,h,w)$. 
\end{definition}

%Let $k,h,w\in \mathbb{N}$. We denote by $\mathsf{InfoFrame}(G,k,h,w)$ the set of info-frames $F=(f,d_f,E_f,U_f)$ of $G$, such that $f$ is bounded by $R_{h,w}$, and $\sizef(d_f,f)\leq k$.

%Let $C=(c,F_1,F_2)$ be an info-cutter of an info-frame $F=(f,d_f,E_f,U_f)$ and let $k\in \mathbb{N}$. We say that $C$ costs at most $k$, if $F_1,F_2\in \mathsf{InfoFrame}(G,k,h,w)$. 

%Let $\mathsf{Classifier}$ be an information function. We denote by $\mathsf{Info}(G,k,h,w)$ the set of $I'\in \mathsf{INF}$ such that there exists $F\in \mathsf{InfoFrame}(G,k,h,w)$ and a drawing $d$ of $F$ such that $\mathsf{Classifier}(F,d)=I'$.  

\begin{lemma}\label{lem:inductiveAl}
	Let $G=(V,E)$ be a graph, let $k,h,w\in \mathbb{N}$ and let $\mathsf{INF}$ be a universe. Let $\mathsf{Classifier}$ be a classifier, and let $A$ be a $\mathsf{Classifier}$-algorithm. Let $F\in \mathsf{InfoFrame}(G,k,h,w)$ that is not a leaf, and let $I\in \mathsf{INF}$. There exists a drawing $d$ of $F$ for which $\mathsf{dtw}(d,f)\leq k$ and $\mathsf{Classifier}(F,d)=I$, if and only if there exist an info-cutter $C=(c,F_1,F_2)$ of $F$ that costs at most $k$ and $I_1,I_2\in \mathsf{INF}$ such that the following conditions are satisfied:
	\begin{enumerate}
		\item There exists a drawing $d_1$ of $F_1$ such that $\mathsf{Classifier}(F_1,d_1)=I_1$ and $\mathsf{dtw}(d_1,f_1(c))\leq k$.
		\item There exists a drawing $d_2$ of $F_2$ such that $\mathsf{Classifier}(F_2,d_2))=I_2$ and $\mathsf{dtw}(d_2,f_2(c))\leq k$.
		\item $A(F,C,I_1,I_2)=I$.
	\end{enumerate}
\end{lemma}

\begin{proof}
	Let $F\in \mathsf{InfoFrame}(G,k,h,w)$ that is not a leaf and let $I\in \mathsf{INF}$. Let $d$ be a drawing of $F$ such that $\mathsf{dtw}(d,f)\leq k$ and $\mathsf{Classifier}(F,d)=I$. Let $({\cal T}=(V_T,E_T),\alpha: V_T\rightarrow \mathsf{Frames})$ be an $f$-frame-tree of $d$ with $\mathsf{dw}({\cal T},\alpha,d,f)\leq k$, where $\cal T$ is a rooted tree with root $v_r\in V_T$. Observe that $\alpha(v_r)=f$, and since $F$ is not a leaf, there exists a cutter $c$ of $f$ such that $\alpha(v_1)=f_1(c)$ and $\alpha(v_2)=f_1(c)$, where $v_1$ and $v_2$ are the children of $v_r$ in ${\cal T}$. Then, by Lemma \ref{splitter}, $\mathsf{splitter}(F,d,c)_c=(c,F_1,F_2)$ is an info-cutter of $F$, $d_1=\mathsf{splitter}(F,d,c)_{d_1}$ is a drawing of $F_1$, and $d_2=\mathsf{splitter}(F,d,c)_{d_2}$ is a drawing of $F_2$. Let $\mathsf{Classifier}(F_1,d_1)=I_1$ and $\mathsf{Classifier}(F_2,d_2)=I_2$. By Lemma \ref{lem:Glueto}, $\mathsf{Glue}(F,C,d_1,d_2)=d$, and since $A$ is a $\mathsf{Classifier}$-algorithm, and $\mathsf{Classifier}(F,d)=I$, then $A(F,C,I_1,I_2)=I$.
	
	Now, let $({\cal T}_1=(V_{T_1},E_{T_1}),\alpha_1: V_{T_1}\rightarrow \mathsf{Frames})$ where ${\cal T}_1$ is the subtree of ${\cal T}$ with root $v_1$, and for every $u\in V_{T_1}$, $\alpha_1(u)=\alpha(u)$. We claim that $({\cal T}_1,\alpha_1)$ is an $f_1(c)$-frame-tree of $d_1$ with $\mathsf{dw}({\cal T}_1,\alpha_1,d_1,f_1(c))\leq k$. It is easy to see that the conditions of Definition \ref{def:DrawnDecof} are satisfied, therefore $({\cal T}_1,\alpha)$ is an $f_1(c)$-frame-tree of $d_1$. Now, we argue that $\mathsf{dw}({\cal T}_1,\alpha_1,d_1,f_1(c))\leq k$. Let $v\in V_{T_1}$, we aim to prove that $\sizef(\alpha(v),d_1)\leq k$. Since $\mathsf{dw}({\cal T},\alpha,d,f)\leq k$, we get that $\sizef(\alpha(v),d)\leq k$. Then, by Observation \ref{obs:turningEq2}, we get that  $\sizef(\alpha(v),d_1)=\sizef(\alpha(v),d)\leq k$, and hence $\mathsf{dw}({\cal T}_1,\alpha_1,d_1,f_1(c))\leq k$, so $\mathsf{dtw}(d_1,f_1(c))\leq k$. Similarly,  let $({\cal T}_2=(V_{T_2},E_{T_2}),\alpha_2: V_{T_2}\rightarrow \mathsf{Frames})$ where ${\cal T}_2$ is the subtree of ${\cal T}$ with root $v_2$, and for every $u\in V_{T_2}$, $\alpha_2(u)=\alpha(u)$.  It can be show similarly to ${\cal T}_1$, that $\mathsf{dw}({\cal T}_2,\alpha_2,d_2,f_2(c))\leq k$. Therefore, $\mathsf{dtw}(d_2,f_2(c))\leq k$.

	%Then, by Lemma, $({\cal T}_1,\alpha_1\rightarrow \mathsf{Frames})$ is a drawn tree-frame decomposition of $(G*(F_1(C)),(F_1(C))$ with $dw({\cal T}_1,\alpha_1)\leq k$, and so $\mathsf{dtw}(d_1,f_1(C))\leq k$. Similarly, let $({\cal T_2}=(V_{T_2},E_{T_2}),\alpha_2: V_{T_1}\rightarrow \mathsf{Frames})$ where ${\cal T_2}$ is the subtree of ${\cal T}$ with root $v_2$, and for every $u\in V_{T_2}$, $\alpha_1(u)=\alpha(u)$. Then, by Lemma, $({\cal T_2},\alpha_2\rightarrow \mathsf{Frames})$ is a drawn tree-frame decomposition of $(G*(F_2(C)),(F_2(C))$ with $dw({\cal T_2},\alpha_1)\leq k$, and so $\mathsf{dtw}(d_2,f_2(C))\leq k$. 
	
	Now, we prove the opposite direction of the lemma. Let $C$ be an info-cutter of $F$ and let $I_1,I_2\in \mathsf{INF}$ such that the conditions of the lemma are satisfied. Let $d=\mathsf{Glue}(F,C,d_1,d_2)$. Notice that by Lemma \ref{lem:glue}, $d$ is a drawing of $F$. Now, since $\mathsf{Classifier}(F_1,d_1)=I_1$, $\mathsf{Classifier}(F_2,d_2)$ $=I_2, A(F,C,I_1,I_2)=I$ and $A$ is a $\mathsf{Classifier}$-algorithm, it follows that $\mathsf{Classifier}(F,d)=I$. Let $({\cal T}_1=(V_{T_1},E_{T_1}),\alpha_1: V_{T_1}\rightarrow \mathsf{Frames})$, where $v_1$ is the root of ${\cal T}_1$, be an $f_1(c)$-frame-tree of $d_1$ for which $\mathsf{dw}({\cal T}_1,\alpha_1,d_1,f_1(c))\leq k$. In addition, let $({\cal T}_2=(V_{T_2},E_{T_2}),\alpha_2: V_{T_2}\rightarrow \mathsf{Frames})$, where $v_2$ is the root of ${\cal T}_2$, be an $f_2(c)$-frame-tree of $d_2$ for which $\mathsf{dw}({\cal T}_2,\alpha_2,d_2,f_2(c))\leq k$. Let ${\cal T}=(V_T,E_T)$ be a tree where $V_T=V_{T_1}\cup V_{T_2}\cup \{v_r\}$ and $E_T=E_{T_1}\cup E_{T_2}\cup \{ \{v_r,v_1\},\{v_r,v_2\}\}$. Furthermore, we define $\alpha: V_T\rightarrow \mathsf{Frames}$ as follows. First, we set $\alpha(v_r)=f$. Then, for every $v\in V_{T_1}$ we set $\alpha(v)=\alpha_1(v)$, and for every $v\in V_{T_2}$ we set $\alpha(v)=\alpha_2(v)$. We claim that $({\cal T}=(V_T,E_T),\alpha: V_T\rightarrow \mathsf{Frames})$ is an $f$-frame-tree of $(d,f)$ with $\mathsf{dw}({\cal T},\alpha,d,f)\leq k$. It is easy to see that $({\cal T},\alpha)$ is an $f$-frame-tree of $d$. Now, we argue that $\mathsf{dw}({\cal T},\alpha,d,f)\leq k$. Let $v\in V_T$. 
	If $v=v_r$, then since $\alpha(v_r)=f$ and $F$ is a frame that costs at most $k$, we get that $\sizef(f,d)\leq k$.
	Now, assume that $v\in V_{T_1}$. Since $\mathsf{dw}({\cal T}_1,\alpha_1,d_1,f_1(c))\leq k$, we get that $\sizef(\alpha_1(v),d_1)\leq k$. By Observation \ref{obs:turningEq1}, we get that $\sizef(\alpha(v),d)=\sizef(\alpha(v),d_1)$. Therefore, $\sizef(\alpha(v),d)\leq k$. If $v\neq v_r$ and $v\notin V_{T_1}$, then $v\in V_{T_2}$. Since $\mathsf{dw}({\cal T}_2,\alpha_2,d_2,f_2(c))\leq k$, we get that $\sizef(\alpha_2(v),d_2)\leq k$. Now, by Observation \ref{obs:turningEq1}, we get that $\sizef(\alpha(v),d)=\sizef(\alpha(v),d_2)$. Therefore, $\sizef(\alpha(v),d)\leq k$. Therefore, we get that $\mathsf{dw}({\cal T},\alpha,d,f)\leq k$, so $\mathsf{dtw}(d,f)\leq k$, and we are done.
\end{proof}

%!TEX root =Main-Movement.tex

\subsection{The Algorithm}

In this subsection we present the algorithm of our scheme. First, we describe the type of problems we aim to solve. To this end, we present some terms and definitions.
Let $G$ be a graph, and let $h,w,k\in\mathbb{N}$. We denote by $\mathsf{InitF}_{h,w}$ the info-frame $(R_{h,w},\emptyset,\emptyset,\emptyset,\emptyset) $. Observe that, $d$ is a polyline grid drawing of $\mathsf{InitF}_{h,w}$ if and only if $d$ is a polyline grid drawing of the graph $G$ strictly bounded by $R_{h,w}$. Now, we define the the problems our algorithm is able to solve.

\begin{definition}[{\bf$(\mathsf{Classifier},h,w,k,{\cal I}_{\mathsf{yes}})$-problem}] \label{def:probSc}
Let $\Pi$ be a decision problem, let $\mathsf{INF}$ be a universe, let $\mathsf{Classifier}$ be a classifier, let ${\cal I}_{\mathsf{yes}}\subseteq \mathsf{INF}$, and let $h,w,k\in\mathbb{N}$. We say that $\Pi$ is a {\em $(\mathsf{Classifier},h,w,k,{\cal I}_{\mathsf{yes}})$-problem}, if (i) for every instance $G=(V,E)$ of $\Pi$, $G$ is a connected graph, and (ii) $G$ is a yes-instance of $\Pi$, if and only if there exists a drawing $d$ of $\mathsf{InitF}_{h,w}$ with $\mathsf{dtw}(d)\leq k$ and $\mathsf{Classifier}(\mathsf{InitF}_{h,w},d)=I'$ for some $I'\in {\cal I}_{\mathsf{yes}}$.
\end{definition}

%Let $\Pi$ be a decision problem, let $\mathsf{Classifier}$ be a classifier with a universe $\mathsf{INF}$, let ${\cal I}_{\mathsf{yes}}\subseteq \mathsf{INF}$, and let $h,w,k\in\mathbb{N}$. We say that $\Pi$ is a {\em $(\mathsf{Classifier},h,w,k,{\cal I}_{\mathsf{yes}})$-drawing decision problem}, if for every instance $G=(V,E)$ of $\Pi$, $G$ is a yes-instance, if and only if, there exists a drawing $d$ of $\mathsf{InitF}_{h,w}$, with $\mathsf{dtw}(d)\leq k$ and $\mathsf{Classifier}(\mathsf{InitF}_{h,w},d)=I'$, for an $I'\in {\cal I}_{\mathsf{yes}}$, where 

In addition to the $\mathsf{Classifier}$-algorithm (see Definition \ref{def:probInf}), we ask from the user also to provide us an algorithm that ``solves'' all the info-frames that are leaves:

%That is, an algorithm that for every info-frame $F$ that is a leaf, and for every $I'\in \mathsf{INF}$, the algorithm returns ``yes'', if and only if, there exists a drawing $d$ of $F$, such that $\mathsf{Classifier}(F,d)=I'$. We call these kind of algorithm, a {\em $\mathsf{LeafSolver}}. 

\begin{definition}[{\bf$\mathsf{Classifier}$-Leaf Solver}] \label{def:probLeGen}
%	Let $\Pi$ be a $(\mathsf{Classifier},h,w,k,{\cal I}_{\mathsf{yes}})$-drawing decision problem for a classifier $\mathsf{Classifier}$ with a universe $\mathsf{INF}$, ${\cal I}_{\mathsf{yes}}\subseteq \mathsf{INF}$, and $h,w,k\in\mathbb{N}$. 

Let $\mathsf{INF}$ be a universe and let $\mathsf{Classifier}$ be a classifier. An algorithm $L$ is a {\em $\mathsf{Classifier}$-leaf solver}, if for every info-frame $F$ that is a leaf, and for every $I'\in \mathsf{INF}$, $L(F,I')$ returns ``yes'' if and only if there exists a drawing $d$ of $F$ such that $\mathsf{Classifier}(F,d)=I'$.
\end{definition}

%When $\mathsf{Classifier}$ is clear from the context, we refer to $\mathsf{Classifier}$-leaf solver as leaf solver.

Next, we define the term {\em area} of an info-frame:

% Given an info-frame $F=(f,d_f,E_f,U_f,\mathsf{V^*Dir}_f)$, the area of $F$ is the area of $f$. We measure the area of $f$ by the number of unite length squares inside $f$:

\begin{definition}[{\bf Area of an Info-Frame}] \label{def:area}
	Let $F=(f,d_f,E_f,U_f,\mathsf{V^*Dir}_f)$ be an info-frame. The {\em area} of $F$ is the number of unite length squares inside $f$. 
\end{definition}

Now, we would like to give an upper bound on the number of info-frames. This bound will be useful for the time analysis of the algorithm we present in this section.  

\begin{lemma}\label{lem:numFr}
Let $G$ be a connected graph and let $k,h,w\in \mathbb{N}$. Then, $|\mathsf{InfoFrame}(G,k,h,w)|$ $\leq \OO((k\cdot h\cdot w \cdot n)^{\OO(k)}\cdot 2^{\OO(\Delta\cdot k)})$, where $\Delta$ is the maximum degree of $G$ and $|V|=n$.
\end{lemma}

\begin{proof}
Consider an info-frame $F=(f,d_f,E_f,U_f,\mathsf{V^*Dir}_f)\in \mathsf{InfoFrame}(G,k,h,w)$. Every frame $f$ has at most $k$ vertices, therefore the number of frames is bounded by $\sum_{i=1}^{k}{h\cdot w \choose i}\leq k\cdot {h\cdot w \choose k}\leq k\cdot (h\cdot w)^{k}$. Now we give an upper bound for the number of different $d_f$ and $E_f$. First, we choose at most $k$ vertices from $V$ and from $V^*$. Observe that, since the labeling of vertices from the set $V^*$ starts by $1$ and is continuous, we choose edges from $E$, which correspond to vertices in $V^*$. We can choose an edge multiple times, where the $i$-th time we choose an edge $\{u,v\}$ corresponds to choosing the vertex $uv_i$. We have $n$ vertices and at most $n^2$ edges, where edges can be chosen multiple times. So we have at most $\sum_{i=1}^{k}{(n+n^2)^i}\leq k\cdot (2n^2)^{k}$ different options. For each vertex we chose, we choose a point in $f$ to place the vertex. Vertices from $V$ are placed on points in $\gp(f)$ and vertices from $V^*$ are placed on points in $\gi(f)$. Observe that $\gp(f)\subseteq \gi(f)\subseteq \gis(\fin)$, and $|\gis(\fin)|\leq (h\cdot w)^2$. As chose at most $k$ vertices we have at most ${(h\cdot w)^2 \choose k}\leq ((h\cdot w)^2)^{k}=(h\cdot w)^{2k}$ different options to place the vertices. 

Now, at the worst case there is an edge in $G$ between every two vertices we chose. We need to guess, for every two vertices, if they have an edge between them on $f$, outside $f$ or inside $f$ (and therefore in $E_f$). Observe that, if they do have an edge on the frame, then we have at most two options to draw this edge, since $f$ is a simple cycle. Therefore, for every two vertices, we have at most four options: two options for placing the edge on the frame if they have an edge on $f$, and if they do not, we have two options: the edge is inside or outside $f$. The edges we chose to be inside $f$ are added to $E_f$. So, we have at most $4^{k^2}$ options for the drawing of the edges. Thus, we have at most $(h\cdot w)^{2k}4^{k^2}$ different options for $d_f$, after we chose the set of vertices on $f$. In addition, recall that $E_f$ is the set of edges drawn strictly inside $f$, except for at least one of their endpoints. Notice that, this boundary covers also the choice of edges in $E_f$ drawn strictly inside $f$, except for both of their endpoint. 

Now, we choose the set of vertices drawn strictly inside $f$, that is, the set $U_f$. Observe that $V(d_f)\cap V$ is a separator in $G$. Therefore, for every connected component $\mathsf{cc}$ in $G\setminus (V(d_f)\cap V)$ we have exactly two options: either all of the vertices of $\mathsf{cc}$ are drawn strictly inside $f$, all of the vertices of $\mathsf{cc}$ are drawn strictly outside $f$. Since the maximum degree is $\Delta$, and we chose at most $k$ vertices from $V$ to place on $f$, we have at most $\Delta\cdot k$ connected components in $G\setminus (V(d_f)\cap V)$, so we have at most $2^{\Delta\cdot k}$ different options for $U_f$. Observe that, once we choose $U_f$, the set of edges drawn strictly inside $f$ except for exactly one endpoint, is fixed, thus, we guessed $E_f$ as well. 

Now, we guess the directions of vertices from $V^*$ we placed on points in\\ $\gis($ $\fin)\setminus \gps(\fin)$. We have at most $k$ such vertices, and since every direction consists with a choice of two grid points, we have at most $(h\cdot w)^2$ options for each vertex, so at most $(h\cdot w)^{2k}$ different guesses for $\mathsf{V^*Dir}_f$. 

In conclusion, we have at most $ k\cdot (h\cdot w)^{k}k\cdot (2n^2)^{k}(h\cdot w)^{2k}4^{k^2}2^{\Delta\cdot k}(h\cdot w)^2=\OO((k\cdot h\cdot w \cdot n)^{\OO(k)}\cdot 2^{\OO(\Delta\cdot k)})$ different info-frames in $\mathsf{InfoFrame}(G,k,h,w)$.
\end{proof}

We are now ready to describe our algorithm defined in Algorithm \ref{alg:schemeiter}. The algorithm is able to solve any $(\mathsf{Classifier},h,w,k,{\cal I}_{\mathsf{yes}})$-problem $\Pi$ (see Definition \ref{def:probSc}), for a given classifier $\mathsf{Classifier}$ with a universe $\mathsf{INF}$ (see Definition \ref{def:probInf}), a subset ${\cal I}_{\mathsf{yes}}\subseteq \mathsf{INF}$, and $h,w,k\in\mathbb{N}$. Let $\Pi$ be such a problem. The input of the algorithm is:

\begin{itemize}
	\item $G,\mathsf{INF},{\cal I}_{\mathsf{yes}},h,w,k$.
	\item A $\mathsf{Classifier}$-leaf solver $L$ (see Definition \ref{def:probLeGen}).
	\item A $\mathsf{Classifier}$-algorithm $A$ (see Definition \ref{def:probsolv}).
\end{itemize}

Now, we describe the steps of Algorithm \ref{alg:schemeiter}.

\smallskip\noindent{\bf Line \ref{alg3:Line2}: Creating the $\mathsf{FramesTable}$.} The $\mathsf{FramesTable}$ is where we intend to store for every $F\in\mathsf{InfoFrame}(G,k,h,w)$ and for every $I\in \mathsf{INF}$, $\mathsf{True}$ if there exists a drawing $d$ of $F$ with $\mathsf{dtw.}(d,f)\leq k$ such that $\mathsf{Classifier}(F,d)=I$; otherwise, we store $\mathsf{False}$.

\smallskip\noindent{\bf Line \ref{alg3:Line33}: Solving the Leaves Frames with $L$.} We use $L$ to solve $(F,I)$ for every $F\in\mathsf{InfoFrame}(G,k,h,w)$ that is a leaf, and for every $I\in \mathsf{INF}$.

\smallskip\noindent{\bf Lines \ref{alg3:Line4}-\ref{alg3:Line7}: Iterating Over Frames By Increasing Order of Their Area.} We solve every $F\in\mathsf{InfoFrame}(G,k,h,w)$, such that $F$ is not a leaf, with area $\mathsf{Area}$, and with every $I\in \mathsf{INF}$, starting from $\mathsf{Area}=1$ to $\mathsf{Area}=h\cdot w$. Observe that, this covers every info-frame in $\mathsf{InfoFrame}(G,k,h,w)$. At every iteration of the while loop at Line \ref{alg3:Line5}, we use only info-frames we already solved, that is, info-frames with strictly less area than $\mathsf{Area}$, in order to compute the value of $(F,I)$.   

\smallskip\noindent{\bf Lines \ref{alg3:Line8}-\ref{alg3:Line14}: Invoking Lemma \ref{lem:inductiveAl} to Solve $(F,I)$ By Using Values of Info-Frames With Less Area.} We iterate over every info-cutter $C=(c,F_1,F_2)$ of $F$ and every $I_1,I_2\in \mathsf{INF}$. Observe that, every frame has at least size one, and for every frame $f$ and a cutter $c$ of $f$, the size of $f$ is exactly the sum of the sizes of $f_1(c)$ and $f_2(c)$. Therefore, $F_1$ and $F_2$ are both with size strictly smaller than the size of $F$, so we already solved $(F_1,I_1)$ and $(F_2,I_2)$. Now, if the values we stored in $\mathsf{FramesTable}$ for $(F_1,I_1)$ and $(F_2,I_2)$ are both $\mathsf{True}$, and $A(F,C,I_1,I_2)=I$, then we store $\mathsf{True}$ for $(F,I)$ in $\mathsf{FramesTable}$. If, these conditions are not satisfied for any guess of $C$, $I_1$ and $I_2$, we store $\mathsf{False}$ for $(F,I)$ in $\mathsf{FramesTable}$. The correctness of the value we store is valid due to the correctness of the values we have allready stored in $\mathsf{FramesTable}$, and due to Lemma \ref{lem:inductiveAl}.

\smallskip\noindent{\bf Lines \ref{alg33:Line5}-\ref{alg3:Line20}: Returning Answer For $G$.} We iterate over every $I\in {\cal I}_{\mathsf{yes}}$, and return ``yes-instance'' if there exists such $I$ where we stored $\mathsf{True}$ in $\mathsf{FramesTable}$ for $(F_{\mathsf{Init}},I)$; otherwise, we return ``no-instance''. The correctness of the value the algorithm returns is valid due to the correctness of values we stored in $\mathsf{FramesTable}$ and since $\Pi$ is $(\mathsf{Classifier},h,w,k,{\cal I}_{\mathsf{yes}})$-problem.

Now, we aim to prove the correctness of Algorithm \ref{alg:schemeiter}. First, we show that in the end of every iteration $i$ of the while loop at Line \ref{alg3:Line5}, we compute the values of info-frames, with area $i$ or less, correctly. That is, at the end of the $i$-th iteration, the value of every pair of $(F,I)$ where $F$ is an info-frame with area $i$ or less and $I\in \mathsf{INF}$, is $\mathsf{True}$ if and only if there exists a drawing $d$ of $F$ with $\mathsf{dtw.}(d,f)\leq k$ such that $\mathsf{Classifier}(F,d)=I$. For this purpose, for every $0\leq i\leq h\cdot w$, we denote $\mathsf{FramesTable}$ obtained at the end of the $i$-th iteration, by $\mathsf{FramesTable}_i$.

\begin{lemma}\label{lem:algIter}
Let $\Pi$ be a $(\mathsf{Classifier},h,w,k,{\cal I}_{\mathsf{yes}})$-problem, for a classifier $\mathsf{Classifier}$ with a universe $\mathsf{INF}$, a subset ${\cal I}_{\mathsf{yes}}\subseteq \mathsf{INF}$ and $h,w,k\in\mathbb{N}$. Let $A$ be a $\mathsf{Classifier}$-algorithm and let $\mathsf{LeavesGen}$ be a $\mathsf{LeafSolver}$ with. Let $0\leq i\leq h\cdot w$, $F\in \mathsf{InfoFrame}(G,k,h,w)$ with area $i$ or less, and let $I\in \mathsf{INF}$. Then $(F,I,\mathsf{True}) \in \mathsf{FramesTable}_i$ if and only if there exists a drawing $d$ of $F$ with $\mathsf{dtw.}(d,f)\leq k$ such that $\mathsf{Classifier}(F,d)=I$.
\end{lemma}

\begin{proof}
We prove this claim by induction on $i$. For $i=0$, that is, before the algorithm enters the while loop at Line \ref{alg3:Line5}, for every $F\in \mathsf{InfoFrame}(G,k,h,w)$ that is a leaf, and for every $I\in \mathsf{INF}$, there exists a drawing $d$ of $F$ with $\mathsf{dtw.}(d,f)\leq k$ and $\mathsf{Classifier}(F,d)=I$, if and only if $(F,I,\mathsf{True})\in\mathsf{FramesTable}_0$. Observe that, every info-frame $F$ with area $0$, $f$ contains $0$ inner points, is a leaf, and therefore the basic case is valid. Now, let $0<i\leq h\cdot w$. By the inductive hypothesis, for every $F\in \mathsf{InfoFrame}(G,k,h,w)$ with area at most $i'<i$, and for every $I\in \mathsf{INF}$, there exists a drawing $d$ of $F$ with $\mathsf{dtw.}(d,f)\leq k$ and $\mathsf{Classifier}(F,d)=I$, if and only if $(F,I,\mathsf{True})\in \mathsf{FramesTable}_{i-1}\subseteq \mathsf{FramesTable}_{i}$. Let $F\in \mathsf{InfoFrame}(G,k,h,w)$ with area $i$ and let $I\in \mathsf{INF}$. If $F$ is a leaf, $(F,I)$ was computed by $L$, then the claim is valid due to the correctness of $L$. Assume that $F$ is not a leaf. By Lemma \ref{lem:inductiveAl}, there exists a drawing $d$ of $F$ with $\mathsf{dtw.}(d,f)\leq k$ and $\mathsf{Classifier}(F,d)=I$, if and only if there exist an info-cutter $C=(c,F_1,F_2)$ of $F$ of size at most $k$, and $I_1,I_2\in \mathsf{INF}$, such that the following conditions are satisfied:
\begin{enumerate}
	\item There exists a drawing $d_1$ of $F_1$ such that $\mathsf{Classifier}(d_1,f_1(c))=I_1$ and $\mathsf{dtw}(d_1,f_1)\leq k$.
	\item There exists a drawing $d_2$ of $F_2$ such that $\mathsf{Classifier}(d_2,f_2(c))=I_2$ and $\mathsf{dtw}(d_2,f_2)\leq k$.
	\item $A(F,C,I_1,I_2)=I$.
\end{enumerate} 

Since each of $f_1(c)$ and $f_2(c)$ has strictly less area than $i$, by the inductive hypothesis, $(F_1,I_1,\mathsf{True}),$ $ (F_2,I_2,\mathsf{True})\in \mathsf{FramesTable}_{i-1}$. Therefore, the conditions of Lemma \ref{lem:inductiveAl}, are satisfied, if and only if, there exist an info-cutter $C=(c,F_1,F_2)$ of $F$, of size at most $k$ and $I_1,I_2\in \mathsf{INF}$, such that $(F_1,I_1,\mathsf{True}),$ $ (F_2,I_2,\mathsf{True})\in \mathsf{FramesTable}_{i-1}$, and $A(F,C,I_1,I_2)=I$. These conditions are satisfied if and only if in the iteration of the loop at Line \ref{alg3:Line6} where $F$ and $I$ are chosen, there is an iteration at Line \ref{alg3:Line8} where the condition in Line \ref{alg3:Line9} is satisfied, if and only if $(F,I,\mathsf{True})\in \mathsf{FramesTable}_{i}$. This ends the proof of the inductive claim. 
\end{proof}

Next, by Invoking Lemma \ref{lem:algIter}, we prove the correctness of Algorithm \ref{alg:schemeiter}.

\begin{lemma}\label{lem:AlgCorr}
	Let $\Pi$ be a $(\mathsf{Classifier},h,w,k,{\cal I}_{\mathsf{yes}})$-problem, for a classifier $\mathsf{Classifier}$ with a universe $\mathsf{INF}$, a subset ${\cal I}_{\mathsf{yes}}\subseteq \mathsf{INF}$ and $h,w,k\in\mathbb{N}$. Let $A$ be a $\mathsf{Classifier}$-algorithm, and let $L$ be a $\mathsf{Classifier}$-leaf solver. Then, Algorithm \ref{alg:schemeiter}, given as input an instance $G=(V,E)$ of $\Pi$, returns ``yes'' if and only if $G$ is a yes-instance of $\Pi$.
\end{lemma}

\begin{proof}
Let $G=(V,E)\in \Pi$, and assume that $G$ is a yes-instance. Then, since $\Pi$ is a $(\mathsf{Classifier},h,$ $w,k,{\cal I}_{\mathsf{yes}})$-problem, there exist a drawing $d$ of $\mathsf{InitF}_{h,w}$ with $\mathsf{dtw}(d)\leq k$ and $I\in {\cal I}_{\mathsf{yes}}$ such that $\mathsf{Classifier}(F_{\mathsf{Initi}(h,w)},d)=I$. Notice that $\mathsf{InitF}_{h,w}$ is with area $h\cdot w$. So, by Lemma \ref{lem:algIter}, by the end of the last iteration of the while loop at Line \ref{alg3:Line5}, we get that $(F_{\mathsf{Initi}(h,w)},I,\mathsf{True})\in \mathsf{FramesTable}_{h\cdot w}$. Therefore, the condition in Line \ref{alg33:Line5} is satisfied, so the algorithm returns ``yes-instance''. If $G$ is a no-instance, since $\Pi$ is a $(\mathsf{Classifier},h,w,k,{\cal I}_{\mathsf{yes}})$-problem, for every $I\in {\cal I}_{\mathsf{yes}}$, there is no drawing $d$ of $\mathsf{InitF}_{h,w}$ with $\mathsf{dtw}(d)\leq k$ such that $\mathsf{Classifier}(F_{\mathsf{Initi}(h,w)},d)=I$.  Again, by Lemma \ref{lem:algIter}, by the end of the last iteration of the while loop at Line \ref{alg3:Line5}, we get that $(F_{\mathsf{Initi}(h,w)},I,\mathsf{False})\in \mathsf{FramesTable}_{h\cdot w}$ for every $I\in {\cal I}_{\mathsf{yes}}$. Therefore, the condition in Line \ref{alg33:Line5} is not satisfied. Thus, the algorithm returns ``no-instance''.
\end{proof}

In the next lemma, we analyze the runtime of Algorithm \ref{alg:schemeiter}.

\begin{lemma}\label{lem:AlgTime}
	Let $\Pi$ be a $(\mathsf{Classifier},h,w,k,{\cal I}_{\mathsf{yes}})$-problem, for a classifier $\mathsf{Classifier}$ with a universe $\mathsf{INF}$, a subset ${\cal I}_{\mathsf{yes}}\subseteq \mathsf{INF}$ and $h,w,k\in\mathbb{N}$. Let $A$ be a $\mathsf{Classifier}$-algorithm, with runtime $\mathsf{Time}A(G,k)$, and let $L$ be a $\mathsf{Classifier}$-leaf solver with runtime $\mathsf{Time}L(G,k))$. Then, Algorithm \ref{alg:schemeiter} given an input $G$, runs in time $\OO ((k\cdot h\cdot w\cdot n)^{\OO (k)}\cdot 2^{\OO (\Delta\cdot k)}\cdot |\mathsf{INF}|^{\OO(1)}\cdot \mathsf{Time}A(G,k)) +\OO ((k\cdot h\cdot w\cdot n)^{\OO (k)}\cdot |\mathsf{INF}|\cdot \mathsf{Time}L(G,k))$, where $|V|=n$ and $\Delta$ is the maximum degree of $G$.   
\end{lemma}

\begin{proof}
	Recall, that from Lemma \ref{lem:numFr}, we have at most $\OO((k\cdot h\cdot w \cdot n)^{\OO(k)}\cdot 2^{\OO(\Delta\cdot k)})$ info-frames in $\mathsf{InfoFrame}(G,k,h,w)$. Therefore, we have at most that number of info-frames which are leaves, so the runtime of Line \ref{alg3:Line33} is $\OO ((k\cdot h\cdot w\cdot n)^{\OO (k)}\cdot |\mathsf{INF}|\cdot \mathsf{Time}L(G,k))$. We continue with the runtime of the rest of the algorithm. Observe that, every info-frame is chosen in exactly one iteration of the iterations the while loop at Line \ref{alg3:Line5}. For any choice of an info-frame and an $I\in \mathsf{INF}$, we have at most $|\mathsf{InfoFrame}(G,k,h,w)|^2|\mathsf{INF}|^2$ choices for an info-cutter and $I_1,I_2\in \mathsf{INF}$. For any of these choices, the runtime of the $\mathsf{Classifier}$-solver $A$ is $\mathsf{Time}A(G,k)$. The rest of the operations in the while loop at Line \ref{alg3:Line5} are done in polynomial time. Therefore, the runtime of the while loop at Line \ref{alg3:Line5} is $\OO((k\cdot h\cdot w \cdot n)^{\OO(k)}\cdot 2^{\OO(\Delta\cdot k)})\cdot |\mathsf{INF}|\cdot (\OO((k\cdot h\cdot w \cdot n)^{\OO(k)}\cdot 2^{\OO(\Delta\cdot k)}))^2\cdot |\mathsf{INF}|^2\cdot \mathsf{Time}A(G,k)=\OO((k\cdot h\cdot w \cdot n)^{\OO(k)}\cdot |\mathsf{INF}|^{\OO(1)}\cdot \mathsf{Time}A(G,k))$. Therefore, the runtime of Algorithm \ref{alg:schemeiter} is $\OO ((k\cdot h\cdot w\cdot n)^{\OO (k)}\cdot 2^{\OO (\Delta\cdot k)}\cdot |\mathsf{INF}|^{\OO(1)}\cdot \mathsf{Time}A(G,k)) +\OO ((k\cdot h\cdot w\cdot n)^{\OO (k)}\cdot |\mathsf{INF}|\cdot \mathsf{Time}L(G,k))$. This ends the proof.
\end{proof}

Now, by invoking Lemmas \ref{lem:AlgCorr} and \ref{lem:AlgTime}, we prove the following theorem.

\begin{theorem}\label{the:AlgSch}
	Let $\Pi$ be a $(\mathsf{Classifier},h,w,k,{\cal I}_{\mathsf{yes}})$-drawing decision problem, for a classifier $\mathsf{Classifier}$ with a universe $\mathsf{INF}$, a subset ${\cal I}_{\mathsf{yes}}\subseteq \mathsf{INF}$ and $h,w,k\in\mathbb{N}$. Let $A$ be a $\mathsf{Classifier}$-algorithm, with runtime $\mathsf{Time}A(G,k)$, and let $L$ be a $\mathsf{Classifier}$-leaf solver with runtime $\mathsf{Time}L(G,$ $k))$. Then, there exists an algorithm that gets as input an instance $G=(V,E)$ of $\Pi$, runs in time $\OO ((k\cdot h\cdot w\cdot n)^{\OO (k)}\cdot 2^{\OO (\Delta\cdot k)}\cdot |\mathsf{INF}|^{\OO(1)}\cdot \mathsf{Time}A(G,k)) +\OO ((k\cdot h\cdot w\cdot n)^{\OO (k)}\cdot |\mathsf{INF}|\cdot \mathsf{Time}L(G,k))$ and returns ``yes'' if and only if $G$ is a yes-instance of $\Pi$, where $|V|=n$ and $\Delta$ is the maximum degree of $G$.   
\end{theorem}

%\begin{algorithm}[!t]
	%\SetKwInOut{Input}{Input}
	%\SetKwInOut{Output}{Output}
	%\medskip
	%{\textbf{function} $\mathsf{ProbSolverScheme}$}$(\langle A,\mathsf{LeavesGenerator},I,{\cal I}_{\mathsf{yes}},h,w,k\rangle)$\;
	%$k\gets 4$\; 
	%\While{$k\leq k$}
	%{\label{alg3:Line3}
	%	$\mathsf{FramesTable}\gets \mathsf{ProbSolverSchemeIter}(A,\mathsf{LeavesGenerator},I,h,w,k)$\;
	%	\If{There exists $I\in {\cal I}_{\mathsf{yes}}$ such that $\mathsf{FramesTable}(F_\mathsf{Init},I)=\mathsf{True}$}
	%	{\label{alg33:Line5}
		%	\Return ``yes-instance''\;
	%	}
	%	$k\gets k+1$\;
%	}
%	\Return ``no-instance''\;
%	\caption{$\mathsf{ProbSolverScheme}$}
%	\label{alg:scheme}
%\end{algorithm}

\begin{algorithm}[!t]
	\SetKwInOut{Input}{Input}
	\SetKwInOut{Output}{Output}
	\medskip
	{\textbf{function} $\mathsf{ProbSolverScheme}$}$(\langle G,\mathsf{INF},{\cal I}_{\mathsf{yes}},h,w,k,L,A\rangle)$\;
	$\mathsf{FramesTable}\gets \emptyset$\; \label{alg3:Line2}
	$\mathsf{FramesTable}\gets L(G,k,h,w,\mathsf{INF})$\; \label{alg3:Line33}
	$\mathsf{Area}\gets 1$\; \label{alg3:Line4}
	\While{$\mathsf{Area}\leq h\cdot w$}
	{\label{alg3:Line5}
		\For{Every $F\in \mathsf{InfoFrame}(G,k,h,w)$ that is not a leaf, with area $\mathsf{Area}$ and for every $I'\in \mathsf{INF}$}
		{\label{alg3:Line6}
			$\mathsf{Found}=\mathsf{False}$\; \label{alg3:Line7}
			\For{Every info-cutter $C$ of $F$ and $I_1,I_2\in \mathsf{INF}$ such that $\mathsf{FramesTable}(F_1,I_1)=\mathsf{True}$ and $\mathsf{FramesTable}(F_1,I_1)=\mathsf{True}$} 
			{\label{alg3:Line8}
				\If{$A(F,C,I_1,I_2)=I'$}
				{\label{alg3:Line9}
					$\mathsf{Found}=\mathsf{True}$\;
				}
			}
			$\mathsf{FramesTable}\gets \mathsf{FramesTable}\cup \{((F,I'),\mathsf{Found})\}$\;
		} \label{alg3:Line14}

$\mathsf{Area}\gets \mathsf{Area}+1$\;
}
	
	\If{There exists $I\in {\cal I}_{\mathsf{yes}}$ such that $\mathsf{FramesTable}(F_\mathsf{Init},I)=\mathsf{True}$}
{\label{alg33:Line5}
	\Return ``yes-instance''\;
}

\Return ``no-instance''\; \label{alg3:Line20}
	\caption{$\mathsf{ProbSolverScheme}$}
	\label{alg:schemeiter}
\end{algorithm}

%!TEX root =Main-Movement.tex

\section{Examples of Using The Scheme}\label{sec:ExampleScheme3}

In this section, we present several examples of applications of our scheme, developing algorithms for different graph drawing problems. First, remind that, one of the things we should provide in order to use the scheme, is a leaf solver. That is, an algorithm $L$, which solves the cases of info-frames which are leaves. Recall, that an info-frame $F=(f,d_f,E_f,U_f,\mathsf{V^*Dir}_f)$ is a leaf if there are no grid points strictly inside $f$. Therefore, if $U_f\neq \emptyset$ we can conclude that there are no drawings of $F$, since every vertex in $U_f$ must be drawn on a grid point strictly inside $f$. Similarly, for every $\{uv_i,uv_{i+1}\}\in E_f$ there is at most one possible way to draw, since there are no grid points inside $f$, and edges might bend only at grid points. Thus, every $\{uv_i,uv_{i+1}\}\in E_f$ must be drawn as a straight line between $uv_i$ and $uv_{i+1}$. Therefore, for $F$ we have at most one drawing. Observe that, we can get this drawing in $\OO(k)$ runtime. This observation is very useful to construct a leaf solver, as we will see in the examples introduced in this section. 

\begin{observation} \label{obs:oneDre}
Let $G$ be a graph, let $F$ be an info-frame which is a leaf. Then, there is at most one drawing of $F$. Moreover, it is possible in runtime $\OO(k)$, to construct the drawing of $F$, if such exists, or conclude that there is no such a drawing. 
\end{observation} 

\subsection{Algorithm for The Grid Recognition Problem}

%In this subsection we show an example of using our scheme.
In this subsection, we use Theorem \ref{the:AlgSch} to design an algorithm for the {\sc Grid Recognition} problem (see Definition~\ref{def:gridRec}) parameterized by $k$ with runtime $n^{\OO (k)}$ where $k$, given as input, bounds the drawn treewidth of the sought realization (if one exists). In turn, this will also yield a runtime of $n^{\OO(\sqrt{n)}}$. Recall that in the {\sc Grid Recognition} problem parameterized by  $k$, given a graph $G=(V,E)$, the goal is to determine whether $G$ has a grid drawing of drawn treewidth at most $k$. We assume that $G$ is connected; otherwise, we apply the algorithm on each of the different connected components separately. 

%Additionally, recall that a drawing $d$ is a grid drawing if for every $u\in V(d)$, $d(u)\in \mathbb{N}_0\times \mathbb{N}_0$ and for every $\{u,v\} \in E(d)$, $|\fr(u)-\fr(v)|+|\fc(u)-\fc(v)|=1$. 

Observe that for every grid drawing $d$ of a graph $G$, a frame $f$ and an edge $\{u,v\}$ of $G$, $\{u,v\}$ does not have no turning points in $f$ with respect to $d$ except for, possibly $d(u)$ or $d(v)$. Therefore, we conclude that any info-frame with vertices from $V^*$ is not useful for us in order to construct a grid drawing. In particular, we will only consider info-frames $F=(f,d_f,E_f,U_f,\mathsf{V^*Dir}_f)$ where $d_f(V)\cap V^*=\emptyset$.   
Let $\mathsf{INF}=\{0,1\}$.

We define a classifier, $\mathsf{GridClassifier}$, as follows:

\begin{definition}[{\bf $\mathsf{GridClassifier}$}]
Let $G$ be a graph, let $F$ be an info-frame and let $d$ be a drawing of $F$. Then $\mathsf{GridClassifier}(F,d)=1$ if $d$ is a grid drawing; otherwise, $\mathsf{GridClassifier}(F,d)=0$.
\end{definition}

 %We define the function $\mathsf{Classifier}$ as follows. For every info-frame $F$, and every drawing $d$ of $F$, $\mathsf{Classifier}(F,d)=1$ if $d$ is a grid drawing; Otherwise $\mathsf{Classifier}(F,d)=0$. 
 
Next, we define an algorithm called $A_{\mathsf{Grid}}$: 

\begin{definition}[{\bf $A_{\mathsf{Grid}}$}]
	Let $G$ be a graph, let $F$ be an info-frame, let $C=(c,F_1,F_2)$ be an info-cutter of $F$ and let $I_1,I_2\in \mathsf{INF}$. Then, $A_{\mathsf{Grid}}(F,C,I_1,I_2)=1$ if $I_1=I_2=1$; otherwise, $A_{\mathsf{Grid}}(F,C,I_1,I_2)=0$. 
\end{definition}

We show that $A_{\mathsf{Grid}}$ is a $\mathsf{GridClassifier}$-algorithm: by showing that the conditions of Definition \ref{def:probsolv} are satisfied:

\begin{lemma}\label{lem:A}
$A_{\mathsf{Grid}}$ is a $\mathsf{GridClassifier}$-algorithm.
\end{lemma}

\begin{proof}
We show that the conditions of Definition \ref{def:probsolv} are satisfied. Condition \ref{def:probSolvCon1} is trivially satisfied by the definition of the algorithm. Now, let $(F,C=(c,F_1,F_2),I_1,I_2)$, where $F_1=(f_1(c),d_{f_1},E_{f_1},U_{f_1}, \mathsf{V^*Dir}_{f_1})$ and $F_2=(f_2(c),d_{f_2},E_{f_2},U_{f_2},\mathsf{V^*Dir}_{f_2}))$, be an input for $A_{\mathsf{Grid}}$. We have the following cases:
\begin{itemize}
	\item First, assume that at least one of $I_1$ and $I_2$ equals $0$. Without loss of generality, assume that $I_1=0$. Therefore, for every drawing $d_1$ of $F_1$, $d_1$ is not a grid graph. Now, let $d_1$ and $d_2$ be drawings of $F_1$ and $F_2$, respectively, and let $d=\mathsf{Glue}(F,C,d_1,d_2)$. Now, since we assume that $d_{f_1}(V)\cap V^*=\emptyset$, observe that for every $u\in V(d)$, $d(u)\in \gps(\fin)$. Since $d_1$ is not a grid drawing, there exists $\{u,v\} \in E(d_1)$ such that 
	$|d_{1\mathsf{x}}(u)-d_{1\mathsf{y}}(v)|+|d_{1\mathsf{x}}(u)-d_{1\mathsf{y}}(v)|\neq 1$. Observe that $d=^{\pp(f_1(c))}_{\mathsf{rename}} d_1$. Therefore, $u,v\in V(d)$, $d(u)=d_1(u)$ and $d(v)=d_1(v)$; so,  $|\fr(u)-\fr(v)|+|\fc(u)-\fc(v)|\neq 1$. Thus, $d$ is not a grid drawing, and hence $\mathsf{GridClassifier}(F,d)=0$. 
	\item Second, assume that $I_1=I_2=1$. Let $d_1$ and $d_2$ be drawings of $F_1$ and $F_2$, respectively, such that $\mathsf{GridClassifier}(F_1,d_1)=1$ and $\mathsf{GridClassifier}(F_2,d_2)=1$. Let $d=\mathsf{Glue}(F,C,d_1,d_2)$. We show that $d$ is a grid drawing. Let $u\in V(d)$. Assume, without loss of generality, that $d(u)$ is inside $f_1(c)$. Then, since $u\in V$ and $d=^{\pp(f_1(c))}_{\mathsf{rename}} d_1$, we get that $d(u)=d_1(u)$. Since $d_1$ is a grid drawing, we get that $d_1(u)\in \gps(\fin)$; so, $d(u)\in \gps(\fin)$. Now, let $\{u,v\}\in E(d)$. Since there are no vertices from $V^*$ in $d_1$ or $d_2$, it follows that $\{u,v\}$ is inside $f_1(c)$ or $f_2(c)$. Assume, without loss of generality, that $\{u,v\}$ is inside $f_1(c)$. Then, since $u,v\in V$ and $d=^{\pp(f_1(c))}_{\mathsf{rename}} d_1$, we get that $d(\{u,v\})=d_1(\{u,v\})$. Now, $d_1$ is a grid drawing, so $|d_{1\mathsf{x}}(u)-d_{1\mathsf{y}}(v)|+|d_{1\mathsf{x}}(u)-d_{1\mathsf{y}}(v)|=1$, and hence $|d_{\mathsf{x}}(u)-d_{\mathsf{y}}(v)|+|d_{\mathsf{x}}(u)-d_{\mathsf{y}}(v)|= 1$. Therefore, $d$ is a grid drawing, and thus $\mathsf{GridClassifier}(F,d)=1$.
\end{itemize}
We conclude that Condition \ref{def:probSolvCon2} is satisfied, so $A_{\mathsf{Grid}}$ is a $\mathsf{GridClassifier}$-algorithm.
	\end{proof}

Observe that $A$ runs in $\OO(1)$ runtime, so $\mathsf{Time}A_{\mathsf{Grid}}(G,k)=\OO(1)$.

Now, we define the algorithm $L_{\mathsf{Grid}}$, which solves info-frames that are leaves. 

%Recall, that an info-frame $F=(f,d_f,E_f,U_f,\mathsf{V^*Dir}_f)$ is a leaf if there are no grid points strictly inside $f$. Therefore, if $U_f\neq \emptyset$ we can conclude that there are no drawings of $F$, since every vertex in $U_f$ must be drawn on a grid point strictly inside $f$. Similarly, for every $\{uv_i,uv_{i+1}\}\in E_f$ there is at most one possible way to draw, since there are no grid points inside $f$, and edges might bend only at grid points. Thus, every $\{uv_i,uv_{i+1}\}\in E_f$ must be drawn as a straight line between $uv_i$ and $uv_{i+1}$. Therefore, for $F$ we have at most one possible drawing. Observe that, we can get this drawing in $\OO(k)$ runtime.

\begin{definition}[{\bf$L_{\mathsf{Grid}}$}]
	Let $G$ be a graph and let $F=(f,d_f,E_f,U_f,\mathsf{V^*Dir}_f)$ be an info-frame that is a leaf and let $I\in \{0,1\}$. Let $d$ be the drawing of $F$ (see Observation \ref{obs:oneDre}), if such exists; otherwise, $L_{\mathsf{Grid}}(F,I)=\mathsf{False}$.
If $d$ is a grid drawing and $I=1$, then $L_{\mathsf{Grid}}(F,I)=\mathsf{True}$. If $d$ is a grid drawing and $I=0$, then $L_{\mathsf{Grid}}(F,I)=\mathsf{False}$. If $d$ is not a grid drawing and $I=1$, then $L_{\mathsf{Grid}}(F,I)=\mathsf{False}$. If $d$ is not a grid drawing and $I=0$, then $L_{\mathsf{Grid}}(F,I)=\mathsf{True}$. 
\end{definition}

It is easy to see that $L_{\mathsf{Grid}}$ is a $\mathsf{GridClassifier}$-leaf solver (see Definition~\ref{def:probLeGen}):

\begin{observation}\label{obs:LGRid}
	$L_{\mathsf{Grid}}$ is a $\mathsf{GridClassifier}$-leaf solver. In addition, $L_{\mathsf{Grid}}$ runs in $\OO(k)$ runtime, so  $\mathsf{Time}L_{\mathsf{Grid}}(G,k)=\OO(k)$.
\end{observation} 

Next, we show that {\sc Grid Recognition} problem parameterized by $k$ is a $(\mathsf{Classifier},n+2,n+2,k,{\cal I}_{\mathsf{yes}}=\{1\})$-problem on connected graphs (see Definition \ref{def:probSc}).

\begin{lemma}\label{gridisprob}
{\sc Grid Recognition} parameterized by $k$ on connected graphs is $(\mathsf{GridClassifier},$ $n+2,n+2,k,{\cal I}_{\mathsf{yes}}=\{1\})$-problem.
\end{lemma}

\begin{proof}
	It easy to see that every grid graph drawing is bounded by the frame $R_{(n+2)\times (n+2)}$.  Therefore, we get that the {\sc Grid Recognition} problem on connected graphs is a\\ $(\mathsf{GridClassifier}, n+2,n+2,k,{\cal I}_{\mathsf{yes}}=\{1\})$-problem.
\end{proof}

%Moreover, in Corollary \ref{cor:uppBouGr}, we saw that for every grid graph drawing $d$, $\mathsf{dtw}(d)\leq \sqrt{n}$.

In conclusion, Lemma \ref{gridisprob} states that {\sc Grid Recognition} parameterized by $k$ on connected graphs is $(\mathsf{GridClassifier},n+2,n+2,k,{\cal I}_{\mathsf{yes}}=\{1\})$-problem. In addition, $\mathsf{GridClassifier}$ is a classifier with universe $\mathsf{INF}=\{0,1\}$. Further, $A_{\mathsf{Grid}}$ is a $\mathsf{GridClassifier}$-algorithm with $\mathsf{Time}A_{\mathsf{Grid}}(G,k)=\OO(1)$, (by Lemma \ref{lem:A}), and $L_{\mathsf{Grid}}$ is a $\mathsf{GridClassifier}$-leaf solver with $\mathsf{Time}L_{\mathsf{Grid}}($ $G,k)=\OO(k)$ (by Observation \ref{obs:LGRid}). Lastly, the maximum degree of a graph in any grid drawing is $4$. 
Therefore, by Theorem \ref{the:AlgSch}, we get that there exists an algorithm that solves the {\sc Grid Recognition} problem on connected graphs (and hence on general graphs) in runtime $\OO ((k\cdot h\cdot w\cdot n)^{\OO (k)}\cdot 2^{\OO (\Delta\cdot k)}\cdot |\mathsf{INF}|^{\OO(1)}\cdot \mathsf{Time}A_{\mathsf{Grid}}(G,k)) +\OO((k\cdot h\cdot w\cdot n)^{\OO (k)}\cdot |\mathsf{INF}|\cdot \mathsf{Time}L_{\mathsf{Grid}}(G,k))=\OO((k\cdot (n+2)\cdot (n+2)\cdot n)^{\OO (k)}\cdot 2^{\OO (4\cdot k)}\cdot 2^{\OO(1)}\cdot \OO(1) +\OO((k\cdot (n+2)\cdot (n+2)\cdot n)^{\OO (k)}\cdot 2\cdot \OO(k))=n^{\OO(k)}$. Thus, Theorem~\ref{the:gridRecRunTime} follows.
%
%\begin{theorem}
%	There exists an algorithm for the {\sc Grid Recognition} problem parameterized $k$ with runtime $n^{\OO(k)}$.
%\end{theorem}
%
%% problem information $\mathsf{Classifier}$ that we defined, with the universe $I=\{0,1\}$, the $\mathsf{Classifier}$-solver $A$, with runtime $A(k)=\OO(k)$, and the $\mathsf{Classifier}$-leaves generator, with runtime $L(k)=\OO(k)$, we defined. 
%
%
%Now, in Corollary \ref{cor:uppBouGr}, we saw that for every grid graph drawing $d$, $\mathsf{dtw}(d)\leq \OO(\sqrt{n})$.
%Therefore, we have the following corollary:
%
%\begin{corollary}
%	There exists an algorithm for the {\sc Grid Recognition} problem with runtime $n^{\OO(\sqrt{n})}$.
%\end{corollary}  

\subsection{Crossing Minimization Problem on Straight-Line Grid Drawings}
In this subsection, we aim to develop an algorithm for {\sc Straight-line Grid Crossing Minimization} problem (see Definition~\ref{def:crossMin}), using our scheme. In a straight-line grid drawing, vertices are mapped to grid points and edges are mapped to line segments, connecting the images of their endpoints. Here, notice that every two edges can intersect in at most one point (else, we have infinitely many crossing). The goal is, given a graph $G$, to construct a straight-line grid drawing of $G$ with minimum number of pairs of crossing edges. As we solve the parameterized version, we only seek drawings of drawn treewidth at most $k$ (where $k$ is given as input).

%\begin{definition}[{\bf Straight-Line Grid Drawing}]
%Let $G$ be a graph. A {\em straight-line grid} drawing $d$ of $G$ is a straight-line drawing $d$ of $G$ such that (i) for every $u\in V$, $d(u)$ is a grid point (ii) For every $\{u,v\},\{u',v'\}\in E$, $d(\{u,v\})$ and $d(\{u',v'\})$ are intersected in at most one point.
%\end{definition}

%\begin{definition}[{\bf Crossing Minimization Problem on Straight-Line Grid Drawings}]
%	Let $G$ be a graph. The {\sc Crossing Minimization} problem on straight-line grid drawings is, given a graph $G$ and , to construct a straight-line grid drawing $d$ of $G$ (if one exists) such that: (i) $d$ is strictly bounded by $R_{h,w}$,(ii) $d$ has minimum number of crossings out of all the straight-line grid drawings of $G$ which are strictly bounded by $R_{h,w}$. If such a drawing does not exists, return ``no-instance''.
%\end{definition}

In order to develop an algorithm for this problem, we need to store the following additional information in $\mathsf{INF}$, for every info-frame. Observe that, here we do not allow edges to bend. Thus, as we do for vertices from $V^*$, which are mapped to points from the set $\gis(\fin)\setminus \gps(\fin)$ (see the discussion before Definition \ref{def:infFr} regarding Condition  \ref{definfFraCon4}), we would like to store the direction of every edge in $E_f$. Therefore, we have a function, $\mathsf{VerDir}_f$, that assigns a direction for every endpoint drawn on $f$ of every edge in $E_f$. In particular, for an edge $\{uv_i,uv_{i+1}\}\in E_f$ with both endpoints on $f$, we store two directions, one for each endpoint. For $\{uv_i,uv_{i+1}\}\in E_f$ with exactly one endpoint on $f$, we store a direction for that endpoint. Formally, $\mathsf{VerDir}$ gets a pair $(\{uv_i,uv_{i+1}\},z)$, where $\{uv_i,uv_{i+1}\}\in E_f$ and $z\in \{uv_i,uv_{i+1}\}$ such that $z$ is drawn on $f$, and $\mathsf{VerDir}$ returns a point, satisfying conditions similar to Condition  \ref{definfFraCon4} of Definition \ref{def:infFr}. We also add to $\mathsf{INF}$ an indicator stating whether $d$ is ``part of'' a straight-line grid drawing. That is, given a drawing $d$ of $F$ we need to verify the following. Recall that, as $d$ is a $G^*$-drawing (see Definition \ref{def:gstdr}), $d$ maps edges in $E$ to paths in ${\cal P}$ and  edges in $E^*$ to paths in ${\cal P}^*$. In our case, we restrict the paths to be straight lines. We also verify that there are no two edges intersecting in more than one point. We say that such a drawing is a {\em $G^*$-straight-line grid drawing}. So, the universe $\mathsf{INF}$ is defined as follows: $\mathsf{INF}=\{(\mathsf{dir},\mathsf{Indicator})~|~\mathsf{dir}\in \mathsf{DirSet},\mathsf{Indicator}\in\{0,1\}\}$, where $\mathsf{DirSet}$ is the set of all valid $\mathsf{VerDir}$. We remark that, we do not distinguish between different values for $\mathsf{dir}$ that refer the same direction. That is, for every $\{uv_i,uv_{i+1}\}$ and an endpoint drawn on $f$ $z\in \{uv_i,uv_{i+1}\}$, the directions $p$ and $p'$ are equal if $\ell(d_f(z),p)$ is on $\ell(d_f(z),p')$, or vice versa. Now, given a drawing $d$ of $F$, we say that $\mathsf{VerDir}$ is the {\em direction induced by $F$ and $d$}, if for every $\{uv_i,uv_{i+1}\}\in E_f$ and for every endpoint $z\in\{uv_i,uv_{i+1}\}$ drawn on $f$, $\mathsf{VerDir}(\{uv_i,uv_{i+1}\},z)$ is the direction defined similarly to \ref{def:inducedInfoFrameC4} of Definition \ref{def:inducedInfoFrame}. We now define a classifier for this problem:

\begin{definition}[{\bf $\mathsf{CroClassifier}$}] \label{def:crosclas}
	Let $G$ be a graph, let $F$ be an info-frame and let $d$ be a drawing of $F$. Let $\mathsf{dir}$ be the direction induced by $F$ and $d$. Let $\mathsf{Indicator}=1$ if $d$ is a $G^*$-straight-line grid drawing; otherwise, $\mathsf{Indicator}=0$. Then, $\mathsf{CroClassifier}(F,d)=(\mathsf{dir},\mathsf{Indicator})$. 
\end{definition}

%We define the function $\mathsf{Classifier}$ as follows. For every info-frame $F$, and every drawing $d$ of $F$, $\mathsf{Classifier}(F,d)=1$ if $d$ is a grid drawing; Otherwise $\mathsf{Classifier}(F,d)=0$. 

Next, we define the algorithm $A_{\mathsf{Cro}}$. 
For this purpose, we first describe a function called $\mathsf{GetDirections}$. Let $F$ be an info-frame, let $C=(c,F_1=(f_1(c),d_{f_1},E_{f_1},U_{f_1},\mathsf{V^*Dir}_{f_1})$, $F_2=(f_2(c),d_{f_2},E_{f_2},U_{f_2},\mathsf{V^*Dir}_{f_2}))$ be an info-cutter of $F$, and let $I_1=(\mathsf{dir}_1,\mathsf{Indicator}_1),I_2=(\mathsf{dir}_2,$ $\mathsf{Indicator}_2)\in \mathsf{INF}$. We will later define the algorithm $A_{\mathsf{Cro}}$ so that it should return the directions of the endpoints of edges in $E_f$. Here, we take the direction for each such an endpoint from $F_1$ and $F_2$. In particular, we take the direction according to the cases described when we glue two drawings of $F_1$ and $F_2$ (see Definition \ref{def:glueEdg}). For example, consider a vertex $uv_i\in V(d_f)$ such that $d_f(uv_i)\in \gp(d_{f_1},f_1(c))$, such that there exists $\{uv_i,z\},\in E_{f}$, and $\{\mathsf{Identify}_{d_{f},d_{f_1}}(uv_i),\mathsf{Identify}_{d_{f},d_{f_1}}(z)\}\in E_{f_1}$. That is, since $\{uv_i,z\}$ is contained in $f_1(c)$, we take the direction of $(\{uv_i,z\},uv_i)$ from $\mathsf{dir}_1$. So, we define  $\mathsf{dir}(\{uv_i,z\},uv_i)=\mathsf{dir}_1(\{\mathsf{Identify}_{d_{f},d_{f_1}}(uv_i),\mathsf{Identify}_{d_{f},d_{f_1}}(z)\},\mathsf{Identify}_{d_{f},d_{f_1}}(uv_i))$. If $\{uv_i,z\}$ is constructed by gluing several edges together (for example, see Condition \ref{con4defglueedge} of Definition \ref{def:glueEf2}), then we take the direction of the first edge. If $\{uv_i,z\}$ is on $c$ (for example, see Condition \ref{con3defglueedge} of Definition \ref{def:glueEf2}), then we can take the direction directly from the edge drawn on $c$ in $d_{f_1}$ or $d_{f_2}$. Accordingly, $\mathsf{GetDirections}(F,C,I_1,I_2)$ returns the function $\mathsf{dir}$ defined as follows; for every $\{uv_i,uv_{i+1}\}\in E_{f}$ and an endpoint $z\in \{uv_i,uv_{i+1}\}$ drawn on $f$, $\mathsf{dir}(\{uv_i,z\},uv_i)$ is define as described.

We now define the algorithm $A_{\mathsf{Cro}}$, which we will later prove to be $\mathsf{CroClassifier}$-algorithm. Let $G$ be a graph, let $F=(f,d_f,E_f,U_f,\mathsf{V^*Dir}_f)$ be an info-frame, let $C=(c,F_1=(f_1(c),d_{f_1},E_{f_1},$ $U_{f_1},\mathsf{V^*Dir}_{f_1}),F_2=(f_2(c),d_{f_2},E_{f_2},U_{f_2},\mathsf{V^*Dir}_{f_2}))$ be an info-cutter of $F$. The algorithm $A_{\mathsf{Cro}}$ should decide if by gluing any two drawings $d_1$ and $d_2$ of $F_1$ and $F_2$, respectively, we get a $G^*$-straight-line grid drawing or not. So, we need to verify the following conditions. Obviously, $d_f$, $d_1$ and $d_2$ should be $G^*$-straight-line grid drawings. In addition, notice that every vertex from $V^*$ drawn on a point in $\gp(c)\setminus \gp(f)$, does not appear in $d$, as we glue its two edges to construct one edge (or part of an edge; see, for example, Condition \ref{con4defglueedge} of Definition \ref{def:glueEf2})). So, if we have such a vertex, which is connected to one edge in $f_1(c)$ and the other edge in $f_2(c)$, we should verify that they have the same direction, to get an edge or part of an edge in $d$ as a straight line. We have the following definition:

\begin{definition}[{\bf The Function $A_{\mathsf{Cro}}$}] \label{AstCro}
	Let $G$ be a graph, let $F=(f,d_f,E_f,U_f,\mathsf{V^*Dir}_f)$ be an info-frame, let $C=(c,F_1=(f_1(c),d_{f_1},E_{f_1},U_{f_1},\mathsf{V^*Dir}_{f_1})$, $F_2=(f_2(c),d_{f_2},E_{f_2},U_{f_2},$ $\mathsf{V^*Dir}_{f_2}))$ be an info-cutter of $F$, and let $I_1=(\mathsf{dir}_1,\mathsf{Indicator}_1),I_2=(\mathsf{dir}_2,\mathsf{Indicator}_2)\in \mathsf{INF}$. Then, $\mathsf{Indicator}=1$ if and only if the following conditions are satisfied:
	\begin{enumerate}
	
		\item $d_f$ is a $G^*$-straight-line grid drawing. \label{AstCro1} 
		\item $\mathsf{Indicator}_1=\mathsf{Indicator}_2=1$. \label{AstCro2}
		\item For every $uv_i\in V(d_{f_1},\gp(c)\setminus \gp(f))\cap V^*$ such that there exists $\{uv_i,z\},\in E_{f_1}$ and $\{\mathsf{Identify}_{d_{f_1},d_{f_2}}(uv_i),z'\},\in E_{f_2}$, $\ell(\mathsf{dir}_1(\{uv_i,z\},uv_i),d_{f_1}(uv_i))$ is on\\ $\ell(\mathsf{dir}_2(\{\mathsf{Identify}_{d_{f_1},d_{f_2}}(uv_i),z'\}$ $,uv_i),d_{f_2}(\mathsf{Identify}_{d_{f_1},d_{f_2}}(uv_i)))$, or vice versa.\label{AstCro3}
		\end{enumerate} 
	In addition, $\mathsf{dir}=\mathsf{GetDirections}(F,C,I_1,I_2)$.
	Then, $A_{\mathsf{Cro}}(F,C,I_1,I_2)=(\mathsf{Indicator},\mathsf{dir})$.
\end{definition}

We show that $A_{\mathsf{Cro}}$ is a $\mathsf{CroClassifier}$-algorithm, by showing that the conditions of Definition \ref{def:probsolv} are satisfied. First, we show that the drawing of $F$, obtained by gluing drawings the of $F_1$ and $F_2$ is a $G^*$-straight-line grid drawing if and only if Conditions \ref{AstCro1}--\ref{AstCro3} of Definition \ref{AstCro} are satisfied.

\begin{lemma}\label{lemCros}
	Let $F$ be an info-frame, let $C=(c,F_1=(f_1(c),d_{f_1},E_{f_1},U_{f_1},\mathsf{V^*Dir}_{f_1})$, $F_2=(f_2(c),d_{f_2},E_{f_2},U_{f_2},\mathsf{V^*Dir}_{f_2}))$ be an info-cutter of $F$ and let $I_1=(\mathsf{dir}_1,\mathsf{Indicator}_1),I_2=(\mathsf{dir}_2,$ $\mathsf{Indicator}_2)\in \mathsf{INF}$. Let $d_1$ and $d_2$ be drawings of $F_1$ and $F_2$, respectively, such that $\mathsf{CroClassifier}(F_1,$ $d_1)=I_1$ and $\mathsf{CroClassifier}(F_2,d_2)=I_2$. Then, $d=\mathsf{Glue}(F,C,d_1,d_2))$ is a $G^*$-straight-line grid drawing if and only if Conditions \ref{AstCro1}--\ref{AstCro3} of Definition \ref{AstCro} are satisfied.
\end{lemma}

\begin{proof}
	 If $d$ is a $G^*$-straight-line grid drawing, it is easy to see that Conditions \ref{AstCro1} and \ref{AstCro2} are satisfied. Now, consider a vertex $uv_i\in V(d_{f_1},\gp(c)\setminus \gp(f))\cap V^*$ such that there exists $\{uv_i,z\}\in E_{f_1}$ and $\{\mathsf{Identify}_{d_{f_1},d_{f_2}}(uv_i),z'\},\in E_{f_2}$. Notice that, every vertex $uv_i\in V^*$ that is on $\pp(c)\setminus \pp(f)$ in $d_1$, does not appear in $d$, as $d$ is a $G^*$-drawing of $F$. So, $\{uv_i,z\}$ and $\{\mathsf{Identify}_{d_{f_1},d_{f_2}}(uv_i),z'\}$ must be glued. Since $d$ is a $G^*$-straight-line grid drawing, there is no bending at the point $d_1(uv_i)$, so the directions of $\mathsf{dir}_1(\{uv_i,z\},uv_i)$ and $\mathsf{dir}_2(\{\mathsf{Identify}_{d_{f_1},d_{f_2}}(uv_i),z'\},\mathsf{Identify}_{d_{f_1},d_{f_2}}(uv_i))$ must be on the same line. That is, $\ell(\mathsf{dir}_1(\{uv_i,z\},uv_i),d_{f_1}(uv_i))$ is on $\ell(\mathsf{dir}_2(\{\mathsf{Identify}_{d_{f_1},d_{f_2}}(uv_i),z'\},uv_i),$ $d_{f_2}(\mathsf{Identify}_{d_{f_1},d_{f_2}}$ $(uv_i)))$ or vice versa. So, Condition \ref{AstCro3} is satisfied. We show now the other direction. Assume that Conditions \ref{AstCro1}--\ref{AstCro3} of Definition \ref{AstCro} are satisfied. It is easy to see that due to Conditions \ref{AstCro1} and \ref{AstCro2}, there are no two edges with more than one intersection point. Now, for every edge $\{uv_i,uv_{i+1}\}\in E(d)$, it is easy to see that if $d(\{uv_i,uv_{i+1}\})$ is contained inside $f_1(c)$ or $f_2(c)$ then $d(\{uv_i,uv_{i+1}\})$ is a straight-line. Otherwise, $d(\{uv_i,uv_{i+1}\})$ must be the result of gluing several edges together (see for example, Condition \ref{con4defglueedge} of Definition \ref{def:glueEf2}). So, observe that due to Condition \ref{AstCro3}, we get that any of these edges glued to construct $d(\{uv_i,uv_{i+1}\})$ are on the same line, thus $d(\{uv_i,uv_{i+1}\})$ is a straight-line. Thus, $d$ is $G^*$-straight-line grid drawing. This completes the proof.
\end{proof}

We are now ready to prove that $A_{\mathsf{Cro}}$ is a $\mathsf{CroClassifier}$-algorithm.

\begin{lemma}\label{lem:A1}
	$A_{\mathsf{Cro}}$ is a $\mathsf{CroClassifier}$-algorithm.
\end{lemma}

\begin{proof}
	We show that the conditions of Definition \ref{def:probsolv} are satisfied. Condition \ref{def:probSolvCon1} is satisfied, by the definition of the function. Let $F=(f,d_f,E_f,U_f,\mathsf{V^*Dir}_f)$ be an info-frame, let $C=(c,F_1=(f_1(c),d_{f_1},E_{f_1},U_{f_1},\mathsf{V^*Dir}_{f_1})$ $F_2=(f_2(c),d_{f_2},E_{f_2},U_{f_2},\mathsf{V^*Dir}_{f_2}))$ be an info-cutter of $F$ and let $I_1=(\mathsf{dir}_1,\mathsf{Indicator}_1),I_2=(\mathsf{dir}_2,\mathsf{Indicator}_2)\in \mathsf{INF}$. Let $d_1$ and $d_2$ be drawings of $F_1$ and $F_2$, respectively, such that $\mathsf{CroClassifier}(F_1,d_1)=I_1$ and $\mathsf{CroClassifier}(F_2,d_2)=I_2$. 
It is easy to see, due to the definitions of $\mathsf{CroClassifier}$ (Definition \ref{def:crosclas}) and $\mathsf{GetDirections}$, that $\mathsf{dir}$ defined in Definition \ref{AstCro} is the direction induced by $F$ and $d$.
From this and Lemma \ref{lemCros}, we get that Condition \ref{def:probSolvCon2} is satisfied. So, $A_{\mathsf{Cro}}$ is a $\mathsf{CroClassifier}$-algorithm.
\end{proof}

Observe that, up until now, we have only made sure that we ``build'' a straight-line grid drawing of $G$. Recall that we are interested in a drawing where we have the minimum number of crossings. For this purpose, we add a variable $\mathsf{MinCross}(F,I)$, for every info-frame $F=(f,d_f,E_f,U_f,\mathsf{V^*Dir}_f)$ and $I\in \mathsf{INF}$. We define $\mathsf{MinCross}(F,I)$ to be the minimum number of crossings strictly inside $f$ of any drawing $d$ of $F$ such that $\mathsf{CroClassifier}(F,d)=I$ or $\mathsf{MinCross}(F,I)=\infty$ if there is no such a drawing. We proceed to define the function $L_{\mathsf{Cro}}$, and how to compute $\mathsf{MinCross}(F,I)$ for every $I\in\mathsf{INF}$ and for every info-frame $F$ that is a leaf.

\begin{definition}[{\bf The Function $L_{\mathsf{Cro}}$}]
	Let $G$ be a graph, let $F=(f,d_f,E_f,U_f,\mathsf{V^*Dir}_f)$ be an info-frame that is a leaf, and let $I=(\mathsf{dir},\mathsf{Indicator})\in \mathsf{INF}$. Let $d$ be the drawing of $F$ (see Observation \ref{obs:oneDre}). If $d$ is a $G^*$-straight-line grid drawing and $\mathsf{dir}$ is the direction induced by $F$ and $d$, then (i) $L_{\mathsf{Cro}}(F,I)=\mathsf{True}$ and (ii) $\mathsf{MinCross}(F,I)$ is the number crossings strictly inside $f$; otherwise, (i) $L_{\mathsf{Cro}}(F,I)=\mathsf{False}$ and (ii) $\mathsf{MinCross}(F,I)=\infty$. 
\end{definition}

It is easy to see that $L_{\mathsf{Cro}}$ is a leaf solver and $\mathsf{MinCross}(F,I)$ is indeed the value we aim for.

Now, we aim to give a general formula for $\mathsf{MinCross}(F,I)$. To this end, for a $G^*$-straight-line grid drawing $d$ and a set of points $P$, we denote by $\mathsf{Cross}(d,P)$ the number of crossing edges on $P$ in $d$. Let $F=(f,d_f,E_f,U_f,\mathsf{V^*Dir}_f)$ be an info-frame, let $C=(c,F_1=(f_1(c),d_{f_1},E_{f_1},U_{f_1},\mathsf{V^*Dir}_{f_1})$, $F_2=(f_2(c),d_{f_2},E_{f_2},U_{f_2},\mathsf{V^*Dir}_{f_2}))$ be an info-cutter of $F$, let $d_1$ and $d_2$ be drawings of $F_1$ and $F_2$, respectively, and let $d=\mathsf{Glue}(F,C,d_1,d_2)$. Observe that the number of crossings strictly inside $f$ in $d$ is exactly the sum of (i) the number of crossings strictly inside $f_1(c)$ in $d_1$, (ii) the number of crossings strictly inside $f_2(c)$ in $d_2$, and (iii) the number of crossings on $\pp(c)\setminus \pp(f)$. For every info-frame $F$ that is not a leaf and for every $I\in \mathsf{INF}$, denote $F_1=(f_1(c),d_{f_1},E_{f_1},U_{f_1},\mathsf{V^*Dir}_{f_1})$ and $F_2=(f_2(c),d_{f_2},E_{f_2},U_{f_2},\mathsf{V^*Dir}_{f_2})$. We have the following formula:  $\mathsf{MinCross}(F,I)=\min\{\mathsf{MinCross}(F_1,I_1)+\mathsf{MinCross}(F_2,I_2)+\mathsf{Cross}(d_{f_1},\pp(c)\setminus \pp(f))~|~C=(c,F_1,F_2)$ is an info-cutter of $F$ of size at most $k$, $\mathsf{CroClassifier}(F_1,d_1)=I_1$,$\mathsf{CroClassifier}(F_2,d_2)$ $=I_2$ and $A_{\mathsf{Cro}}(C,I_1,I_2)=I\}$. The proof for this formula is similar to the proof of Lemma \ref{lem:inductiveAl}. In addition, observe that, given an info-cutter $C$ of $F$, $\mathsf{MinCross}(F,I)=\min\{\mathsf{MinCross}(F_1,I_1)+\mathsf{MinCross}(F_2,I_2)+\mathsf{Cross}(d_{f_1},\pp(c)\setminus \pp(f))$ is computed in $\OO(k)$ runtime. We add this computation to Algorithm \ref{alg:schemeiter} after Line \ref{alg3:Line9}, and we update an additional variable for each $I\in \mathsf{INF}$ and info-frame $F$, to store the minimum value. Notice that this does not change the given time complexity bound of the algorithm. We remark that one can use our scheme without adding the computation of $\mathsf{MinCross}(F,I)$, by adding to $\mathsf{INF}$ also the number of crossings. Notice that this number is bounded by $n^4$, where $n=|V|$, as the number of edges is bounded by $n^2$, and every two edges can intersect in one point at most. 

Now, it is easy to see that {\sc Straight-line Grid Crossing Minimization} is a\\ $(\mathsf{CroClassifier},$ $h,w,k,I_{\mathsf{yes-In}}=\{(\mathsf{True},\emptyset)\})$-problem (see Definition \ref{def:probSc}). In addition, observe that, for a given info-frame $F=(f,d_f,E_f,U_f,\mathsf{V^*Dir}_f)$, we have at most $k$ vertices for which we need to guess directions, so there are at most $n^{\OO(k)}$ different choices. Then, by Theorem \ref{the:AlgSch}, Theorem~\ref{the:crossingMinTime} follows.
We remark that, the value returned by the algorithm is $\mathsf{MinCross}(F_{\mathsf{Init}},(\mathsf{True},\emptyset))$.

\subsection{Algorithm for The Orthogonal Compaction Problem}

In this subsection, we aim to develop an algorithm for the {\sc Orthogonal Compaction} problem (see Definition \ref{def:OrtCom}), using our scheme. In this problem, we get a connected graph $G$. We assume to have an order on the vertices, that is, for every $u,v\in V$ such that $u\neq v$, either $u>v$ or $v<u$. In addition to $G$, we have, for every $\{u,v\}\in E$ where $u>v$, the relative position of $v$ compered to $u$, that is, the {\em direction} of the $\{u,v\}$ from $u$ to $v$. We denote these directions by $\mathsf{U},\mathsf{D},\mathsf{L}$ and $\mathsf{R}$; this stands for ``up'', ``down'', ``left'' and ``right'', respectively. We assume that there exists a planar rectilinear grid drawing of $G$ such that for every $\{u,v\}\in E$, the relative position of $v$ compered to $u$ is as given as input. Our goal is to find such a drawing of minimum area. Similar to the two problems we discussed in the previous subsections, we only seek drawings of drawn treewidth at most $k$ (where $k$ is given as input).

We first assume that $h$ and $w$ are given as input. So, we need to find a planar rectilinear grid drawing $d$ of $G$ such that $d$ respects the directions of the edges and it is strictly bounded by $R_{h,w}$.
Now, similarly to the concept of a $G^*$-drawing, we would like to define the way parts of such a drawing look like. First, since we look for a planar rectilinear grid drawing, we expect that every vertex is drawn on a grid point and every edge or a part of an edge is parallel to the axis, in every part of the drawing. In addition, every part of the drawing is obviously planar too. Thus, every part of the drawing should be a planar rectilinear grid drawing. Now, recall that we wish that the drawing respects every direction of every edge. To this end, we introduce an alternative version of Definition~\ref{def:dir} for $G^*$-drawings. First, for every $\{u,v\}\in E$, we want that all vertices from $V^*_{\{u,v\}}$ will be on the same row or column, depending on the direction of $\{u,v\}$ that is given as input. In addition, if $\{u,uv_1\}\in E(d)$ (or $\{uv_{\mathsf{index}(u,v)},v\}\in E(d)$), then we would like $u$ (or $v$) is drawn in the direction relatively to $uv_1$ (or $uv_{\mathsf{index}}$). Thus, we have the following definition:

%Recall that, the labeling of vertices from the set $V_{\{u,v\}}$ are ordered by the order of appearance of the vertices from $u$ to $v$. Therefore, if for example $\mathsf{dir}_{\{u,v\}}=\mathsf{U}$, then $v$ should appear above $u$. So, every $uv_{i+1}$, represents a point of the edge $\{u,v\}$, should be above $uv_{i}$, represents a point of the edge $\{u,v\}$ that is closer to $u$ than $uv_{i+1}$ is. Thus, for every $\{uv_i,uv_{i+1}\}$ in a part of $d$, it should respect the same direction as $\mathsf{dir}_{\{u,v\}}$. In this case, we say that the $G^*$ drawing respects $\mathsf{dir}_{\{u,v\}}$. We call such a drawing a {\em $G^*$-directed rectilinear drawing}:

\begin{definition}[{\bf $G^*$-Drawing Respects an Edge Direction}]
	Let $G$ be a connected graph, let $\{u,v\}\in E$, such that $u>v$, let $\mathsf{dir}_{\{u,v\}}\in \{\mathsf{U},\mathsf{D},\mathsf{L},\mathsf{R}\}$, and let $d$ be $G^*$-drawing. Then, $d$ {\em respects $\mathsf{dir}_{\{u,v\}}$} if for every $uv_i,uv_j\in V(d)$ such that $i\neq j$ the following conditions are satisfied.
	\begin{enumerate}
		\item If $\mathsf{dir}_{\{u,v\}}=\mathsf{U}$ or $\mathsf{dir}_{\{u,v\}}=\mathsf{D}$, then $\fr(uv_i)=\fr(uv_j)$.
		\item If $\mathsf{dir}_{\{u,v\}}=\mathsf{L}$ or $\mathsf{dir}_{\{u,v\}}=\mathsf{R}$, then $\fc(uv_i)=\fc(uv_j)$.
	\end{enumerate}  
In addition, if $\{u,uv_1\}\in E(d)$, then the following conditions are satisfied.
\begin{enumerate}
	\item If $\mathsf{dir}_{\{u,v\}}=\mathsf{U}$, then $\fc(uv_1)>\fc(u)$.
	\item If $\mathsf{dir}_{\{u,v\}}=\mathsf{D}$, then $\fc(uv_1)<\fc(u)$.
	\item If $\mathsf{dir}_{\{u,v\}}=\mathsf{L}$, then $\fr(uv_1)<\fr(u)$.
	\item If $\mathsf{dir}_{\{u,v\}}=\mathsf{R}$, then $\fr(uv_1)>\fr(u)$.
\end{enumerate}  
 In addition, if $\{uv_{\mathsf{index}(u,v)},v\}\in E(d)$, then the following conditions are satisfied.
 \begin{enumerate}
 	\item If $\mathsf{dir}_{\{u,v\}}=\mathsf{U}$, then $\fc(uv_{\mathsf{index}(u,v)})<\fc(v)$.
 	\item If $\mathsf{dir}_{\{u,v\}}=\mathsf{D}$, then $\fc(uv_{\mathsf{index}(u,v)})>\fc(v)$.
 	\item If $\mathsf{dir}_{\{u,v\}}=\mathsf{L}$, then $\fr(uv_{\mathsf{index}(u,v)})>\fr(v)$.
 	\item If $\mathsf{dir}_{\{u,v\}}=\mathsf{R}$, then $\fr(uv_{\mathsf{index}(u,v)})<\fr(v)$.
 \end{enumerate}  
\end{definition}

Now, we give the definition for parts of drawings with the properties discussed so far, called {\em $G^*$-directed rectilinear drawings}:

\begin{definition}[{\bf $G^*$-Directed Rectilinear Drawing}]
		Let $G$ be a connected graph, let $\mathsf{dir}_{\{u,v\}}$ $\in \{\mathsf{U},\mathsf{D},\mathsf{L},\mathsf{R}\}$, for every $\{u,v\}\in E$, and let $d$ be a $G^*$-drawing. Then, $d$ is a {\em $G^*$-directed rectilinear drawing} if (i) $d$ is a planar rectilinear drawing, and (ii) for every $\{uv_i,uv_{i+1}\}$, $d$ respects $\mathsf{dir}_{\{u,v\}}$.
	\end{definition}

Now, here, the only information we need to store for every info-frame $F$ is whether there exists a drawing $d$ of $F$ that is a $G^*$-directed rectilinear drawing. So, we define $\mathsf{INF}=\{0,1\}$, where $1$ indicates that there exists such a drawing of $F$. Accordingly, we define the classifier for the problem: 

\begin{definition}[{\bf $\mathsf{OClassifier}$}] \label{def:cosclas}
	Let $G$ be a connected graph, let $\mathsf{dir}_{\{u,v\}}\in \{\mathsf{U},\mathsf{D},\mathsf{L},\mathsf{R}\}$ for every $\{u,v\}\in E$, let $F$ be an info-frame and let $d$ be a drawing of $F$. Then, $\mathsf{OClassifier}(F,d)=1$ if $d$ is a $G^*$-directed rectilinear drawing; otherwise, $\mathsf{OClassifier}(F,d)=0$.  
\end{definition}

Next, we define the algorithm $A_{\mathsf{OC}}$.

\begin{definition}[{\bf The Function $A_{\mathsf{OC}}$}] \label{AOC}
		Let $G$ be a connected graph, let $\mathsf{dir}_{\{u,v\}}\in \{\mathsf{U},\mathsf{D},\mathsf{L},\mathsf{R}\}$ for every $\{u,v\}\in E$, let $F=(f,d_f,E_f,U_f,\mathsf{V^*Dir}_f)$ be an info-frame, let $C=(c,F_1,F_2)$ be an info-cutter of $F$ and let $I_1,I_2\in \mathsf{INF}$. Then, if (i) $I_1=I_2=1$ and (ii) $d_f$ is a $G^*$-directed rectilinear drawing, then $A_{\mathsf{OC}}(F,C,I_1,I_2)=1$; otherwise, $A_{\mathsf{OC}}(F,C,I_1,I_2)=0$.
\end{definition}

We now show that $A_{\mathsf{OC}}$ is an $\mathsf{OCClassifier}$-algorithm:

\begin{lemma}\label{lem:AC1}
	$A_{\mathsf{OC}}$ is an $\mathsf{OCClassifier}$-algorithm.
\end{lemma}

\begin{proof}
Let $F=(f,d_f,E_f,U_f,\mathsf{V^*Dir}_f)$ be an info-frame, let $C=(c,F_1,F_2)$ be an info-cutter of $F$ and let $I_1,I_2\in \mathsf{INF}$.  Let $d_1$ and $d_2$ be drawings of $F_1$ and $F_2$, respectively, such that $\mathsf{OClassifier}(F_1,d_1)=I_1$ and $\mathsf{OClassifier}(F_2,d_2)=I_2$. Let $d=\mathsf{Glue}(F,C,d_1,d_2))$. First, observe that if $d_f$ is not a $G^*$-directed rectilinear drawing, then neither is $d$, so $\mathsf{OClassifier}(F,d)=0$ and $A_{\mathsf{OC}}(F,C,I_1,I_2)=0$. So, assume that $d_f$ is a $G^*$-directed rectilinear drawing. Now, assume that at least one among $I_1$ and $I_2$ is $0$; without loss of generality, assume that $I_1=0$. Then, $d_1$ is not a $G^*$-directed rectilinear drawing. It is easy to see that if $d_1$ is not planar or not a rectilinear drawing, then neither is $d$, and then, $d$ is not a $G^*$-directed rectilinear drawing. So, assume that $d_1$ is a planar rectilinear drawing. Since $d_1$ is not a $G^*$-directed rectilinear drawing, there exists $\{u,v\}\in E$ such that $d_1$ does not respect $\mathsf{dir}_{\{u,v\}}$; without loss of generality, assume that $\mathsf{dir}_{\{u,v\}}=\mathsf{U}$. Now, we have few cases. Assume, that there exists $uv_i,uv_j\in V(d_1)$ such that $d_{1\mathsf{y}}(uv_i)\neq d_{1\mathsf{y}}(uv_j)$. If $d_1(uv_i)$ and $d_1(uv_j)$ are on $f$, then we get that $d_f$ is not a $G^*$-directed rectilinear drawing, a contradiction to the assumption that it is a $G^*$-directed rectilinear drawing. Otherwise, at least one among $uv_i$ and $uv_j$ is drawn on a point in $\gp(f)\setminus \gps(c)$; assume that $uv_i$ is such. Therefore, there must be $\{uv_t,uv_{t+1}\}\in E_f$ such that $d_1(uv_i)\in d(\{uv_t,uv_{t+1}\})$. So, either $d(\{uv_t,uv_{t+1}\})$ is not a straight line, and then $d$ is not a rectilinear drawing, or $\fc(uv_t)\neq\fc(uv_{t+1})$, and then $d$ does not respect $\mathsf{dir}_{\{u,v\}}$. Either way, we get that $d$ is not a $G^*$-directed rectilinear drawing, so $\mathsf{OClassifier}(F,d)=0$, and $A_{\mathsf{OC}}(F,C,I_1,I_2)=0$. 
Now, assume that $\{u,uv_1\}\in E(d_1)$ and $d_{1\mathsf{y}}(uv_1)\leq d_{1\mathsf{y}}(u)$. So, observe that $\{u,uv_1)\}\in E(d)$ and either $\fc(uv_1)\leq\fc(u)$, or there exists $d(\{u,uv_2\})$ that is not a straight line. Either way, we get that $d$ is not a $G^*$-directed rectilinear drawing, so $\mathsf{OClassifier}(F,d)=0$ and $A_{\mathsf{OC}}(F,C,I_1,I_2)=0$. The other cases are similar.

Now, assume that $d_f$ is a $G^*$-directed rectilinear drawing and $I_1=I_2=1$. Since $d_1$ and $d_2$ are planar, it is easy to see that $d$ is planar too. In addition, since $d_1$ and $d_2$ are rectilinear drawings, we get that every $v\in V(d)$ is drawn on a grid point. Now, let $\{uv_i,uv_{i+1}\}\in E(d)$, and assume, without loss of generality, that $\mathsf{dir}_{\{u,v\}}=\mathsf{U}$. We have a few cases corresponding to the different cases of Definition \ref{def:glueEdg}. Assume that $d(\{uv_i,uv_{i+1}\})$ is constructed by Condition \ref{con4defglueedge} of Definition \ref{def:glueEf2}. So, since $d_1$ and $d_2$ are rectilinear drawings and respect $\mathsf{dir}_{\{u,v\}}$, it can be proved by induction, that for every $0\leq t\leq \ell$, $\mathsf{PartEdge}_t$ is a straight line, parallel to the $y$-axis. Therefore, $d(\{uv_i,uv_{i+1}\})$ is a straight line, thus $d$ is a rectilinear drawing. Now, we show that $d$ respects $\mathsf{dir}_{\{u,v\}}$. Assume that $\{u,uv_1\}\in E(d)$. Again, we have a few cases. Assume that $\{u,uv_1\}$ is drawn inside $f_1(c)$ or $f_2(c)$. Then, since $d_1$ and $d_2$ respect $\mathsf{dir}_{\{u,v\}}$, we get that $\fc(uv_1)>\fc(u)$. Otherwise, $\{u,uv_1\}$ is drawn inside neither $f_1(c)$ nor $f_2(c)$. Assume, without loss of generality, that $u$ is drawn inside $f_1(c)$. Since $\{u,uv_1\}$  is not drawn inside $d_1$, there must be a a turning point of $\{u,uv_1\}$ in $f_1(c)$. Let $p$ be the closest one to $u$ among the aforementioned points. So, since $d_1$ is a $G^*$-drawing, $\{u,uv_1\}\in E(d_1)$ where $d_1(uv_1)=p$, and since $d_1$ respects $\mathsf{dir}_{\{u,v\}}$, $d_{1\mathsf{y}}(uv_1)>d_{1\mathsf{y}}(u)$. Now, as $d$ is a rectilinear drawing, we get that $\fc(uv_1)>\fc(u)$ (recall that $uv_1$ in $d_1$ and $uv_1$ in $d$ might be drawn on different points). The rest of the cases are similar. This completes the proof.
%Assume that $u$ is drawn strictly inside $f_1(c)$ or $f_2(c)$, without loss of generality, assume that $u$ is drawn strictly inside $f_1(c)$. T
 %If $uv_i$ is drawn inside $f_1(c)$, then $\fr(uv_i)>\fr(u)$ since $d_1$ respects $\mathsf{dir}_{\{u,v\}}$. If $uv_i$ is drawn outside $f_1(c)$, then, since $d_1$ is a $G^*$-drawing, there must be $uv_j\in V(d_1)$ drawn on $f_1(c)$. If $uv_j$ is drawn on $f$, then ther
\end{proof}

Observe that, the runtime of $A_{\mathsf{OC}}$ is $\OO(k)$. 
Now, we define the leaf solver $L_{\mathsf{OC}}$:

 \begin{definition}[{\bf The Function $L_{\mathsf{OC}}$}]
	Let $G$ be a connected graph, let $\mathsf{dir}_{\{u,v\}}\in \{\mathsf{U},\mathsf{D},\mathsf{L},\mathsf{R}\}$, for every $\{u,v\}$, let $F=(f,d_f,E_f,U_f,\mathsf{V^*Dir}_f)$ be an info-frame that is a leaf and let $I\in \mathsf{INF}$. Let $d$ be the drawing of $F$ (see Observation \ref{obs:oneDre}). Then:
	\begin{itemize}
		\item If $d$ is a $G^*$-directed rectilinear drawing and $I=1$, then $L_{\mathsf{OC}}(F,d)=\mathsf{True}$.
		\item If $d$ is a $G^*$-directed rectilinear drawing and $I=0$, then $L_{\mathsf{OC}}(F,d)=\mathsf{False}$.
		\item If $d$ is not a $G^*$-directed rectilinear drawing and $I=1$, then $L_{\mathsf{OC}}(F,d)=\mathsf{False}$.
		\item If $d$ is not  a $G^*$-directed rectilinear drawing and $I=0$, then $L_{\mathsf{OC}}(F,d)=\mathsf{True}$.
	\end{itemize} 
\end{definition}

It is easy to see that $L_{\mathsf{OC}}$  is indeed a leaf solver. Notice that $L_{\mathsf{OC}}$  runs in $\OO(k)$ runtime.
In addition, observe that if in a rectilinear grid drawing $d$ of $G$ we have more than $n+1$ rows (columns), then there is a row (column) with no vertex drawn on it. We can ``delete" this row (column) and get a rectilinear grid drawing $d$ of $G$ with less rows (columns). Observe that the relative position of every two vertices is preserved in the modified drawing. Therefore, we can assume that $h\leq n+1$ and $w\leq n+1$. 

Thus, by Theorem \ref{lem:AlgTime} and since $\Delta\leq 4$, we have that there exists an algorithm that solves the {\sc Orthogonal Compaction} problem in time $\OO((k\cdot h\cdot w\cdot n)^{\OO (k)}\cdot 2^{\OO (4\cdot k)}\cdot 2^{\OO(1)}\cdot \OO(k)) +\OO((k\cdot h\cdot w\cdot n)^{\OO (k)}\cdot 2\cdot \OO(k))=\OO((k\cdot n\cdot n\cdot n)^{\OO (k)}=n^{\OO(k)}$. Thus, Theorem~\ref{the:orthoComTime} follows.

\section{Additional Upper and Lower Bounds}\label{sec:uppBound}

\subsection{Upper Bounds for Orthogonal Grid Drawings}

In this subsection, we aim to give upper bounds for the drawn treewidth of some classes of grid drawings. In particular, we give an upper bound for orthogonal grid drawings. Observe that, in orthogonal grid drawings, vertices are drawn on grid points, and edges are mapped to straight-line paths, which might bend only on grid points. So, the ``changes'' in the drawing, such as turning points or drawing of vertices, might happen only on grid points.  
Therefore, a cost of a frame is an order of the sum ``elements'' intersected, that is, vertices and grid points contained in the drawing of an edge (i.e., for example, if two edges intersect the frame at the same grid point, we count it as two elements). We first begin with defining the number of {\em elements intersected by a cutter} in an orthogonal grid drawing:

\begin{definition}[{\bf Cutter and Elements Intersection}]
Let $G=(V,E)$ be a graph, let $d$ be an orthogonal grid drawing of $G$, let $f$ be a frame and let $c$ be a cutter $f$. The number of {\em elements intersected by $c$} in $d$, is the sum of (i) the number of vertices drawn on $c$ and (ii) for every $\{u,v\}\in E$, the number of grid points intersected by $d(\{u,v\})$ and $c$. 
\end{definition}

We say that a cutter $c$ of a frame $f$ is with {\em $s$ element intersections} in $d$, if the number of elements intersected by $c$ in $d$ is $s$.

Let $G=(V,E)$ be a graph, let $h,w\in\mathbb{N}$, let $d$ be an orthogonal grid drawing of $G$ and let $\{u,v\}\in E$. The {\em length} of $\{u,v\}$ in $d$, is the number of grid points intersected by $d(\{u,v\})$.

Now, we show that in an orthogonal grid drawings, for any frame that is a rectangles (with grid points strictly inside it), we can always find a straight cutter that is ``small'' enough. 

\begin{lemma}\label{lem:smallCut}
Let $G=(V,E)$ be a graph, let $h,w\in\mathbb{N}$  and let $d$ be an orthogonal grid drawing of $G$, strictly bounded by $R_{h,w}$.  Then, for every $R_{h',w'}$ which is inside $R_{h,w}$ and has grid points strictly inside, there exists a straight cutter of $R_{h',w'}$ with $\OO (\sqrt{\Delta\cdot \ell\cdot n}\cdot \mathsf{maxInt})$ element intersections in $d$, where (i) $\Delta$ is the maximum degree in $G$, (ii) $\ell$ is the average length of the edges of $G$ in $d$, and (iii) $\mathsf{maxInt}$ is the maximum number of edges and vertices intersected in a grid point in $d$.
\end{lemma}

\begin{proof}
	Observe that, the sum of the vertices in $G$ and the sum of lengths of all the edges in $d$ is bounded by $n+n\cdot \Delta \cdot \ell$. Now, if $R_{h',w'}$ has at most $\sqrt{(\Delta\cdot \ell+1)n}$ rows (columns) strictly inside it, then we take as a cutter $c$, a column (a row). It is easy to see that $c$ is indeed a cutter of $R_{h',w'}$, that intersects at most $\sqrt{(\Delta\cdot \ell+1)n}$ grid points, and therefore with at most $\sqrt{(\Delta\cdot \ell+1)n}\cdot \mathsf{maxInt}=\OO (\sqrt{\Delta\cdot \ell\cdot n}\cdot \mathsf{maxInt})$ element intersections in $d$. Otherwise, $R_{h',w'}$ has more than $\sqrt{(\Delta\cdot \ell+1)n}$ rows (columns). Therefore, there must be a row (column) with at most $\sqrt{(\Delta\cdot \ell+1)n}$ element intersections in $d$. We take this row (column) as a cutter $c$. It is easy to see that $c$ is indeed a cutter of $R_{h',w'}$, and we are done.
\end{proof}

Now, consider an orthogonal grid drawing $d$ of a graph $G$, with maximum degree $\Delta$ and a frame $f$. Observe that, an element might contribute at most $\OO(\Delta)$ to the cost of the frame. For example, a vertex can be also a turning point in $f$ for every one of its edges, so it contributes at most $\Delta+1$ to the cost. An element which is a part of an edge that is not an endpoint, contributes at most $1$ to the cost of the frame, if it is a turning points of this edge in $f$. Therefore, if we have a frame that intersects $s$ elements, its cost is bounded by $\OO(\Delta\cdot s)$. Now, assume that we have a drawn frame-tree of $d$, where every frame is a rectangle, which every one of its edges are contained inside a cutter with $s$ element intersections. We get that, every frame intersects $\OO(s)$ elements, so the cost of every such a frame is $\OO(\Delta\cdot s)$, and so is the width of the drawing frame-tree:

\begin{observation}\label{obs:costF}
Let $G=(V,E)$ be a graph with maximum degree $\Delta$, let $h,w\in\mathbb{N}$  and let $d$ be an orthogonal grid drawing of $G$, strictly bounded by $R_{h,w}$. Let $({\cal T},\alpha)$ be a drawn frame-tree of $d$, such that, for every $v\in V_T$, $\alpha(v)=R_{h',w'}$ for some $h',w'\in \mathbb{N}$, and every edge of $R_{h',w'}$ is contained in some cutter with $\OO(s)$ element intersections, or in $R_{h,w}$. Then, $\mathsf{width}({\cal T},\alpha)\leq \OO(\Delta\cdot s)$. 
\end{observation}

Next, we assume that for every rectangle in an orthogonal grid drawing $d$, we can find a ``small'' cutter. Then, we use Observation \ref{obs:costF} in order to construct a drawn decomposition of $d$ with small width.

\begin{lemma}\label{lem:smalldt}
Let $G=(V,E)$ be a graph, let $h,w\in\mathbb{N}$  and let $d$ be an orthogonal grid drawing of $G$, strictly bounded by $R_{h,w}$.  If that, for every $R_{h',w'}$ which is inside $R_{h,w}$ and has grid points strictly inside, there exists a straight cutter of $R_{h',w'}$ with $\OO(s)$ element intersections in $d$, then $\mathsf{dtw}(d)\leq \OO(\Delta\cdot s)$.
	\end{lemma}

\begin{proof}
	We give a drawn frame-tree $({\cal T}=(V_T,E_T),\alpha: V_T\rightarrow \mathsf{Frames})$ where $\mathsf{width}({\cal T},\alpha)\leq \OO(s)$. We construct ${\cal T}$ by induction, starting from the root $v_r$, and $\alpha(v_r)=R_{h,w}$. The inductive hypothesis is as follows. At every stage of the construction, the following conditions are satisfied:
	\begin{enumerate}
		\item ${\cal T}$ is a tree. \label{Con1}
			\item For every $v\in V_T$ that is not a leaf, there exists a straight cutter $c_v$ of $f$, with $\OO(s)$ element intersections in $d$, such that $\alpha(v_1)=f_1(c_v)$ and $\alpha(v_2)=f_2(c_v)$, where $\alpha(v)=f$ and $v_1, v_2$ are the children of $v$ in ${\cal T}$. \label{Con2}
		\item For every vertex $v\in V_T$, $\alpha(v)=R_{h',w'}$ for some $h',w'\in \mathbb{N}$, such that $R_{h',w'}$ is inside $R_{h,w}$. Moreover, every edge of $R_{h',w'}$ is contained in some cutter previously chosen or in $R_{h,w}$. \label{Con3}
	\end{enumerate} 
We denote by $({\cal T},\alpha)_i$ the pair obtained after the $i$-th step. We show that for every step $i$, $({\cal T},\alpha)_i$ satisfies Conditions \ref{Con1}--\ref{Con3}.
It is easy to see that for the basic case, $({\cal T},\alpha)_1$ satisfies Conditions \ref{Con1}--\ref{Con3}. Let $1\leq i$, and let $({\cal T},\alpha)_i$ be the pair obtained after the $i$-th step. If for every $v\in V_T$ that is a leaf, there are no grid points strictly inside $\alpha(v)$, then we are done. Otherwise, let $v\in V_T$ be a leaf, where there is a grid point strictly inside $\alpha(v)$. By the inductive hypothesis, we get that $\alpha(v)=R_{h',w'}$ for some $h',w'\in \mathbb{N}$, such that $R_{h',w'}$ is inside $R_{h,w}$. From the assumption of the lemma, there exists a straight cutter $v_c$ with $\OO(s)$ element intersections. We obtain $({\cal T},\alpha)_{i+1}$ from $({\cal T},\alpha)_{i}$ as follows. We add two vertices $v_1$ and $v_2$ to $V_T$ and $\{v,v_1\},\{v,v_2\}$ to $E_T$. Denote $R_{h',w'}$ by $f$, we define $\alpha(v_1)=f_1(c_v)$ and $\alpha(v_1)=f_2(c_v)$. We show that $({\cal T},\alpha)_{i+1}$ satisfies Conditions \ref{Con1}--\ref{Con3}. It is easy to see that ${\cal T}$ is a tree, so Condition \ref{Con1} is satisfied. By the inductive hypothesis, and by the construction, it is also easy to see that Condition \ref{Con2} is satisfied. In addition, since $c_v$ is a straight line, then $\alpha(v_1)=f_1(c_v)$ and $\alpha(v_2)=f_2(c_v)$ are rectangles. Now, from the inductive hypothesis, every edge of $R_{h',w'}$ is contained in a cutter previously chosen or in $R_{h,w}$. So, one edge of $f_1(c_v)$ (and similarly, $f_2(c_v)$) is $c_v$ and each one of the other three edges, is contained in an edge of $R_{h',w'}$, and thus, it is contained in a cutter previously chosen or in $R_{h,w}$. Therefore, we get that $({\cal T},\alpha)_{i+1}$ satisfies Conditions \ref{Con1}--\ref{Con3}, so the induction is proved. Observe that, this process is finite. In addition, at the end of the process, we get that, a vertex $v\in V_T$ is a leaf if and only if there are no grid points in the interior of $\alpha(v)$. Therefore, the conditions of Definition \ref{def:DrawnFrDeco} are satisfied, so $({\cal T},\alpha)$ is a drawn frame-tree of $d$. Moreover, by Conditions \ref{Con2} and \ref{Con3}, for every $v\in V_T$, $\alpha(v)=R_{h',w'}$ for some $h',w'\in \mathbb{N}$, such that every edge of $R_{h',w'}$ is contained in some cutter with $\OO(s)$ element intersections, or in $R_{h,w}$. Thus, from Observation \ref{obs:costF}, we get that $\mathsf{width}({\cal T},\alpha)\leq \OO(\Delta\cdot s)$. This completes the proof.
\end{proof}

From Lemmas \ref{lem:smallCut} and \ref{lem:smalldt}, we get the following corollary:

\begin{corollary}\label{cor:uppBou}
Let $G=(V,E)$ be a graph and let $d$ be an orthogonal grid drawing of $G$. Then, $\mathsf{dtw}(d)\leq \OO (\Delta \cdot \sqrt{\Delta\cdot \ell\cdot n}\cdot \mathsf{maxInt})$, where (i) $\Delta$ is the maximum degree in $G$, (ii) $\ell$ is the average length of the edges of $G$ in $d$, and (iii) $\mathsf{maxInt}$ is the maximum number of edges and vertices intersected in a grid point in $d$.
\end{corollary}

In addition, from Corollary \ref{cor:uppBou}, we get the following corollary regarding grid drawings:

\begin{corollary}\label{cor:uppBouGr}
	Let $G=(V,E)$ be a graph and let $d$ be a grid drawing of $G$. Then, $\mathsf{dtw}(d)\leq \OO \sqrt{n}$. Moreover, there exists a drawn tree decomposition $({\cal T},\beta,\alpha)$ of $d$, where every cutter is a straight line, and $\mathsf{width}({\cal T},\beta,\alpha)\leq\OO(\sqrt{n})$.
\end{corollary}

Next, we give another upper bound for orthogonal grid drawings $d$, with respect to $\mathsf{min}\{h,w\}$, where $h,w\in \mathbb{N}$ and $d$ is strictly inside $R_{h,w}$. Assume, without loss of generality, that $h\leq w$. Observe that, every column we take as a cutter, is with at most $\OO (h\cdot\mathsf{maxInt})$ element intersections in $d$, where $\mathsf{maxInt}$ is the maximum number of edges and vertices intersected in a grid point in $d$. Therefore, similarly to the drawn frame-tree we build in Lemma \ref{lem:smalldt}, we can inductively choose columns as cutters, and obtain a drawn frame-tree $({\cal T},\alpha)$. Notice that, every frame in $({\cal T},\alpha)$ is a rectangle, where at most two of its edges are cutters, and the other contained in $R_{h,w}$. Therefore, by Observation \ref{obs:costF}, $\mathsf{width}({\cal T},\alpha)\leq \OO(\Delta\cdot h\cdot\mathsf{maxInt})$. Thus, we have the following lemma (without a proof):

\begin{lemma}\label{lem:ubhw}
	Let $G=(V,E)$ be a graph and let $d$ be an orthogonal grid drawing of $G$. Then, $\mathsf{dtw}(d)\leq \OO(\Delta\cdot \mathsf{min}\{h,w\}\cdot\mathsf{maxInt})$, where (i) $\Delta$ is the maximum degree in $G$, (ii) $d$ is strictly bounded by $R_{h,w}$, and (iii) $\mathsf{maxInt}$ is the maximum number of edges and vertices intersected in a grid point in $d$.
\end{lemma}

From Lemma \ref{lem:ubhw}, we have the following corollary regarding planar orthogonal grid drawings:

 \begin{corollary}\label{lem:ubhwP}
 	Let $G=(V,E)$ be a graph and let $d$ be a planar orthogonal grid drawing of $G$. Then, $\mathsf{dtw}(d)\leq \OO(\mathsf{min}\{h,w\})$, where $d$ is strictly bounded by $R_{h,w}$.
 \end{corollary}

\subsection{Lower Bound for Balanced Drawing Treewidth in Rectilinear Grid Drawings}

\begin{figure}[t]
	\centering
	\includegraphics[page=74, width=0.31\textwidth]{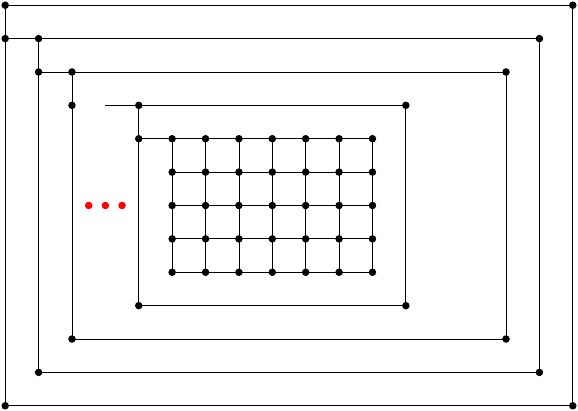}
	\caption{Example for lower bound.}
	\label{fig:lowerBound}
\end{figure}

In this subsection, we give a lower bound for balanced drawing treewidth in rectilinear grid drawings. Consider the drawing given in Figure~\ref{fig:lowerBound}. Let $({\cal T},\beta,\alpha)$ be a balanced drawn tree decomposition of $d$ and let $v_r\in V_T$ be the root. Let $f$ be $R_{h,w}$, such that $d$ is strictly bounded by $R_{h,w}$, and let $c_r$ be the cutter associated with $v_r$. By Definition \ref{def:BalaDrawnDeco}, $|\mathsf{VerIn}(f_1(c_r))|\leq \frac{2}{3}|\mathsf{VerIn}(f)|$ and $|\mathsf{VerIn}(f_2(c_r))|\leq \frac{2}{3}|\mathsf{VerIn}(f)|$. Now, observe that, there are $\frac{1}{100}\cdot n$ cycles, each consists of five vertices, so there are $\frac{5}{100}\cdot n$ vertices in all the square cycles around the full grid square inside. So, $c_r$ must get inside the full grid square, and therefore $c_r$ intersects all of the square cycles. Thus, $\mathsf{width}({\cal T},\beta,\alpha)\geq \frac{1}{100}\cdot n= \Omega (n)$. Therefore, we have the following corollary:

 \begin{corollary}\label{lem:LbR}
There exists a connected graph $G$ and a planar rectilinear grid drawing $d$ of $G$, such that, for every balanced drawn tree decomposition $({\cal T},\beta,\alpha)$ of $d$, $\mathsf{width}({\cal T},\beta,\alpha)\geq \Omega (n)$.
\end{corollary}

Now, we give an intuition for the existence of a balanced tree decomposition for grid drawings, with width bounded by $\OO(\sqrt{n})$. In particular, we show how to find the first cutter.
Let $G$ be a graph, let $h,w\in \mathbb{N}$ and let $d$ be a grid drawing of $G$ strictly bounded by $R_{h,w}$. First, assume that there are less than $\sqrt{n}$ rows (columns). Then, in every column (row) there are at most $\sqrt {n}$ vertices, so we can easily find a row (column) as a cutter, which is balanced (assuming $n$ is large enough). Otherwise, there are more than $\sqrt{n}$ columns, so there is a column with at most $\sqrt {n}$ vertices. We say that such a column is {\em nice}.

\smallskip\noindent{\bf Case 1.}  Assume that for every nice column there are more than $\frac{1}{3}\cdot n$ vertices to its right and more than $\frac{1}{3}\cdot n$ vertices to its left. Then, it is easy to see that every nice column is balanced. 

\smallskip\noindent{\bf Case 2.} Assume, that there is no nice column, which there are at most $\frac{1}{3}\cdot n$ vertices to its left (right), and there is a nice column which there are at most $\frac{1}{3}\cdot n$ vertices to its right (left). Let $c_1$ be the leftmost (rightmost) such a column. 

\smallskip\noindent{\bf Case 2.a} If there are more than $\sqrt{n}$ columns between $c_1$ and the left (right) edge of $R_{h,w}$, then there must be a nice column $c_2$ left (right) to $c_1$. Now, recall that we assumed that there are no nice columns where there are $\frac{1}{3}\cdot n$ vertices to its left (right), and $c_1$ is the leftmost (rightmost) nice column where there at most $\frac{1}{3}\cdot n$ vertices to its right (left). So, we get that there are more than $\frac{1}{3}\cdot n$ vertices to its right and more than $\frac{1}{3}\cdot n$ vertices to its left, so $c_2$ is balanced. 

\smallskip\noindent{\bf Case 2.b} Now, assume that there are less than $\sqrt{n}$ columns between $c_1$ and left (right) edge of $R_{h,w}$. Then, in every row between $c_1$ to left (right) edge of $R_{h,w}$ there are at most $\sqrt{n}$ vertices. Therefore, there must be a row $r$, such that a cutter that starts from $c_1$ and goes left to the left edge of $R_{h,w}$ is balanced. Observe that, this cutter is also with width bounded by $\sqrt{n}$.

\smallskip\noindent{\bf Case 3.} Now, assume that there is a nice column which there are at most $\frac{1}{3}\cdot n$ vertices to its right, and there exists one which there are at most $\frac{1}{3}\cdot n$ vertices to its left. Let $c_1$ be the rightmost nice column which there are at most $\frac{1}{3}\cdot n$ vertices to its left, and let $c_2$ be the leftmost nice column which there are at most $\frac{1}{3}\cdot n$ vertices to its right. Observe that, $c_1$ is left to $c_2$. 

\smallskip\noindent{\bf Case 3.a} Now, if there are more than $\sqrt{n}$ columns between $c_1$ and $c_2$, then there must be a nice column $c$ between them. Therefore, by our assumption regarding the choices of $c_1$ and $c_2$, we get that there are more than $\sqrt{n}$ vertices left to $c$, and more than $\sqrt{n}$ right to $c$. So, $c$ is balanced. 

\smallskip\noindent{\bf Case 3.b} Otherwise, there are at most $\sqrt{n}$ columns between $c_1$ and $c_2$, and this case is similar to Case 2.b.

\section{Conclusion}\label{sec:conclusion}
In this paper, we introduced the concept of \emph{drawn tree decomposition}, a new form of tree decomposition based on the drawing of the graph, and provided a general scheme to apply it to several graph drawing problems. Due to its nature, our work also opens up various research directions, some of which we list below.
\begin{itemize}
	\item Our concept of drawn tree decomposition takes into account only polyline grid drawings of a graph, which are in particular, drawings in which the vertices are placed on grid points. An immediate question is to further extend the definition of drawn tree decomposition so as to more general drawings of the graph, e.g., by removing the restriction of placing vertices to grid points.
	\item It would be interesting to explore whether drawn tree decompositions have some of the nice properties of tree decompositions such as duality theorems, which can be exploited in contexts of Bidimensionality, logic and linear-time recognition for graphs of bounded width.
	\item Our scheme provides running times which are XP with respect to the drawn treewidth of the graph. So, a natural research direction is to either prove that such running times are optimal for some of the problems, or to modify or improve upon our scheme so as to get FPT algorithms with respect to the drawn treewidth of the graph for some of the problems.
	\item We gave an upper bound and a lower bound for drawn treewidth for orthogonal grid drawing and rectilinear grid drawing, respectively. So, another interesting research direction is to study upper bounds and lower bounds for other types of drawings.
\end{itemize}

\newpage

\bibliography{RefsSurvey,RefsSid}
	
\end{document}